\documentclass[UTF-8,reqno]{amsart}
\usepackage[margin=1.2in]{geometry}
\usepackage{enumerate}
\usepackage{enumitem}
\usepackage{amssymb,url,color, booktabs,nccmath}
\usepackage{mathrsfs}
\usepackage{tikz}
\usepackage{mhequ}
\usepackage{mhsymb}
\usepackage{longtable}

\usepackage{color}
\usepackage[colorlinks=true]{hyperref}
\hypersetup{
	linkcolor=blue,          
	citecolor=red,        
	filecolor=blue,      
	urlcolor=cyan
}

\usetikzlibrary{shapes.misc}
\usetikzlibrary{shapes.symbols}
\usetikzlibrary{snakes}
\usetikzlibrary{decorations}
\usetikzlibrary{decorations.markings}

\def\drawx{\draw[-,solid] (-3pt,-3pt) -- (3pt,3pt);\draw[-,solid] (-3pt,3pt) -- (3pt,-3pt);}
\tikzset{
	root/.style={circle, draw=black, fill=white, inner sep=0pt, minimum size=0.7mm},
	dot/.style={circle,fill=black,draw=black, solid,inner sep=0pt,minimum size=0.5mm},
	dot1/.style={circle,fill=red,draw=red, solid,inner sep=0pt,minimum size=0.5mm},
	square/.style={rectangle,fill=red,draw=red, solid,inner sep=0pt,minimum size=0.5mm},
	empty/.style={circle,fill=white,draw=white, solid,inner sep=0pt,minimum size=0.5mm},
	var/.style={circle,fill=black!10,draw=black,inner sep=0pt, minimum size=
		2mm},
	symb/.style={circle,fill=symbols,draw=symbols, solid,inner sep=0pt,minimum size=0.5mm},
	yy/.style={circle,fill=gray!20,draw=black,inner sep=0pt,minimum size=0.8mm},
	>=stealth,
	dotred/.style={circle,fill=black!50,inner sep=0pt, minimum size=2mm},
	generic/.style={semithick,shorten >=1pt,shorten <=1pt},
	dist/.style={ultra thick,draw=testcolor,shorten >=1pt,shorten <=1pt},
	testfcn/.style={ultra thick,testcolor,shorten >=1pt,shorten <=1pt,<-},
	testfcnx/.style={ultra thick,testcolor,shorten >=1pt,shorten <=1pt,<-,
		postaction={decorate,decoration={markings,mark=at position 0.6 with {\drawx}}}},
	kprime/.style={semithick,shorten >=1pt,shorten <=1pt,densely dashed,->},
	kprimex/.style={semithick,shorten >=1pt,shorten <=1pt,densely dashed,->,
		postaction={decorate,decoration={markings,mark=at position 0.4 with {\drawx}}}},
	kernel/.style={semithick,shorten >=1pt,shorten <=1pt,->},
	multx/.style={shorten >=1pt,shorten <=1pt,
		postaction={decorate,decoration={markings,mark=at position 0.5 with {\drawx}}}},
	kernelx/.style={semithick,shorten >=1pt,shorten <=1pt,->,
		postaction={decorate,decoration={markings,mark=at position 0.4 with {\drawx}}}},
	kernel1/.style={->,semithick,shorten >=1pt,shorten <=1pt,postaction={decorate,decoration={markings,mark=at position 0.45 with {\draw[-] (0,-0.1) -- (0,0.1);}}}},
	kernel2/.style={->,semithick,shorten >=1pt,shorten <=1pt,postaction={decorate,decoration={markings,mark=at position 0.45 with {\draw[-] (0.05,-0.1) -- (0.05,0.1);\draw[-] (-0.05,-0.1) -- (-0.05,0.1);}}}},
	kernelBig/.style={semithick,shorten >=1pt,shorten <=1pt,decorate, decoration={zigzag,amplitude=1.5pt,segment length = 3pt,pre length=2pt,post length=2pt}},
	rho/.style={dotted,semithick,shorten >=1pt,shorten <=1pt},
	renorm/.style={shape=circle,fill=white,inner sep=1pt},
	labl/.style={shape=rectangle,fill=white,inner sep=1pt},
	xi/.style={circle,fill=symbols!10,draw=symbols,inner sep=0pt,minimum size=1.2mm},
	xix/.style={crosscircle,fill=symbols!10,draw=symbols,inner sep=0pt,minimum size=1.2mm},
	xib/.style={circle,fill=symbols!10,draw=symbols,inner sep=0pt,minimum size=1.6mm},
	xibx/.style={crosscircle,fill=symbols!10,draw=symbols,inner sep=0pt,minimum size=1.6mm},
	not/.style={circle,fill=symbols,draw=symbols,inner sep=0pt,minimum size=0.5mm},
	midarrow/.style={postaction={decorate,decoration={markings,mark=at position 0.6 with {\arrow{>}}}}},
	midarrow1/.style={postaction={decorate,decoration={markings,mark=at position 0.45 with {\arrow{>}}}}},
	midarrow2/.style={postaction={decorate,decoration={markings,mark=at position 0.75 with {\arrow{>}}}}},
	>=stealth,
}

\colorlet{symbols}{blue!90!black}

\makeatletter
\def\DeclareSymbol#1#2#3{\expandafter\gdef\csname MH@symb@#1\endcsname{\tikz[baseline=#2,scale=0.15]{#3}}%
	\expandafter\gdef\csname MH@symb@#1s\endcsname{\scalebox{0.6}{\tikz[baseline=#2,scale=0.15]{#3}}}}
\def\<#1>{\csname MH@symb@#1\endcsname}
\makeatother

\DeclareSymbol{X}{-2.4}{\node[dot] {};}
\DeclareSymbol{X1}{-2.4}{\node[dot1] {};}
\DeclareSymbol{1}{0}{\draw[white] (-.4,0) -- (.4,0); \draw (0,0)  -- (0,1.2) node[dot] {};}
\DeclareSymbol{1r}{0}{\draw[white] (-.4,0) -- (.4,0); \draw (0,0)  -- (0,1.2) node[dot1] {};}
\DeclareSymbol{1K}{0}{\draw[white] (-.4,0) -- (.4,0); \draw (0,0)  -- (0,1.2)  {};}
\DeclareSymbol{1R}{0}{\draw[white] (-.4,0) -- (.4,0); \draw [red] (0,0)  to (0,1.2)  {};}
\DeclareSymbol{2}{0}{\draw (-0.5,1.2) node[dot] {} -- (0,0.2) -- (0.5,1.2) node[dot] {};}
\DeclareSymbol{2yc}{0}{\draw (-0.5,1.2) node[dot1] {} -- (0,0.2) -- (0.5,1.2) node[dot1] {};}
\DeclareSymbol{2m}{0}{\draw (-0.5,1.2) node[dot] {} -- (0,0.2) -- (0.5,1.2) node[dot1] {};}
\DeclareSymbol{2v}{0}{\draw (-0.5,1.5) node[dot] {} -- (0,0.5) -- (0.5,1.5) node[dot] {}; \draw  [red](0,0.5) to (0,0);}
\DeclareSymbol{2vm}{0}{\draw (-0.5,1.5) node[dot] {} -- (0,0.5) -- (0.5,1.5) node[dot1] {}; \draw  [red](0,0.5) to (0,-0.2);}

\DeclareSymbol{10n}{0}{ \draw  (0,0.9) to (0,0); \node[dot] at (0,0.9) {} ;}

\DeclareSymbol{10cn}{0}{ \draw  (0,0.9) to (0,0); \node[dot1] at (0,0.9) {};}

\DeclareSymbol{20vm}{0}{\draw (-0.5,1.7) node[dot] {} -- (0,0.7) -- (0.5,1.7) node[dot1] {}; \draw  [blue, thick](0,0.7) to (0,0);}

\DeclareSymbol{20vm1}{0}{\draw (-0.5,1.7) node[dot] {} -- (0,0.7) -- (0.5,1.7) node[dot1] {}; \draw  (0,0.7) to (0,0);}

\DeclareSymbol{20vm2}{0}{\draw (-0.5,1.7) node[dot] {} -- (0,0.7) -- (0.5,1.7) node[dot] {}; \draw  (0,0.7) to (0,0);}

\DeclareSymbol{20vm2y}{0}{\draw (-0.5,1.7) node[dot1] {} -- (0,0.7) -- (0.5,1.7) node[dot1] {}; \draw  (0,0.7) to (0,0);}

\DeclareSymbol{22vm}{-3}{\draw [blue,thick] (0,0)  to (0,-0.7) ; \draw [red] (0,-0.7) to (0.6,-0.2); \draw (0.6,-0.2) -- (1.0,0.6) node[dot] {}; \draw (0.6,-0.2) -- (1.4,0) node[dot1] {};
	\draw (-.5,0.9) node[dot1] {} -- (0,0) -- (.5,0.9) node[dot] {};  \node [root] at (0,-0.7) {};}
	
	\DeclareSymbol{22vm-W}{-3}{\draw [green,thick] (0,0)  to (0,-0.7) ; \draw [red] (0,-0.7) to (0.6,-0.2); \draw (0.6,-0.2) -- (1.0,0.6) node[dot] {}; \draw (0.6,-0.2) -- (1.4,0) node[dot1] {};
		\draw (-.5,0.9) node[dot1] {} -- (0,0) -- (.5,0.9) node[dot] {};  \node [root] at (0,-0.7) {};}
		
			\DeclareSymbol{22ccvm-Wy}{-3}{\draw (0,0) -- (0,-0.7) -- (0.7,0) node[dot1] {}; \draw (0,0) [thick]-- (0,-0.7){};\draw (0,-0.7) -- (-0.7,0) node[dot1] {};\draw (-.5,0.8) node[dot] {} -- (0,0) -- (.5,0.8) node[dot] {};  \node [root] at (0,-0.7) {};}
			
				\DeclareSymbol{22ccvm-Wy1}{-3}{\draw (0,0) -- (0,-0.7) -- (0.7,0) node[dot1] {}; \draw (0,0) [thick]-- (0,-0.7){};\draw (0,-0.7) -- (-0.7,0) node[dot1] {};\draw (-.5,0.8) node[dot] {} -- (0,0) -- (.5,0.8) node[dot1] {};  \node [root] at (0,-0.7) {};}
				
					\DeclareSymbol{22ccvm-Wy2}{-3}{\draw (0,0) -- (0,-0.7) -- (0.7,0) node[dot] {}; \draw (0,0) [thick]-- (0,-0.7){};\draw (0,-0.7) -- (-0.7,0) node[dot1] {};\draw (-.5,0.8) node[dot] {} -- (0,0) -- (.5,0.8) node[dot1] {};  \node [root] at (0,-0.7) {};}
			
				\DeclareSymbol{22vm-Wy}{-3}{\draw (0,0) -- (0,-0.7) -- (0.7,0) node[dot1] {}; \draw (0,0) [green, thick]-- (0,-0.7){};\draw (0,-0.7) -- (-0.7,0) node[dot1] {};\draw (-.5,0.8) node[dot] {} -- (0,0) -- (.5,0.8) node[dot1] {};  \node [root] at (0,-0.7) {};}

			\DeclareSymbol{22vm-Wy1}{-3}{\draw (0,0) -- (0,-0.7) -- (0.7,0) node[dot] {}; \draw (0,0) [green, thick]-- (0,-0.7){};\draw (0,-0.7) -- (-0.7,0) node[dot1] {};\draw (-.5,0.8) node[dot] {} -- (0,0) -- (.5,0.8) node[dot1] {};  \node [root] at (0,-0.7) {};}
	
	\DeclareSymbol{22vmy}{-3}{\draw  (0,0)  to (0,-0.7) ; \draw [red] (0,-0.7) to (0.6,-0.2); \draw (0.6,-0.2) -- (1.0,0.6) node[dot] {}; \draw (0.6,-0.2) -- (1.4,0) node[dot1] {};
		\draw (-.5,0.9) node[dot1] {} -- (0,0) -- (.5,0.9) node[dot] {};  \node [root] at (0,-0.7) {};}
		
			\DeclareSymbol{22vmy1}{-3}{\draw  (0,0)  to (0,-0.7) ; \draw [red] (0,-0.7) to (0.6,-0.2); \draw (0.6,-0.2) -- (1.0,0.6) node[dot] {}; \draw (0.6,-0.2) -- (1.4,0) node[dot1] {};
			\draw (-.5,0.9) node[dot] {} -- (0,0) -- (.5,0.9) node[dot] {};  \node [root] at (0,-0.7) {};}

\DeclareSymbol{22cvm}{-3}{\draw (0,0)  to (0,-0.7) ; \draw [red] (0,-0.7) to (0.6,-0.2); \draw (0.6,-0.2) -- (1.0,0.6) node[dot] {}; \draw (0.6,-0.2) -- (1.4,0) node[dot1] {};
	\draw (-.5,0.5)  {} [red]-- (0,0){}; \draw (0,0) -- (.5,0.9) node[dot] {}; \draw (-.5,0.5) to (-.5,1.2) node[dot1]{}; \node [root] at (0,-0.7) {};}

\DeclareSymbol{22ccvm}{-3}{\draw (0,0)  to (0,-0.7) ; \draw [red] (0,-0.7) to (0.6,-0.2); \draw (0.6,-0.2) -- (1.0,0.6) node[dot] {}; \draw (0.6,-0.2) -- (1.4,0) node[dot1] {};
	\draw (-.5,0.5)  {} [red]-- (0,0){}; \draw (0,0) -- (.5,0.9) node[dot] {}; \draw (-.5,0.5) to (-.5,1.2) node[dot]{}; \node [root] at (0,-0.7) {};}
	
	\DeclareSymbol{22ccvm-W}{-3}{\draw (0,0) [green, thick]-- (0,-0.7){};\draw (0,-0.7) [red] -- (0.8,-0.3) {}; \draw (0.8,-0.3) -- (1.2,0.4) node[dot1] {}; \draw (0,-0.7) -- (-0.7,0) node[dot1] {};\draw (-.5,0.8) node[dot1] {} -- (0,0) -- (.5,0.8) node[dot] {};  \node [root] at (0,-0.7) {};}
		
			\DeclareSymbol{22ccvm-W1}{-3}{\draw (0,0) [green, thick]-- (0,-0.7){};\draw (0,-0.7) [red] -- (0.8,-0.3) {}; \draw (0.8,-0.3) -- (1.2,0.4) node[dot] {}; \draw (0,-0.7) -- (-0.7,0) node[dot1] {};\draw (-.5,0.8) node[dot1] {} -- (0,0) -- (.5,0.8) node[dot] {};  \node [root] at (0,-0.7) {};}
			
				\DeclareSymbol{22ccvm-W2}{-3}{\draw (0,0) [thick]-- (0,-0.7){};\draw (0,-0.7) [red] -- (0.8,-0.3) {}; \draw (0.8,-0.3) -- (1.2,0.4) node[dot1] {}; \draw (0,-0.7) -- (-0.7,0) node[dot1] {};\draw (-.5,1.2) node[dot] {}-- (-.5,0.4){}; \draw (0,0) -- (.5,0.8) node[dot] {}; \draw (-.5,0.4) [red]-- (0,0) {};\node [root] at (0,-0.7) {};}
				
					\DeclareSymbol{22ccvm-W3}{-3}{\draw (0,0) [thick]-- (0,-0.7){};\draw (0,-0.7) [red] -- (0.8,-0.3) {}; \draw (0.8,-0.3) -- (1.2,0.4) node[dot1] {}; \draw (0,-0.7) -- (-0.7,0) node[dot1] {};\draw (-.5,1.2) node[dot1] {}-- (-.5,0.4){}; \draw (0,0) -- (.5,0.8) node[dot] {}; \draw (-.5,0.4) [red]-- (0,0) {};\node [root] at (0,-0.7) {};}
					
						\DeclareSymbol{22ccvm-W4}{-3}{\draw (0,0) [thick]-- (0,-0.7){};\draw (0,-0.7) [red] -- (0.8,-0.3) {}; \draw (0.8,-0.3) -- (1.2,0.4) node[dot] {}; \draw (0,-0.7) -- (-0.7,0) node[dot1] {};\draw (-.5,1.2) node[dot1] {}-- (-.5,0.4){}; \draw (0,0) -- (.5,0.8) node[dot] {}; \draw (-.5,0.4) [red]-- (0,0) {};\node [root] at (0,-0.7) {};}

\DeclareSymbol{2c}{0}{\draw (-0.5,1.5) node[dot1] {} -- (-0.5,0.5);\draw (-0.5,0.5) [red] (-0.5,0.5) to (0,0); \draw (0,0) -- (0.5,1.0) node[dot] {};}

\DeclareSymbol{2v}{0}{\draw (-0.5,1.5) node[dot] {} -- (-0.5,0.5);\draw (-0.5,0.5) [red] (-0.5,0.5) to (0,0); \draw (0,0) -- (0.5,1.0) node[dot] {};}

\DeclareSymbol{2vc}{0}{\draw (-0.5,1.5) node[dot1] {} -- (-0.5,0.5);\draw (-0.5,0.5) [red] (-0.5,0.5) to (0,0); \draw (0,0) -- (0.5,1.0) node[dot1] {};}

\DeclareSymbol{2c0}{0}{\draw (-0.5,1.8) node[dot1] {} -- (-0.5,1.0);\draw (-0.5,1.0) [red] (-0.5,1.0) to (0,0.5); \draw (0,0.5) -- (0.5,1.2) node[dot] {}; \draw (0,-0.2) to (0,0.5);}

\DeclareSymbol{2v0}{0}{\draw (-0.5,1.8) node[dot] {} -- (-0.5,1.0);\draw (-0.5,1.0) [red] (-0.5,1.0) to (0,0.5); \draw (0,0.5) -- (0.5,1.2) node[dot] {}; \draw (0,-0.2) to (0,0.5);}

\DeclareSymbol{2vc0}{0}{\draw (-0.5,1.8) node[dot1] {} -- (-0.5,1.0);\draw (-0.5,1.0) [red] (-0.5,1.0) to (0,0.5); \draw (0,0.5) -- (0.5,1.2) node[dot1] {}; \draw (0,-0.2) to (0,0.5);}

\DeclareSymbol{2v2v}{-3}{\draw (0,0) [blue, thick]-- (0,-0.7){};\draw (0,-0.7) [red] -- (0.8,-0.3) {}; \draw (0.8,-0.3) -- (1.2,0.4) node[dot1] {}; \draw (0,-0.7) -- (-0.7,0) node[dot] {};\draw (-.5,0.8) node[dot1] {} -- (0,0) -- (.5,0.8) node[dot] {};  \node [root] at (0,-0.7) {};}

\DeclareSymbol{2v2vr}{-3}{\draw (0,0) [blue, thick]-- (0,-0.7){};\draw (0,-0.7) [red] -- (0.8,-0.3) {}; \draw (0.8,-0.3) -- (1.2,0.4) node[dot] {}; \draw (0,-0.7) -- (-0.7,0) node[dot] {};\draw (-.5,0.8) node[dot1] {} -- (0,0) -- (.5,0.8) node[dot] {};  \node [root] at (0,-0.7) {};}

\DeclareSymbol{2cv2v}{-3}{\draw (0,0) -- (0,-0.7){};\draw (0,-0.7) [red] -- (0.8,-0.3) {}; \draw (0.8,-0.3) -- (1.2,0.4) node[dot1] {}; \draw (0,-0.7) -- (-0.7,0) node[dot] {};\draw (-.5,1.2) node[dot1] {}-- (-.5,0.4){}; \draw (0,0) -- (.5,0.8) node[dot] {}; \draw (-.5,0.4) [red]-- (0,0) {};\node [root] at (0,-0.7) {};}

\DeclareSymbol{2cvr2v}{-3}{\draw (0,0) -- (0,-0.7){};\draw (0,-0.7) [red] -- (0.8,-0.3) {}; \draw (0.8,-0.3) -- (1.2,0.4) node[dot1] {}; \draw (0,-0.7) -- (-0.7,0) node[dot] {};\draw (-.5,1.2) node[dot] {}-- (-.5,0.4){}; \draw (0,0) -- (.5,0.8) node[dot] {}; \draw (-.5,0.4) [red]-- (0,0) {};\node [root] at (0,-0.7) {};}

\DeclareSymbol{2cv2v0}{-3}{\draw (0,0) -- (0,-0.7){};\draw (0,-0.7) [red] -- (0.8,-0.3) {}; \draw (0.8,-0.3) -- (1.2,0.4) node[dot] {}; \draw (0,-0.7) -- (-0.7,0) node[dot] {};\draw (-.5,1.2) node[dot1] {}-- (-.5,0.4){}; \draw (0,0) -- (.5,0.8) node[dot1] {}; \draw (-.5,0.4) [red]-- (0,0) {};\node [root] at (0,-0.7) {};}

\DeclareSymbol{2cv2vr}{-3}{\draw (0,0) -- (0,-0.7){};\draw (0,-0.7) [red] -- (0.8,-0.3) {}; \draw (0.8,-0.3) -- (1.2,0.4) node[dot] {}; \draw (0,-0.7) -- (-0.7,0) node[dot] {};\draw (-.5,1.2) node[dot] {}-- (-.5,0.4){}; \draw (0,0) -- (.5,0.8) node[dot1] {}; \draw (-.5,0.4) [red]-- (0,0) {};\node [root] at (0,-0.7) {};}

\DeclareSymbol{3v}{0}{\draw (-0.9,1.5) node[dot] {} -- (-0.5,0.5) --(-0.2,1.5) node[dot] {} ;\draw (-0.5,0.5) [red] (-0.5,0.5) to (0,0); \draw (0,0) -- (0.5,1.0) node[dot] {};}

\DeclareSymbol{3vm}{0}{\draw (-0.9,1.5) node[dot] {} -- (-0.5,0.5) --(-0.2,1.5) node[dot1] {} ;\draw (-0.5,0.5) [red] (-0.5,0.5) to (0,0); \draw (0,0) -- (0.5,1.0) node[dot] {};}

\DeclareSymbol{3v0}{0}{\draw (-0.9,2.0) node[dot] {} -- (-0.5,1.0) --(-0.2,2.0) node[dot] {} ;\draw (-0.5,1.0) [red] (-0.5,1.0) to (0,0.5); \draw (0,0.5) -- (0.5,1.5) node[dot] {};
	\draw (0,0.5) to (0,0);}

\DeclareSymbol{3v0m}{0}{\draw (-0.9,2.0) node[dot] {} -- (-0.5,1.0) --(-0.2,2.0) node[dot1] {} ;\draw (-0.5,1.0) [red] (-0.5,1.0) to (0,0.5); \draw (0,0.5) -- (0.5,1.5) node[dot] {};
	\draw (0,0.5) to (0,-0.2);}

\DeclareSymbol{3v0m1}{-2}{\draw (-0.9,2.0) node[dot] {} -- (-0.5,1.0) --(-0.2,2.0) node[dot1] {} ;\draw (-0.5,1.0) [red] (-0.5,1.0) to (0,0.5); \draw (0,0.5) -- (0.5,1.5) node[dot] {};	\draw (0,0.5) [blue, thick] to (0,-0.4); \draw (0,-0.4) to (0.7,0.4) node[dot]{}; }

\DeclareSymbol{3v0m1ci}{-2}{\draw (-0.9,2.0) node[dot] {} -- (-0.5,1.0) --(-0.2,2.0) node[dot1] {} ;\draw (-0.5,1.0) [red] (-0.5,1.0) to (0,0.5); \draw (0,0.5) -- (0.5,1.5) node[dot] {};	\draw (0,0.5) [blue, thick] to (0,-0.4); \draw (0,-0.4) to (0.7,0.4) node[dot]{};  \node[root] at (0,-0.4){};  }

\DeclareSymbol{3v0m1ci-Wn}{-2}{\draw (-0.9,2.0) node[dot] {} -- (-0.5,1.0) --(-0.2,2.0) node[dot1] {} ;\draw (-0.5,1.0) [red] (-0.5,1.0) to (0,0.5); \draw (0,0.5) -- (0.5,1.5) node[dot] {};	\draw (0,0.5) [green, thick] to (0,-0.4); \draw (0,-0.4) to (0.7,0.4) node[dot]{};  \node[root] at (0,-0.4){};  }

\DeclareSymbol{3v0m1ci-W}{-2}{\draw (-0.9,2.0) node[dot] {} -- (-0.5,1.0) --(-0.2,2.0) node[dot1] {} ;\draw (-0.5,1.0) [red] (-0.5,1.0) to (0,0.5); \draw (0,0.5) -- (0.5,1.5) node[dot] {};	\draw (0,0.5) [thick] to (0,-0.4); \draw (0,-0.4) to (0.7,0.4) node[dot]{};  \node[root] at (0,-0.4){};  }

\DeclareSymbol{3v0m1ci2}{-2}{\draw (-0.9,2.0) node[dot] {} -- (-0.5,1.0) --(-0.2,2.0) node[dot1] {} ;\draw (-0.5,1.0) [red] (-0.5,1.0) to (0,0.5); \draw (0,0.5) -- (0.5,1.5) node[dot] {};	\draw (0,0.5)  to (0,-0.4); \draw (0,-0.4) to (0.7,0.4) node[dot1]{};  \node[root] at (0,-0.4){};  }

\DeclareSymbol{3v0m1ci2-W}{-2}{\draw (-0.9,2.0) node[dot] {} -- (-0.5,1.0) --(-0.2,2.0) node[dot1] {} ;\draw (-0.5,1.0) [red] (-0.5,1.0) to (0,0.5); \draw (0,0.5) -- (0.5,1.5) node[dot] {};	\draw [thick] (0,0.5)  to (0,-0.4); \draw (0,-0.4) to (0.7,0.4) node[dot1]{};  \node[root] at (0,-0.4){};  }

\DeclareSymbol{3v0m1ci3}{-2}{\draw (-0.9,2.0) node[dot] {} -- (-0.5,1.0) --(-0.2,2.0) node[dot1] {} ;\draw (-0.5,1.0) [red] (-0.5,1.0) to (0,0.5); \draw (0,0.5) -- (0.5,1.5) node[dot] {};	\draw (0,0.5)  to (0,-0.4); \draw (0,-0.4) to (0.7,0.4) node[dot1]{};    }

\DeclareSymbol{3v0m1ci1}{-2}{\draw (-0.9,2.0) node[dot] {} -- (-0.5,1.0) --(-0.2,2.0) node[dot1] {} ;\draw (-0.5,1.0) [red] (-0.5,1.0) to (0,0.5); \draw (0,0.5) -- (0.5,1.5) node[dot] {};	\draw [thick] (0,0.5)  to (0,-0.4){};
	\draw (-1.1,1.0) node[dot1] {} -- (-0.7,0.2); \draw (-0.7,0.2) [red] (-0.7,0.2) to (0,-0.4); \draw (0,-0.4) to (0.7,0.4) node[dot]{};\node[root] at (0,-0.4){}; }

\DeclareSymbol{3v0m1ci10}{-2}{\draw (-0.9,2.0) node[dot] {} -- (-0.5,1.0) --(-0.2,2.0) node[dot1] {} ;\draw (-0.5,1.0) [red] (-0.5,1.0) to (0,0.5); \draw (0,0.5) -- (0.5,1.5) node[dot] {};	\draw (0,0.5)  to (0,-0.4) {};
	\draw (-1.1,1.0) node[dot1] {} -- (-0.7,0.2); \draw (-0.7,0.2) [red] (-0.7,0.2) to (0,-0.4); \node[root] at (0,-0.4){};}

\DeclareSymbol{3v0m1ci1o}{-2}{\draw (-0.9,2.0) node[dot] {} -- (-0.5,1.0) --(-0.2,2.0) node[dot1] {} ;\draw (-0.5,1.0) [red] (-0.5,1.0) to (0,0.5); \draw (0,0.5) -- (0.5,1.5) node[dot] {};	\draw (0,0.5)  to (0,-0.4) {};
	\draw (-1.1,1.0) node[dot1] {} -- (-0.7,0.2); \draw (-0.7,0.2) [red] (-0.7,0.2) to (0,-0.2); \draw (0,-0.4) to (0.7,0.4) node[dot]{}; }

\DeclareSymbol{3v0m1cci1}{-2}{\draw (-0.9,2.0) node[dot] {} -- (-0.5,1.0) --(-0.2,2.0) node[dot1] {} ;\draw (-0.5,1.0) [red] (-0.5,1.0) to (0,0.5); \draw (0,0.5) -- (0.5,1.5) node[dot1] {};	\draw [thick] (0,0.5)  to (0,-0.4) {};
	\draw (-1.1,1.0) node[dot] {} -- (-0.7,0.2); \draw (-0.7,0.2) [red] (-0.7,0.2) to (0,-0.4); \draw (0,-0.4) to (0.7,0.4) node[dot]{};  \node[root] at (0,-0.4){};}

\DeclareSymbol{3v0m1cci10}{-2}{\draw (-0.9,2.0) node[dot] {} -- (-0.5,1.0) --(-0.2,2.0) node[dot1] {} ;\draw (-0.5,1.0) [red] (-0.5,1.0) to (0,0.5); \draw (0,0.5) -- (0.5,1.5) node[dot] {};	\draw (0,0.5)  to (0,-0.4){};
	\draw (-1.1,1.0) node[dot] {} -- (-0.7,0.2); \draw (-0.7,0.2) [red] (-0.7,0.2) to (0,-0.4); \node[root] at (0,-0.4){};  }

\DeclareSymbol{3v0m1cci1o}{-2}{\draw (-0.9,2.0) node[dot] {} -- (-0.5,1.0) --(-0.2,2.0) node[dot1] {} ;\draw (-0.5,1.0) [red] (-0.5,1.0) to (0,0.5); \draw (0,0.5) -- (0.5,1.5) node[dot1] {};	\draw (0,0.5)  to (0,-0.2) node{};
	\draw (-1.1,1.0) node[dot] {} -- (-0.7,0.2); \draw (-0.7,0.2) [red] (-0.7,0.2) to (0,-0.4); \draw (0,-0.4) to (0.7,0.4) node[dot]{}; }

\DeclareSymbol{3v0m1c}{-2}{\draw (-0.9,2.0) node[dot] {} -- (-0.5,1.0) --(-0.2,2.0) node[dot1] {} ;\draw (-0.5,1.0) [red] (-0.5,1.0) to (0,0.5); \draw (0,0.5) -- (0.5,1.5) node[dot1] {};	\draw (0,0.5) [blue, thick] to (0,-0.4); \draw (0,-0.4) to (0.7,0.4) node[dot]{}; }

\DeclareSymbol{3v0m1cci}{-2}{\draw (-0.9,2.0) node[dot] {} -- (-0.5,1.0) --(-0.2,2.0) node[dot1] {} ;\draw (-0.5,1.0) [red] (-0.5,1.0) to (0,0.5); \draw (0,0.5) -- (0.5,1.5) node[dot1] {};	\draw (0,0.5) [blue, thick] to (0,-0.4) {}; \draw (0,-0.4) to (0.7,0.4) node[dot]{};\node[root] at (0,-0.4){};  }

\DeclareSymbol{3v0m1cci-Wn}{-2}{\draw (-0.9,2.0) node[dot] {} -- (-0.5,1.0) --(-0.2,2.0) node[dot1] {} ;\draw (-0.5,1.0) [red] (-0.5,1.0) to (0,0.5); \draw (0,0.5) -- (0.5,1.5) node[dot1] {};	\draw (0,0.5) [green, thick] to (0,-0.4) {}; \draw (0,-0.4) to (0.7,0.4) node[dot]{};\node[root] at (0,-0.4){};  }

\DeclareSymbol{3v0m1cci-W}{-2}{\draw (-0.9,2.0) node[dot] {} -- (-0.5,1.0) --(-0.2,2.0) node[dot1] {} ;\draw (-0.5,1.0) [red] (-0.5,1.0) to (0,0.5); \draw (0,0.5) -- (0.5,1.5) node[dot1] {};	\draw (0,0.5) [thick] to (0,-0.4) {}; \draw (0,-0.4) to (0.7,0.4) node[dot]{};\node[root] at (0,-0.4){};  }

\DeclareSymbol{3v0m1cci2}{-2}{\draw (-0.9,2.0) node[dot] {} -- (-0.5,1.0) --(-0.2,2.0) node[dot1] {} ;\draw (-0.5,1.0) [red] (-0.5,1.0) to (0,0.5); \draw (0,0.5) -- (0.5,1.5) node[dot1] {};	\draw (0,0.5)  to (0,-0.4) {}; \draw (0,-0.4) to (0.7,0.4) node[dot]{};\node[root] at (0,-0.4){};  }

\DeclareSymbol{3v0m1cci2-W}{-2}{\draw (-0.9,2.0) node[dot] {} -- (-0.5,1.0) --(-0.2,2.0) node[dot1] {} ;\draw (-0.5,1.0) [red] (-0.5,1.0) to (0,0.5); \draw (0,0.5) -- (0.5,1.5) node[dot1] {};	\draw [thick] (0,0.5)  to (0,-0.4) {}; \draw (0,-0.4) to (0.7,0.4) node[dot]{};\node[root] at (0,-0.4){};  }

\DeclareSymbol{3v0m1cci2y}{-2}{\draw (-0.9,2.0) node[dot] {} -- (-0.5,1.0) --(-0.2,2.0) node[dot1] {} ;\draw (-0.5,1.0) [red] (-0.5,1.0) to (0,0.5); \draw (0,0.5) -- (0.5,1.5) node[dot] {};	\draw (0,0.5)  to (0,-0.4) {}; \draw (0,-0.4) to (0.7,0.4) node[dot]{};\node[root] at (0,-0.4){};  }

\DeclareSymbol{3v0m1cci3}{-2}{\draw (-0.9,2.0) node[dot] {} -- (-0.5,1.0) --(-0.2,2.0) node[dot1] {} ;\draw (-0.5,1.0) [red] (-0.5,1.0) to (0,0.5); \draw (0,0.5) -- (0.5,1.5) node[dot1] {};	\draw (0,0.5)  to (0,-0.4) {}; \draw (0,-0.4) to (0.7,0.4) node[dot]{};  }

\DeclareSymbol{3v0m2}{-2}{\draw (-0.9,2.0) node[dot] {} -- (-0.5,1.0) --(-0.2,2.0) node[dot1] {} ;\draw (-0.5,1.0) [red] (-0.5,1.0) to (0,0.5); \draw (0,0.5) -- (0.5,1.5) node[dot] {};
	\draw (0,0.5) to (0,-0.2); \draw (0,-0.2) to (0.6,0.5) node[dot1]{}; \draw (0,-0.2) [red] to (0,-0.9){}; \draw (0,-0.9) to (0.7, -0.1) node[dot]{};}

\DeclareSymbol{3v0m2o}{-2}{\draw (-0.9,2.0) node[dot] {} -- (-0.5,1.0) --(-0.2,2.0) node[dot1] {} ;\draw (-0.5,1.0) [red] (-0.5,1.0) to (0,0.5); \draw (0,0.5) -- (0.5,1.5) node[dot] {};
	\draw (0,0.5) to (0,-0.2); \draw (0,-0.2) to (0.6,0.5) node[dot1]{}; \draw (0,-0.2) [red] to (0,-0.9){}; \draw (0,-0.9) to (0.7, -0.1) node[dot]{};\node[root] at (0,-0.9){}; }

\DeclareSymbol{3v0m2c}{-2}{\draw (-0.9,2.0) node[dot] {} -- (-0.5,1.0) --(-0.2,2.0) node[dot1] {} ;\draw (-0.5,1.0) [red] (-0.5,1.0) to (0,0.5); \draw (0,0.5) -- (0.5,1.5) node[dot1] {};
	\draw (0,0.5) to (0,-0.2); \draw (0,-0.2) to (0.6,0.5) node[dot]{}; \draw (0,-0.2) [red] to (0,-0.9){}; \draw (0,-0.9) to (0.7, -0.1) node[dot]{};}

\DeclareSymbol{3v0m2co}{-2}{\draw (-0.9,2.0) node[dot] {} -- (-0.5,1.0) --(-0.2,2.0) node[dot1] {} ;\draw (-0.5,1.0) [red] (-0.5,1.0) to (0,0.5); \draw (0,0.5) -- (0.5,1.5) node[dot1] {};
	\draw (0,0.5) to (0,-0.2); \draw (0,-0.2) to (0.6,0.5) node[dot]{}; \draw (0,-0.2) [red] to (0,-0.9){}; \draw (0,-0.9) to (0.7, -0.1) node[dot]{};\node[root] at (0,-0.9){}; }

\DeclareSymbol{3v02vm}{-3}{\draw  (0,0)  to (0,-0.9) ; \draw [red] (0,-0.9) to (0.6,-0.2); \draw (0.6,-0.2) -- (1.0,0.6) node[dot] {}; \draw (0.6,-0.2) -- (1.4,0) node[dot1] {};
	\draw (-0.9,1.5) node[dot] {} -- (-0.5,0.5) --(-0.2,1.5) node[dot1] {} ;\draw (-0.5,0.5) [red] (-0.5,0.5) to (0,0); \draw (0,0) -- (0.5,1.0) node[dot] {};
	\draw (0,0) to (0,-0.9){};
	\node [root] at (0,-0.9) {};}
	
	\DeclareSymbol{3v02vm-W1}{-3}{\draw  (0,0)  to (0,-0.9) ; \draw [red] (0,-0.9) to (0.6,-0.2); \draw (0.6,-0.2) -- (1.0,0.6) node[dot] {}; \draw (0.6,-0.2) -- (1.4,0) node[dot1] {};
		\draw (-0.9,1.5) node[dot] {} -- (-0.5,0.5) --(-0.2,1.5) node[dot1] {} ;\draw (-0.5,0.5) [red] (-0.5,0.5) to (0,0); \draw (0,0) -- (0.5,1.0) node[dot] {};
		\draw [thick](0,0) to (0,-0.9){};
		\node [root] at (0,-0.9) {};}
	
	\DeclareSymbol{3v02vm-W}{-3}{\draw  (0,0)  to (0,-0.9) ; \draw [red] (0,-0.9) to (0.6,-0.2); \draw (0.6,-0.2) -- (1.0,0.6) node[dot] {}; \draw (0.6,-0.2) -- (1.4,0) node[dot1] {};
		\draw (-0.9,1.5) node[dot] {} -- (-0.5,0.5) --(-0.2,1.5) node[dot1] {} ;\draw (-0.5,0.5) [red] (-0.5,0.5) to (0,0); \draw (0,0) -- (0.5,1.0) node[dot1] {};
		\draw [thick] (0,0) to (0,-0.9){};
		\node [root] at (0,-0.9) {};}
	
	\DeclareSymbol{3v02vmyn}{-3}{\draw  (0,0)  to (0,-0.9) ; \draw [red] (0,-0.9) to (0.6,-0.2); \draw (0.6,-0.2) -- (1.0,0.6) node[dot] {}; \draw (0.6,-0.2) -- (1.4,0) node[dot1] {};
		\draw (-0.9,1.5) node[dot] {} -- (-0.5,0.5) --(-0.2,1.5) node[dot1] {} ;\draw (-0.5,0.5) [red] (-0.5,0.5) to (0,0); \draw (0,0) -- (0.5,1.0) node[dot] {};
		\draw (0,0) to (0,-0.9){};
	}

\DeclareSymbol{3v02vmp}{-3}{\draw [blue, thick] (0,0)  to (0,-0.9) ; \draw [red] (0,-0.9) to (0.6,-0.2); \draw (0.6,-0.2) -- (1.0,0.6) node[dot] {}; \draw (0.6,-0.2) -- (1.4,0) node[dot1] {};
	\draw (-0.9,1.5) node[dot] {} -- (-0.5,0.5) --(-0.2,1.5) node[dot1] {} ;\draw (-0.5,0.5) [red] (-0.5,0.5) to (0,0); \draw (0,0) -- (0.5,1.0) node[dot] {};
	\draw (0,0) to (0,-0.9){};
}

\DeclareSymbol{3}{0}{\draw (0,0) -- (0,1.2) node[dot] {}; \draw (-.7,1) node[dot] {} -- (0,0) -- (.7,1) node[dot] {};}

\DeclareSymbol{3m}{0}{\draw (0,0) -- (0,1.2) node[dot] {}; \draw (-.7,1) node[dot1] {} -- (0,0) -- (.7,1) node[dot] {};}

\DeclareSymbol{3mc}{0}{\draw (0,0) -- (0,1.2) node[dot1] {}; \draw (-.7,1) node[dot1] {} -- (0,0) -- (.7,1) node[dot] {};}

\DeclareSymbol{31}{-3}{\draw (0,0) -- (0,-0.7) node[root] {} -- (0.9,0.1) node[dot] {}; \draw (0,0) -- (0,1.2) node[dot] {}; \draw (-.7,1) node[dot] {} -- (0,0) -- (.7,1) node[dot] {};}

\DeclareSymbol{30}{-3}{\draw (0,0) -- (0,-0.7); \draw (0,0) -- (0,1.2) node[dot] {}; \draw (-.7,1) node[dot] {} -- (0,0) -- (.7,1) node[dot] {};}

\DeclareSymbol{30m}{-3}{\draw (0,0) -- (0,-0.7); \draw (0,0) -- (0,1.2) node[dot] {}; \draw (-.7,1) node[dot1] {} -- (0,0) -- (.7,1) node[dot] {};}

\DeclareSymbol{30mc}{-3}{\draw (0,0) -- (0,-0.7); \draw (0,0) -- (0,1.2) node[dot1] {}; \draw (-.7,1) node[dot1] {} -- (0,0) -- (.7,1) node[dot] {};}

\DeclareSymbol{22m}{-3}{\draw (0,0) -- (0,-0.7) -- (0.7,0) node[dot] {}; \draw (0,0) -- (0,-0.7) -- (-0.7,0) node[dot] {};\draw (-.5,0.8) node[dot] {} -- (0,0) -- (.5,0.8) node[dot] {};  \node [root] at (0,-0.7) {};}

\DeclareSymbol{22oc}{-3}{\draw (0,0) -- (0,-0.7) -- (0.7,0) node[dot] {}; \draw (0,0) -- (0,-0.7) -- (-0.7,0) node[dot1] {};\draw (-.5,0.8) node[dot] {} -- (0,0) -- (.5,0.8) node[dot1] {};  \node [root] at (0,-0.7) {};}

\DeclareSymbol{22oc-W}{-3}{\draw  (0,-0.7) -- (0.7,0) node[dot] {}; \draw  (0,-0.7) -- (-0.7,0) node[dot1] {};\draw (-.5,0.8) node[dot] {} -- (0,0) -- (.5,0.8) node[dot1] {};\draw [red] (0,0) to (0,-0.7); }

\DeclareSymbol{22oc-W1}{-3}{\draw (0,0) -- (0,-0.7) -- (0.7,0) node[dot] {}; \draw (0,0) -- (0,-0.7) -- (-0.7,0) node[dot1] {};\draw (-.5,0.8) node[dot] {} -- (0,0) -- (.5,0.8) node[dot1] {};  \node [root] at (0,-0.7) {};}

\DeclareSymbol{22ocy}{-3}{\draw [blue,thick] (0,0)  to (0,-0.7) ;\draw  (0,-0.7) -- (0.7,0) node[dot] {}; \draw  (0,-0.7) -- (-0.7,0) node[dot1] {};\draw (-.5,0.8) node[dot] {} -- (0,0) -- (.5,0.8) node[dot1] {};  \node [root] at (0,-0.7) {};}

\DeclareSymbol{22ocy1}{-3}{\draw [blue,thick] (0,0)  to (0,-0.7) ;\draw (0,-0.7) -- (0.7,0) node[dot] {}; \draw  (0,-0.7) -- (-0.7,0) node[dot] {};\draw (-.5,0.8) node[dot] {} -- (0,0) -- (.5,0.8) node[dot1] {};  \node [root] at (0,-0.7) {};}

\DeclareSymbol{22ocn}{-3}{\draw (0,0) -- (0,-0.7) -- (0.7,0) node[dot] {}; \draw (0,0) -- (0,-0.7) -- (-0.7,0) node[dot1] {};\draw (-.5,0.8) node[dot] {} -- (0,0) -- (.5,0.8) node[dot1] {};  \draw (0,0) [blue,thick] to (0,-0.7){}; }

\DeclareSymbol{22ocn1}{-3}{\draw (0,0) [blue, thick]-- (0,-0.7){};\draw (0,-0.7) [red] -- (-0.8,-0.3) {}; \draw (-0.8,-0.3) -- (-1.2,0.4) node[dot1] {}; \draw (0,-0.7) -- (0.7,0) node[dot] {};\draw (-.5,0.8) node[dot] {} -- (0,0) -- (.5,0.8) node[dot1] {};  }

\DeclareSymbol{22ocn2}{-3}{\draw (0,0) -- (0,-0.7) -- (0.7,0) node[dot] {}; \draw (0,0) -- (0,-0.7) -- (-0.7,0) node[dot1] {};\draw (-.5,0.8) node[dot] {} -- (0,0) -- (.5,0.8) node[dot1] {};  \draw (0,0) [blue,thick] to (0,-0.7){}; \node[square] at (0,-0.7){}; }

\DeclareSymbol{22oc1}{-3}{\draw (0,0) -- (0,-0.7) -- (0.7,0) node[dot1] {}; \draw (0,0) -- (0,-0.7) -- (-0.7,0) node[dot] {};\draw (-.5,0.8) node[dot] {} -- (0,0) -- (.5,0.8) node[dot] {};  \node [root] at (0,-0.7) {};}

\DeclareSymbol{22oc2}{-3}{\draw (0,0) -- (0,-0.7) -- (0.7,0) node[dot] {}; \draw (0,0) -- (0,-0.7) -- (-0.7,0) node[dot] {};\draw (-.5,0.8) node[dot1] {} -- (0,0) -- (.5,0.8) node[dot1] {};  \node [root] at (0,-0.7) {};}

\DeclareSymbol{22oc3}{-3}{\draw (0,0) -- (0,-0.7) -- (0.7,0) node[dot] {}; \draw (0,0) -- (0,-0.7) -- (-0.7,0) node[dot] {};\draw (-.5,0.8) node[dot] {} -- (0,0) -- (.5,0.8) node[dot1] {};  \node [root] at (0,-0.7) {};}

\DeclareSymbol{31oc}{-3}{\draw (0,0) -- (0,-0.7) -- (0.7,0) node[dot] {};  \draw (0,0) -- (0,1) node[dot] {}; \draw (-.7,0.8) node[dot] {} -- (0,0) -- (.7,0.8) node[dot1] {}; \node [root] at (0,-0.7) {};}

\DeclareSymbol{31oc1}{-3}{\draw (0,0) -- (0,-0.7) -- (0.7,0) node[dot] {};  \draw (0,0) -- (0,1) node[dot] {}; \draw (-.7,0.8) node[dot] {} -- (0,0) -- (.7,0.8) node[dot1] {}; }

\DeclareSymbol{31oc2}{-3}{\draw (0,0) -- (0,-0.7) -- (0.7,0) node[dot] {};  \draw (0,0) -- (0,1) node[dot] {}; \draw (-.7,0.8) node[dot1] {} -- (0,0) -- (.7,0.8) node[dot1] {}; }

\DeclareSymbol{31oc3}{-3}{\draw (0,0) -- (0,-0.7) -- (0.7,0) node[dot] {};  \draw (0,0) -- (0,1) node[dot] {}; \draw (-.7,0.8) node[dot1] {} -- (0,0) -- (.7,0.8) node[dot1] {}; \node [root] at (0,-0.7) {};}

\DeclareSymbol{31oc4}{-3}{\draw (0,0) -- (0,-0.7) -- (0.7,0) node[dot1] {};  \draw (0,0) -- (0,1) node[dot] {}; \draw (-.7,0.8) node[dot] {} -- (0,0) -- (.7,0.8) node[dot1] {};
	\node[root] at (0,-0.7){}; }

\DeclareSymbol{31oc5}{-3}{\draw (0,0) -- (0,-0.7) -- (0.7,0) node[dot1] {};  \draw (0,0) -- (0,1) node[dot] {}; \draw (-.7,0.8) node[dot] {} -- (0,0) -- (.7,0.8) node[dot1] {}; }

\DeclareSymbol{32oc}{-3}{\draw (0,0) -- (0,-0.7) -- (0.7,0) node[dot] {}; \draw  (0,-0.7)  -- (-0.7,0) node[dot1] {}; \draw (0,0) -- (0,1) node[dot] {}; \draw (-.7,0.8) node[dot] {} -- (0,0) -- (.7,0.8) node[dot1] {}; \node [root] at (0,-0.7) {};}

\DeclareSymbol{32oc1}{-3}{\draw (0,0) -- (0,-0.7) -- (0.7,0) node[dot] {}; \draw  (0,-0.7)  -- (-0.7,0) node[dot1] {}; \draw (0,0) -- (0,1) node[dot] {}; \draw (-.7,0.8) node[dot] {} -- (0,0) -- (.7,0.8) node[dot1] {}; }

\DeclareSymbol{32oc2}{-3}{\draw (0,0) -- (0,-0.7) -- (0.7,0) node[dot] {}; \draw  (0,-0.7)  -- (-0.7,0) node[dot] {}; \draw (0,0) -- (0,1) node[dot] {}; \draw (-.7,0.8) node[dot1] {} -- (0,0) -- (.7,0.8) node[dot1] {}; }

\DeclareSymbol{32oc3}{-3}{\draw (0,0) -- (0,-0.7) -- (0.7,0) node[dot] {}; \draw  (0,-0.7)  -- (-0.7,0) node[dot] {}; \draw (0,0) -- (0,1) node[dot] {}; \draw (-.7,0.8) node[dot1] {} -- (0,0) -- (.7,0.8) node[dot1] {}; \node [root] at (0,-0.7) {};}

\DeclareSymbol{10}{-3}{\draw (0,0) -- (0,-1); \draw (-.8,1) node[dot] {} -- (0,0);}
\DeclareSymbol{32}{-3}{\draw (0,0) -- (0,-0.7) -- (0.7,0) node[dot] {}; \draw  (0,-0.7)  -- (-0.7,0) node[dot] {}; \draw (0,0) -- (0,1) node[dot] {}; \draw (-.7,0.8) node[dot] {} -- (0,0) -- (.7,0.8) node[dot] {}; \node [root] at (0,-0.7) {};}
\DeclareSymbol{12}{-3}{\draw (0,0) -- (0,-1) -- (1,0) node[dot] {}; \draw (0,0) -- (0,-1) -- (-1,0) node[dot] {}; \draw (-.8,1) node[dot] {} -- (0,0);}

\DeclareSymbol{22}{-3}{\draw (0,0) -- (0,-0.7) -- (0.7,0) node[dot] {}; \draw (0,0) -- (0,-0.7) -- (-0.7,0) node[dot] {};\draw (-.5,0.8) node[dot] {} -- (0,0) -- (.5,0.8) node[dot] {};  \node [root] at (0,-0.7) {};}

\DeclareSymbol{20}{-3}{\draw (0,0) -- (0,-0.8);\draw (-.6,0.8) node[dot] {} -- (0,0) -- (.6,0.8) node[dot] {};}
\DeclareSymbol{32W}{-3}{\draw  [densely dotted, semithick] (-1.5,0)node[dot] {}  -- (0,-1.5) node[yy]{} -- (1.5,0) node[dot] {}; \draw (.8,-2.3)node{\tiny $y$} ;\draw[semithick] (0,0) -- (0,-1.5) node[yy] {}; \draw [densely dotted, semithick] (0,0) -- (0,1.8) node[dot] {}; \draw[densely dotted, semithick] (-1,1.5) node[dot] {} -- (0,0) -- (1,1.5) node[dot] {};}
\DeclareSymbol{32WW}{-3}{\draw  [densely dotted, semithick] (-1.5,0)node[dot] {}  -- (0,-1.5) node[yy]{} -- (1.5,0) node[dot] {}; \draw (.8,-2.3)node{\tiny $\overline{y}$} ;\draw[semithick] (0,0) -- (0,-1.5) node[yy] {}; \draw [densely dotted, semithick] (0,0) -- (0,1.8) node[dot] {}; \draw[densely dotted, semithick] (-1,1.5) node[dot] {} -- (0,0) -- (1,1.5) node[dot] {};}

\DeclareSymbol{s1}{0}{\draw[white] (-.4,0) -- (.4,0); \draw[symbols] (0,0)  -- (0,1.2) node[symb] {};}
\DeclareSymbol{s2}{0}{\draw[symbols] (-0.5,1.2) node[symb] {} -- (0,0) -- (0.5,1.2) node[symb] {};}
\DeclareSymbol{s3}{0}{\draw[symbols] (0,0) -- (0,1.2) node[symb] {}; \draw[symbols] (-.7,1) node[symb] {} -- (0,0) -- (.7,1) node[symb] {};}
\DeclareSymbol{s31}{-3}{\draw[symbols] (0,0) -- (0,-1) -- (1,0) node[symb] {}; \draw[symbols] (0,0) -- (0,1.2) node[symb] {}; \draw[symbols] (-.7,1) node[symb] {} -- (0,0) -- (.7,1) node[symb] {};}
\DeclareSymbol{s30}{-3}{\draw[symbols] (0,0) -- (0,-1); \draw[symbols] (0,0) -- (0,1.2) node[symb] {}; \draw[symbols] (-.7,1) node[symb] {} -- (0,0) -- (.7,1) node[symb] {};}
\DeclareSymbol{s32}{-3}{\draw[symbols] (0,0) -- (0,-1) -- (1,0) node[symb] {}; \draw[symbols] (0,0) -- (0,-1) -- (-1,0) node[symb] {}; \draw[symbols] (0,0) -- (0,1.2) node[symb] {}; \draw[symbols] (-.7,1) node[symb] {} -- (0,0) -- (.7,1) node[symb] {};}
\DeclareSymbol{s22}{-3}{\draw[symbols] (0,0.3) -- (0,-1) -- (1,0) node[symb] {}; \draw[symbols] (0,0.3) -- (0,-1) -- (-1,0) node[symb] {};\draw[symbols] (-.7,1) node[symb] {} -- (0,0.3) -- (.7,1) node[symb] {};}
\DeclareSymbol{s20}{-3}{\draw[symbols] (0,0) -- (0,-1);\draw[symbols] (-.7,1) node[symb] {} -- (0,0) -- (.7,1) node[symb] {};}
\DeclareSymbol{s21}{-3}{\draw[symbols] (0,0) -- (0,-1);\draw[symbols] (-.7,1) node[symb] {} -- (0,0) -- (.7,1) node[symb] {}; \draw[symbols] (0,0) -- (0,-1) -- (1,0) node[symb] {};}
\DeclareSymbol{s12}{-3}{\draw[symbols] (0,0) -- (0,-1) -- (1,0) node[symb] {}; \draw[symbols] (0,0) -- (0,-1) -- (-1,0) node[symb] {}; \draw (-.8,1) node[symb] {} -- (0,0);}
\DeclareSymbol{s11}{-3}{\draw[symbols] (0,0) -- (0,-1) -- (-1,0) node[symb] {}; \draw (-.8,1) node[symb] {} -- (0,0);}
\DeclareSymbol{s10}{-3}{\draw[symbols] (0,0) -- (0,-1); \draw[symbols] (-.8,1) node[symb] {} -- (0,0);}

\definecolor{darkergreen}{rgb}{0.0, 0.5, 0.0}

\numberwithin{equation}{section}
\def\theequation{\arabic{section}.\arabic{equation}}
\newcommand{\be}{\begin{eqnarray}}
	\newcommand{\ee}{\end{eqnarray}}
\newcommand{\ce}{\begin{eqnarray*}}
	\newcommand{\de}{\end{eqnarray*}}
\newtheorem{theorem}{Theorem}[section]
\newtheorem{lemma}[theorem]{Lemma}
\newtheorem{remark}[theorem]{Remark}
\newtheorem{definition}[theorem]{Definition}
\newtheorem{proposition}[theorem]{Proposition}
\newtheorem{Examples}[theorem]{Example}
\newtheorem{corollary}[theorem]{Corollary}

\newcommand{\LL}{\mathscr{L}}

\def\Wick#1{\,\colon\! #1 \colon}

\def\Re{{\mathrm{Re}}}

\def\eps{\varepsilon}

\def\p{\partial}
\def\d{\dif}

\def\la{\langle}
\def\ra{\rangle}
\def\[{{\Big[}}
\def\]{{\Big]}}
\def\({{\Big(}}
\def\){{\Big)}}

\def\bx{{\mathbf{x}}}
\def\tr{\mathrm {tr}}

\def\dif{{\mathord{{\rm d}}}}

\def\min{{\mathord{{\rm min}}}}

\def\no{\nonumber}
\def\={&\!\!=\!\!&}
\newcommand{\gF}{\mathfrak{F}}

\newcommand{\norm}[1]{ \left| \! \left| #1 \right| \! \right| }
\def\bB{{\mathbf B}}
\def\bC{{\mathbf C}}

\def\cD{{\mathcal D}}
\def\cE{{\mathcal E}}
\def\cF{{\mathcal F}}
\def\cG{{\mathcal G}}
\def\cH{{\mathcal H}}

\def\cL{{\mathcal L}}
\def\cM{{\mathcal M}}
\def\cN{{\mathcal N}}

\def\cQ{{\mathcal Q}}
\def\cR{{\mathcal R}}
\def\cS{{\mathcal S}}

\def\cU{{\mathcal U}}
\def\cV{{\mathcal V}}
\def\cW{{\mathcal W}}

\def\cY{{\mathcal Y}}
\def\cZ{{\mathcal Z}}

\def\mB{{\mathbb B}}
\def\mC{{\mathbb C}}

\def\mH{{\mathbb H}}

\def\mM{{\mathbb M}}
\def\mN{{\mathbb N}}

\def\mR{{\mathbb R}}

\def\mT{{\mathbb T}}

\def\mW{{\mathbb W}}

\def\mY{{\mathbb Y}}
\def\mZ{{\mathbb Z}}

\def\bB{{\mathbf B}}
\def\bP{{\mathbf P}}

\def\1{{\mathbf{1}}}

\def\sA{{\mathscr A}}

\def\sD{{\mathscr D}}

\def\sH{{\mathscr H}}
\def\sI{{\mathscr I}}
\def\sJ{{\mathscr J}}

\def\sR{{\mathscr R}}

\def\sV{{\mathscr V}}

\def\sZ{{\mathscr Z}}
\def\E{\mathbf E}

\newcommand{\bT}{\mathbb{T}}
\newcommand{\bH}{\mathbb{H}}
\newcommand{\bZ}{\mathbb{Z}}

\newcommand{\bA}{\mathbb{A}}
\newcommand{\bW}{\mathbb{W}}
\newcommand{\gH}{\mathfrak{H}}
\newcommand{\ii}{\infty}
\renewcommand{\ge}{\geqslant}
\renewcommand{\le}{\leqslant}
\newcommand\dGamma{{\rm d}\Gamma}

\def\geq{\geqslant}
\def\leq{\leqslant}
\def\ge{\geqslant}
\def\le{\leqslant}

\def\C{\mathbb{C}}

\def\iint{\int\!\!\!\int}
\def\Re{{\mathrm{Re}}}

\def\eps{\varepsilon}

\def\p{\partial}
\def\d{\dif}

\def\la{\langle}
\def\ra{\rangle}
\def\[{{\Big[}}
\def\]{{\Big]}}
\def\({{\Big(}}
\def\){{\Big)}}

\def\bx{{\mathbf{x}}}
\def\tr{\mathrm {Tr}}

\def\dif{{\mathord{{\rm d}}}}

\def\min{{\mathord{{\rm min}}}}

\def\no{\nonumber}
\def\={&\!\!=\!\!&}
\def\bt{\begin{theorem}}
	\def\et{\end{theorem}}
\def\bl{\begin{lemma}}
	\def\el{\end{lemma}}
\def\br{\begin{remark}}
	\def\er{\end{remark}}
\def\bx{\begin{Examples}}
	\def\ex{\end{Examples}}
\def\bd{\begin{definition}}
	\def\ed{\end{definition}}
\def\bp{\begin{proposition}}
	\def\ep{\end{proposition}}
\def\bc{\begin{corollary}}
	\def\ec{\end{corollary}}
	\newcommand{\dG}{\mathrm{d}\Gamma}
	\newcommand{\gesim}{\gtrsim}
	\newcommand{\lesim}{\lesssim}
	\newcommand{\ssim}{\backsimeq}
	
	\newcommand{\cZl}{\cZ_\lambda}

	\newcommand{\Gammat}{\widetilde{\Gamma}}

\newcommand{\cHcl}{\cH_{\rm cl}}

\def\geq{\geqslant}
\def\leq{\leqslant}
\def\ge{\geqslant}
\def\le{\leqslant}

\def\iint{\int\!\!\!\int}

\def\bH{{\mathbb H}}
\def\bW{{\mathbb W}}
\def\bP{{\mathbf P}}

\def\bT{{\mathbb T}} \def\R{\mathbb R}\def\nn{\nonumber}
 \def\R{\mathbb R}    
\def\N{\mathbb N}

\allowdisplaybreaks
\DeclareMathOperator{\Tr}{{\rm Tr}}

\begin{document}

\title{$\Phi^4_3$ theory from many-body quantum Gibbs states}


\author{Phan Th\`anh Nam}
\address[P. T. Nam]{LMU Munich, Department of Mathematics, Theresienstrasse 39, 80333 Munich, Germany}
\email{nam@math.lmu.de}

\author{Rongchan Zhu}
\address[R. Zhu]{Department of Mathematics, Beijing Institute of Technology, Beijing 100081, China}
\email{zhurongchan@126.com}

\author{Xiangchan Zhu}
\address[X. Zhu]{ Academy of Mathematics and Systems Science,
	Chinese Academy of Sciences, Beijing 100190, China}
\email{zhuxiangchan@126.com}
\date{\today}

\begin{abstract} We derive the $\Phi^4_3$ measure on the torus as a rigorous limit of the quantum Gibbs state of an interacting Bose gas. To be precise, starting from many-body quantum mechanics, where the problem is linear and regular but involving non commutative operators, we justify the emergence of the $\Phi^4_3$ measure as a semiclassical limit which captures the formation of Bose--Einstein condensation just above the critical temperature. 	We employ and develop several tools from both stochastic quantization and many-body quantum mechanics. Since the quantum problem is typically formulated using a nonlocal interaction potential, our first key step involves approximating the $\Phi^4_3$ measure through a Hartree measure with nonlocal interaction, achieved by developing new techniques in paracontrolled calculus. The connection between the quantum problem and the Hartree measure emerges through a  variational interplay between classical and quantum models.

\end{abstract}

\maketitle

\setcounter{tocdepth}{1}
\tableofcontents

\section{Introduction}\label{sec:defs}

The (complex) $\Phi^4_3$ measure on a domain $D\subset \R^3$  is a probability measure over complex-valued distributions $\Phi:D\to \mathbb{C}$ which is formally defined by
\begin{align}\label{eq:Phi-measure1-intro}
	\d\nu(\Phi) = \frac{1}{\sZ}\exp\bigg(-\int_{D} \Big( |\nabla \Phi(x)|^2+m_0 |\Phi(x)|^2 + \frac12 |\Phi(x)|^4 \Big) \d x \bigg) \cD \Phi
\end{align}
with a given parameter $m_0\in \R$, where $\cD \Phi=\prod_{x\in D}\dif \Phi(x)$ is the formal Lebesgue measure on the space of fields. This  formula is purely formal since all the relevant terms, namely the kinetic energy $\int_{D} |\nabla \Phi(x)|^2 \d x $, the mass  $\int_{D} |\Phi(x)|^2 \d x$ and the interaction energy $\int_{D}|\Phi(x)|^4 \d x$, are {\em infinite} almost surely in the support of the measure. Therefore, $\nu$ must be defined via an appropriate renormalization procedure.

Historically,  the measure in \eqref{eq:Phi-measure1-intro} is a typical example of a class of nonlinear Gibbs measures which were first defined in the 1960s and 1970s  in the context of \emph{Constructive Quantum Field Theory} ~\cite{Symanzik-66,Nelson-73,Simon-74,GliJafSpe-74,GueRosSim-75,GliJaf-87}. In this direction, the Gibbs measures serve as a tool to construct interacting quantum fields. One of the main achievements of the constructive quantum field theory is the mathematically rigorous construction of the $\Phi^4_3$ field in the Euclidean framework.

Afterwards, starting from the work of Lebowitz, Rose, and Speer \cite{LebRosSpe-88}, and a series of papers by Bourgain \cite{Bou,Bourgain-96,Bourgain-97}, the same Gibbs measures have been used to study
{\em nonlinear Schr\"odinger equations} (NLS). Heuristically, since $\nu$ in  \eqref{eq:Phi-measure1-intro} is formally the invariant measure of the cubic NSL,
it is helpful to establish the well-posedness with low regular data;  see  \cite{Tzvetkov-08,ThoTzv-10,BurThoTzv-13,BouBul-14a,BouBul-14b,CacSuz-14,OhTho-18, DNY19, DNY22} for more recent developments.

In another direction, in 1981, Parisi and Wu \cite{PW81} introduced a framework for Euclidean quantum field theory that seeks to obtain Gibbs states of classical statistical mechanics as limiting distributions of stochastic processes, particularly through solutions to nonlinear stochastic differential equations. These stochastic partial differential equations (SPDEs) can then be employed to study the properties of Gibbs states, a procedure known as  {\em stochastic quantization}. We refer to \cite{JLM85, AR91, DD03, Hai14,Kup16,RocZhuZhu-16,TsaWeb-18,MW17, MW18, CC15, GH18, GH18a} for some further works in this approach. In particular, for the $\Phi^4_d$ measure, the corresponding stochastic quantization equation is referred to as the dynamical $\Phi^4_d$ model.

Initiated by \cite{LebRosSpe-88}, the idea of using Gibbs measures is also widely applied in {quantum statistical mechanics}; see, e.g.,  \cite{ZinnJustin-89,Cardy-96,ZinnJustin-13,BauBrySla-19}. In particular, the $\Phi^4_3$ measure is closely related to the description of the \emph{Bose--Einstein phase transition} \cite{ArnMoo-01,BayBlaiHolLalVau-99,BayBlaiHolLalVau-01,HolBay-03,KasProSvi-01}.

The aim of the present paper is to provide a rigorous derivation of the $\Phi^4_3$ measure by making a link between two areas: stochastic partial differential equations and the physics of Bose gases. Starting from many-body quantum mechanics, where the problem is  linear and regular but involving non commutative operators, we will justify the emergence of the $\Phi^4_3$ as a semiclassical limit which captures the formation of the Bose--Einstein condensation just above the critical temperature, thus resolving a natural question  raised from a series of works  \cite{LewNamRou-15,FroKnoSchSoh-17,LewNamRou-18,FroKnoSchSoh-19,LewNamRou-21,FKSS22,FKSS23}. By developing techniques from stochastic quantization and many-body quantum mechanics, and in particular leveraging an interplay between classical and quantum models, we derive uniform estimates that are crucial to pass from the quantum setting involving a non-local interaction to the model \eqref{eq:Phi-measure1-intro} with a local interaction.

More precisely, we will study an interacting Bose gas on the torus $\bT^3$ in the grand canonical ensemble, which is described by the quantum Gibbs state $\Gamma_\lambda=\cZ_\lambda^{-1}e^{-\bH_\lambda}$ on the bosonic Fock space $\gF=\gF(L^2(\bT^3))$ with
the Hamiltonian
\begin{align*}
	\bH_\lambda = 0 \oplus (\lambda  (-\Delta -\vartheta) )  \bigoplus_{n=2}^\infty  \left( \lambda \sum_{j=1}^n (-\Delta_{x_j} -\vartheta) + \lambda^2 \sum_{1\le j<k\le n} v^\eps(x_j-x_k) \right).
\end{align*}
We can further write it via the second quantization formalism as follows:
\begin{align}\label{eq:many body hamil-x-intro}
	\bH_\lambda = \lambda \int_{\bT^3}  a_x^* (-\Delta_x - \vartheta) a_x  \d x+ \frac{\lambda^2}{2} \int_{\bT^3\times \bT^3}  v^\eps(x-y) a_x^* a_y^*  a_x a_y \, \d x \d y,
\end{align}
 where  the creation and annihilation operators $a_x^*,a_x$ satisfy the canonical commutation relations (CCR)
\begin{equation}\label{eq:CCR-x}
	[a_x, a_y] = 0 = [a^*_x,a^*_y], \quad [a_x,a^*_y] = \delta_0(x-y),\quad  \forall x,y\in \bT^3.
\end{equation}
We will take $\lambda\to 0^+$ as the inverse temperature and choose $\eps=\eps(\lambda) \to 0^+$ so that $v^\eps$ converges to the Dirac delta function $\delta_0$. Heuristically,  the reader may also think of $\lambda$ as a semiclassical parameter, and interpret $\sqrt{\lambda} a_x$ and $\sqrt{\lambda} a_x^*$ as  the second quantized version of $\Phi(x)$ and $\overline{\Phi(x)}$ in the classical field theory. In this way, $\bH_\lambda$ formally gives rise to the energy functional in a $\Phi^4_3$ measure.

The important parameter $\vartheta \in \mathbb{R}$, called the {\em chemical potential}, is used to adjust the number of particles (or equivalently, the density) of the system.  We will choose the chemical potential such that the number of particles in the zero-momentum mode is comparable to that in every nonzero momentum mode, which is appropriate to derive the classical field theory and essentially places the Gibbs state just slightly above the critical point of the Bose--Einstein phase transition. In particular, one of our main tasks is to adjust $\vartheta$ in order to correctly capture the renormalization required to define the classical $\Phi^4_3$ measure.

We are interested in the following question: {\em Is it possible to derive the $\Phi^4_3$ theory as a limiting description of the quantum Gibbs state $\Gamma_\lambda$?}  In statistical physics, the  macroscopic properties of a system are typically characterized by its volume, temperature, and number of particles. Here we have already chosen the first two parameters, so it is desirable to show that the nonlinear measure in \eqref{eq:Phi-measure1-intro} can be obtained  under an appropriate choice of the  chemical potential $\vartheta$. This type of question, in the more general context of nonlinear Gibbs measures,  was proposed in \cite{LewNamRou-15} and further studied in \cite{FroKnoSchSoh-17,LewNamRou-18,FroKnoSchSoh-19,LewNamRou-21,FKSS22,FKSS23}. While these works cover many important cases, including $\Phi^4_1$ theory \cite{LewNamRou-15,FroKnoSchSoh-17,LewNamRou-18,FroKnoSchSoh-19}, Hartree-type measures with  non-local interactions in 2D and 3D \cite{LewNamRou-21,FKSS22}, and $\Phi^4_2$ theory in  \cite{FKSS23}, the derivation of the $\Phi^4_3$ theory remains open. Our goal is to resolve this issue. As already argued in \cite[Section 1.3]{FKSS22}, we believe that this result reveals a profound duality between Euclidean field theory and interacting Bose gases, generating reciprocal conjectures: established results for Bose gases inspire new insights into $\Phi^4_3$ theories, while known properties of $\Phi^4_3$ theories suggest previously unrecognized features of Bose gases. In particular, the presence of a phase transition in $\Phi^4_3$ field theory (see \cite{FSS76}) strongly suggests that Bose-Einstein condensation arises in three-dimensional, translation-invariant Bose gases with repulsive two-body interactions, and we hope that our approach will stimulate further contributions to the mathematical theory of Bose-Einstein phase transitions.

On the mathematical side, our derivation of the $\Phi^4_3$ theory requires several tools from both stochastic quantization and many-body quantum mechanics. Since the quantum problem is typically formulated using a nonlocal interaction potential, our first key step involves approximating the $\Phi^4_3$ measure through a Hartree measure with nonlocal interaction, achieved by developing techniques in paracontrolled calculus. The connection between the quantum problem and the Hartree measure emerges through a deep variational interplay between classical and quantum models. Below we provide further details about these key ingredients.

The main challenge in deriving the $\Phi^4_3$ theory, compared to previous works \cite{LewNamRou-15,FroKnoSchSoh-17,LewNamRou-18,FroKnoSchSoh-19,LewNamRou-21,FKSS22,FKSS23}, is that Wick renormalization alone is insufficient. In particular, determining the correct counterterm is a subtle task. As a first step, we will study a simplified version of this problem at the classical field theory level. We will show that the $\Phi^4_3$ measure can be approximated by the Hartree measure
\begin{align}\label{eq:Phi_eps-measure1-intro}
	\d\nu^\eps(\Psi) = \frac{1}{\sZ_\eps}\exp\bigg(&-\int_{\mathbb T^3} (|\nabla \Psi(x)|^2+ m_\eps |\Psi(x)|^2) \d x
	-\frac12\int :\!{|\Psi(x)|^2}\!: v^\eps(x-y) :\! {|\Psi(y)|^2}\!: \d x \d y \bigg) \cD \Psi
\end{align}
in the limit $\eps \to 0$, under a suitable choice of $m_\eps \to \infty$. Here, $v^\eps(x)=\eps^{-3}v(\eps^{-1}x)$ is a nonlocal version of the delta function $\delta_0$, and $:\!{|\Psi(x)|^2}\!:$ is formally defined as $|\Psi(x)|^2 - \langle |\Psi(x)|^2 \rangle_{\mu_0}$, the Wick renormalization of $|\Psi(x)|^2$ with the (complex) Gaussian free field $\mu_0$ with covariance $(-\Delta+1)^{-1}$.

For every fixed $\eps>0$, the Hartree measure \eqref{eq:Phi_eps-measure1-intro} is absolutely continuous with respect to $\mu_0$, and its construction is standard. For example, this measure was constructed by Bourgain \cite{Bourgain-97} as the invariant measure for the Hartree (also called the Gross-Pitaevskii) equation
$$
\imath\partial_t u = -(\Delta-m_\eps) u + (v^\eps* |u|^2)u.
$$
Recently, the Hartree measure has been derived from the Gibbs state of quantum Bose gases through two independent methods by Lewin, Nam and Rougerie \cite{LewNamRou-21}, and Fr\"ohlich, Knowles, Schlein, and Sohinger \cite{FKSS22}.

On the other hand, in the present paper we have to deal with the limit $\eps\to 0$ of the Hartree measure. Since the $\Phi^4_3$ measure is {\em singular} with respect to the Gaussian free field $\mu_0$, the limit of the Hartree measure \eqref{eq:Phi_eps-measure1-intro} is well-defined only under a precise choice of $m_\eps$. The first contribution of $m_\eps$ is $a^\eps = \int_{\bT^3} v^\eps(y) G(y)  \d y \backsimeq \eps^{-1}$, where $G(x)$ is the Green's function of $-\Delta + 1$, which behaves as $(4\pi |x|)^{-1}$ as $|x| \to 0$ (see, e.g., \cite[Lemma 5.4]{RouSer-16}). The counterterm $a_\eps$ still arises from Wick renormalization, and it must be corrected by an additional counterterm $-6b_\eps \backsimeq \log \eps$, which will be given in detail later. Combining these, we find that the corresponding Hartree measure
\begin{equs}[e:Phi_eps-measure1-intro]
	\dif\nu^\eps(\Psi)= \frac{1}{\sZ_\eps}\exp\bigg(&-\int_{\mathbb T^3} (|\nabla \Psi|^2+m |\Psi|^2) \,\dif x
	-\frac12\int \Wick{|\Psi(x)|^2}v^\eps(x-y)\Wick{|\Psi(y)|^2}\,\dif x\dif y\no
	\\&+{(a^\eps- 6b^\eps)\int_{\mathbb T^3} \Wick{|\Psi(x)|^2}\,\dif x}\bigg)\mathcal D \Psi
\end{equs}
converges to a $\Phi^4_3$ measure in a suitable sense.

At this point, an interesting aspect is the appearance of a further correction of order $1$ to the mass coefficient $m$ in the limiting $\Phi^4_3$ measure. In fact, this correction would disappear if we take the alternative choice
\begin{equs}[e:Phi_eps-measure-intro]
	\dif\widetilde\nu^\eps(\Psi) =\frac{1}{\widetilde\sZ_\eps}\exp\bigg(&-\int_{\mathbb T^3} (|\nabla \Psi|^2+m |\Psi|^2) \,\dif x
	-\frac12\int \Wick{|\Psi(x)|^2}v^\eps(x-y)\Wick{|\Psi(y)|^2}\,\dif x\dif y\no
	\\&+{\int\Re(\overline{\Psi} (x) v^\eps(x-y) G(x-y) \Psi(y))\,\dif y\dif x-6b^\eps\int_{\mathbb T^3} \Wick{|\Psi(x)|^2}\dif x}\bigg)\mathcal D \Psi,
\end{equs}
with the nonlocal counterterm $\int\Re(\overline{\Psi} (x) v^\eps(x-y) G(x-y) \Psi(y))\,\dif y\dif x$. However, only \eqref{e:Phi_eps-measure1-intro} allows us to make a rigorous connection to the many-body quantum problem (since we can only alter the chemical potential).

Mathematically, defining the $\Phi^4_3$ measure from the Hartree measure $\nu^\eps$ requires that the measure is tight as $\eps\to 0$, whose proof turns out to be  quite demanding. Our first main contribution is to
prove this uniform bound and related estimates for $\nu^\eps$  by using the stochastic quantization approach. More precisely, we will employ the fact that the Hartree measure $\nu^\eps$ in  \eqref{eq:Phi_eps-measure1-intro} is the invariant measure of the equation on $\mT^3$
\begin{equs}[eq:mainnew-intro]
	(\p_t - \Delta + 1)  \Psi^\eps= -( v^\eps*\Wick{|\Psi^\eps|^2})\Psi^\eps+{(a^\eps-6b^\eps+1-m) \Psi^\eps}+\xi,
\end{equs}
where $\xi$ denotes complex-valued space-time white noise on a probability space $(\Omega, \mathcal{F}, \mathbf{P})$.

The construction of $\Phi^4_3$ theory by SPDE method has been a subject of many works in the last decades. In two spatial dimensions, the dynamical $\Phi_2^4$ model  was previously analyzed in \cite{JLM85, AR91, DD03, MW17}. However, the more irregular three-dimensional case ($\Phi_3^4$) remained unsolved for much longer, as it required fundamentally new ideas. A breakthrough was achieved with Hairer's theory of regularity structures \cite{Hai14}, which for the first time gave meaning to the dynamical $\Phi^4_3$ model. Now, local well-posedness for the dynamical $\Phi^4_3$ model can also be established using other methods such as paracontrolled calculus which was proposed by Gubinelli, Imkeller, and Perkowski \cite{GIP15} and applied to $\Phi^4_3$ model by Catellier and Chouk \cite{CC15}, renormalization group methods \cite{Kup16, Duc21, DGR24}, or the diagram-free approach \cite{LOT21, LOTT21, OSSW18, OW19}. These theories enable the treatment of a wide class of singular subcritical SPDEs (see \cite{BHZ19, CH16, BCCH21}), including the Kardar--Parisi--Zhang (KPZ) equation, the generalized parabolic Anderson model, and the stochastic quantization equations for quantum fields such as Yang--Mills model, Sine-Gordon model (see \cite{Hai14, CC15, GP17, CCHS20, CCHS22, CHS18, GM24, She21, BC23, BC24, BC24a} and references therein).

For the dynamical $\Phi^4_3$ model, global well-posedness has been established due to the strong damping term $-\Phi^3$, as demonstrated in \cite{MW18, AK17, GH18, MW20, JP21}. Additionally, a new PDE-based construction of the $\Phi^4_3$ field was developed in \cite{GH18a}. Recent progress has also been made in the construction of subcritical $\Phi^4$ fields \cite{DGR24}, the Abelian-Higgs model \cite{BC24}, the Sine-Gordon model \cite{CFW24, BC24a}, and the large $N$ limit of the $O(N)$ linear sigma model \cite{SSZZ20, SZZ21, SZZ23}, all through stochastic quantization.

We also mention the variational approach developed by Barashkov and Gubinelli in \cite{BG18}. Based on this technique the Hartree-type classical field for general potential $v$ has been constructed in \cite{Bri, OOT24}. Since the variational approach is very helpful in the quantum problem \cite{LewNamRou-15,LewNamRou-18,LewNamRou-21}, we also apply and develop this approach to give uniform bounds on the partition functions of the classical model, which also helps to derive sharper bounds for the partition functions of the quantum Gibbs state and improve the constraints on the parameters $\lambda$ and $\eps$.

Now we go back to the quantum model \eqref{eq:many body hamil-x-intro}. Given the insights from the classical field theory, we will choose the chemical potential in $\bH_\lambda$ as
\begin{equs}[def:gamma-intro]
	\frac{\zeta\left(\frac{3}2\right)}{(4\pi)^{3/2}} \lambda^{-1/2}+ a^\eps - 6b^\eps
\end{equs}
with $\zeta$ being the Riemann zeta function, plus a correction of order 1 to adjust the shift of the mass term in the limiting $\Phi^4_3$ measure. Here $(4\pi)^{-3/2} \zeta\left(\frac{3}2\right)=(2\pi)^{-3} \int_{\R^3} (e^{|p|^2}-1)^{-1}$ is the critical density of the ideal gas, which goes back to the computation of  Einstein \cite{Einstein-24}.

Using the variational approach in \cite{LewNamRou-15,LewNamRou-18,LewNamRou-21}, we  compare the correlation functions of $\Gamma_\lambda=\cZ_{\lambda}^{-1}e^{-\bH_\lambda}$ with those of the Hartree measure in \eqref{e:Phi_eps-measure1-intro}. In this aspect, we will use another characterization of the Hartree measure, namely it is the unique minimizer for the variational problem
\begin{equation} \label{eq:zr-rel-intro}
	-\log \sZ_\eps = \min_{\substack{\nu \text{~proba. meas.}\\ \nu ' \ll\mu_0}} \left\{ \cH_{\rm cl}(\nu', \mu_0) + \int \cD[u]\,\d\nu'(u)\right\}
\end{equation}
where
$$ \cH_{\rm cl} (\nu',\mu_0)\eqdef \int \frac{\d\nu'}{\d\mu_0}(u)\log\left(\frac{\d\nu'}{\d \mu_0 }(u)\right)\,\d\mu_0(u) \ge 0 $$
is the classical relative entropy between the probability measure $\nu'$ and the Gaussian free field $\mu_0$, and
\begin{equation} \label{eq:cDu1}
	\cD[u] = \frac12\int_{\bT^3\times \bT^3} :\!{|u(x)|^2}\!: v^\eps(x-y) :\! {|u(y)|^2}\!: \d x \d y  - {(a^\eps-6 b^\eps-m+1)\int_{\mathbb T^3} :\!{|u(x)|^2}\!: \d x}
\end{equation}
is the renormalized interaction associated with \eqref{e:Phi_eps-measure1-intro}.  This classical variational problem can be related to its quantum analogue in finite dimensions by using semiclassical techniques. Heuristically, the effect of the interaction should be only visible in low momenta. In high momenta, the particles move too fast and they do not really interact with the others, and hence the behavior of the interacting system is simply comparable to the non-interacting one. The main mathematical challenge in this approach is justifying the finite dimension reduction, namely removing the contribution from high momenta.

For the interacting Gibbs state $\Gamma_\lambda$, controlling the interaction contribution from high momenta requires second-order moment estimates on the Gibbs state $\Gamma_\lambda$. This task cannot be achieved by a standard perturbation method (since we cannot deal with a perturbation that is not bounded by the original Hamiltonian). The key idea from  \cite{LewNamRou-21} is that second-order moment estimates on the Gibbs state $\Gamma_\lambda$ can be obtained via first-order moment estimates on a family of perturbed Gibbs states, and the latter can be handled by a Hellmann--Feynman argument.  Rigorously, this was captured by an abstract correlation inequality in  \cite[Theorem~7.1]{LewNamRou-21}, where second-order moments of $\Gamma_\lambda$ are estimated via Duhamel two-point functions with perturbations. In the present paper, we will use two new correlation inequalities from the recent work of Deuchert, Nam, and Napirkowski \cite{DeuNamNap-25},  one of which is a considerably improved version of the previous one in \cite{LewNamRou-21} and is built on a sharp estimate of Duhamel two-point functions via Stahl's theorem \cite{Sta13} (also known as the Bessis--Moussa--Villani conjecture).

Note that if we assume $|\log \lambda|^{-\eta}\le  \eps \to0 $ for a constant  $0<\eta<1/2$, a condition similar to the one used in the derivation of the $\Phi^4_2$ measure in \cite{FKSS23}, then the analysis simplifies greatly. In this case, the techniques in \cite{LewNamRou-21} and the new tools in  \cite{DeuNamNap-25} allow us to derive straightforwardly the convergence of the free energy, and obtain the convergence of the correlation functions in the strong Hilbert--Schmidt topology:
\begin{align}
	\lim_{\lambda\to 0}\Tr \Big(\Big| n! \lambda^n \Gamma_\lambda^{(n)} - \int |\Psi^{\otimes n}\rangle \langle \Psi^{\otimes n}| \d\nu^\eps(\Psi)\Big|^2 \Big) =0, \quad \forall n\ge 1
\end{align}
provided that $\eps=\eps(\lambda)\to 0$ is chosen appropriately. Putting together with the convergence from the Hartree measure to the $\Phi^4_3$ theory, we thus conclude the proof.

In the present paper, we will consider the more general case where $ \lambda^{\eta}\le  \eps$ for a sufficiently small constant $\eta>0$. This condition is much more physically relevant, but requires substantial improvements over the approach in \cite{LewNamRou-21}. In particular, we will need a bootstrap argument: first, we combine the analysis of the quantum problem in the simpler case where $|\log \lambda|^{-\eta}\le  \eps \to0$ with the variational approach from \cite{BG18} to deduce quantitative approximations for the classical measures. This provides an interesting application of quantum methods to obtain information about classical problems. Then we consider the general case $ \lambda^{\eta}\le  \eps$ and derive the convergence of quantum correlation functions via purely the analysis of de Finetti measures, since the uniform Hilbert-Schmidt estimates for the correlation functions from \cite{LewNamRou-21} are no longer available. As a consequence, our analysis also applies to two-dimensional cases, thereby extending the derivation of the $\Phi^4_2$ measure in \cite{FKSS23} to a significantly larger class of potentials.

The precise mathematical setting and statements will be given in the next section.

\bigskip
\noindent
{\bf Acknowledgments.}  P.T.N. would like to thank Nguyen Viet Dang, Mathieu Lewin, and Nicolas Rougerie for various helpful discussions, and thank Quoc Hung Nguyen for his warm hospitality during a visit to Chinese Academy of Sciences in 2024, when the project was initiated. P.T.N. was  supported by the European Research Council via the ERC Consolidator Grant RAMBAS (Project No. 10104424).  R.Z. and X.Z. are grateful to the financial supports by National Key R\&D Program of China (No. 2022YFA1006300) and the financial supports of the NSFC (No. 12426205).
R.Z. is grateful to the financial supports of the NSFC (No. 12271030). X.Z. is grateful to the financial supports in part by National Key R\&D Program
of China (No. 2020YFA0712700) and the NSFC (No. 12090014, 12288201) and the support by key Lab of Random
Complex Structures and Data Science, Youth Innovation Promotion Association (2020003), Chinese Academy of
Science and the financial supports  by the Deutsche Forschungsgemeinschaft (DFG, German Research Foundation) – Project-ID 317210226--SFB 1283.
\section{Main results and Structure of the proof}

\subsection{Main result on the quantum Gibbs state}
We consider a homogeneous system of bosons in the torus $\mathbb{T}^3=[0,2\pi]^3$. The underlying many-body Hilbert space is the (complex) {\em bosonic Fock space}
\begin{equation}\label{eq:Fock}
	\gF = \gF(L^2(\bT^3)) = \C \oplus \gH \oplus \ldots \oplus \gH^n \oplus \ldots, \quad \gH^n = L^2_{\rm sym} (\mathbb{T}^{3n}).
\end{equation}
We are interested in the Gibbs state on Fock space
\begin{align}\label{eq:Gibbs-state-def}
	\Gamma_{\lambda} = \cZ_{\lambda}^{-1} e^{-\bH_\lambda},\quad \cZ_{\lambda} = \Tr e^{-\bH_\lambda},
\end{align}
with the many-body interacting Hamiltonian
\begin{align}\label{eq:bH-lambda}
	\bH_\lambda = 0 \oplus (\lambda  (-\Delta -\vartheta) )  \bigoplus_{n=2}^\infty  \left( \lambda \sum_{j=1}^n (-\Delta_{x_j} -\vartheta) + \lambda^2 \sum_{1\le j<k\le n} v^\eps(x_j-x_k) \right)
\end{align}
Here $\lambda\to 0^+$ is a semiclassical parameter, which plays the same role of the inverse temperature. We will choose $\eps=\eps(\lambda) \to 0^+$ and $\vartheta=\vartheta(\lambda)\to +\infty$ appropriately.
The interaction potential  $v^\eps: \bT^3\to \R$ will be chosen such that it is a positive-type interaction, namely $\widehat{ v^\eps}(k) \ge0$ and ${\widehat {v^\eps}(0)=1}$, and it converges to the delta function $\delta_0$ when $\eps\to 0^+$. More concrete assumptions will be given later. Under these conditions, $\bH_\lambda$ is bounded from below and it can be defined as a self-adjoint operator by Friedrichs' method. In fact, for any given $\lambda>0$, the partition function $\cZ_\lambda$ is finite and $\Gamma_\lambda$ is well defined quantum state, namely a nonnegative trace class operator on $\gF$ with $\Tr[\Gamma_\lambda]=1$.

\subsubsection*{Reduced density matrices} In the above grand canonical setting, the number of particles is not fixed and its expectation in a given quantum state $\Gamma$ is given by $\Tr [\cN \Gamma]$ with the {\em number operator} $\cN =\bigoplus_{n=0}^\infty n \1_{\fH^n}$. If  $\Gamma$ commutes with $\cN$, then it can be written in the diagonal form $\Gamma=\bigoplus_{\ell=0}^\infty (\Gamma)_\ell$, and we can define the  {\em reduced density matrices} (also called the correlation functions) of $\Gamma$ via partial traces:
$$
\Gamma^{(n)}= \sum_{\ell \ge n} {\ell \choose n} \Tr_{n+1\to \ell}[(\Gamma)_\ell] ,\quad \forall n \ge 1.
$$
Thus the $n$-body density matrix $\Gamma^{(n)}$ of a state $\Gamma$ is a nonnegative trace class operator on $\gH^n$ with
$$n! \Tr[\Gamma^{(n)}]=\Tr[\cN(\cN-1) \cdots (\cN-n+1)\Gamma].$$
The reader may think of the reduced density matrices as the quantum analogue of the marginal probability density functions, since for every self-adjoint operator $A_n$ on $\gH^n$ we have
$$
\Tr [A_n \Gamma^{(n)}] = \Tr [ \bA_n \Gamma],
$$
where $\bA_n$ is the second quantization of $A_n$, given by
\begin{equs}[def:mAn]
	\bA_n = 0\oplus \cdots \oplus A_n \oplus \bigoplus_{\ell=n+1}^\infty \left( \sum_{1\le i_1<i_2<\cdots <i_n\le \ell } (A_n)_{i_1,\cdots,i_n}\right).
\end{equs}

\subsubsection*{The ideal Bose gas} In the non-interacting case, the Gaussian quantum state $\Gamma_0=\cZ_0^{-1}e^{-\bH_0}$ with
\begin{align}\label{eq:GFF-quantum}
	\cZ_0=\tr[e^{-\bH_0}],\quad\bH_0 = 0 \oplus \bigoplus_{n=1}^\infty \left(\lambda \sum_{j=1}^n (-\Delta_{x_j}+1)\right)
\end{align}
is exactly solvable. It can be interpreted as the quantum analogue of the Gaussian free field
$$
\d\mu_0(u) = \frac{1}{\sZ_0} \exp \bigg(-\int_{\bT^3} \Big( |\nabla u(x)|^2 +  |u(x)|^2 \Big) \d x \bigg)
$$
over complex valued distributions $u$. In fact, for all $n\ge 1$ we have
\begin{align}\label{eq:CV-dm-Gamma0}
	n! \lambda^n \Gamma_0^{(n)} = n! \left( \frac{\lambda}{e^{\lambda(-\Delta+1)}-1} \right)^{\otimes n} \to n! \left( \frac{1}{-\Delta+1} \right)^{\otimes n} = \int |u^{\otimes n}\rangle \langle u^{\otimes n}|\, \d\mu_0(u)
\end{align}
strongly in the Hilbert--Schmidt topology when $\lambda\to 0$.  Here in the last expression in \eqref{eq:CV-dm-Gamma0},  $|u^{\otimes n}\rangle \langle u^{\otimes n}|$ is an operator on $\gH^n$ given by the quadratic form  $\langle f, ( |u^{\otimes n}\rangle \langle u^{\otimes n}| ) g\rangle = \langle f, u^{\otimes n}\rangle \langle u^{\otimes n}, g\rangle$, which explains the reasoning of the bra--ket notation. Here $|u^{\otimes n}\rangle \langle u^{\otimes n}|$ is {\em unbounded} $\mu_0$-almost surely, but the expression can be interpreted in a distributional sense. Moreover, since the system is translation-invariant, the one-body density $\rho_{0} = \Gamma_0^{(1)}(x,x)$  is just a constant,  namely
\begin{align}\label{eq:rho0}
	\rho_{0} =  \frac{1}{(2\pi)^3}\sum_{p\in \mathbb{Z}^3}  \frac{1}{e^{\lambda(|p|^2+1)}-1} = \frac{\zeta\left(\frac{3}2\right)}{(4\pi \lambda)^{3/2}} +\frac{C_0}{\lambda}+  O\left( \frac 1 {\lambda^{1/2}} \right)
\end{align}
with
\begin{align}\label{def:C0}
	C_0=-\frac{1}{4\pi} + \frac{1}{4\pi}\sum_{\ell\in2\pi\mZ^3\setminus\{0\}}\frac{e^{-|\ell|}}{|\ell|}.
\end{align}
Here in the leading term, the factor $(4\pi)^{-3/2} \zeta\left(\frac{3}2\right)$, {with $\zeta$ being the Riemann zeta function}, is the critical density of the ideal Bose gas in the infinite-volume limit \cite{Einstein-24}. The second order term is a finite-volume correction since we are working on the fixed  torus $\bT^3$. {Since we fix the chemical potential of the ideal gas as $1$, we place  the corresponding Gibbs state just slightly above the critical point of the Bose--Einstein phase transition, where the number of particles in each (zero or nonzero) momentum mode is of order $\lambda^{-1}$. We refer to Appendix \ref{sec:ideal-gas} for a more detailed discussion of the ideal Bose gas.}

\subsubsection*{The interacting case} We will establish a  picture similar to \eqref{eq:CV-dm-Gamma0}, where the limit of the reduced density matrices of the interacting Gibbs state $\Gamma_\lambda$ is described by the $\Phi^4_3$ measure
\begin{equs}[e:Phi-measure1-main]
	\dif\nu(\Phi)\eqdef \frac{1}{\sZ}\exp\bigg(-\int_{\mathbb T^3} \Big(|\nabla \Phi|^2+ m_0 |\Phi|^2 +  \frac12 |\Phi(x)|^4 \bigg) \dif x\bigg)\mathcal D \Phi
\end{equs}
over $\mC$ valued fields $\Phi$, where  $\sZ$ is normalization constant (partition function).  It is important to note that the formula \eqref{e:Phi-measure1-main} is only formal, as $\nu$ is singular with respect to the Gaussian free field. A rigorous construction can be obtained via stochastic quantization (see Theorem \ref{th:phi43} below).

We will  need the following concrete assumption on the interaction potential. To this end, we introduce the following notations:
The Fourier transform
and the inverse of the Fourier transform  which are defined by
$$\mathcal{F}f(k)=\hat f(k)=\int_{\mathbb{R}^3} f(x) e^{-\imath x\cdot k}\,\dif x,\quad \mathcal{F}^{-1}f(k)=(2\pi)^{-3}\int_{\mathbb{R}^3} f(x) e^{\imath x\cdot k}\,\dif x,\qquad k\in\mR^3.$$

\medskip

\noindent
{\bf Assumption (Hv).} Let $v \in L^1(\R^3)$ be a nonnegative function whose Fourier transform satisfies
\begin{align}\label{eq:def-v}
	\hat v(0)=1,\quad 0\le \hat v(k)\lesssim \frac1{1+|k|^{3+\delta_0}},\quad |D^m \hat v(k)|\lesssim \frac1{1+|k|^{m}},\quad \text{for }m\in \{1,2\},
\end{align}
where the constant $\delta_0>0$ and the implicit constants are independent of $k$,  and $|D^l v|=\sup_{|\alpha|=l}|\partial^\alpha v|$ for $l\in\mathbb{N}$.

For $\eps>0$ we define $v^\eps={\eps^{-3}}v(\eps^{-1}\cdot)$ which approximates the Dirac measure $\delta_0$ as $\eps\to 0$. By an abuse of notation we periodize $v^\eps$ so that $v^\eps$ is treated as periodic function defined on $\mathbb{T}^3$. In fact
\begin{align}\label{eq:v-eps-per-def}
	v^\eps (x)=\frac{1}{(2\pi)^3}\sum_{k\in \bZ^3} \hat v(\eps k) e^{\imath k\cdot x}.
\end{align}
The second condition in \eqref{eq:def-v} essentially means that $\sum_{k \in \bZ^3} \hat v(k) (1+|k|)^{\delta_0-} <\infty$, while the third condition in \eqref{eq:def-v} imposes some weak decay requirement for the derivative of $\hat v$.  Examples satisfying conditions $(\textbf{Hv})$  include the Bessel potential, where $\hat v(k) =(1+|k|^2)^{-\beta/2}$ with $\beta > 3$, and $\hat v(k) = e^{-c|k|^2}$ for some $c > 0$, as well as convex combinations of these functions.

Moreover, to renormalize the interaction in the limit, we choose the chemical potential $\vartheta$ in $\Gamma_\lambda$ in terms of the following {\em counterterms} from the classical field theory:
\begin{equs}[eq:counterterms]
	a^\eps &= \int_{\bT^3}v^\eps(y) G(y)\, \d y,\quad 	
	6b^{\eps}=\sum_{k_1,k_2\in\mZ^3}\frac{\widehat{v^\eps}(k_1)^2+\widehat{v^\eps}(k_2)\widehat{v^\eps}(k_1)}{(2\pi)^6\la k_1\ra^2\la k_2\ra^2\la k_1+k_2\ra^2},
	\\
	C_1 &=\frac1{(2\pi)^{6}}\int_{|x|\leq1} \frac{\hat v(x)(\hat v(y)-\hat v(x-y)-x\cdot \nabla \hat v(y))}{2|x|^4|y|^2}\,\d x\d y\\
	&\qquad+\frac1{(2\pi)^{6}}\int_{|x|>1} \frac{\hat v(x)(\hat v(y)-\hat v(x-y))}{2|x|^4|y|^2}\,\d x\d y, \\
	C_2 &=\frac1{(2\pi)^{6}}\Big(\int_{|x|\leq1} \frac{\hat v(y)-\hat v(x-y)-x\cdot \nabla \hat v(y)}{2|x|^4|y|^2}\,\d x\d y+\int_{|x|>1} \frac{\hat v(y)-\hat v(x-y)}{2|x|^4|y|^2}\,\d x\d y\Big).
\end{equs}
Here in the definition of $a^\eps$, we used the Green function $G(x-y)=(1-\Delta)^{-1}(x,y)$. Note that $G(x,y) \ssim (4\pi)^{-1}|x-y|^{-1}$ when $|x-y|\to 0$ (see, e.g.,  \cite[Lemma 5.4]{RouSer-16}), and hence $a^\eps\ssim \eps^{-1}$. This term comes from Wick normal order.  The term $6b^\eps \ssim |\log \eps| $ is the correction to the Wick renormalization, which is crucial to construct the $\Phi^4_3$ theory.  The two constants $C_1$, $C_2$ are finite thanks to our condition on $v$. The derivation of these terms will be explained later.

Now we are ready to state the main result of the paper.

\begin{theorem}\label{thm:main10} Let $\nu$ be the $\Phi^4_3$ measure associated with $m_0\in \mathbb{R}$ as in \eqref{e:Phi-measure1-main}. Consider the Gibbs state $\Gamma_\lambda=\cZ_\lambda^{-1}e^{-\bH_\lambda}$ in \eqref{eq:Gibbs-state-def}  with $v$ satisfying $(\mathbf{Hv})$  and the chemical potential
	\begin{equs}[def:gamma]
		\vartheta = \frac{\zeta\left(\frac{3}2\right)}{(4\pi)^{3/2}} \lambda^{-1/2} + C_0 + a^\eps - 6b^\eps + 2C_1 + 2C_2 -m_0,
	\end{equs}
	where $C_0$ is given in \eqref{eq:rho0}, and $a^\eps, 6b^\eps,C_1,C_2$ are given in \eqref{eq:counterterms}.  We consider the limit $\lambda,\eps \to 0$ with $ \lambda^{\eta}\le  \eps$ for a sufficiently small constant  $\eta>0$. Then for every continuous function $f:[0,\infty)\to \R$ such that $|f(x)-f(y)|\le C |x-y|(1+x^5+y^5)$ with a constant $C>0,$ and for every fixed function $\varphi\in L^2(\bT^3)$ with $\hat \varphi$ finitely supported, we have
	\begin{equs}[eq:CV-pdm-f]
 \Tr [  f(\lambda a^*(\varphi)a(\varphi)) \Gamma_\lambda)  ] \to \int f(|\langle \varphi, u\rangle|^2) \d\nu (u).
	\end{equs}
Consequently, we have the convergence of the density matrices with $6\ge n\ge 1$:
	\begin{equs}[eq:CV-pdm-f-x6]
 n! \lambda^n \Big\langle \varphi^{\otimes n},\Gamma_\lambda^{(n)}  \varphi^{\otimes n}\Big\rangle  \to  \int  |\langle \varphi, \Phi\rangle|^{2n} \,\d \nu(\Phi),
	\end{equs}
	for all  $\varphi\in L^2(\bT^3)$ with $\hat \varphi$ finitely supported. If we assume further that $|\log \lambda|^{-\eta}\le  \eps \to0 $ for a constant  $0<\eta<1/2$, then the convergence \eqref{eq:CV-pdm-f-x6} holds for all $n\ge 1$.
\end{theorem}

\br The result in \eqref{eq:CV-pdm-f} applies to $f(x)=e^{-\alpha x}$ with $\alpha>0$, for instance, thus providing a unique characterization of the $\Phi^4_3$ measure $\d\nu (u)$ in the limit. The convergence of the correlation function \eqref{eq:CV-pdm-f-x6} follows from  \eqref{eq:CV-pdm-f} and the choice $f(x)=x^n$ with $1\le n\le 6$. This convergence can be interpreted as
\begin{equs}[def:gamma-convergence]
	n! \lambda^n \Gamma_\lambda^{(n)}  \to  \int |u^{\otimes n}\rangle \langle u^{\otimes n}|\, \d \nu(u)
\end{equs}
in the distributional sense. We expect this convergence to hold for all $n\ge 1$ under the polynomial-type condition $ \lambda^{\eta}\le  \eps$, but we can currently prove it only either for $n\le 6$ or under the stronger condition $|\log \lambda|^{-\eta}\le  \eps\to0$. The constant $\eta$ can be determined explicitly from our proof, but since it depends on the constant $\delta_0$ in condition $(\mathbf{Hv})$   and we do not try to optimize this dependence, we decide to leaves it as an implicit constant. 
\er	
	
\br Here we add $C_1+C_2$ in the chemical potential. If we do not put them there, then we will see a new mass term in the limiting measure (see Remark \ref{C1C2e} below).
\er

\br We can also interpret the complex-valued $\Phi^4$ model as an $\mathbb{R}^2$-valued $\Phi^4$ model, i.e., as an $O(2)$ model. From the proof below, we see that Theorem \ref{thm:main10} holds for any $N$ complex components. In other words, the convergence of the {many body quantum Gibbs states with $\mathbb{R}^{2N}$-valued components} to the $O(2N)$ model follows (see \cite{SSZZ20, SZZ21} for the large-$N$ limit of the $O(N)$ model via stochastic quantization).
\er

\br Our approach naturally extends to the two-dimensional problem, leading to an improvement over the results in \cite{FKSS23}. We will discuss the 2D case separately in Section \ref{sec:2D}. 
\er

As we mentioned in the introduction, our result is inspired by a series of works \cite{LewNamRou-15,FroKnoSchSoh-17,LewNamRou-18,FroKnoSchSoh-19,LewNamRou-21,FKSS22,FKSS23} where nonlinear Gibbs measures are derived from many-body quantum mechanics. The initial work \cite{LewNamRou-15} already contained a derivation of the $\Phi^4_1$ measure. In this case, there is no renormalization needed to construct the measure, and the main difficulty lies in the implementation of semiclassical techniques in infinite dimensions.  The $\Phi^4_2$ theory, which requires Wick renormalization, was derived recently in \cite{FKSS23}. Based on the earlier work \cite{FKSS22} which handled the derivation of the Hartree-type measure, the paper \cite{FKSS23} made the analysis quantitative to allow $\eps\to 0$ and then resolved the rigorous connection between the nonlocal field theory and the local one in two dimensions via Nelson's argument, under the condition that $\exp(-|\log \lambda|^{\frac12-c})\le  \eps \to0 $ for $c>0$.

The problem in three dimensions is significantly harder since Wick renormalization is insufficient. While our proof of Theorem \ref{thm:main10} follows the same overall strategy in \cite{FKSS23}, namely the $\Phi^4_3$ measure and the many-body problem are related via an $\eps$-dependent Hartree measure, the implementation of this strategy requires several new key ideas.  On one hand we refine the approach in \cite{LewNamRou-21} and employ correlation inequalities from \cite{DeuNamNap-25} to derive the Hartree measure from the quantum Gibbs state. On the other hand we use an approach via stochastic quantization inspired by  the SPDE
method in  \cite{AK17, GH18a, SSZZ20, SZZ21} to establish the convergence from the Hartree measure to the $\Phi^4_3$ theory.  Moreover, in order to extend the result to the more physically relevant regime $ \lambda^{\eta}\le  \eps$, we combine the variational approach to the classical measure as proposed in \cite{BG18} and a deeper interplay between quantum and classical techniques. These ingredients are of independent interest and will be discussed in detail in the next subsections.

\subsection{Convergence of Hartree measure to $\Phi^4_3$}

Let $v,v^\eps$ as in $\mathbf{(Hv)}$ and $a^\eps,b^\eps$ as in \eqref{eq:counterterms}. We consider the following Hartree measure
\begin{equs}[e:Phi_eps-measure1]
	\dif\nu^\eps(\Psi)\eqdef  \frac{1}{\sZ_\eps}\exp\bigg(&-\int_{\mathbb T^3} (|\nabla \Psi|^2+m |\Psi|^2) \,\dif x
	-\frac12\int \Wick{|\Psi(x)|^2}v^\eps(x-y)\Wick{|\Psi(y)|^2}\,\dif x\dif y
	\\&+{(a^\eps- 6b^\eps)\int_{\mathbb T^3} \Wick{|\Psi(x)|^2}\dif x}\bigg)\mathcal D \Psi,
\end{equs}
over $\mC$ valued fields $\Psi$, with $m\in \mathbb{R}$ (the reader may think of $m=m_0-2C_1-2C_2$).
Here $\Wick{|\Psi|^2}$ denotes the Wick renormalization, formally defined
\begin{align}\Wick{|\Psi(x)|^2}  = |\Psi(x)|^2 - \langle |\Psi(x)|^2 \rangle_{\mu_0}, \label{eq:Wick}
\end{align}
and $\sZ_\eps$ represents the normalization constant (partition function).

For fixed $\eps > 0$, the measure $\nu^\eps$ can be rigorously constructed as a probability measure that is absolutely continuous with respect to the Gaussian free field, provided that $v$ satisfies condition $(\mathbf{Hv})$ (cf. \cite{LewNamRou-21} and Theorem \ref{th:global} below). As $\eps \to 0$, we expect that the limit of $\nu^\eps$ corresponds to the $\Phi^4_3$ field, which is formally expressed as
\begin{equs}[e:Phi-measure1]
	\dif\nu(\Phi)\eqdef \frac{1}{\sZ}\exp\bigg(-\int_{\mathbb T^3} \Big(|\nabla \Phi|^2+(m+{2C_1+2C_2}) |\Phi|^2 \Big)\dif x
	+ \frac12\int |\Phi(x)|^4\,\dif x\bigg)\mathcal D \Phi
\end{equs}
over $\mC$ valued fields $\Phi$,
where $C_1$ and $C_2$ are finite constants depending on $v$, as given \eqref{eq:counterterms}, and  $\sZ$ is normalization constant (partition function). Recall that the formula \eqref{e:Phi-measure1} is only formal and a rigorous construction can be found in  Theorem \ref{th:phi43} below.

To state the main result of this subsection, for $n\geq1$, we define $n$-point correlation function of $\nu^\eps$ and $\nu$ given by
\begin{align*}
	(\gamma_n^\eps)_{\bf{x},\tilde{\bf{x}}}\eqdef \int \overline{\Psi}(\tilde x_1)\dots \overline{\Psi}(\tilde x_n)\Psi(x_1)\dots \Psi(x_n)\,\dif \nu^\eps(\Psi),
\end{align*}
and
\begin{align*}
	(\gamma_n)_{\bf{x},\tilde{\bf{x}}}\eqdef \int \overline{\Phi}(\tilde x_1)\dots \overline{\Phi}(\tilde x_n)\Phi(x_1)\dots \Phi(x_n)\,\dif \nu(\Phi).
\end{align*}
Note that $\nu^\eps$ and $\nu$ are supported on $\bC^{-\frac12-\kappa}$ for $\kappa>0$. Hence, $\Phi(x_i)$, $\Psi(x_i)$ and $\overline{\Phi}(\tilde x_j)$, $\overline{\Psi}(\tilde x_j)$ are understood in the distributional sense.

{{The following key result allows us to approximate the singular $\Phi^4_3$ measure in \eqref{e:Phi-measure1} by the Hartree measure  $\nu^\eps$ in \eqref{e:Phi_eps-measure1}, which is more regular and easier to connect to the quantum model.}}

\bt\label{th:conm} Assume that $v$ satisfies condition $\mathbf{(Hv)}$. As $\eps \to 0$, the probability measure $\nu^\eps$ in \eqref{e:Phi_eps-measure1} converges weakly to $\nu$ in $\bC^{-\frac12-\kappa}$, $\kappa>0$, where $\nu$ represents the $\Phi^4_3$ field \eqref{e:Phi-measure1}. Moreover, any correlation function $\gamma_n^\eps$ associated with $\nu^\eps$ converges to the $n$-point correlation function $\gamma_n$ of the $\Phi^4_3$ field in $\cS'(\mT^{6n})$.
\et

\br If we consider the following measure
\begin{equs}[e:Phi_eps-measure]
	\dif\widetilde\nu^\eps(\Psi)\eqdef  \frac{1}{\widetilde\sZ_\eps}\exp\bigg(&-\int_{\mathbb T^3} (|\nabla \Psi|^2+m |\Psi|^2)\, \dif x
	-\frac12\int \Wick{|\Psi(x)|^2}v^\eps(x-y)\Wick{|\Psi(y)|^2}\,\dif x\dif y
	\\&+{\int\Re(\overline{\Psi} (x) v^\eps(x-y) G(x-y) \Psi(y))\,\dif y\dif x-6b^\eps\int_{\mathbb T^3} \Wick{|\Psi(x)|^2}\dif x}\bigg)\mathcal D \Psi,
\end{equs}
with a natural Wick renormalization counterterm $\int\Re(\overline{\Psi} (x) v^\eps(x-y) G(x-y) \Psi(y))\,\dif y\dif x$,
then its limit measure is given by the following $\Phi^4_3$ measure without new mass term $C_1+C_2$ in \eqref{e:Phi-measure1}:
\begin{equs}[e:Phi-measure]
	\dif\nu(\Phi)\eqdef \frac{1}{\sZ}\exp\bigg(-\int_{\mathbb T^3} (|\nabla \Phi|^2+m |\Phi|^2)\, \dif x
	+ \frac12\int |\Phi(x)|^4\,\dif x\bigg)\mathcal D \Phi
\end{equs}
over $\mC$ valued fields $\Phi$,
where $\sZ_\eps$, $\sZ$ are normalization constants (partition functions).
For further explanation, see Remark \ref{C1C2e} below.

The main reason for considering the measure \eqref{e:Phi_eps-measure1} instead of \eqref{e:Phi_eps-measure} is that \eqref{e:Phi_eps-measure1} is more directly related to the quantum many-body problem. The renormalization counterterm $(a^\eps - 6b^\eps ) \int_{\mT^3}\Wick{|\Psi^\eps|^2}\dif x$ can be understood as the limit of the chemical potential in quantum many-body problems. {As we explained in the introduction, the chemical potential is the only parameter we can adjust in the quantum problem,  and hence having a mass renormalization instead of the nonlocal renormalization keeps the model  more physically relevant.}
\er

\subsection{From quantum model to Hartree measure}\label{sec:defs Fock}

Recall the Hartree measure $\nu^\eps$ in \eqref{e:Phi_eps-measure1} with $m=m_0-2C_1-2C_2$, namely 
$$
\d\nu^\eps(\Psi)=\frac{e^{-\cD [\Psi]} \; \d\mu_{0}(\Psi)}{\int  e^{-\cD [\Psi]} \; \d\mu_{0}(\Psi)}
$$
with 
	$$ \cD[\Psi] = \frac12\int :\!{|\Psi(x)|^2}\!: v^\eps(x-y) :\! {|\Psi(y)|^2}\!: \d x \d y  - {(a^\eps-6 b^\eps-m+1)\int_{\mathbb T^3} :\!{|\Psi(x)|^2}\!: \d x}.
$$
 We have the following connection from the quantum Gibbs state $\Gamma_\lambda$ to this Hartree measure.

\begin{theorem}\label{thm:main1}
	Let $v$  satisfy $(\textbf{Hv})$. Consider the Gibbs state $\Gamma_\lambda=\cZ_\lambda^{-1}e^{-\bH_\lambda}$ in \eqref{eq:Gibbs-state-def}. When $\lambda,\eps \to 0^+$ and $\lambda^\eta \le \eps$ for a sufficiently small parameter $\eta>0$, we have the convergence of the Hartree partition function
\begin{equs}[eq:main-CV-Hartree-1]
	\left| \log \frac{\cZ_\lambda}{\cZ_0} -\log \left(  \int  e^{-\cD [\Psi]} \; \d\mu_{0}(\Psi) \right) \right| \to 0.
\end{equs}
Moreover, for every continuous function $f:[0,\infty)\to \R$ such that $|f(x)-f(y)|\le C |x-y|(1+x^5+y^5)$ with a constant $C>0$, and for every fixed function $\varphi\in L^2(\bT^3)$ with $\hat \varphi$ finitely supported, we have
\begin{equs}[eq:main-CV-Hartree-2]
\left| 	\int f(|\langle \varphi, \Psi\rangle|^2) \d\nu^\eps (\Psi) -  \Tr [  f(\lambda a^*(\varphi)a(\varphi)) \Gamma_\lambda  ] \right| \to 0.
\end{equs}
Consequently, we have the convergence of the density matrices with $6\ge k\ge 1$ in the sense that
\begin{equs}[eq:main-CV-Hartree-3]
\left| 	\int |\langle \varphi, \Psi\rangle|^{2k}\; \d\nu^\eps (\Psi) - k! \lambda^k  \langle \varphi^{\otimes k}, \Gamma_\lambda^{(k)} \varphi^{\otimes k}\rangle   \right| \to 0.
\end{equs}
If we assume further that $|\log \lambda|^{-\eta}\le  \eps \to0 $ for a constant  $0<\eta<1/2$, then the convergence \eqref{eq:main-CV-Hartree-3} holds for all $k\ge 1$.
\end{theorem}

From Theorem \ref{thm:main1} and Theorem \ref{th:conm} we obtain the desired result in Theorem \ref{thm:main10}.

\begin{remark} Theorem \ref{thm:main1} extends previous results with $\eps\ssim 1$ in \cite{LewNamRou-21,FKSS22} and improves the recent result in \cite[Proposition 5.2]{FKSS23} - which was limited to the case $\exp(-|\log \lambda|^{\frac12-c})\le  \eps \to0 $ for $c>0$. While the proof in \cite{FKSS23} is based on a refinement of the functional integral approach in \cite{FKSS23}, we employ the variational approach in both the quantum side \cite{LewNamRou-21} and the classical side \cite{BG18}. 
\end{remark}

\subsection{Ideas of the proof and structure of the paper}

\text{}

The rest of the paper is devoted to the proofs of Theorem \ref{th:conm} and Theorem \ref{thm:main1}.  Let us first explain the proof strategy of Theorem \ref{th:conm}.
We consider  the stochastic quantization of the measure \eqref{eq:mainnew}, given by:
\begin{equs}[eq:psi1-i]
	\LL \Psi^\eps= -( v^\eps*\Wick{|\Psi^\eps|^2})\Psi^\eps+{(a^\eps-6b^\eps+1-m) \Psi^\eps}+\xi,
\end{equs}
and
the stochastic quantization of the measure  \eqref{e:Phi-measure1}  formally given by
\begin{equs}[eq:phi1-i]
	\LL \Phi= -|\Phi|^2\Phi-{(2C_1+2C_2-1+m)\Phi}+\xi.
\end{equs}
Here and in the following, we define $\LL = \p_t - \Delta + 1$, with $m \in \mathbb{R}$, and $\xi$ denotes complex-valued space-time white noise on a probability space $(\Omega, \mathcal{F}, \mathbf{P})$. We refer to Section \ref{sec:ren} and \eqref{eq:mainn1} for more explain on the Wick renormalization $\Wick{|\Psi^\eps|^2}$.

The core idea behind the proof of Theorem \ref{th:conm} is to use the dynamics of \eqref{eq:psi1-i} and \eqref{eq:phi1-i} to establish the convergence of the stationary measures \eqref{e:Phi_eps-measure1} associated with the dynamics of \eqref{eq:psi1-i} to the stationary measure \eqref{e:Phi-measure1} corresponding to equation \eqref{eq:phi1-i}. To this end, we  first prove the convergence of the solutions to equation \eqref{eq:psi1-i} to the solutions of equation \eqref{eq:phi1-i} if the initial data converges in suitable Besov space $\bC^{-\frac12-\kappa}$, $\kappa>0$. We then focus on deriving uniform in $\eps$ estimates (coming down from infinity) for stationary solutions to \eqref{eq:psi1-i}. This is done by decomposing the equation into its low-frequency and high-frequency components, and then  combining an $L^2$-energy bound with Schauder estimates. Then, using the uniform in $\eps$ moments bounds for the stationary solutions, we establish tightness of the quantum fields \eqref{e:Phi-measure1} in $\bC^{-\frac12-\kappa}$.  Finally, by leveraging the convergence of the dynamics and the uniqueness of the invariant measure for equation \eqref{eq:phi1-i}, as provided by the theory of singular SPDEs, we identify  the tight limit as the $\Phi^4_3$ field, completing the proof of Theorem \ref{th:conm}.

In Section \ref{first}, we begin by analyzing the stochastic quantization equations via
paracontrolled calculus and demonstrate that if the initial data converges in $\bC^{-\frac12-\kappa}$, then the solutions to the corresponding dynamics also converge (see Theorem \ref{th:1} below). To compare the Wick renormalization between $v^\eps*\Wick{|\Psi^\eps|^2}$ and $\Psi^\eps$ with the mass renormalization $a^\eps \Psi^\eps$ in \eqref{eq:psi1-i}, we introduce a  new operator $\cR$. Formally, this operator vanishes in the limit. However, owing to the singular nature of the solution to \eqref{eq:phi1-i}, $\cR$ induces an additional mass term $2(C_1+C_2)\Phi$ in the limiting equation \eqref{eq:phi1-i}.
Section \ref{sec:uni} is focused on deriving global in $\eps$ estimates  for stationary solutions to \eqref{eq:mainnew} and establishing the convergence of the fields, which leads to the proof of Theorem \ref{th:conm}. Concerning the uniform in $\eps$ moments bounds for the stationary solutions,  the term $-|\Phi|^2\Phi$ plays a crucial role in the dynamical $\Phi^4_3$ model. In that case, we can apply $L^p$-energy estimates or the maximum principle to obtain uniform control of the solutions. However, in our modified setting, the nonlinear term  $-v^\eps * |\Psi|^2\Psi$ only provides a weaker damping effect in $L^2$-energy estimates. As a result, we need to present a more refined analysis of the nonlinearity to achieve uniform control (see Remark \ref{reglobal} below).
All stochastic estimates are collected in Section \ref{sec:sto} below.

At this stage, the analysis of the relationship between the Hartree measure and the $\Phi^4_3$ measure, as summarized in Theorem \ref{th:conm}, is complete. The remainder of the proof focuses on establishing convergence from the quantum Gibbs state to the Hartree measure, as stated in  Theorem \ref{thm:main1}. 

Our starting point is the variational approach proposed in \cite{LewNamRou-15,LewNamRou-18,LewNamRou-21}. On one hand, the Hartree measure $\nu^\eps$ in \eqref{e:Phi_eps-measure1} is the unique minimizer for the variational problem
\begin{equation} \label{eq:zr-rel}
	-\log \sZ_\eps = \min_{\substack{\nu \text{~proba. meas.}\\ \nu ' \ll\mu_0}} \left\{ \cH_{\rm cl}(\nu', \mu_0) + \int \cD[u]\,\d\nu'(u)\right\},
\end{equation}
where
$$ \cH_{\rm cl} (\nu',\mu_0)\eqdef \int \frac{\d\nu'}{\d\mu_0}(u)\log\left(\frac{\d\nu'}{\d \mu_0 }(u)\right)\,\d\mu_0(u) \ge 0 $$
is the classical relative entropy between the probability measure $\nu'$ and the Gaussian free field $\mu_0$, and
\begin{equation} \label{eq:cDu}
	\cD[u] = \frac12\int :\!{|u(x)|^2}\!: v^\eps(x-y) :\! {|u(y)|^2}\!: \d x \d y  - {(a^\eps-6 b^\eps-m+1)\int_{\mathbb T^3} :\!{|u(x)|^2}\!: \d x}
\end{equation}
is the Wick renormalization of the interaction term. On the other hand, the interacting Gibbs state $\Gamma_\lambda=\cZ_\lambda^{-1}e^{-\bH_\lambda}$ in \eqref{eq:Gibbs-state-def} is the {\em unique minimizer} for the variational problem
\begin{equation} \label{eq:rel-energy}
	-\log \frac{\cZ_\lambda}{\cZ_0} = \min_{\substack{\Gamma\ge 0\\ \Tr \Gamma=1}} \Big\{ \cH(\Gamma,\Gamma_{0}) +  \Tr\left[\bW \Gamma\right] \Big\}.
\end{equation}
Here
$$
\cH(\Gamma,\Gamma_{0}) = \Tr [\Gamma (\log \Gamma-\log \Gamma_0)] \ge 0
$$
is the relative entropy between a quantum state $\Gamma$ and the non-interacting Gibbs state $\Gamma_0=\cZ_0^{-1}e^{-\bH_0}$ in \eqref{eq:GFF-quantum}, and $\bW$ is the renormalized interaction in $\bH_\lambda$. We have chosen the chemical potential in $\bH_\lambda$ such that $\bW$ is exactly the second quantized version of the classical term $\cD[u]$ in \eqref{eq:cDu}. 

In this approach, we will derive upper and lower bound for the (relative) free
energy, and then use the minimizing property of $\Gamma_\lambda$ to deduce desired properties of the state. In contrast to \cite{LewNamRou-21}, our approach in the present paper requires deriving new uniform estimates for singular potentials $v^\eps$, which involves an interesting interplay between quantum and classical methods in multiple steps. We now outline these steps in detail.

Since our overall strategy is to reduce the quantum problem to a finite-dimensional setting via the variational approach, as a first step we will need to establish similar bounds for the classical partition function and correlation functions with suitable finite-dimensional cutoffs. This analysis will be presented in Section \ref{sec:part}. The desired bounds do not follow from the dynamical approach in the previous sections, and we will extend the variational method proposed in \cite{BG18} for this purpose. 
Moreover, we will establish a uniform estimate for the classical partition function with suitable perturbations. To be precise, we prove in particular the following uniform bound in Theorem \ref{th:partition-revised}
\begin{equs}[mom:mutildeP-revised-intro]
	\Big|\log\int &\exp\Big(-\cD(P_N\Psi) +c_0|\la \varphi,\Psi\ra|^2 + c_1 \int \Wick{|P_N\Psi|^2} \Big)\,\dif \mu_0(\Psi)\\ &- \log  \int \exp\Big(-\cD(P_N\Psi)\Big)\d \mu_0(\Psi) \Big| 
\lesssim 1,
\end{equs}
where  $P_N$ is a smooth momentum cutoff localized to $|p|\leq N$ (with $N\to\infty$), and the bound depends only on the constants $c_0,c_1$ and the test function $\varphi$. This bound is inspired by the linear response argument from the quantum problem, eventually leading to sharper bounds for the quantum partition function under refined constraints on the parameters $\lambda$ and $\eps$. On the technical side, unlike the classical variational approach used for the quantum part in \eqref{eq:zr-rel}, our method relies on the Boué--Dupuis variational formula (Lemma \ref{lem:B-D}), which reformulates the partition functions in terms of a stochastic control problem. This transformation allows us to reduce the problem to deriving analytic estimates akin to the $L^2$-energy estimates in stochastic quantization. Our analysis is robust against mass perturbations, a key feature needed to establish \eqref{mom:mutildeP-revised-intro}.

In Section \ref{sec:gibbs}, we start the consideration of the quantum model. Similarly to \cite{LewNamRou-21}, we will estimate the free energy by splitting to low and high momenta. The high-momentum estimate will be done by using a new abstract correlation inequality from \cite{DeuNamNap-25}, which helps greatly to improve the analysis and obtain good quantitative estimates in terms of $\eps$. Then in Section \ref{sec:CV-energy}, we perform a semiclassical approximation in finite dimensions, where the nonlinear classical field theory naturally emerges from a quantitative quantum de Finetti theorem. Here we need to adjust the analysis in \cite{LewNamRou-21} in order to obtain error bounds that remain small when $\eps$ depends only polynomially on $\lambda$. Recall that the de Finetti measure at scale $\lambda$, associated with a state $\Gamma$ and a finite-dimensional projection $P$, is defined via the \emph{lower symbol} (c.f. \eqref{eq:Husimi}):
		\begin{equs}[eq:Husimi]
			\d\mu_{P,\Gamma}^{\lambda}(u) \eqdef (\lambda\pi)^{-\Tr(P)}  \Tr \Big[ |W(u/\sqrt{\lambda}) \rangle \langle W(u/\sqrt{\lambda})|  \Gamma \Big]  \d u ,
		\end{equs}
		where $W(u/\sqrt{\lambda}) = \exp(a^{*}(u)-a(u)) |0\rangle$ denotes the standard coherent state. One of the new ingredients of our proof is a pointwise inequality between the de Finetti measure of the non-interacting Gibbs state and the cylindrical projection of the Gaussian free field, which is of independent interest.  This allows us to approximate the quantum partition function by the classical one (see Theorem \ref{thm:local-W-P})
\begin{equs}[eq:partition-ul-1-intro]
	-\log \frac{\cZ_\lambda}{\cZ_0} = -\log \left(  \int  e^{-\cD(P_N u)} \; \d\mu_{0}(u) \right) + O(\eps^{-8}{N^{-1/8}}) + O(\lambda^{\delta}),
\end{equs}
where the constant $\delta>0$ depending only on $\delta_0>0$ in \eqref{eq:def-v}.

Unlike \cite{LewNamRou-21}, the convergence of the quantum correlation functions in our case cannot follow directly from the free energy convergence in \eqref{eq:partition-ul-1-intro}. This occurs because uniform Hilbert-Schmidt estimates for the quantum correlation functions in \cite{LewNamRou-21} are unavailable (the bound is only good under the strong restriction $|\log \lambda|^{-\eta}\le  \eps \to0 $ for a constant  $0<\eta<1/2$), and a deeper analysis is needed. Roughly speaking, our approach to quantum correlation functions relies on a sharper understanding of the relation between the de Finetti and Hartree measures. In particular, by refining the energy estimate in \eqref{eq:partition-ul-1-intro} and applying Pinsker’s inequality (see, e.g., \cite{CarLie-14}), we obtain a quantitative $L^1$-comparison between these two measures in finite dimensions (see \eqref{eq:Chiribella-app-1-proof}). The convergence of the quantum correlation function then follows from this $L^1$-bound together with an interpolation argument that exploits higher moment estimates.

		More precisely, as a first step, in Section \ref{sec:gibbs-higher} we replace the Hilbert–Schmidt estimates with higher-order moment estimates, obtained via a correlation inequality from \cite{DeuNamNap-25} together with crucial input on the classical Hartree field from Section \ref{sec:part}. In particular, in Theorem \ref{thm:higher-moment} we prove the moment bound (see \eqref{eq:moment-m})
		\begin{equs}[eq:moment-m-intro]
			\lambda^k \Tr \big[ (a^*(\varphi)a(\varphi))^k \Gamma_\lambda\big] \le C_{\varphi,k} \eps^{-6}
		\end{equs}
		for all test functions $\varphi$ and $k \le 7$. The condition $k \le 7$ here is the main reason for our restriction on the growth of the function $f$ in Theorems \ref{thm:main10} and \ref{thm:main1} (if we assume $|\log \lambda|^{-\eta}\le  \eps \to0 $ for a constant  $0<\eta<1/2$, then the condition $k \le 7$ can be removed by using \eqref{eq:moment-m-weak}).

In the second step, in Section \ref{sec:approximation-classical-measure} we explain how quantum estimates can be used to improve the classical measure analysis. In particular, we prove in Theorem \ref{th:partition-revised-improved} a quantitative convergence of the correlation functions of the Hartree  field with a momentum cut-off: 
\begin{equs}[sec:9-2-intro]
	\text{} \left|  \int |\langle \varphi,u\rangle|^{2k} \d \mu^\eps_N (u) - \int |\langle \varphi,u\rangle|^{2k}   \d \nu^\eps(u)  \right|  \le   C_{\varphi,k} (\eps^{-8} N^{-1/8 })^{{1/2}}\end{equs}
with $\d \mu_N^\eps(u)\propto \exp(-\cD(P_N u))\dif \mu_0(u)$. 
Here, the freedom in choosing the semiclassical parameter $\lambda$ offers a crucial advantage. This powerful interplay between classical and quantum techniques proves particularly effective for obtaining precise quantitative estimates, a feature we hope will be helpful in other contexts.

Finally, in Section \ref{sec:CV-Gamma} we establish the convergence of the quantum correlation functions. We introduce an explicit representation of the quantum de Finetti measure in terms of generalized quantum correlation functions of the form $\Tr[f(\lambda a^*(\varphi)a(\varphi))\Gamma_\lambda]$ (see Lemma \ref{lem:deF-measure-f}), which is of independent interest, and then prove higher-moment convergence through $L^1$-norm convergence and interpolation.


The rest of the paper is devoted to proving Theorems \ref{th:conm} and \ref{thm:main1}  as outlined above. Moreover, in Section \ref{sec:2D} we include remarks on extending our approach to the two-dimensional problem.

\subsection{Notations}

Throughout the paper, we use the notation $a\lesssim b$ if there exists a constant $c>0$ such that $a\leq cb$, and we write $a\backsimeq b$ if $a\lesssim b$ and $b\lesssim a$. {We will denote by $C$ a general positive constant independent of relevant variables, whose value may change from line to line.} For $k\in\mR^d$, let $\la k\ra^2=1+|k|^2$.

We use the convention that the inner product $\langle \cdot, \cdot\rangle$ of a (complex) Hilbert space is linear in the second argument and antilinear in the first. In particular, we set  $\langle f,g\rangle=\int_{\mT^3} \overline f g\,\dif x$ for $f,g\in L^2(\mT^3)$
and we write $L^p=L^p(\mT^3)$. We also set $e_k(x)=(2\pi)^{-\frac{3}{2}}e^{\imath x\cdot k},x\in\mathbb{T}^3$.

For index variables $i$ and $j$ of Littlewood-Paley decomposition we write $i\lesssim j$ if $2^i\lesssim 2^j$, meaning  there exists $N\in\mN$, independent of $j$ such that $i\leq j+N$, and we use $i\sim j$ to denote both $i\lesssim j$ and $j\lesssim i$.  Given a Banach space $E$ with norm $\|\cdot\|_E$ and $T\in(0,\infty]$, we write $C_TE=C([0,T];E)$ for the space of continuous functions from $[0,T]$ to $E$, equipped with the supremum norm $\|\cdot\|_{C_TE}$. We also write $CE=C_\infty E$, and $L^pE=L^p([0,\infty);E)$ equipped with the  $L^p$-norm $\|\cdot\|_{L^pE}$.  For $\alpha\in(0,1)$ we also define $C^\alpha_TE$ as the space of $\alpha$-H\"{o}lder continuous functions from $[0,T]$ to $E$, endowed with the seminorm $\|f\|_{C^\alpha_TE}=\sup_{s,t\in[0,T],s\neq t}\frac{\|f(s)-f(t)\|_E}{|t-s|^\alpha}.$

We have for $\mathcal{F}^{-1}m\in L^1(\mR^3)$ and  $f\in L^2(\mT^3)$
\begin{align}\label{eq:fourier}
	\sum_k m(k)\langle f,e_k\rangle e_k=\cF^{-1}m*f.
\end{align}
Here and in what follows, the convolution is taken on $\mathbb{R}^3$, and we consider $f$ as a periodic function on $\mathbb{R}^3$. We denote by $\bB^\alpha_{p,q}$ the Besov spaces on the torus, with general indices $\alpha \in \mathbb{R}$ and $p, q \in [1, \infty]$. Additionally, we use the notations $\bC^\alpha = \bB^\alpha_{\infty, \infty}$, $\bB^\alpha_p = \bB^\alpha_{p, \infty}$, and $H^\alpha = \bB^\alpha_{2,2}$. The definition of Besov spaces and some useful lemmas are provided in Appendix \ref{sec:pre}.

Our proof in Sections \ref{first} through \ref{sec:part} is based on the paracontrolled calculus introduced in \cite{GIP15} and the variational approach developed in \cite{BG18}. This calculus relies on Bony's paraproducts \cite{Bon81}, specifically $f \prec g$, $f \succ g$, and the resonant term $f \circ g$. For the definitions and basic estimates related to paracontrolled calculus, we refer to Appendix \ref{sec:pre}.

\section{Convergence of the dynamics associated with \eqref{e:Phi_eps-measure1}}\label{first}

We now consider the system of SPDEs arising from the stochastic quantization of the measure \eqref{e:Phi_eps-measure1}, given by:
\begin{equs}[eq:mainnew]
	\LL \Psi^\eps= -( v^\eps*\Wick{|\Psi^\eps|^2})\Psi^\eps+{(a^\eps-6b^\eps+1-m) \Psi^\eps}+\xi.
\end{equs}
Here and in the following, we define $\LL = \p_t - \Delta + 1$, with $m \in \mathbb{R}$, and $\xi$ denotes complex-valued space-time white noise on a probability space $(\Omega, \mathcal{F}, \mathbf{P})$. More precisely, we write $\xi = \xi_1 + \imath \xi_2$, where $(\xi_1, \xi_2)$ are independent, real-valued space-time white noises. The term $\Wick{|\Psi^\eps|^2}$ refers to Wick renormalization, as defined in Section \ref{sec:ren} and \eqref{eq:mainn1} below.

In this section, we assume condition $\mathbf{(Hv)}$. Under this assumption, we obtain that $(1 + |x|^\kappa)v \in L^1(\mathbb{R}^3)$ for $\kappa \in [0, 1/2)$. In fact,
\begin{equs}[v1]&\|(1+|x|^{\kappa})v\|_{L^1(\mR^3)}=\|(1+|x|^\kappa)\mathcal{F}^{-1}\hat v\|_{L^1(\mR^3)}\lesssim \|(1+|x|^2)\mathcal{F}^{-1}\hat v\|_{L^2(\mR^3)}
	\\\lesssim&\|\mathcal{F}^{-1}(I-\Delta)\hat v\|_{L^2(\mR^3)}\lesssim \|(I-\Delta)\hat v\|_{L^2(\mR^3)}<\infty.
\end{equs}

The stochastic quantization of the measure \eqref{e:Phi-measure1} is formally given by
\begin{equs}[eq:mainnew1]
	\LL \Phi= -|\Phi|^2\Phi-{(2C_1+2C_2-1+m)\Phi}+\xi.
\end{equs}
Here $C_1, C_2$ are defined in \eqref{eq:counterterms}.
We use $(\mathbf{Hv})$ to derive  the finiteness of $C_1,C_2$.

Equation \eqref{eq:mainnew1} is referred to as the dynamical $\Phi^4_3$ model. Due to the singular nature of the space-time white noise, the nonlinear term $|\Phi|^2\Phi$ from \eqref{eq:mainnew1} does not have a classical interpretation. Instead, it must be understood through advanced frameworks such as regularity structures \cite{Hai14} or paracontrolled calculus \cite{GIP15}, which introduce a specific structure for the solutions and allow the analytically ill-defined products to be made sense of using probabilistic tools and renormalization techniques.

The primary objective of this section is to demonstrate the convergence of the stochastic quantization \eqref{eq:mainnew} of the measure \eqref{e:Phi_eps-measure1} to the dynamical $\Phi^4_3$ model \eqref{eq:mainnew1}, as stated in Theorem \ref{th:1} below. The main strategy is to compare the respective nonlinear terms: $v^\eps * \Wick{|\Psi^\eps|^2} \Psi^\eps$ from equation \eqref{eq:mainnew}, and $|\Phi|^2 \Phi$ from the dynamical $\Phi^4_3$ model \eqref{eq:mainnew1}, using paracontrolled calculus. To this end, we express the convolution $v^\eps * f$ in terms of a shift operator $\tau_y f \eqdef f(\cdot - y)$ as $\int v^\eps(y) \tau_y f\,\d y$, which allows us to apply paracontrolled calculus to equation \eqref{eq:mainnew}.

\br \label{C1C2e} An interesting aspect is the appearance of a new linear term $2(C_1 + C_2)\Phi$ in the limiting equation \eqref{eq:mainnew1}, where $C_1$ and $C_2$ depend on $v$. We can also consider the stochastic quantization of the measure \eqref{e:Phi_eps-measure} given by
\begin{equs}[eq:main]
	\LL \Psi^\eps= -( v^\eps*\Wick{|\Psi^\eps|^2})\Psi^\eps+{(v^\eps G)*\Psi^\eps+(1-m-6b^\eps) \Psi^\eps}+\xi.
\end{equs}
The nonlocal term $(v^\eps G) * \Psi^\eps$ arises from the Wick renormalization for $(v^\eps * \Wick{|\Psi^\eps|^2}) \Psi^\eps$, which is the natural renormalization for the Hartree $\Phi^4_3$ model and also appears as a renormalization term in \cite{Bri, OOT24}.

As $\eps \to 0$, the limit of equation \eqref{eq:main} corresponds to equation \eqref{eq:mainnew} with $C_1 = C_2 = 0$. The transition from the Wick renormalization term $(v^\eps G) * \Psi^\eps$ in \eqref{eq:main} to $a^\eps \Psi^\eps$ in equation \eqref{eq:mainnew} creates a new mass term $(2C_1 + 2C_2) \Phi$ in the limiting equation \eqref{eq:mainnew1}. Formally, this transition introduces a new term $\cR\Psi^\eps \in \bC^{-\frac32-\kappa}$ (see Lemma \ref{comnew}) in the equation, which leads to new renormalization terms $c_1^\eps$ and $c_2^\eps$ when dealing with products involving $\cR\Psi^\eps$, similar to the appearance of $b^\eps$ (see Proposition \ref{thcomm1} Proposition \ref{prop:Zn} below). Through direct computation, these renormalization terms converge to finite constants $C_1$ and $C_2$, owing to the special structure of the operator $\cR$.
\er

The basic idea to analyze equations \eqref{eq:mainnew} and \eqref{eq:mainnew1} is to decompose them into an irregular part and a regular part. The most irregular part of the equation is given by the solution to the following linear equation:
\begin{equs}[eq:lin]
	\LL Z = \xi,
\end{equs}
where $Z$ denotes the stationary solution to the linear equation \eqref{eq:lin}. Since space-time white noise is a random distribution with space-time regularity $-\frac52 - \kappa$ (for $\kappa > 0$ under parabolic scaling), we obtain $Z \in C_T\bC^{-\frac12 - \frac\kappa2} \cap C_T^{\frac12-\sigma} \bC^{-\frac32+2\sigma-\frac\kappa2}$, $\bP$-a.s., for every $\kappa, \sigma \in( 0,\frac12)$. Here and in the sequel, $T \in (0, \infty)$ represents an arbitrary finite time.
By substituting $Z$ into the nonlinear terms of equations \eqref{eq:mainnew} and \eqref{eq:mainnew1}, the singularity of $Z$ causes the products involving $Z$ to lose their meaning in the classical analytical framework. Renormalization and probabilistic tools are required for these terms.

In the following, we first introduce the stochastic objects required for the rigorous formulation of the dynamical $\Phi^4_3$ model \eqref{eq:mainnew1} and the approximation equation \eqref{eq:mainnew}. Then, we present the paracontrolled solution framework, incorporating the paracontrolled ansatz for \eqref{eq:mainnew1} and \eqref{eq:mainnew}. Finally, we prove the convergence of equation \eqref{eq:mainnew} to equation \eqref{eq:mainnew1}.

\subsection{Stochastic Objects and Renormalization}\label{sec:ren}

In this section we   introduce the renormalized terms. To begin, we define some useful notations.  Let $P_tf=e^{t(\Delta-1)}f$ and $\sI f=\LL^{-1}f=\int_0^\cdot P_{\cdot-s}f\,\dif s.$
Let $Z_{\delta}$ be the stationary solution to $\LL {Z}_{\delta}=\xi_{\delta}$ with $\xi_\delta$ being $\xi$ mollified with a bump function.

We also introduce graph notation as follows. We use $\cZ^{\<X>}$ to denote space--time white noise $\xi$ and use $\cZ^{\<X1>}$ to denote the conjugate of space--time white noise $\overline{\xi}$. We use $\cZ^{\<1>}$ to denote $Z$ and use  $\cZ^{\<1r>}$ to denote $\overline{Z}$.
After subtracting the $Z$-term from equation \eqref{eq:mainnew1}, we encounter terms involving the square and cubic powers of
$Z$ in the remainder equation. These terms are defined through the Wick product, as follows:
\begin{equs}[e:wick-tilde]
	\begin{array}{lllllll}
		\cZ^{\<2m>}&\eqdef& \lim_{\delta\to0}\cZ_\delta^{\<2m>}\eqdef
		\lim\limits_{\delta\to0}(|Z_{\delta}|^2-a_\delta),  & & \cZ^{\<2>}&\eqdef& \lim_{\delta\to0}\cZ_\delta^{\<2>}\eqdef
		\lim\limits_{\delta\to0}Z_{\delta}^2,
		\\
		\cZ^{\<20vm1>}&\eqdef & \sI(\cZ^{\<2m>}), & &\cZ^{\<20vm2>}&\eqdef& \sI(\cZ^{\<2>}),
		\\
		\cZ^{\<30m>}&\eqdef&
		\lim_{\delta\to0}\sI( |Z_\delta|^2Z_\delta-2a_\delta Z_\delta), & & \cZ^{\<30mc>}&\eqdef&\lim_{\delta\to0}\sI(|Z_\delta|^2\overline Z_\delta-2a_\delta \overline{Z}_\delta),
	\end{array}
\end{equs}
where $a_\delta=\mathbf{E}[ |Z_{\delta}|^2(t)]\sim \frac1\delta$  is a divergent constant independent of $t$ due to stationarity. In the remainder equation for $\Phi-Z$ the most irregualr part is the third Wick power and we need to continue with the decomposition in the same spirit to cancel it (see \eqref{eq:decphi} below).

We further introduce the corresponding $\eps$ dependent stochastic objects from the nonlinear term of equation \eqref{eq:mainnew}:
\begin{equs}[e:wick-tilde-eps]
	\begin{array}{lllllll}
		\cZ_\eps^{\<2c>}&\eqdef&
		(v_\eps*\overline Z) Z-a^\eps,
		& &
		\cZ_\eps^{\<2v>}&\eqdef&
		v_\eps* Z Z,
		\\ \cZ_\eps^{\<2vc>}&\eqdef&
		(v_\eps* \overline Z) \overline Z& &	\cZ^{\<2vm>}_\eps&\eqdef&
		v_\eps* \cZ^{\<2m>}, \\
		\cZ^{\<20vm>}_\eps&\eqdef&\sI( \cZ^{\<2vm>}),
		& &
		\cZ_\eps^{\<2c0>}&\eqdef&\sI(\cZ_\eps^{\<2c>}), \\ \cZ_\eps^{\<2v0>}&\eqdef&\sI(\cZ_\eps^{\<2v>}),
		& &
		\cZ_\eps^{\<2vc0>}&\eqdef&\sI(\cZ_\eps^{\<2vc>}),
	\end{array}
\end{equs}
and
\begin{equs}
	\cZ_\eps^{\<3vm>}\eqdef&
	\cZ^{\<2vm>}_\eps Z-(v^\eps G)*Z, \qquad \cZ_\eps^{\<3v0m>}\eqdef\sI(\cZ_\eps^{\<3vm>}),
\end{equs}
where we recall $a^\eps=\E [v_\eps*\overline Z Z]$ and  $G(x-y)=\E Z(x)\overline Z(y)$.

The renormalization described in \eqref{e:wick-tilde} is insufficient for handling the dynamical $\Phi^4_3$ model \eqref{eq:mainnew1}. Even further expansions do not completely resolve the issue, as ill-defined products persist. To address this, we employ paracontrolled calculus, which requires the construction of the following stochastic objects:
\begin{equs}[sto:3]
	\begin{array}{llllll}
		\cZ^{\<22oc>}&\eqdef&\lim_{\delta\to0}(\sI(\cZ^{\<2m>}_\delta)\circ \cZ^{\<2m>}_\delta-\tilde{b}_\delta),&  \cZ^{\<22oc1>}&\eqdef&\lim_{\delta\to0}(\sI(\cZ^{\<2>}_\delta)\circ \cZ^{\<2m>}_\delta),
		\\ \cZ^{\<22oc2>}&\eqdef&\lim_{\delta\to0}(\sI(\overline\cZ^{\<2>}_\delta)\circ \cZ^{\<2>}_\delta-2\tilde{b}_\delta), &\cZ^{\<22oc3>}&\eqdef&\lim_{\delta\to0}(\sI(\cZ^{\<2m>}_\delta)\circ \cZ^{\<2>}_\delta),
		\\ \cZ^{\<31oc>}&\eqdef&\lim_{\delta\to0}\cZ^{\<30m>}_\delta\circ Z_\delta,&  \cZ^{\<31oc3>}&\eqdef&\lim_{\delta\to0}\cZ^{\<30mc>}_\delta\circ Z_\delta,\\
		\cZ^{\<31oc4>}&\eqdef&\lim_{\delta\to0}\cZ^{\<30m>}_\delta\circ \overline{Z}_\delta,& &
		\\ \cZ^{\<32oc>}&\eqdef&\lim_{\delta\to0}(\cZ^{\<30m>}_\delta\circ \cZ^{\<2m>}_\delta-2\tilde{b}_\delta Z_\delta),&  \cZ^{\<32oc3>}&\eqdef&\lim_{\delta\to0}(\cZ^{\<30mc>}_\delta\circ \cZ^{\<2>}_\delta-2\tilde{b}_\delta Z_\delta),
	\end{array}
\end{equs}
with $\tilde{b}_\delta=\E (\sI(\cZ^{\<2m>}_\delta)\circ \cZ^{\<2m>}_\delta)$. Here $\tilde{b}_\delta$ depends on $t$ and satisfies
$\sup_{t\geq0} |\tilde{b}_\delta(t)|\lesssim |\log \delta|$.

Before proceeding we decompose $6b^\eps$ in \eqref{eq:counterterms} in the following result.

\bl\label{lem:b0} It holds that
\begin{align*}
	6b^\eps &=b_1^\eps+2b_2^\eps+2b_3^\eps+b_4^\eps ,\end{align*}
with
\begin{align*}
	\begin{array}{llllll}
		b_1^\eps&\eqdef&\sum_{k_1,k_2\in\bZ^3}\widehat{v^\eps}(k_1+k_2)^2 b(k_1,k_2),& b_2^\eps&\eqdef&\sum_{k_1,k_2\in\bZ^3}\widehat{v^\eps}(k_1+k_2)\widehat{v^\eps}(k_1)b(k_1,k_2),\\
		b_3^\eps&\eqdef&\sum_{k_1,k_2\in\bZ^3}\widehat{v^\eps}(k_1)^2b(k_1,k_2),& b_4^\eps&\eqdef&\sum_{k_1,k_2\in\bZ^3}\widehat{v^\eps}(k_1)\widehat{v^\eps}(k_2)b(k_1,k_2),
	\end{array}
\end{align*}
and
\begin{equs}[def:bk1k2]
	b(k_1,k_2) \eqdef\frac1{(2\pi)^{6}\la k_1\ra^2\la k_2\ra^2(\la k_1\ra^2+\la k_2\ra^2+\la k_1+k_2\ra^2)}.
\end{equs}
\el
\begin{proof}
	Let \begin{equs}[def-F]F(k_1,k_2)\eqdef&\,\widehat{v^\eps}(k_1+k_2)^2+\widehat{v^\eps}(k_1+k_2)\widehat{v^\eps}(k_1)
		\\&+\widehat{v^\eps}(k_1+k_2)\widehat{v^\eps}(k_2)+\widehat{v^\eps}(k_1)^2+\widehat{v^\eps}(k_2)^2+\widehat{v^\eps}(k_2)\widehat{v^\eps}(k_1).\end{equs}
	Using symmetry we obtain
	\begin{equs}[eq:b-s]
		6b^{\eps}=&\,\frac13\sum_{k_1,k_2} \frac{F(k_1,k_2)}{(2\pi)^6\la k_1\ra^2\la k_2\ra^2\la k_1+k_2\ra^2}
		\\=&\,\frac13\sum_{k_1,k_2}\frac{F(k_1,k_2)}{(2\pi)^6(\la k_1\ra^2+\la k_2\ra^2+\la k_1+k_2\ra^2)}
		\\&\bigg(\frac1{\la k_1\ra^2\la k_1+k_2\ra^2}+\frac1{\la k_2\ra^2\la k_1+k_2\ra^2}+\frac1{\la k_1\ra^2\la k_2\ra^2}\bigg)
		\\=&\,\sum_{k_1,k_2}\frac{F(k_1,k_2)}{(2\pi)^6(\la k_1\ra^2+\la k_2\ra^2+\la k_1+k_2\ra^2)\la k_1\ra^2\la k_2\ra^2}.
	\end{equs}
	The RHS is just $b_1^\eps+2b_2^\eps+2b_3^\eps+b_4^\eps$ by symmetry and	the result then follows.
\end{proof}
We further introduce the corresponding $\eps$-dependent  renormalization terms:
\begin{equs}[sto:4]
	\begin{array}{llllllllll}
		\cZ_\eps^{\<22vm>}\eqdef&\cZ_\eps^{\<20vm>}\circ \cZ_\eps^{\<2vm>}-{\tilde b_1^\eps},&
		\cZ_\eps^{\<22cvm>}\eqdef&\cZ^{\<2c0>}_\eps\circ \cZ^{\<2vm>}_\eps-{\tilde b_2^\eps},&
		\cZ_\eps^{\<22ccvm>}\eqdef&\cZ^{\<2v0>}_\eps\circ \cZ^{\<2vm>}_\eps,
	\end{array}
\end{equs}
use the condition $\E ZZ = 0$ to ensure that several renormalization constants vanish. Specifically, we define $\tilde b_1^\eps = \E \cZ_\eps^{\<20vm>} \circ \cZ_\eps^{\<2vm>}$ and $\tilde b_2^\eps = \E \cZ_\eps^{\<2c0>} \circ \cZ_\eps^{\<2vm>}$. Additionally, we introduce
$\tilde b_4^\eps=\E \cZ_\eps^{\<2c0>}\circ \cZ_\eps^{\<2c>}, \tilde{b}_5^\eps=\E \cZ_\eps^{\<2vc0>}\circ \cZ_\eps^{\<2v>},\tilde b_3^\eps=\tilde b_5^\eps-\tilde b_4^\eps,$
which will be used below.
Note that 	$\tilde b_i^\eps$ is given as $b_i^\eps$ in Lemma \ref{lem:b0} with $b(k_1,k_2)$ in Lemma \ref{lem:b0} replaced by $\tilde b_t(k_1,k_2)$
\begin{align}\label{def:tib}
	\tilde b_t(k_1,k_2)\eqdef\, b(k_1,k_2)(1-e^{-t(\la k_1\ra^2+\la k_2\ra^2+\la k_1+k_2\ra^2)}).
\end{align}
We have for $i=1,2,3, 4, $ and $\kappa>0$
$$\sup_{\eps>0}|\tilde{b}_i^\eps(t)-b_i^\eps|\lesssim t^{-\kappa}.$$
We also introduce the following $\eps$-dependent renormalization terms:
\begin{equs}[sto:5]
	\begin{array}{llllll}
		\cZ_\eps^{\<3v0m1ci>}&\eqdef&(v^\eps*\cZ_\eps^{\<3v0m>})\circ Z,& \cZ_\eps^{\<3v0m1cci>}&\eqdef&(v^\eps*\overline{\cZ}_\eps^{\<3v0m>})\circ Z,\\
		\cZ_\eps^{\<3v0m1ci2>}&\eqdef&\cZ^{\<3v0m>}_\eps\circ \overline Z,& \cZ_\eps^{\<3v0m1cci2>}&\eqdef&\overline{\cZ}_\eps^{\<3v0m>}\circ {Z},\\
		\cZ_\eps^{\<3v0m1ci10>}&\eqdef&\cZ^{\<3v0m>}_\eps\circ v^\eps*\overline{Z},& \cZ_\eps^{\<3v0m1cci10>}&\eqdef&\cZ^{\<3v0m>}_\eps\circ v^\eps*Z,
		\\
		\cZ^{\<3v02vm>}_\eps&\eqdef& \cZ^{\<3v0m>}_\eps\circ \cZ_\eps^{\<2vm>}-{(\tilde b_1^\eps+\tilde b_2^\eps) Z}.&
	\end{array}
\end{equs}
We also set
\begin{align*}
	\cZ^{\<3v02vmyn>}_\eps	\eqdef\cZ^{\<3v02vm>}_\eps+\cZ_\eps^{\<3v0m>}(\prec+\succ)\cZ_\eps^{\<2vm>}.
\end{align*}
The nonlinearity in equation \eqref{eq:mainnew} involves convolution with $v^\eps$. To apply paracontrolled calculus to this term, we introduce a shift operator defined for $y \in \mathbb{R}^3$ as follows:
\begin{equs}[def:tauy]
	\tau_y f\eqdef f(\cdot-y),
\end{equs}
and express $v^\eps*f=\int v(y)\tau_yf\dif y$.
Using this shift operator, we define the following random fields that depend on $y$:
\begin{equs}[sto:ny0]
	\cZ^{\<2m>}_y\eqdef&\Wick{\tau_y \overline ZZ}\eqdef\lim_{\delta\to0}\big(\tau_y \overline Z_\delta Z_\delta-\E[\tau_y \overline Z_\delta Z_\delta]\big)= \lim_{\delta\to0}\big(\tau_y \overline Z_\delta Z_\delta -G_\delta(y)\big),\\
	\cZ^{\<2>}_y\eqdef &\Wick{\tau_y  ZZ}\eqdef\lim_{\delta\to0}\big(\tau_y  Z_\delta Z_\delta\big), \qquad \cZ^{\<2yc>}_y\eqdef \overline{\cZ}^{\<2>}_y\\
	\cZ^{\<20vm1>}_y\eqdef & \sI(\cZ^{\<2m>}_y), \qquad \cZ^{\<20vm2>}_y\eqdef \sI(\cZ^{\<2>}_y),\qquad \cZ^{\<20vm2y>}_y\eqdef \sI(\cZ^{\<2yc>}_y).
\end{equs}
with the limit in $L^p(\Omega;C_T\bC^{-1-\kappa})$ for $p\geq1, \kappa>0$. Here $G_\delta$ is $G$ mollified with a bump function. As in \eqref{sto:3} we also define the following $y$-dependent  random fields and the related renormalization counterterms:
\begin{equs}[sto:ny2]
	\begin{array}{llllll}
		
		\cZ_{y,\eps}^{\<3v0m1ci>}&\eqdef&(\tau_y\cZ_\eps^{\<3v0m>})\circ Z,& \cZ_{y,\eps}^{\<3v0m1cci>}&\eqdef&(\tau_y\overline{\cZ}_\eps^{\<3v0m>})\circ Z,\\
		\cZ_{y,\eps}^{\<3v0m1ci2>}&\eqdef&\cZ^{\<3v0m>}_\eps\circ \tau_y\overline Z,& \cZ_{y,\eps}^{\<3v0m1cci2y>}&\eqdef&\cZ_\eps^{\<3v0m>}\circ \tau_y{Z},\\

		\cZ_{y,\eps}^{\<22vmy>}&\eqdef&\cZ^{\<20vm1>}_y\circ\cZ^{\<2vm>}_\eps-\tilde b_2^\eps(t,y),& \cZ_{y,\eps}^{\<22vmy1>}&\eqdef&\cZ^{\<20vm2>}_y\circ\cZ^{\<2vm>}_\eps,
		\\
		\cZ^{\<22ocy>}_{y,\eps}&\eqdef&\tau_y\cZ^{\<20vm>}_\eps\circ \cZ^{\<2m>}_y-\tilde b_{2}^\eps(t,y),&\cZ^{\<22ocy1>}_{y,\eps}&\eqdef&\tau_y\cZ^{\<20vm>}_\eps\circ \cZ^{\<2>}_y,
		\\
		\cZ^{\<22oc1>}_{y,y_1}&\eqdef&\tau_y \cZ^{\<20vm2>}_{y_1}\circ \cZ^{\<2m>}_y,&	\cZ^{\<22oc>}_{y,y_1}&\eqdef&\tau_y \cZ^{\<20vm1>}_{y_1}\circ \cZ^{\<2m>}_y-\tilde b_3(t,y,y_1),\\
		\cZ^{\<22oc3>}_{y,y_1}&\eqdef&\tau_y\overline{\cZ}^{\<20vm1>}_{y_1}\circ \cZ^{\<2>}_y,&	\cZ^{\<22oc2>}_{y,y_1}&\eqdef&\tau_y\cZ^{\<20vm2y>}_{y_1}\circ \cZ^{\<2>}_y-\tilde b_5(t,y,y_1),
	\end{array}
\end{equs}
where the RHS of the last two lines can be defined as in \eqref{sto:3} via probabilistic estimates and
\begin{equs}[def:tby]
	\tilde b_2^\eps(t,y)\eqdef&\,\E[\tau_y\cZ^{\<20vm1>}\circ\cZ^{\<2vm>}]=\sum_{k_1,k_2\in\mathbb{Z}^3}\widehat{v^\eps}(k_1+k_2)b_t(k_1,k_2)e^{-\imath k_1y},\\
	\tilde b_3(t,y,y_1)\eqdef&\,\E[\tau_y \cZ^{\<20vm1>}_{y_1}\circ \cZ^{\<2m>}_y]= \sum_{k_1,k_2\in\mathbb{Z}^3} b_t(k_1,k_2)e^{-\imath k_1(y+y_1)},\\
	\tilde b_5(t,y,y_1)\eqdef&\,\E[\tau_y\cZ^{\<20vm2y>}_{y_1}\circ\cZ^{\<2>}_y]=\sum_{k_1,k_2\in\mathbb{Z}^3}b_t(k_1,k_2)(e^{-\imath (k_1y+k_2y_1)}+e^{-\imath k_1(y+y_1)}).
\end{equs}
Here $\tilde b_3(y,y_1)$ and $\tilde b_5(y,y_1)$ can be viewed as $L^2$-functions w.r.t. the variables $y,y_1$.
We also introduce the following random fields through $y$-dependent random fields:
\begin{equs}[sto:ny1]
	\begin{array}{llllll}
		\mathbf{Z}^{\<2v2v>}_\eps&\eqdef&	\int v^\eps(y)\tau_y\cZ^{\<20vm>}_\eps\circ \cZ^{\<2m>}_y\dif y-\tilde b_2^\eps,&
		\mathbf{Z}^{\<2cvr2v>}_\eps&\eqdef&\int v^\eps(y)v^\eps(y_1)\tau_y\cZ^{\<20vm2>}_{y_1} \circ \cZ^{\<2m>}_y\dif y\dif y_1,\\
		\mathbf{Z}^{\<2cv2v>}_\eps&\eqdef&\int v^\eps(y)v^\eps(y_1)\tau_y\cZ^{\<20vm1>}_{y_1}\circ \cZ^{\<2m>}_y\dif y\dif y_1-\tilde b_3^\eps,&
		\mathbf{Z}^{\<2v2vr>}_\eps&\eqdef&	\int v^\eps(y)\tau_y\cZ^{\<20vm>}_\eps\circ \cZ^{\<2>}_y\dif y,\\
		\mathbf{Z}^{\<2cv2vr>}_\eps&\eqdef&\int v^\eps(y)v^\eps(y_1)\tau_y\overline{\cZ}^{\<20vm1>}_{y_1}\circ \cZ^{\<2>}_y\dif y\dif y_1,&
		\mathbf{Z}^{\<2cv2v0>}_\eps&\eqdef&\int v^\eps(y)v^\eps(y_1)\tau_y\cZ^{\<20vm2y>}_{y_1}\circ \cZ^{\<2>}_y\dif y\dif y_1-\tilde b_5^\eps.
	\end{array}
\end{equs}

For $y$-dependent random fields in \eqref{sto:ny2} and $\eps$-dependent random fields in \eqref{sto:4}, \eqref{sto:5} and \eqref{sto:ny1} we have
\begin{equs}[sto:re]
	\begin{array}{llllll}
		\cZ^{\<22cvm>}_\eps&=&\int v^\eps(y)	\cZ^{\<22vmy>}_{y,\eps}\dif y, & \cZ_\eps^{\<22ccvm>}&=&\int v^\eps(y)\cZ_{y,\eps}^{\<22vmy1>}\dif y,\\
		\cZ^{\<3v0m1ci>}_\eps&=&\int v^\eps(y)\cZ^{\<3v0m1ci>}_{y,\eps}\dif y,& \cZ_{\eps}^{\<3v0m1cci>}&=&\int v^\eps(y)\cZ_{y,\eps}^{\<3v0m1cci>}\dif y,\\
		\cZ_\eps^{\<3v0m1ci10>}&=&\int v^\eps(y)	\cZ_{y,\eps}^{\<3v0m1ci2>}\dif y,& \cZ_\eps^{\<3v0m1cci10>}&=&\int v^\eps(y)	\cZ_{y,\eps}^{\<3v0m1cci2>}\dif y,\\
		\mathbf{Z}^{\<2v2v>}_\eps&=&\int v^\eps(y)	\cZ^{\<22ocy>}_{y}\dif y,&	\mathbf{Z}^{\<2v2vr>}_\eps&=&\int v^\eps(y)	\cZ^{\<22ocy1>}_{y}\dif y,\\
		\mathbf{Z}^{\<2cvr2v>}_\eps&=&\int v^\eps(y)v^\eps(y_1)\cZ^{\<22oc1>}_{y,y_1,\eps}\dif y\dif y_1,
		&\mathbf{Z}^{\<2cv2v>}_\eps&=&\int v^\eps(y)v^\eps(y_1)\cZ^{\<22oc>}_{y,y_1,\eps}\dif y\dif y_1,
		\\
		\mathbf{Z}^{\<2cv2vr>}_\eps&=&\int v^\eps(y)v^\eps(y_1)\cZ^{\<22oc3>}_{y,y_1}\dif y\dif y_1,&
		\mathbf{Z}^{\<2cv2v0>}_\eps&=&\int v^\eps(y)v^\eps(y_1)\cZ^{\<22oc2>}_{y,y_1,\eps}\dif y\dif y_1,
	\end{array}
\end{equs}
and
\begin{equs}[sto:re1]
	\int v^\eps(y)\tilde b_2^\eps(t,y)\dif y=\tilde b_2^\eps(t),\qquad\int v^\eps(y)v^\eps(y_1)\tilde b_i(t,y,y_1)\dif y\dif y_1=\tilde b_i^\eps(t),
\end{equs}
for $i=3, 5$.

We then recall the following moments bounds for the stochastic terms from \cite{CC15}, \cite{Hos}. To this end we also introduce $|\tau|$ for every $\cZ^\tau$ in the following table.

\begin{table}[!htbp]
	\begin{tabular}{c|c|c|c|c|c|c|c}
		\toprule
		$\cZ^\tau$  &  $Z$ &$\cZ^{\<2m>}$ & $\cZ^{\<2>}$& $\cZ^{\<30mc>}$& $\cZ^{\<30m>}$& $\cZ^{\<22oc>}$& $\cZ^{\<22oc1>}$
		\\ \hline
		$|\tau|$ & $-\frac12-\frac\kappa2$ & $-1-\kappa$ & $-1-\kappa$& $\frac12-\kappa$& $\frac12-\kappa$ &$-\kappa$ &$-\kappa$
		\\ \hline \hline
		$\cZ^\tau$  & $\cZ^{\<22oc2>}$& $\cZ^{\<22oc3>}$& $\cZ^{\<31oc>}$& $\cZ^{\<31oc3>}$& $\cZ^{\<31oc4>}$& $\cZ^{\<32oc>}$& $\cZ^{\<32oc3>}$
		\\ \hline
		$|\tau|$ & $-\kappa$& $-\kappa$& $-\kappa$& $-\kappa$& $-\kappa$& $-\frac12-\kappa$& $-\frac12-\kappa$
		\\\midrule
		\bottomrule
	\end{tabular}
	\caption{Regularity of stochastic objects: $\cZ^\tau$ }\label{ta:1}
\end{table}

\bl\label{lem:bdZ} It holds that for every $p\geq 1, \kappa>0$
and every $\cZ^\tau$ from Table \ref{ta:1}
\begin{align*}
	\E \|\cZ^{\tau}\|_{C_T\bC^{|\tau|}}^p+\E\|\cZ^{\<30m>}\|^p_{C_T^{\frac18}L^\infty}\lesssim 1.
\end{align*}
The convergence in \eqref{e:wick-tilde} and \eqref{sto:3} hold in $L^p(\Omega;C_T\bC^{|\tau|})$.
\el

We use $\mZ=(\cZ^\tau)$ for the tree $\tau$ appear in Lemma \ref{lem:bdZ}   and $\|\mZ\|$ to denote the smallest number bigger than $1$ and  $C_T\bC^{|\tau|}$-norm of  $\cZ^\tau$ from Table \ref{ta:1} and $C_T^{\frac18}L^\infty$-norm of $\cZ^{\<30m>}$.

Using probabilistic methods analogous to those in \cite{CC15, ZZ18}, we establish the convergence of the corresponding renormalized terms. For clarity, we summarize the renormalized quantity $\cZ_\eps^{\tau_\eps}$ and its limiting counterpart $\cZ^\tau$ in the  table below.

\begin{table}[!htbp]
	\begin{tabular}{c|c|c|c|c|c|c|c|c}
		\toprule
		$\cZ_\eps^{\tau_\eps}$   &$\cZ_\eps^{\<2c>}$ &$\cZ_\eps^{\<2vm>}$ & $\cZ^{\<2v>}$& $\cZ^{\<2vc>}$&  $\cZ_\eps^{\<3v0m>}$ &  $\cZ_\eps^{\<22vm>}$ & $\cZ_\eps^{\<22cvm>}$& $\cZ_\eps^{\<22ccvm>}$
		\\ \hline
		$\cZ^{\tau}$  &$\cZ^{\<2m>}$ & $\cZ^{\<2m>}$& $\cZ^{\<2>}$ & $\overline{\cZ}^{\<2>}$ & $\cZ^{\<30m>}$&  $\cZ^{\<22oc>}$  &$\cZ^{\<22oc>}$& $\cZ^{\<22oc1>}$
		\\ \hline \hline
		$\cZ^{\tau_\eps}_\eps$  &     $\cZ_\eps^{\<3v0m1ci>}$ &$ \cZ_\eps^{\<3v0m1cci>}$& $\cZ_\eps^{\<3v0m1ci10>}$ &$ \cZ_\eps^{\<3v0m1cci10>}$& $\cZ_\eps^{\<3v0m1ci2>}$& $\cZ_\eps^{\<3v0m1cci2>}$ &$\cZ^{\<3v02vm>}_\eps$
		\\ \hline
		$\cZ^\tau$ & $\cZ^{\<31oc>}$& $\cZ^{\<31oc3>}$& $\cZ^{\<31oc4>}$&  $\cZ^{\<31oc3>}$& $\cZ^{\<31oc4>}$&$\cZ^{\<31oc3>}$ &$\cZ^{\<32oc>}$ 
		\\\midrule
		\bottomrule
	\end{tabular}
	\caption{Convergence of stochastic objects $\cZ^{\tau_\eps}_\eps$ to $\cZ^\tau$ }\label{ta:2}
\end{table}

\bp\label{prop:Z} For every $\cZ^{\tau_\eps}_\eps$ and corresponding $\cZ^\tau$ in Table \ref{ta:2} it holds that for every $\kappa>0$ and $p\geq1$
\begin{align*}
	\E\|\cZ^{\tau_\eps}_\eps-\cZ^\tau\|_{C_T\bC^{|\tau|-\kappa}}^p+\E\|\cZ_\eps^{\<3v0m>}-\cZ^{\<30m>}\|^p_{C_T^{\frac18}L^\infty}\lesssim \eps^{\kappa p},
\end{align*}
and
\begin{align*}
	\E\|\cZ^{\tau_\eps}_\eps\|_{C_T\bC^{|\tau|}}^p+\E\|\cZ_\eps^{\<3v0m>}\|^p_{C_T^{\frac18}L^\infty}\lesssim 1,
\end{align*}
with the proportional constant independent of $\eps$.
\ep
\begin{proof}
	The proof  follows the same probabilistic calculations as in \cite{CC15, ZZ18}, utilizing the chaos expansion of the stochastic terms, Gaussian hypercontractivity and the Besov embedding Lemma \ref{lem:emb}. The only difference is the appearance of the additional term $ v^\eps $, which introduces a Fourier multiplier $ \hat{v}^\eps(k) \to 1 $ as $ \eps \to 0 $.  We replace $1$ with $\widehat{v^\eps}(k)$, and by applying $|\widehat{v^\eps}(k)-1|\lesssim \eps^\kappa |k|^\kappa$ we derive the result using the same approach.
\end{proof}

We collect the renormalized quantity $\mathbf{Z}_\eps^{\tau_\eps}$ and the corresponding limit $\cZ^\tau$ in the following table.

\begin{table}[!htbp]
	\begin{tabular}{c|c|c|c|c|c|c}
		\toprule
		$\cZ_\eps^{\tau_\eps}$   & $\mathbf{Z}_\eps^{\<2v2v>}$ & $\mathbf{Z}_\eps^{\<2cv2v>}$& $\mathbf{Z}_\eps^{\<2cvr2v>}$&$\mathbf{Z}_\eps^{\<2cv2v0>}$&
		$\mathbf{Z}_\eps^{\<2cv2vr>}$  &    $\mathbf{Z}_\eps^{\<2v2vr>}$
		\\ \hline
		$\cZ^{\tau}$   &  $\cZ^{\<22oc>}$&  $\cZ^{\<22oc>}$ & $\cZ^{\<22oc1>}$ &  $\cZ^{\<22oc2>}$
		&$\cZ^{\<22oc3>}$ &  $\cZ^{\<22oc3>}$
		\\\midrule
		\bottomrule
	\end{tabular}
	\caption{Convergence of stochastic objects $\mathbf{Z}^{\tau_\eps}_\eps$ to $\cZ^\tau$ }\label{ta:3}
\end{table}

\bp\label{prop:Zny} For every $\mathbf{Z}^{\tau_\eps}_\eps$ and corresponding $\cZ^\tau$ in Table \ref{ta:3} it holds that for every $\kappa>0$ and $p\geq1$
\begin{align*}
	\E\|\mathbf{Z}^{\tau_\eps}_\eps-\cZ^\tau\|_{C_T\bC^{|\tau|-\kappa}}^p\lesssim \eps^{\kappa p},
\end{align*}
and
\begin{align*}
	\E\|\mathbf{Z}^{\tau_\eps}_\eps\|_{C_T\bC^{|\tau|}}^p+\E\sup_{y\in\mathbb{R}^3}\|\cZ^{\tau}_y\|_{C_T\bC^{-1-\kappa}}^p\lesssim 1,
\end{align*}
for $\cZ^{\tau}_y\in\{\cZ^{\<2m>}_y, \cZ^{\<2>}_y\}$ and
\begin{align*}\E\sup_{y\in\mR^3}\|\cZ^{\tau_\eps}_{y,\eps}\|^p_{C_T\bC^{-\kappa}}+\E\sup_{y,y_1\in\mR^3}\|\cZ^{\tau}_{y,y_1}\|^p_{C_T\bC^{-\kappa}}\lesssim 1,
\end{align*}
for $\cZ^{\tau_\eps}_{y,\eps}\in\{\cZ_{y,\eps}^{\<22vmy>},  \cZ_{y,\eps}^{\<22vmy1>},
\cZ^{\<22ocy>}_{y,\eps}, \cZ^{\<22ocy1>}_{y,\eps},\cZ_{y,\eps}^{\<3v0m1ci>},\cZ_{y,\eps}^{\<3v0m1cci>},\cZ_{y,\eps}^{\<3v0m1ci2>},\cZ_{y,\eps}^{\<3v0m1cci2y>}\}$  and $\cZ^{\tau}_{y,y_1}\in\{ \cZ^{\<22oc1>}_{y,y_1}, 	\cZ^{\<22oc>}_{y,y_1},
\cZ^{\<22oc3>}_{y,y_1},\cZ^{\<22oc2>}_{y,y_1}\}$,
where the proportional constants are independent of $\eps$. Moreover, we have
\begin{align*} \E\sup_{y\in\mathbb{R}^3}\Big(\frac{\|\cZ^{\<2m>}_y-\cZ^{\<2m>}\|_{C_T\bC^{-1-2\kappa}}+\|\cZ^{\<2>}_y-\cZ^{\<2>}\|_{C_T\bC^{-1-2\kappa}}}{|y|^{\kappa}}\Big)^p\lesssim 1.\end{align*}	
\ep
\begin{proof}
	The convergence results follow exactly as in Proposition \ref{prop:Z}. For further details, we also refer to the proof of Proposition \ref{thcomm1} below. Regarding the uniform bounds in $y$ for $\cZ^{\tau}_y$, $\cZ^{\tau\eps}_{y,\eps}$, and $\cZ^{\tau}_{y,y_1}$, we can use the estimate
	\begin{equs}[eq:dify]
		|e^{-\imath k \cdot y_1}-e^{-\imath k\cdot y_2}|\lesssim |y_1-y_2|^\lambda |k|^{\lambda},\end{equs}
	for $0<\lambda<1$, and apply Kolmogorov's continuity criterion to conclude the result for $y \in \mathbb{T}^3$. Since $\cZ^{\tau}_y$, $\cZ^{\tau\eps}_{y,\eps}$, and $\cZ^{\tau}_{y,y_1}$ are periodic in $y$, we extend this result to $y \in \mathbb{R}^3$. For the final result, we also use \eqref{eq:dify} to deduce that it holds for $\sup_{y \in [-\pi, \pi]^3}$ by the same argument, which can be easily extended to $y \in \mathbb{R}^3$, thereby implying the last result.
\end{proof}

We set
$$\cZ_{y,\eps}=(\cZ^{\tau}_y,\cZ^{\tau_\eps}_{y,\eps}, \cZ^{\tau}_{y,y_1} ),$$
for  $\cZ^{\tau}_y\in\{\cZ^{\<2m>}_y, \cZ^{\<2>}_y\}$ and
$\cZ^{\tau_\eps}_{y,\eps}\in\{\cZ_{y,\eps}^{\<22vmy>},  \cZ_{y,\eps}^{\<22vmy1>},
\cZ^{\<22ocy>}_{y,\eps}, \cZ^{\<22ocy1>}_{y,\eps}, \cZ_{y,\eps}^{\<3v0m1ci>},\cZ_{y,\eps}^{\<3v0m1cci>},\cZ_{y,\eps}^{\<3v0m1ci2>},\cZ_{y,\eps}^{\<3v0m1cci2y>}\}$ and $\cZ^{\tau}_{y,y_1}\in\{ \cZ^{\<22oc1>}_{y,y_1}, 	\cZ^{\<22oc>}_{y,y_1},$
$\cZ^{\<22oc3>}_{y,y_1},\cZ^{\<22oc2>}_{y,y_1}\}$. We also
set $\mZ_\eps=(\cZ_\eps^{\tau_\eps}, \mathbf{Z}_\eps^{\tau_\eps}, \cZ_{y,\eps})$ for the tree $\tau_\eps$, which appears in Propositions  \ref{prop:Z}, \ref{prop:Zny},   and use $\|\mZ_\eps\|$ to denote the smallest number bigger than $1$ and  $C_T\bC^{|\tau|}$-norm of  $\cZ_\eps^{\tau_\eps}$ and $\mathbf{Z}_\eps^{\tau_\eps}$ from Table \ref{ta:2} and Table \ref{ta:3}. Additionally, $\|\mZ_\eps\|$ is also greater than $\sup_y \|\cZ_y^\tau\|_{C_T\bC^{-1-\kappa}}$  and $\sup_y \|\cZ^{\tau_\eps}_{y,\eps}\|_{C_T\bC^{-\kappa}}$  and $\sup_{y,y_1} \|\cZ^{\tau}_{y,y_1}\|_{C_T\bC^{-\kappa}}$.

In the following we  fix $\kappa>0$ small enough.

\subsection{Paracontrolled calculus for equations}
In this section we apply paracontrolled calculus for the dynamical $\Phi^4_3$ model \eqref{eq:mainnew1} and equation \eqref{eq:mainnew}.

\subsubsection{Paracontrolled calculus for equation \eqref{eq:mainnew1}}

At the level of approximation, equation \eqref{eq:mainnew1} can be viewed as the limiting case of the following equations as $\delta$ approaches $0$:
\begin{equs}[eq:phidelta]
	\LL \Phi_\delta= -(|\Phi_\delta|^2\Phi_\delta-(2a_\delta-6 b_\delta)\Phi_\delta)-(2C_1+2C_2+m-1)\Phi_\delta+\xi_\delta,
\end{equs}
where $a_\delta$ is given in Section \ref{sec:ren}, $b_\delta=\E (\int_{-\infty}^\cdot P_{\cdot-s}\cZ^{\<2m>}_\delta\dif s\circ \cZ^{\<2m>}_\delta)$ is a constant, and $\xi_\delta$ is the mollification of $\xi$ with a bump function. Since both $a_\delta$ and $b_\delta$ diverge to infinity as $\delta \to 0$, directly analyzing equation \eqref{eq:phidelta} becomes challenging. To address this, we decompose equation \eqref{eq:mainnew1} into an irregular part and a regular part as follows:
\begin{equs}[eq:decphi]
	\Phi=Z-\cZ^{\<30m>}+\phi,
\end{equs}
with $\phi$ solving the following equation
\begin{equs}[phi]
	\LL \phi=&\, 2 \cZ^{\<32oc1>}-2{\cZ^{\<2m>}\phi}-\cZ^{\<2>}\overline{\phi}- |-\cZ^{\<30m>}+\phi|^2Z-|-\cZ^{\<30m>}+\phi|^2(-\cZ^{\<30m>}+\phi)\\&+\cZ^{\<32oc2>}-2\Re[(-\cZ^{\<30m>}+\phi)\overline Z](-\cZ^{\<30m>}+\phi)
	\\&-(2C_1+2C_2+m-1+6b)(Z-\cZ^{\<30m>}+\phi),
\end{equs}
where
\begin{align*}
	\cZ^{\<32oc1>}\eqdef\cZ^{\<32oc>}+\cZ^{\<30m>}(\prec+\succ)
	\cZ^{\<2m>},\qquad 	\cZ^{\<32oc2>}\eqdef \cZ^{\<32oc3>}+\cZ^{\<30mc>}(\prec+\succ)
	\cZ^{\<2>},
\end{align*}
with $\cZ^\tau$ being random objects introduced in Section \ref{sec:ren}. From \eqref{sto:3} $\cZ^{\<32oc1>}$ and $\cZ^{\<32oc2>}$ evolve renormalization from $\tilde b_\delta$ and $$b(t)=\lim_{\delta\to0}( b_\delta-\tilde b_\delta(t))=\sum_{k_1,k_2}e^{-(|k_1+k_2|^2+|k_1|^2+|k_2|^2+3)t}b(k_1,k_2)\lesssim t^{-\kappa}$$
for any $\kappa>0$. Here $b(k_1,k_2)$ is given in \eqref{def:bk1k2}.  Note that the Wick renormalization part $2a_\delta \Phi_\delta$ from the approximation \eqref{eq:phidelta} has been incoporated into $\cZ^{\<30m>}$ and $ \cZ^{\<32oc1>}$, $\cZ^{\<2m>}$ in equation \eqref{phi}, whereas $6\tilde b_\delta \Phi_\delta$ appears in the definition of  $\cZ^{\<32oc>}$, $\cZ^{\<32oc3>}$ in \eqref{sto:3},  as well as in the definition of ${\cZ^{\<2m>}\circ\phi}$  and ${\cZ^{\<2>}\circ\overline\phi}$ below. Additionally, in \eqref{eq:phidelta} and the construction of the $\Phi^4_3$ field in \eqref{e:Phi-measure1} we use the renormalization constant $b_\delta$, while in the definition of $\cZ^{\<32oc>}$, $\cZ^{\<32oc3>}$ we use renormlization involving  $\tilde b_\delta$ to ensure that $\cZ^{\<32oc>}$, $\cZ^{\<32oc3>}$ stay in $C_T\bC^{-\frac12-\kappa}$ $\bP$-a.s.. As a result, we have an additional $b$ term in the last line of equation \eqref{phi}.

The main difficulty to understand equation \eqref{phi} comes from  ${\cZ^{\<2m>}\phi}$  and ${\cZ^{\<2>}\overline\phi}$. In fact these two terms are not well-defined in the classical sense as the expected sum of their regularities is not strictly positive for the resonant product ${\cZ^{\<2m>}\circ\phi}$  and ${\cZ^{\<2>}\circ\overline\phi}$ to be well-defined, cf. Lemma \ref{lem:para}. Collecting the terms which makes $\phi$ too irregular leads to the  following paraproduct ansatz
\begin{equs}[para:phi]
	\phi=&\,2(\cZ^{\<30m>}-\phi)\prec \cZ^{\<20vm1>}+({\cZ}^{\<30mc>}-\overline\phi)\prec\cZ^{\<20vm2>}
	+\phi^\sharp.
\end{equs}
Then $\phi^\sharp$ becomes more regular than $\phi$ since
\begin{align*}\phi^\sharp=&\,\phi+2\sI[(\phi-\cZ^{\<30m>})\prec \cZ^{\<2m>}]+\sI[ (\overline\phi-{\cZ}^{\<30mc>})\prec\cZ^{\<2>}]
	\\&-2[\sI,(\phi-\cZ^{\<30m>})\prec] \cZ^{\<2m>}-[\sI, (\overline\phi-{\cZ}^{\<30mc>})\prec]\cZ^{\<2>},
\end{align*}
where $[\sI,f\prec]g$ denotes the commutator between $\sI$ and $f\prec$ given by
$$[\sI,f\prec]g\eqdef \sI(f\prec g)-f\prec \sI(g).$$
Here the second and the third terms on the right hand side cancel the irregular terms in $\phi$ whereas the last two terms have better regularity by Lemma \ref{lem:com2}.
Using \eqref{para:phi} and   the  commutator Lemma \ref{lem:com1} we define $\cZ^{\<2m>}\circ \phi, \overline\phi \circ \cZ^{\<2>}$ as follows
\begin{equs}[def:phicirc]
	\phi\circ \cZ^{\<2m>}\eqdef &\,2[(\cZ^{\<30m>}-\phi)\prec \cZ^{\<20vm1>}]\circ\cZ^{\<2m>} +[({\cZ}^{\<30mc>}-\overline\phi)\prec\cZ^{\<20vm2>}]\circ \cZ^{\<2m>}+\phi^\sharp\circ \cZ^{\<2m>},\\
	\overline{\phi}\circ \cZ^{\<2>}\eqdef &\,2[(\cZ^{\<30mc>}-\overline{\phi})\prec \cZ^{\<20vm1>}]\circ\cZ^{\<2>} +[({\cZ}^{\<30m>}-\phi)\prec\overline{\cZ}^{\<20vm2>}]\circ \cZ^{\<2>}+\overline{\phi}^\sharp\circ \cZ^{\<2>},
\end{equs}
with
\begin{equs}[phi:circ]
	(	(\cZ^{\<30m>}-\phi)\prec \cZ^{\<20vm1>})\circ \cZ^{\<2m>}\eqdef &\,C(\cZ^{\<30m>}-\phi, \cZ^{\<20vm1>}, \cZ^{\<2m>})+(\cZ^{\<30m>}-\phi)  \cZ^{\<22oc>},\\
	(	({\cZ}^{\<30mc>}-\overline\phi)\prec \cZ^{\<20vm2>})\circ \cZ^{\<2m>}\eqdef&\, C({\cZ}^{\<30mc>}-\overline\phi, \cZ^{\<20vm2>}, \cZ^{\<2m>})+({\cZ}^{\<30mc>}-\overline\phi) \cZ^{\<22oc1>},
	\\ 	(	(\cZ^{\<30mc>}-\overline{\phi})\prec \cZ^{\<20vm1>})\circ \cZ^{\<2>}\eqdef &\,C(\cZ^{\<30mc>}-\overline{\phi}, \cZ^{\<20vm1>}, \cZ^{\<2>})+(\cZ^{\<30mc>}-\overline{\phi})\cZ^{\<22oc3>},\\
	(({\cZ}^{\<30m>}-\phi)\prec 	\overline{\cZ}^{\<20vm2>})\circ \cZ^{\<2>}\eqdef&\, C({\cZ}^{\<30m>}-\phi, \overline{\cZ}^{\<20vm2>}, \cZ^{\<2>})+({\cZ}^{\<30m>}-\phi)  \cZ^{\<22oc2>},
\end{equs}
where $C$ is the  commutator from Lemma \ref{lem:com1}.
Note that the stochastic terms $\cZ^{\<22oc>}$,  $\cZ^{\<22oc2>}$ defined in \eqref{sto:3} require the renormalization  $\tilde b_\delta$.

Let us now formulate the definition of paracontrolled solution to \eqref{phi}. We also note that $|-\cZ^{\<30m>}+\phi|^2Z$ and $2\Re[(-\cZ^{\<30m>}+\phi)\overline Z]\cZ^{\<30m>}$ also evolve terms not well-defined in the classical sense and will be defined in Lemma \ref{lem:conn1} and Lemma \ref{lem:dif1} below using stochastic objects introduced in \eqref{sto:3}. We introduce the following solution space: for a small number $\kappa>0$
\begin{align*}\sD=\{(\phi,\phi^\sharp)\in (C_T\bC^{-\frac12-\kappa})^2, &\sup_{t\in[0,T]}\|\phi(t)\|_{\bC^{-\frac12-\kappa}}+\sup_{t\in[0,T]}t^{\frac{1+3\kappa}{2}}\|\phi(t)\|_{\bC^{\frac{1}{2}+2\kappa}}
	\\&+\sup_{t\in[0,T]}t^{\frac{3+8\kappa}{4}}\|\phi^{\sharp}(t)\|_{\bC^{1+3\kappa}}<\infty\}.\end{align*}

\bd\label{def:sol:phi}We say that a pair $(\phi,\phi^\sharp)\in \sD$ is a
paracontrolled solution to equation \eqref{phi} provided
$\phi^\sharp$ given by \eqref{para:phi}
and $\phi$ satisfies \eqref{phi}  in the analytic weak sense with ${\cZ^{\<2m>}\circ\phi}$ and $\overline{\phi}\circ\cZ^{\<2>} $ given by \eqref{def:phicirc}.
\ed

Based on this definition we also call $\Phi$ is a solution to the dynamical $\Phi^4_3$ model if $\Phi$ satisfies \eqref{eq:decphi} with $\phi$ in \eqref{eq:decphi} is a solution to equation \eqref{phi} in the sense of Definition \ref{def:sol:phi}.

We recall the following result from \cite{CC15,MW18}.

\bt\label{th:phi4} For any $\Phi(0)\in \bC^{-\frac12-\kappa}$ there exists a global in time  unique solutions $(\phi,\phi^\sharp)\in \sD$ $\bP$-a.s. to equation \eqref{phi}.
\et

We take  solution $\phi$ from Theorem \ref{th:phi4} and define
\begin{align*}
	\Phi=Z-\cZ^{\<30m>}+\phi,
\end{align*}
which gives a unique solution to the dynamical $\Phi^4_3$ model \eqref{eq:mainnew1}.
According to \cite{HM18}, these solutions also constitute a Markov process within the space $\bC^{-\frac12-\kappa}$.
By leveraging the uniform estimates provided in \cite{MW18} or \cite{GH18a}, one can utilize the dynamical $\Phi^4_3$
model \eqref{eq:mainnew1} to offer a novel construction of the $\Phi^4_3$ field $\nu$.
This can be achieved through the Krylov-Bogoliubov method applied to the Markov semigroup or via lattice approximation (c.f. \cite{GH18a}). Additionally, we employ general ergodicity results from \cite{HM18, HS22} along with lattice approximation techniques from \cite{GH18a, HM18a, ZZ18} to establish the following conclusion.

\bt\label{th:phi43} There exists a unique invariant measure to  the dynamical $\Phi^4_3$ model \eqref{eq:mainnew1} given by $\nu$ from \eqref{e:Phi-measure1}.
\et

\subsubsection{Paracontrolled calculus for equation \eqref{eq:mainnew}}\label{sec:parapsi}

In equation \eqref{eq:mainnew} $\Wick{|\Psi^\eps|^2}$ denotes the Wick renormalization, which is understood as follows: More precisely, we decompose $\Psi^\eps=Z+u^\eps$,
where $u^\eps$ satisfies the  equation:
\begin{equs}[eq:mainn1]
	\LL u^\eps= -\big( v^\eps*(\cZ^{\<2m>}+2\Re(u^\eps\overline Z)+|u^\eps|^2)\big)(Z+u^\eps)+(a^\eps -6b^\eps-m+1)(Z+u^\eps).
\end{equs}
In this equation we interpret $\Wick{|\Psi^\eps|^2}$ as $\cZ^{\<2m>}+2\Re(u^\eps\overline Z)+|u^\eps|^2$, where the Wick renormalization counterpart is encapsulated in $\cZ^{\<2m>}$. For fixed $\eps$, due to the smoothing effect of $v$, we  obtain local well-posedeness of solutions to equation \eqref{eq:mainn1} and also \eqref{eq:mainnew}, similar to the approach used for the dynamical $\Phi^4_2$ model (c.f. \cite{DD03, MW17}).  We first state  the following result by using Bourgain's argument \cite{Bou}.

\bt\label{th:global} For fixed $\eps>0$ and any $\Psi^\eps(0)\in \bC^{-\frac12-\kappa}$  there exist a unique solution to equation \eqref{eq:mainnew} in $C([0,T];\bC^{-\frac12-\kappa})$ $\bP$-a.s. Moreover, $\nu^\eps$  from \eqref{e:Phi_eps-measure1} is  the unique invariant measure of the solutions to equation  \eqref{eq:mainnew}.
\et
The proof of this result is standard and we put it in appendix.

As we need to prove the solution to equation \eqref{eq:mainnew} converges to the solution of equation \eqref{eq:mainnew1}, paracontrolled calculus is required to further analyze equation \eqref{eq:mainnew}. To facilitate this, we need a more detailed decomposition of equation  \eqref{eq:mainn1}. Notably the renormalized random field $\cZ_\eps^{\<3vm>}$ for
$v^\eps*\cZ^{\<2m>}Z$ corresponds to the Wick renormalization  $(v^\eps G)*Z$, while in equation \eqref{eq:mainnew1} we transform this Wick renormalization to mass renormalization $a^\eps \Psi^\eps$, as explained in Remark \ref{C1C2e}. To this end, we rewrite equation \eqref{eq:mainnew} as
\begin{equs}[eq:mainn]
	\LL \Psi^\eps= -( v^\eps*\Wick{|\Psi^\eps|^2})\Psi^\eps+(v^\eps G)*\Psi^\eps+\cR\Psi^\eps-6b^\eps\Psi^\eps-(m-1)\Psi^\eps+\xi,
\end{equs}
with the operator $\cR$ defined by
\begin{equs}[def:cR]
	\cR f\eqdef a^\eps f-(v^\eps G)*f.
\end{equs}
Compared with \eqref{eq:main} we have an additional term $\cR \Psi^\eps$, which depends linearly on the solution. For fixed $\eps$, the operator $\cR f$ has the same regularity as $f$, but this regularity is dependent on $\eps$. The following lemma regarding $\cR f$ demonstrates that this operator uniformly reduces regularity by one in $\eps$.

\bl\label{comnew} Let $f\in \bB_p^\alpha,\alpha\in\mathbb{R},p\in[1,\infty]$. Then for $\eta\in (0,1)$
\begin{align*}
	\|\cR f\|_{\bB_p^{\alpha-1-\eta}}\lesssim \eps^{\eta}\|f\|_{\bB_p^{\alpha}}.
\end{align*}
Here the implicit constant is independent of $\eps$.
\el
We put the proof of this lemma in Subsection \ref{sec:comnew}.

To address the singular SPDE \eqref{eq:mainn} as well as the new term $\cR \Psi^\eps$, we introduce the following decomposition for equation \eqref{eq:mainnew}, omitting
$\eps$ in $\Psi$ and renormalized random fields for simplicity in the sequel:
\begin{equ}[eq:decPsi]
	\Psi=Z-\cZ^{\<3v0m>}+\sI(\cR Z)+\psi.
\end{equ}
Since $Z\in C_T\bC^{-\frac{1+\kappa}2}$, by Lemma \ref{comnew} we have $\cR Z\in C_T\bC^{-\frac32-\kappa}$ $\bP$-a.s. and $\sI(\cR Z)\to 0$ in $C_T\bC^{\frac12-\kappa}$ $\bP$-a.s..
We can then easily check that $\psi$ satisfies the following equation:
\begin{equs}[eq:main2new]
	\LL \psi=& \, \cZ^{\<3v02vmyn>}-\sI(\cR Z)\cZ^{\<2vm>}-\cZ^{\<2vm>}\psi-2\big(v^\eps*\Re[(Y+\psi)\Wick{\overline Z]\big)Z}
	\\&-6\tilde{b}^\eps(Y+\psi)-(\tilde{b}_2^\eps+2\tilde{b}_3^\eps+\tilde{b}_4^\eps)Z\\&-\big(v^\eps* |Y+\psi|^2\big)Z-\big(v^\eps*|Y+\psi|^2\big)(Y+\psi)-2\big(v^\eps*\Re[(Y+\psi)\overline Z]\big)(Y+\psi)
	\\&+\cR Y+\cR\psi+(6\tilde{b}^\eps-6b^\eps-m+1)\Psi,
\end{equs}
where $6\tilde{b}^\eps=\tilde{b}_1^\eps+2\tilde{b}_2^\eps+2\tilde{b}_3^\eps+\tilde{b}_4^\eps$ and $b^\eps$ given in \eqref{e:Phi_eps-measure}
\begin{equs}[def:Y]
	Y\eqdef-\cZ^{\<3v0m>}+\sI(\cR Z).
\end{equs}
Here and in subsequent discussions, we adopt the following notation for the Wick product  over $\overline ZZ$:
\begin{equs}[wickf]2\big(v^\eps*\Re[f:\overline Z]\big)Z:\eqdef2\big(v^\eps*\Re[f\overline Z]\big)Z-(v^\eps G)*f.\end{equs}
Regarding the renormalization counterterms $(v^\eps G)*\Psi$ and $6 b^\eps \Psi$ in equation \eqref{eq:mainn} we have
\begin{itemize}
	\item $(v^\eps G)*Z$ is incoporated into $\cZ^{\<3v0m>}$, while $(v^\eps G)*(Y+\psi)$ from \eqref{eq:mainn} is incoporated into  $2\big(v^\eps*\Re[(Y+\psi)\Wick{\overline Z]\big)Z}$, respectively;
	\item By \eqref{sto:5} $(\tilde{b}_1^\eps+\tilde{b}_2^\eps) Z$ is  incoporated into $\cZ^{\<3v02vmyn>}$ and $(\tilde{b}_2^\eps+2\tilde{b}_3^\eps+\tilde{b}_4^\eps)Z$ serves as the renomalization counterterm for ${2\big(v^\eps*\Re[Y\Wick{\overline Z]\big)Z}}$;
	\item The term $6\tilde{b}^\eps(Y+\psi)$ serves as the renormalization counterterm for $\cZ^{\<2vm>}\psi$
	and ${2\big(v^\eps*\Re[\psi\Wick{\overline Z]\big)Z}}$.
\end{itemize}

Using Lemma \ref{comnew} we know that $\sI(\cR Z)\in \bC^{\frac12-\kappa}$ $\bP$-a.s. uniform in $\eps$. Hence, we note that the convergence of $\sI(\cR Z)\circ Z$ and  several new stochastic terms invovling $\cR Z$  requires renormalization (c.f. Lemma \ref{lem:para}). We employ probabilistic calculations to prove that they converge to zero in suitable spaces, after subtracting suitable renormalization counterterms.
\bp\label{prop:Zn}
It holds that  for $\kappa>0$ and $p\geq1$
\begin{equs}[con:RZ]
	\E\|\big(v^\eps*\sI(\cR Z)\big)\circ Z\|^p_{C_T\bC^{-2\kappa}}+\E\| \big(v^\eps*\sI(\overline{\cR Z})\big)\circ Z- c^\eps_1\|^p_{C_T\bC^{-2\kappa}}\lesssim \eps^{p\kappa},
\end{equs}
and
\begin{equs}[con:RZ1]
	\E\|(v^\eps*\overline Z)\circ\sI(\cR Z)-c^\eps_1\|^p_{C_T\bC^{-2\kappa}}+\E\| (v^\eps* Z)\circ\sI(\cR Z)\|^p_{C_T\bC^{-2\kappa}}\lesssim \eps^{p\kappa},
\end{equs}
\begin{equs}[con:RZ2]
	\E\|\Re[\overline{Z}\circ\sI(\cR Z)]-c^\eps_2\|^p_{C_T\bC^{-2\kappa}}\lesssim \eps^{p\kappa},
\end{equs}
where \begin{equs}[c1c2]c_1^\eps\eqdef\E\big[\big(v^\eps*\sI(\overline{\cR Z})\big)\circ Z\big]=\E\big[(v^\eps*\overline Z)\circ\sI(\cR Z)\big],\qquad c_2^\eps\eqdef\E\Re[\overline{Z}\circ\sI(\cR Z)]
\end{equs} depending on $t$. Moreover,
$$|c_1^\eps(t)-C_1|\lesssim (t^{-\kappa}+1)\eps^\kappa,\quad |c_2^\eps(t)-C_2|\lesssim (t^{-\kappa}+1)\eps^{\kappa}.$$
Here $C_1$ and $C_2$ are defined in  \eqref{eq:counterterms}, and the proportional constants are independent of $\eps$.
\ep

The proof of Proposition \ref{prop:Zn} is given in Section \ref{proof2}.

\br
From Proposition \ref{prop:Zn}, we deduce that the renormalization constants from $\cR Z$ converge to finite values $C_1$ and $C_2$, which correspond to the newly introduced mass term in the limiting equation \eqref{eq:mainnew1}. Furthermore, we will observe the emergence of such new mass terms in the limiting behavior of the nonlinear terms $\big(v^\eps* |Y+\psi|^2\big)Z,\big(v^\eps*\Re[(Y+\psi)\overline Z]\big)(Y+\psi)$, and $\sI(\cR Z)\circ\cZ^{\<2vm>}$, $v^\eps*\Re[\sI(\cR Z)\circ\Wick{\overline Z]Z}$ (see Proposition \ref{prop:Zn1}, Lemmas \ref{lem:conn1} and \ref{lem:dif1} below).
\er

Analogous to the approach outlined in \eqref{phi},  deriving  uniform in $\eps$ estimates for $\cZ^{\<2vm>}\psi, 2\big(v^\eps*\Re[(Y+\psi)\Wick{\overline Z]\big)Z}$ necessitates the use of paracontrolled calculus. The latter term $2v^\eps*\Re[(Y+\psi)\Wick{\overline Z]Z}$ involves the convolution with $v^\eps$. To analyze this term, we employ the paraproduct decomposition and leverage the random fields introduced in \eqref{sto:ny0}, expressing it as follows:
\begin{equs}[vepsde]&
	2v^\eps*\Re[(Y+\psi)\Wick{\overline Z]Z}
	\\=&\int v^\eps(y)\tau_y(Y+\psi)\cZ^{\<2m>}_y\,\dif y+\int v^\eps(y)\tau_y\overline{(Y+\psi)}\cZ^{\<2>}_y\,\dif y
	\\=&\int v^\eps(y)\tau_y(Y+\psi)\prec\cZ^{\<2m>}_y\,\dif y+\int v^\eps(y)\tau_y\overline{(Y+\psi)}\prec\cZ^{\<2>}_y\,\dif y
	\\&+\int v^\eps(y)\tau_y(Y+\psi)\succcurlyeq\cZ^{\<2m>}_y\,\dif y+\int v^\eps(y)\tau_y\overline{(Y+\psi)}\succcurlyeq\cZ^{\<2>}_y\,\dif y.
\end{equs}
Here for fixed $\eps>0$, $Y, \psi\in \bC^{\frac32-\kappa}$ $\bP$-a.s., and each term on both sides is well-defined.
To prove convergence to the corresponding terms in the dynamical $\Phi^4_3$ model, we need to introduce renormalization counterterms for the terms involving the paraproduct $\circ$.
For the purely stochastic term independent of $\psi$, we have the following decomposition
\begin{align*}&
	\int v^\eps(y)\tau_y Y\circ\cZ^{\<2m>}_y\dif y+\int v^\eps(y)\tau_y\overline{Y}\circ\cZ^{\<2>}_y\,\dif y
	\\=&-\int v^\eps(y)\tau_y \cZ^{\<3v0m>}\circ\cZ^{\<2m>}_y\,\dif y-\int v^\eps(y)\tau_y\overline{\cZ}^{\<3v0m>}\circ\cZ^{\<2>}_y\,\dif y
	\\&+\int v^\eps(y)\tau_y \sI(\cR Z)\circ\cZ^{\<2m>}_y\,\dif y+\int v^\eps(y)\tau_y\overline{\sI(\cR Z)}\circ\cZ^{\<2>}_y\,\dif y.
\end{align*}

\bp\label{thcomm1} It holds that for $\kappa>0,p\geq1$
\begin{align*}\Big(\E \Big\|\int v^\eps(y)\tau_y \cZ^{\<3v0m>}\circ\cZ^{\<2m>}_y\,\dif y-(\tilde b_2^\eps+\tilde b_3^\eps)Z-\cZ^{\<32oc>}\Big\|^p_{C_T\bC^{-\frac12-\kappa}}\Big)^{\frac1p}\lesssim \eps^{\frac\kappa2},
\end{align*}
and
\begin{align*}\Big(\E \Big\|\int v^\eps(y)\tau_y\overline{\cZ}^{\<3v0m>}\circ\cZ^{\<2>}_y\,\dif y-\tilde b_5^\eps Z-\cZ^{\<32oc3>}\Big\|^p_{C_T\bC^{-\frac12-\kappa}}\Big)^{\frac1p}\lesssim \eps^{\frac\kappa2},
\end{align*}
with the proportional constant independent of $\eps$.
\ep
The proof of this result is purely probabilistic, and we will provide the proof of Proposition \ref{thcomm1} in Section \ref{proof1}.

Using Lemma \ref{comnew} we know that $\sI(\cR Z)\in \bC^{\frac12-\kappa}$ $\bP$-a.s.. Hence, we note that $\sI(\cR Z)\circ Z$ and  several new stochastic terms involving $\cR Z$  do not make sense in the classical setting (c.f. Lemma \ref{lem:para}). We employ probabilistic calculations to prove that they converge to zero in suitable spaces, after subtracting suitable renormalization counterterms.
\bp\label{prop:Zn1}
It holds that  for $\kappa>0$ and $p\geq1$
\begin{align*}
	\E\|\sI(\cR Z)\circ\cZ^{\<2vm>}-c^\eps_1Z\|^p_{C_T\bC^{-\frac12-\kappa}}\lesssim \eps^{\frac{p\kappa}2},
\end{align*}
and
\begin{align*}	
	\E\Big\|\int v^\eps(y)\tau_y \sI(\cR Z)\circ\cZ^{\<2m>}_y\,\dif y-c^\eps_2 Z\Big\|^p_{C_T\bC^{-\frac12-\kappa}}\lesssim \eps^{\frac{\kappa p}2},\end{align*}
\begin{align*}	
	\E\Big\|\int v^\eps(y)\tau_y \sI(\cR \overline Z)\circ\cZ^{\<2>}_y\,\dif y-c^\eps_1 Z-c^\eps_2 Z\Big\|^p_{C_T\bC^{-\frac12-\kappa}}\lesssim \eps^{\frac{\kappa p}2}.\end{align*}
Here $c_1^\eps,c_2^\eps$ are defined in \eqref{c1c2} and the proportional constant is independent of $\eps$.
Moreover, it holds that
\begin{align*}\E\sup_y\Big\|\tau_y\overline{\sI(\cR Z)}\circ Z\Big\|_{C_T\bC^{-2\kappa}}^p+\E\sup_y\Big\|\tau_y\sI(\cR Z)\circ Z\Big\|_{C_T\bC^{-2\kappa}}^p\lesssim 1,\end{align*}
\begin{align*}\E\sup_y\Big\|{\sI(\cR Z)}\circ\tau_y\overline Z\Big\|_{C_T\bC^{-2\kappa}}^p+\E\sup_y\Big\|\sI(\cR Z)\circ \tau_yZ\Big\|_{C_T\bC^{-2\kappa}}^p\lesssim 1.\end{align*}
\ep
The proof of Proposition \ref{prop:Zn1} is given in Section \ref{proof2}.

By Lemma \ref{lem:para} the worst part of the RHS of \eqref{eq:main2new} is
$$-(\psi+Y)\prec \cZ^{\<2vm>}-\int v^\eps(y)\tau_y(Y+\psi)\prec\cZ^{\<2m>}_y\,\dif y-\int v^\eps(y)\tau_y\overline{(Y+\psi)}\prec\cZ^{\<2>}_y\,\dif y,$$ which leads to the following paracontrolled ansatz:
\begin{equs}[para]
	\psi=&-(\psi+Y)\prec \cZ^{\<20vm>}-\int v^\eps(y)\tau_y(Y+\psi)\prec\cZ^{\<20vm1>}_y\,\dif y
	-\int v^\eps(y)\tau_y\overline{(Y+\psi)}\prec\cZ^{\<20vm2>}_y\,\dif y+\psi^\sharp.
\end{equs}

Similar to the discussion after \eqref{para:phi}, $\psi^\sharp$ exhibits better regularity compared to $\psi$. Specifically, we have $\psi^\sharp \in \bB_p^{1+2\kappa}$.
As with the dynamical $\Phi^4_3$ model, the paracontrolled ansatz \eqref{para} can be used to decompose the terms ${\cZ^{\<2vm>}\circ\psi}$ , $\int v^\eps(y)\tau_y\psi\circ\cZ^{\<2m>}_y\,\dif y,  \int v^\eps(y)\tau_y\overline{\psi}\circ\cZ^{\<2>}_y\,\dif y $ and give uniform in $\eps$ bounds after adding suitable renormalization counterterms.

For $\cZ^{\<2vm>}\circ\psi$ we have the following result.
\bl\label{lem:sto1-n} It holds that for any $p\geq1$ and $\kappa\in (0,\frac12)$
\begin{equs}
	&\|\cZ^{\<2vm>}\circ\psi+(\tilde b_1^\eps+\tilde b_2^\eps)(Y+\psi)\|_{\bB_p^{-\kappa}}
	\lesssim \|Y+\psi\|_{\bB^{3\kappa}_p}(\eps^{\frac\kappa2}\|\mZ_\eps\|+\|\mZ_\eps\|^2)+\|\psi^\sharp\|_{\bB^{1+2\kappa}_p}\|\mZ_\eps\|.
\end{equs}
\el
\begin{proof} Similar to \eqref{def:phicirc}, we substitute the paracontrolled ansatz \eqref{para} into $\cZ^{\<2vm>}\circ\psi$ and have the following decomposition :
	\begin{equs}[psi:sto]
		&\cZ^{\<2vm>}\circ\psi+(\tilde b_1^\eps+\tilde b_2^\eps)(Y+\psi)=-\sum_{i=1}^3J_i+\psi^\sharp\circ \cZ^{\<2vm>},
	\end{equs}
	with
	\begin{equs}
		J_1\eqdef&\,[(Y+\psi)\prec \cZ^{\<20vm>}]\circ \cZ^{\<2vm>}-\tilde b_1^\eps(Y+\psi),
		\\
		J_2\eqdef&\int v^\eps(y)[\tau_y(Y+\psi)\prec \cZ_y^{\<20vm1>}]\circ \cZ^{\<2vm>}\,\dif y -\tilde b_2^\eps(Y+\psi),
		\\
		J_3\eqdef&\int v^\eps(y)[\tau_y\overline{(Y+\psi)}\prec\cZ^{\<20vm2>}_y]\circ \cZ^{\<2vm>}\,\dif y.
	\end{equs}
	Using the paraproduct estimates Lemma \ref{lem:para} we have that
	\begin{align*}
		\|\psi^\sharp\circ \cZ^{\<2vm>}\|_{\bB_p^{-\kappa}}\lesssim \|\psi^\sharp\|_{\bB^{1+2\kappa}_p}\|\cZ^{\<2vm>}\|_{\bC^{-1-\kappa}}.
	\end{align*}
	Furthermore $J_1$ can be decomposed as follows by using the classical commutator from Lemma \ref{lem:com1}:
	\begin{equs}
		&J_1 =C(\psi+Y, \cZ^{\<20vm>}, \cZ^{\<2vm>})+(Y+\psi)\cZ^{\<22vm>},
	\end{equs}
	where $\tilde b_1^\eps(Y+\psi)$ is incoporated into $\cZ^{\<22vm>}(Y+\psi)$ from \eqref{sto:4}.
	Using Lemma \ref{lem:com1} and Proposition \ref{prop:Z}, we obtain
	\begin{align*}
		\|J_1\|_{\bB^{-\kappa}_p}\lesssim \|Y+\psi\|_{\bB^{3\kappa}_p}\|\mZ_\eps\|^2.
	\end{align*}
	For $J_2$ further refinement is required. In fact, we use the $y$-dependent random fields introduced in \eqref{sto:ny2} to have
	\begin{equs}[dec:y1]
		J_2=&\int v^\eps(y) C(\tau_y(Y+\psi),\cZ^{\<20vm1>}_y, \cZ^{\<2vm>})\,\dif y+\int v^\eps(y)\tau_y(Y+\psi)\cZ^{\<22vmy>}_{y}\,\dif y
		\\&+\int v^\eps(y)\tau_y(Y+\psi)\tilde b_2^\eps(y)\,\dif y-\tilde b_2^\eps(Y+\psi)
		\\=&\sum_{i=1}^3J_{2i},\end{equs}
	with
	\begin{equs}
		J_{21}\eqdef&\int v^\eps(y) C(\tau_y(Y+\psi),\cZ^{\<20vm1>}_y, \cZ^{\<2vm>})\,\dif y+(Y+\psi)\cZ^{\<22cvm>}
		\\
		J_{22}\eqdef&\int v^\eps(y)\big[\tau_y(Y+\psi)-(Y+\psi) \big]\cZ^{\<22vmy>}_{y}\,\dif y
		\\
		J_{23}\eqdef&\int v^\eps(y)[\tau_y(Y+\psi)-(Y+\psi)]\tilde b_2^\eps(y)\,\dif y.
	\end{equs}
	for $\tilde b_2^\eps(y)=\tilde b_2^\eps(t,y)$ defined in \eqref{def:tby}, where we used \eqref{sto:re} and \eqref{sto:re1}.

	For $J_{23}$ we  have by Lemma \ref{veps} and \eqref{v1}
	\begin{equs}[vepsb] \|J_{23}\|_{\bB_p^{-\kappa}}
		\lesssim&\,\eps^{-\frac\kappa2}\|\mZ_\eps\|\int |v^\eps(y)||y|^{\kappa}\|Y+\psi\|_{\bB_p^{2\kappa}}\dif y
		\lesssim\eps^{\frac\kappa2}\|\mZ_\eps\|\|Y+\psi\|_{\bB_p^{2\kappa}},
	\end{equs}
	where we used $|\widehat{v^\eps}(k_1+k_2)|\lesssim \eps^{-\frac\kappa2}|k_1+k_2|^{-\frac\kappa2}$ to deduce $|\tilde b_2^\eps(y)|\lesssim \eps^{-\frac\kappa2}.$
	For $J_{22}$ we have by Lemma \ref{veps} and Proposition \ref{prop:Zny} and \eqref{v1}
	\begin{equs}[vepsnew]\|J_{22}\|_{\bB^{-\kappa}_p}
		\lesssim &\int |v^\eps(y)|\|\tau_y(Y+\psi)-(Y+\psi)\|_{\bB^{2\kappa}_p}\|\cZ^{\<22vmy>}_{y}\|_{\bC^{-\kappa}}\,\dif y
		\\\lesssim &\int |v^\eps(y)||y|^\kappa\|Y+\psi\|_{\bB^{3\kappa}_p}\,\dif y\sup_y\|\cZ^{\<22vmy>}_{y}\|_{\bC^{-\kappa}}
		\lesssim \eps^\kappa \|Y+\psi\|_{\bB^{3\kappa}_p}\|\mZ_\eps\|.
	\end{equs}
	For $J_{21}$ we use Lemma \ref{lem:com1} and Proposition \ref{prop:Z} and Proposition \ref{prop:Zny} and \eqref{eq:tauy} to have
	\begin{align}\label{est:J21}
		\|J_{21}\|_{\bB^{-\kappa}_p}\lesssim \|Y+\psi\|_{\bB^{3\kappa}_p}\|\mZ_\eps\|^2.
	\end{align}
	
	For $J_3$ in \eqref{psi:sto} we use similar decomposition to obtain
	\begin{equs}
		J_3
		=&\int v^\eps(y) C(\tau_y(\overline{Y+\psi}),\cZ^{\<20vm2>}_y, \cZ^{\<2vm>})\,\dif y+\overline{(Y+\psi)}\cZ^{\<22ccvm>}\\&+\int v^\eps(y)[\tau_y(\overline{Y+\psi})-\overline{Y+\psi} ]\cZ^{\<22vmy1>}_{y}\,\dif y
		.
	\end{equs}
	The desired bounds for $\|J_3\|_{\bB^{-\kappa}_p}$ follow the same line as that for $J_{21}$ and $J_{22}$. Hence, the result follows.
\end{proof}

Before proceeding we introduce the following operator, which will be used in the sequel:
\begin{align*}
	\cR_1f\eqdef\sum_k \widehat{\cR_1}(k)\langle f,e_k\rangle e_k ,
\end{align*}
with
\begin{equs}[defg]
	\widehat{\cR_1}(k)\eqdef\sum_{k_1,k_2} \tilde b_t(k_1,k_2)(\widehat{v^\eps}(k+k_1)^2-\widehat{v^\eps}(k_1)^2),
\end{equs}
for $\tilde b_t$ defined in \eqref{def:tib}.
We have the following result for $\cR_1$.

\bl\label{boundr1} It holds that for $p\geq1$, $\kappa\in(0,\frac12)$
\begin{align*}\|\mathcal R_1 f\|_{\bB^{\kappa}_p}\lesssim \eps^\kappa\|f\|_{\bB^{3\kappa}_p},
\end{align*}
with the proportional constant independent of $\eps$.
\el
The proof of this lemma is provided in Section \ref{sec:comnew}.
Furthermore, we obtain the following result for $\int v^\eps(y)\tau_y\psi\circ \cZ^{\<2m>}_y\,\dif y+(\tilde b_2^\eps+\tilde b_3^\eps)(Y+\psi)$ and $\int v^\eps(y)\tau_y\overline{\psi}\circ\cZ^{\<2>}_y\,\dif y +\tilde b_5^\eps(Y+\psi)$, which appear in the RHS of \eqref{vepsde}.

\bl \label{lem:sto1} It holds that for any $p\geq1$
\begin{equs}
	&\Big\|\int v^\eps(y)\tau_y\psi\circ \cZ^{\<2m>}_y\,\dif y+(\tilde b_2^\eps+\tilde b_3^\eps)(Y+\psi)\Big\|_{\bB_p^{-\kappa}}+\Big\| \int v^\eps(y)\tau_y\overline{\psi}\circ\cZ^{\<2>}_y\,\dif y +\tilde b_5^\eps(Y+\psi)\Big\|_{\bB_p^{-\kappa}}
	\\\,\lesssim &\|Y+\psi\|_{\bB^{3\kappa}_p}(\eps^{\kappa}\|\mZ_\eps\|+\|\mZ_\eps\|^2)+\|\psi^\sharp\|_{\bB^{1+2\kappa}_p}\|\mZ_\eps\|,
\end{equs}
\el
\begin{proof} We start with the first term and  substitute the paracontrolled ansatz \eqref{para} to obtain
	\begin{equs}[psi:sto1]
		\int v^\eps(y)\tau_y\psi\circ \cZ^{\<2m>}_y\dif y+(\tilde b_2^\eps+\tilde b_3^\eps)(Y+\psi)&=-\sum_{i=1}^3J_i+\int v^\eps(y)\tau_y\psi^\sharp\circ \cZ^{\<2m>}_y\dif y\end{equs}
	with
	\begin{equs}
		J_1\eqdef &\int v^\eps(y)[\tau_y(Y+\psi)\prec \tau_y\cZ^{\<20vm>}]\circ \cZ^{\<2m>}_y\,\dif y-\tilde b_2^\eps(Y+\psi),\\
		J_2\eqdef & \int v^\eps(y)v^\eps(y_1)[\tau_{y+y_1}\overline{(Y+\psi)}\prec \tau_y \cZ^{\<20vm2>}_{y_1}]\circ \cZ^{\<2m>}_y\,\dif y\dif y_1,\\
		J_3\eqdef &\int v^\eps(y)v^\eps(y_1)[\tau_{y+y_1}(Y+\psi)\prec\tau_y \cZ^{\<20vm1>}_{y_1}]\circ \cZ^{\<2m>}_y\,\dif y\dif y_1-\tilde b_3^\eps(Y+\psi)
		.
	\end{equs}
	We use the paraproduct estimates Lemma \ref{lem:para} and Proposition \ref{prop:Zny} and \eqref{eq:tauy} to have
	\begin{align*}
		\Big\|\int v^\eps(y)\tau_y\psi^\sharp\circ \cZ^{\<2m>}_y\,\dif y\Big\|_{\bB^{\kappa}_p}\lesssim \|\psi^\sharp\|_{\bB^{1+2\kappa}_p}\|\mZ_\eps\|.
	\end{align*}
	
	Furthermore, we use $\mathbf{Z}^{\<2v2v>}_\eps$, $\mathbf{Z}^{\<2cvr2v>}_\eps$, $\mathbf{Z}^{\<2cv2v>}_\eps$ introduced in \eqref{sto:ny1} and $	\cZ^{\<22ocy>}_{y,\eps}$, $
	\cZ^{\<22oc1>}_{y,y_1}$, $\cZ^{\<22oc>}_{y,y_1}$ introduced in \eqref{sto:ny2}
	to decompose $J_1, J_2,J_3$ as in \eqref{dec:y1}. We only give details for $J_3$ and others follow the same line: Similar as in \eqref{dec:y1} we have
	\begin{align*}
		J_3= &\sum_{i=1}^3J_{3i},
	\end{align*}
	with
	\begin{align*}
		J_{31}\eqdef& \int v^\eps(y)v^\eps(y_1)C(\tau_{y+y_1}(Y+\psi),\tau_y \cZ^{\<20vm1>}_{y_1}, \cZ^{\<2m>}_y)\,\dif y\dif y_1+(Y+\psi)\mathbf{Z}^{\<2cv2v>}_\eps,
		\\
		J_{32}\eqdef&\int v^\eps(y)v^\eps(y_1)[\tau_{y+y_1}(Y+\psi)-(Y+\psi)]\cZ^{\<22oc>}_{y,y_1}\,\dif y\dif y_1,
		\\
		J_{33}\eqdef&\int v^\eps(y)v^\eps(y_1)[\tau_{y+y_1}(Y+\psi)-(Y+\psi)]\tilde b_3(y,y_1)\,\dif y\dif y_1,
	\end{align*}
	where
	$\tilde b_3(y,y_1)=\tilde b_3(t,y,y_1)$ defined in \eqref{def:tby} and  we used \eqref{sto:re} and \eqref{sto:re1}.
	
	For $J_{31}$, we apply Lemma \ref{lem:com1}, Proposition \ref{prop:Z}, Proposition \ref{prop:Zny} as in \eqref{est:J21} and \eqref{eq:tauy} to have
	\begin{align*}
		\|J_{31}\|_{\bB^{-\kappa}_p}\lesssim \|Y+\psi\|_{\bB^{3\kappa}_p}\|\mZ_\eps\|^2.
	\end{align*}
	We note that $J_{33}$ cannot be bounded as in \eqref{vepsb}, since we cannot bound $|\tilde b_3(y,y_1)|$ by $\eps^{-\frac\kappa2}$. Instead, we employ the Fourier basis to obtain
	\begin{align*}
		J_{33} = \mathcal{R}_1 (Y+\psi).
	\end{align*}
	By using Lemma \ref{boundr1}  we have
	\begin{align*}
		\|J_{33}\|_{\bB^{\kappa}_p}\lesssim \eps^\kappa \|Y+\psi\|_{\bB^{3\kappa}_p}.
	\end{align*}
	Similar as in \eqref{vepsnew} we obtain for $J_{32}$
	and $p\geq1$
	\begin{equs}[vepsnew1]\|J_{32}\|_{\bB^{-\kappa}_p}
		\lesssim &\int |v^\eps(y)v^\eps(y_1)|\|\tau_{y+y_1}(Y+\psi)-(Y+\psi)\|_{\bB^{2\kappa}_p}\|\cZ^{\<22oc>}_{y,y_1}\|_{\bC^{-\kappa}}\,\dif y\dif y_1
		\\\lesssim &\int |v^\eps(y)v^\eps(y_1)|(|y|^\kappa+|y_1|^\kappa)\|Y+\psi\|_{\bB^{3\kappa}_p}\,\dif y\dif y_1\sup_{y,y_1}\|\cZ^{\<22oc>}_{y,y_1}\|_{\bC^{-\kappa}}
		\\\lesssim &\,\eps^\kappa \|Y+\psi\|_{\bB^{3\kappa}_p}\|\mZ_\eps\|,
	\end{equs}
	where in the second step we used Lemma \ref{veps} and \eqref{v1}.
	
	Similarly we can also use $\mathbf{Z}^{\<2v2vr>}_\eps$, $\mathbf{Z}^{\<2cv2vr>}_\eps$ and $	\mathbf{Z}^{\<2cv2v0>}_\eps$ introduced in \eqref{sto:ny1} and $\cZ^{\<22ocy1>}_{y,\eps}$, $\cZ^{\<22oc3>}_{y,y_1}$,
	$\cZ^{\<22oc2>}_{y,y_1}$ from \eqref{sto:ny2}
	to  decompose  $ \int v^\eps(y)\tau_y\overline{\psi}\circ(\cZ^{\<2>}_y)\,\dif y +\tilde b_5^\eps(Y+\psi)$ and derive the same bounds. Thus the result follows.
\end{proof}

\subsection{Control the difference}\label{sec:dif}
In this section, we present the proof of the convergence of the dynamics stated in Theorem \ref{th:1}. The main approach is to compare the corresponding terms in equations \eqref{eq:main2new} and \eqref{phi}. We can estimate each term through  paracontrolled calculus and employ similar arguments as in \cite{ZZ18}.

\bt\label{th:1} Suppose that $\Psi^\eps(0)\to \Phi(0)$ in $\bC^{-\frac12-\kappa}$ for $\kappa>0$ and $v$ satisfies $(\mathbf{Hv})$. It holds that the unique solutions $\Psi^\eps$ and $\Phi$ to equations \eqref{eq:mainnew} and \eqref{eq:mainnew1} satisfy for any $T>0$
$$\|\Psi^\eps-\Phi\|_{C_T \bC^{-\frac12-\kappa}}\to^{\bP} 0,\qquad \eps\to0.$$
\et

The proof of Theorem \ref{th:1} follows the same reasoning as in \cite[Section 5]{ZZ18}, while accounting for the differences in the corresponding terms as given in Lemma \ref{lem:con1}-- Lemma \ref{lem:dif1},  as well as \eqref{bdd:mZ}, \eqref{dif:esb}--\eqref{bd:cRY},  below. Further details can be found in Appendix \ref{app:pro}.

In the sequel we compare the corresponding terms on the RHS of  equation \eqref{eq:main2new} and equation \eqref{phi}. To this end, we introduce the following notations:
We adjust $\|\mZ_\eps\|$ from Section \ref{sec:ren} to be greater than the $\|\mZ_\eps\|$ from Section \ref{sec:ren} and also larger than
$$ \sup_{y\in\mathbb{R}^3}\frac{\|\cZ^{\<2m>}_y-\cZ^{\<2m>}\|_{C_T\bC^{-1-2\kappa}}+\|\cZ^{\<2>}_y-\cZ^{\<2>}\|_{C_T\bC^{-1-2\kappa}}}{|y|^{\kappa}}$$ and  $ \|\cZ_\eps^{\<3v0m>}\|_{C_T^{\frac18}L^\infty}$, as well as
\begin{equs}\sup_y\Big\|\tau_y\overline{\sI(\cR Z)}\circ Z\Big\|_{C_T\bC^{-2\kappa}},\qquad \sup_y\Big\|\tau_y\sI(\cR Z)\circ Z\Big\|_{C_T\bC^{-2\kappa}}, \\
	\sup_y\Big\|{\sI(\cR Z)}\circ\tau_y\overline Z\Big\|_{C_T\bC^{-2\kappa}},\qquad \sup_y\Big\|\sI(\cR Z)\circ \tau_yZ\Big\|_{C_T\bC^{-2\kappa}}\end{equs}
from Proposition \ref{prop:Zn1}.
Furthermore, we also use $\|\mZ-\mZ_\eps\|$ to denote the smallest number bigger than  $\|\cZ_{\eps}^{\tau_\eps}-\cZ^{\tau}\|_{C_T\bC^{|\tau|-\kappa}}$  from  Proposition \ref{prop:Z} and $\|\mathbf{Z}_{\eps}^{\tau_\eps}-\cZ^{\tau}\|_{C_T\bC^{|\tau|-\kappa}}$  from  Proposition \ref{prop:Zny} and $C_T\bC^{|\tau|-\kappa}$-norm of all the random fields $\cZ_\eps^{\tau_\eps}-\cZ^\tau$ from
Propositions \ref{prop:Zn} and \ref{thcomm1}, and the first three terms in Proposition \ref{prop:Zn1},  and $\|\cZ_\eps^{\<3v0m>}-\cZ^{\<30m>}\|_{C_T^{\frac18}L^\infty}$.
Here, we also view the random fields from Propositions  \ref{prop:Zn} and \ref{thcomm1} and the first three terms in Proposition \ref{prop:Zn1} as $\cZ_\eps^{\tau_\eps}-\cZ^\tau$.

Using Propositions \ref{prop:Z}, \ref{prop:Zny},  \ref{prop:Zn1} and Lemma \ref{lem:bdZ} we also have for every $p\geq1$
\begin{align}\label{bd:mZ}
	\sup_\eps\E\|\mZ_\eps\|^p+\E\|\mZ\|^p\lesssim 1,
\end{align}
with the proportional constant independent of $\eps$.
Using Propositions \ref{prop:Z}, \ref{prop:Zny}  and Propositions \ref{prop:Zn}, \ref{prop:Zn1}, we obtain
\begin{align}\label{bdd:mZ}
	\E\|\mZ_\eps-\mZ\|^p\lesssim \eps^{\frac{p\kappa}2},
\end{align}
where the proportional constant is independent of $\eps$.

Now we analyze each term in equation \eqref{eq:main2new} as follows. We define
\begin{equs}[def:M]
	\mM\eqdef&\, \|\psi\|_{\bC^{\frac12+2\kappa}}+\|\phi\|_{\bC^{\frac12+2\kappa}}+\|\mZ\|+\|\mZ_\eps\|+1,
	\\
	\mM_\infty\eqdef&\, \|\psi\|_{L^\infty}+\|\phi\|_{L^\infty}+\|\mZ\|+\|\mZ_\eps\|+1.
\end{equs}

We first consider the terms $\cZ^{\<2vm>}(Y+\psi)$ and $2\big(v^\eps*\Re[(Y+\psi):\overline Z]\big)Z:$ from the first two lines of \eqref{eq:main2new}.

\bl\label{lem:con1} It holds that for $t>0$
\begin{equs}[con1]&\|2\big(v^\eps*\Re[(Y+\psi):\overline Z]\big)Z:+(\tilde b_3^\eps+\tilde b_2^\eps+\tilde b_5^\eps)(Z+Y+\psi)
	\\&+\cZ^{\<32oc1>}- \cZ^{\<2m>}\phi+\cZ^{\<32oc2>}-\cZ^{\<2>}\overline\phi-(C_1+2C_2)Z\|_{\bC^{-1-2\kappa}}
	\\\lesssim&\, \|\mZ_\eps-\mZ\|\mM\|\mZ\|+\|\psi-\phi\|_{\bC^{\frac12+2\kappa}}\|\mZ\|^2+\eps^{\kappa/2} \mM\|\mZ_\eps\|^2(1+t^{-\kappa})\\&+\|\psi^\sharp-\phi^\sharp\|_{\bC^{1+2\kappa}}\|\mZ\|+(\|\mZ_\eps-\mZ\|+\eps^\kappa)\|\psi^\sharp\|_{\bC^{1+3\kappa}},
\end{equs}
and
\begin{align*}&\|\cZ^{\<2vm>}(Y+\psi)+(\tilde b_1^\eps+\tilde b_2^\eps)(Z+Y+\psi)+\cZ^{\<32oc1>}- \phi\cZ^{\<2m>}-C_1Z\|_{\bC^{-1-2\kappa}}
	\\\lesssim&\, \|\mZ_\eps-\mZ\|\mM\|\mZ\|+\|\psi-\phi\|_{\bC^{\frac12+2\kappa}}\|\mZ\|^2+\eps^{\kappa/2} \mM\|\mZ_\eps\|^2(1+t^{-\kappa})\\&+\|\psi^\sharp-\phi^\sharp\|_{\bC^{1+2\kappa}}\|\mZ\|+(\|\mZ_\eps-\mZ\|+\eps^\kappa)\|\psi^\sharp\|_{\bC^{1+3\kappa}},
\end{align*}
where
$$\cZ^{\<2m>}\phi\eqdef \cZ^{\<2m>}(\prec+\succ)\phi+\cZ^{\<2m>}\circ\phi,\qquad \cZ^{\<2>}\overline\phi\eqdef \cZ^{\<2>}(\prec+\succ)\overline\phi+\cZ^{\<2>}\circ\overline\phi,$$
with $\cZ^{\<2m>}\circ\phi$ and  $\cZ^{\<2>}\circ\overline\phi$ defined in \eqref{def:phicirc}.
\el
\begin{proof} We only prove the first one and the second one is similar. Recall \eqref{vepsde} and it suffices to prove
	\begin{equs}[confirst]&
		\Big\|\int v^\eps(y)\tau_y(Y+\psi)\cZ^{\<2m>}_y\,\dif y+(\tilde b_3^\eps+\tilde b_2^\eps)(Z+Y+\psi)+\cZ^{\<32oc1>}- \cZ^{\<2m>}\phi-C_2Z\Big\|_{\bC^{-1-2\kappa}}
	\end{equs}
	controlled by the RHS of \eqref{con1},
	and the other part follows the same line.
	We first have
	\begin{equs}&\int v^\eps(y)\tau_y(Y+\psi)(\prec+\succ) \cZ^{\<2m>}_y\,\dif y
		\\=&\int v^\eps(y)[\tau_y(Y+\psi)-(Y+\psi)](\prec+\succ) \cZ^{\<2m>}_y\,\dif y+(Y+\psi)(\prec+\succ) \int v^\eps(y)\cZ^{\<2m>}_y\,\dif y.
	\end{equs}
	By using the paraproduct estimates in Lemma \ref{lem:para} and Lemma \ref{veps} and \eqref{v1} we obtain
	\begin{align*}&\Big\|\int v^\eps(y)[\tau_y(Y+\psi)-(Y+\psi)](\prec+\succ) \cZ^{\<2m>}_y\,\dif y\Big\|_{\bC^{-1-\kappa}}
		\\\lesssim&\int |v^\eps(y)|\|\tau_y(Y+\psi)-(Y+\psi)\|_{\bC^{\kappa}}\|\cZ^{\<2m>}_y\|_{\bC^{-1-\kappa}}\,\dif y
		\lesssim \,\eps^\kappa \|Y+\psi\|_{\bC^{2\kappa}}\|\mZ_\eps\|,
	\end{align*}
	which by Propositions \ref{prop:Z}, \ref{prop:Zny} and Lemma \ref{comnew} implies that
	\begin{equs}[eq:con1]
		&\Big\|\int v^\eps(y)\tau_y(Y+\psi)(\prec+\succ) \cZ^{\<2m>}_y\,\dif y- \cZ^{\<2m>}(\prec+\succ)(-\cZ^{\<30m>}+\phi)\Big\|_{\bC^{-1-2\kappa}}
		\\\lesssim& \,\|\mZ_\eps-\mZ\|\mM+\|\psi-\phi\|_{\bC^{2\kappa}}\|\mZ\|+\eps^\kappa \mM\|\mZ_\eps\|.
	\end{equs}
	Here we also used \eqref{v1} and
	\begin{equs}[est:difZy]
		\|\cZ^{\<2m>}_y-\cZ^{\<2m>}\|_{\bC^{-1-2\kappa}}\lesssim |y|^{\kappa}\|\mZ_\eps\|.
	\end{equs}
	
	We then consider $\int v^\eps(y)\tau_y\psi\circ\cZ^{\<2m>}_y\dif y+(\tilde b_3^\eps+\tilde b_2^\eps)(Y+\psi)$ and recall Lemma \ref{lem:sto1} and we only concentrate on $J_3$ on the RHS of \eqref{psi:sto1} and the bounds for the other terms are similar. From the proof of Lemma \ref{lem:sto1} we have $J_3=\sum_{i=1}^3J_{3i}$ and we  know that the $\bC^{-\kappa}$-norms of $J_{32},J_{33}$ are bounded by the RHS of \eqref{con1}.
	For  $J_{31}$ we  write it as follows
	\begin{equs}[veps1]&\int v^\eps(y)v^\eps(y_1)C(\tau_{y+y_1}(Y+\psi)-(Y+\psi),\tau_y \cZ^{\<20vm1>}_{y_1}, \cZ^{\<2m>}_y)\,\dif y\dif y_1
		\\&+\int v^\eps(y)v^\eps(y_1)C(Y+\psi,\tau_y \cZ^{\<20vm1>}_{y_1}, \cZ^{\<2m>}_y- \cZ^{\<2m>})\,\dif y\dif y_1
		\\&+C(Y+\psi,\int v^\eps(y)v^\eps(y_1)\tau_y \cZ^{\<20vm1>}_{y_1}\,\dif y\dif y_1,  \cZ^{\<2m>})+(Y+\psi)\mathbf{Z}^{\<2cv2v>}_\eps
	\end{equs}
	Using Proposition \ref{prop:Zny} and \eqref{eq:tauy} for $\sup_{y,y_1}\|\tau_y \cZ^{\<20vm1>}_{y_1}\|_{\bC^{1-\kappa}}, \sup_y\|\cZ^{\<2m>}_y\|_{\bC^{-1-\kappa}}$, along with \eqref{est:difZy}, Lemma \ref{veps} for $\tau_{y+y_1}(Y+\psi)-(Y+\psi)$, Lemma \ref{lem:com1}, and \eqref{v1}, we conclude that the $\bC^\kappa$-norms of the first two terms on the RHS of \eqref{veps1} are bounded by the RHS of \eqref{con1}. Furthermore, the $\bC^{-\kappa}$-norm of the difference between the last line in \eqref{veps1} and $C(-\cZ^{\<30m>}+\phi,\cZ^{\<20vm1>},  \cZ^{\<2m>})+(-\cZ^{\<30m>}+\phi)\cZ^{\<22oc>}$ in the definition of $\cZ^{\<2m>}\circ\phi$ from \eqref{def:phicirc} can also be bounded by the RHS of \eqref{con1}. Moreover, $J_1$ and $J_2$ can be bounded in a similar manner using Proposition \ref{prop:Zny}, Lemma \ref{veps}, Lemma \ref{lem:para}, and Lemma \ref{lem:com1}. Hence, we obtain
	$$\Big\|\int v^\eps(y)\tau_y\psi\circ\cZ^{\<2m>}_y\,\dif y+(\tilde b_3^\eps+\tilde b_2^\eps)(Y+\psi)-\cZ^{\<2m>}\circ\phi \Big\|_{\bC^{-\kappa}}$$ can be bounded by the RHS of \eqref{con1}.

	It remains to consider $\int v^\eps(y)\tau_yY\circ\cZ^{\<2m>}_y\,\dif y-C_2Z$.  We use Propositions \ref{prop:Zn} and \ref{prop:Zn1} to  obtain
	\begin{align}
		\Big\|\int v^\eps(y)\tau_y\sI(\cR Z)\circ\cZ^{\<2m>}_y\,\dif y-C_2Z\Big\|_{\bC^{-\frac12-2\kappa}}\lesssim \|\mZ_\eps-\mZ\|+\eps^{\kappa}\|\mZ_\eps\|(1+t^{-\kappa}).
	\end{align}
	Thus, the conclusion that \eqref{confirst} is controlled by the RHS of \eqref{con1} follows from the application of Proposition \ref{thcomm1}.
\end{proof}

We then consider the last line of equation \eqref{eq:main2new}.

\textbf{I}.  $(6b^\eps-6\tilde{b}^\eps+m-1)\Psi-(6b+m-1)\Phi$: By definition we know that
$$b_1^\eps-\tilde{b}_1^\eps(t)=\sum_{k_1,k_2}e^{-(\la k_1\ra^2+\la k_2\ra^2+\la k_1+k_2\ra^2)t}\widehat{v^\eps}(k_1+k_2)^2 b(k_1,k_2).$$
Thus $|\tilde{b}_1^\eps-b_1^\eps-b|\lesssim \eps^{\kappa}t^{-\kappa}.$ For other $b_i^\eps$ and $\tilde b_i^\eps$ we have similar bounds. Hence we obtain
\begin{equs}[dif:esb]
	&\|(6b^\eps-6\tilde{b}^\eps+m-1)\Psi-(6b+m-1)\Phi\|_{\bC^{-\frac12-\kappa}}
	\\\lesssim &\,(t^{-\kappa}+1) \big(\eps^{\kappa}(\|\mZ\|+\|\psi\|_{\bC^{-\frac12-\kappa}})+\|\psi-\phi\|_{\bC^{-\frac12-\kappa}}+\|\mZ_\eps-\mZ\|\big).\end{equs}

\textbf{II}. $\cR Y$ and $\cR\psi$:  By Lemma \ref{comnew} we know that
\begin{align}\label{bd:cRY}
	\|\cR Y\|_{\bC^{-\frac12-2\kappa}}\lesssim \eps^{\kappa}\|\mZ_\eps\|,\qquad \|\cR\psi\|_{\bC^{-\frac12+\kappa}}\lesssim \eps^\kappa \|\psi\|_{\bC^{\frac12+2\kappa}}.
\end{align}

Next, we analyze the contribution from the third line of equation \eqref{eq:main2new} in the following lemmas. We begin by examining the difference between
$\big(v^\eps* |Y+\psi|^2\big)Z$ from the RHS of equation \eqref{eq:main2new} and $|-\cZ^{\<30m>}+\phi|^2Z$ from equation \eqref{phi}.

\bl\label{lem:conn1} It holds that for $t>0$
\begin{align}\label{bd:sto1}
	&\|\big(v^\eps* |Y|^2\big)Z(t)\|_{\bC^{-\frac12-\kappa}}\lesssim (\|\mZ_\eps\|^3+1)(t^{-\kappa}+1),
\end{align}
and
\begin{equs}[est:dif1]
	&\big\|v^\eps* |Y+\psi|^2Z-( |-\cZ^{\<30m>}+\phi|^2Z +C_1(-\cZ^{\<30m>}+\phi))\big\|_{\bC^{-\frac12-\kappa}}
	\\\lesssim&\,\mM \mM_\infty\big(\|\mZ_\eps-\mZ\|+\eps^\kappa(\|\mZ\|+\|\mZ_\eps\|)\big)(1+t^{-\kappa})
	+\big(\|\psi-\phi\|_{\bC^{\frac12+2\kappa}}\mM_\infty+\|\psi-\phi\|_{L^\infty}\mM\big)\|\mZ\|,\end{equs}
where the proportional constants are independent of $\eps$ and
\begin{align}\label{def:zz}
	|-\cZ^{\<30m>}+\phi|^2Z\eqdef |\cZ^{\<30m>}|^2Z- \phi\cZ^{\<31oc2>}-\overline{\phi}\cZ^{\<31oc1>}+|\phi|^2Z, \end{align}
with
\begin{align*}
	|\cZ^{\<30m>}|^2\circ Z\eqdef(\cZ^{\<30m>}\circ {\cZ}^{\<30mc>})\circ Z+C(\cZ^{\<30m>}, \cZ^{\<30mc>}, Z)+\cZ^{\<30m>}\cZ^{\<31oc3>}+C( \cZ^{\<30mc>},\cZ^{\<30m>}, Z)+\cZ^{\<30mc>}\cZ^{\<31oc>},
\end{align*}
\begin{align*}
	\cZ^{\<31oc1>}\eqdef\cZ^{\<30m>}(\prec+\succ)Z+\cZ^{\<31oc>},\qquad \cZ^{\<31oc2>}\eqdef\cZ^{\<30mc>}(\prec+\succ)Z+\cZ^{\<31oc3>}.
\end{align*}
Here $C$ is the commutator given by Lemma \ref{lem:com1} and all the $\cZ^\tau$ are introduced in \eqref{sto:3}.
\el
\begin{proof} First note that by Lemma \ref{comnew} \begin{equs}[est:IRZ]
		\|\sI(\cR Z)\|_{\bC^{\frac12-2\kappa}}\lesssim \eps^\kappa \|\mZ\|.
	\end{equs} We have the following decomposition
	\begin{equs}[eq:dec1]
		\big(v^\eps* |Y+\psi|^2\big)Z
		=v^\eps* |Y|^2Z+v^\eps*|\psi|^2Z+2v^\eps*\Re[\psi \overline Y]Z.
	\end{equs}
	We compare  the RHS of  \eqref{eq:dec1} with the RHS of \eqref{def:zz}.
	
	\textbf{I. }	For the first term we decompose
	\begin{align}\label{eq:decz}
		v^\eps* |Y|^2Z=&(v^\eps* |Y|^2)(\prec+\succ) Z+(v^\eps* |Y|^2)\circ Z.
	\end{align}
	By Lemmas \ref{lem:para}, \ref{veps} and Proposition \ref{prop:Z}, \eqref{est:IRZ}, we know that for 	the first term
	\begin{align*}
		&\|(v^\eps* |Y|^2)(\prec+\succ) Z-|\cZ^{\<30m>}|^2(\prec+\succ)Z\|_{\bC^{-\frac12-\kappa}}
	\end{align*}
	can be controlled by the RHS of \eqref{est:dif1}.
	
	For the second term on the RHS of \eqref{eq:decz} we have the decomposition
	\begin{equs}[eq:dec2l]
		(v^\eps* |Y|^2)\circ Z=&\,(v^\eps* (Y\prec \overline Y))\circ Z+(v^\eps* (Y\succ \overline Y))\circ Z+(v^\eps* (Y\circ \overline Y))\circ Z.
	\end{equs}
	Using Lemmas \ref{lem:para}, \ref{veps} and Proposition \ref{prop:Z}, \eqref{est:IRZ}, we know that for	the last term on the RHS of \eqref{eq:dec2l}
	\begin{align*}
		\|	(v^\eps* (Y\circ \overline Y))\circ Z-(\cZ^{\<30m>}\circ {\cZ}^{\<30mc>})\circ Z\|_{\bC^{\frac12-4\kappa}}
	\end{align*}
	can be bounded by the RHS of \eqref{est:dif1}.
	
	The first two terms on the RHS of \eqref{eq:dec2l} are similar and we only concentrate on the first one. We use the commutator $C$ from Lemma \ref{lem:com1} to have
	\begin{equs}[vepsc]&v^\eps* (Y\prec \overline Y)\circ Z=\int v^\eps(y) (\tau_y Y\prec \tau_y \overline Y)\circ Z\,\dif y
		\\=& \int v^\eps(y)C(\tau_yY, \tau_y\overline Y, Z)\,\dif y+\int v^\eps(y)\tau_yY(\tau_y\overline Y\circ Z)\,\dif y.
	\end{equs}
	For the first term we use Lemmas \ref{lem:com1}, \ref{veps} and \eqref{v1}, \eqref{est:IRZ} to have
	\begin{align*}
		\Big\|\int v^\eps(y)C(\tau_yY, \tau_y\overline Y, Z)\,\,\dif y-C(\cZ^{\<30m>}, \cZ^{\<30mc>}, Z)\Big\|_{\bC^{\kappa}}
	\end{align*}
	controlled by the RHS of \eqref{est:dif1}. For the second term we write it as
	\begin{align*}\int v^\eps(y)\tau_yY(\tau_y\overline Y\circ Z)\,\dif y=\int v^\eps(y)(\tau_yY-Y)(\tau_y\overline Y\circ Z)\,\dif y+Y(v^\eps*\overline Y\circ Z).
	\end{align*}
	Here we used \eqref{sto:re} to have
	\begin{equs}[sto:re2]
		&v^\eps*\overline Y\circ Z=- \cZ^{\<3v0m1cci>}+v^\eps*\sI(\cR \overline Z)\circ Z
		\\=&\int v^\eps(y)\big(-\cZ^{\<3v0m1cci>}_y+\tau_y\sI(\cR \overline Z)\circ Z\big)\,\dif y =\int v^\eps(y)\tau_y\overline Y\circ Z\,\dif y,\end{equs}
	as the both sides are well-defined for fixed $\eps$. Using Lemma \ref{veps}, Proposition \ref{prop:Zny} and Proposition \ref{prop:Zn1}, \eqref{v1} we obtain that the first term is bounded by the RHS of \eqref{est:dif1}.
	By Lemma \ref{lem:para} and Propositions \ref{prop:Z}, \ref{prop:Zn} , we have
	\begin{align*}
		&	\|Y(v^\eps*\overline Y\circ Z)+\cZ^{\<30m>}C_1-\cZ^{\<30m>}\cZ^{\<31oc3>}\|_{\bC^{-2\kappa}}
		\\\lesssim &\,\|Y\|_{\bC^{\frac12-\kappa}}\|\cZ^{\<31oc3>}+v^\eps*\overline Y\circ Z- C_1\|_{\bC^{-2\kappa}}+\|Y+\cZ^{\<30m>}\|_{\bC^{\frac12-\kappa}}(\|\cZ^{\<31oc3>}\|_{\bC^{-2\kappa}}+C_1).
	\end{align*}
	In fact, we write
	\begin{align*}
		\|\cZ^{\<31oc3>}+v^\eps*\overline Y\circ Z- C_1\|_{\bC^{-2\kappa}}\lesssim \|\cZ^{\<31oc3>}-\cZ^{\<3v0m1cci>}\|_{\bC^{-2\kappa}}+\|v^\eps*\sI(\cR \overline Z)\circ Z- C_1\|_{\bC^{-2\kappa}}
	\end{align*}
	and apply Proposition \ref{prop:Z} and \eref{con:RZ}.
	
	Thus  using \eqref{est:IRZ} $\|v^\eps* |Y|^2Z-( |\cZ^{\<30m>}|^2Z -C_1\cZ^{\<30m>})\|_{\bC^{-\frac12-\kappa}} $ is bounded by the RHS of \eqref{est:dif1}.  \eqref{bd:sto1} follows similarly.

	\textbf{II.}	Concerning the second term on the RHS of \eqref{eq:dec1}, using Lemmas \ref{lem:para} and \ref{veps} we have
	\begin{equs}[est:dif3.4n]
		&\|v^\eps*|\psi|^2Z-|\phi|^2Z\|_{\bC^{-\frac12-\kappa}}\\\lesssim& \Big[\|\psi-\phi\|_{\bC^{\frac12+\kappa}}(\|\psi\|_{L^\infty}+\|\phi\|_{L^\infty})+\|\psi-\phi\|_{L^\infty}(\|\psi\|_{\bC^{\frac12+\kappa}}+\|\phi\|_{\bC^{\frac12+\kappa}})
		\\&+\eps^\kappa \|\psi\|_{\bC^{\frac12+2\kappa}}\|\psi\|_{L^\infty}\Big]\|Z\|_{\bC^{-\frac12-\frac\kappa2}}.
	\end{equs}
	
	\textbf{III.}	For the third term on the RHS of \eqref{eq:dec1}, we  have
	\begin{equs}[v2]
		2v^\eps*\Re[\psi \overline Y]Z=\int v^\eps(y) \tau_y \psi \tau_y\overline Y Z\,\dif y +\int v^\eps(y) \tau_y \overline\psi {\tau_yY}Z\,\dif y.
	\end{equs}
	We only concentrate on the first one and the second one is similar. For the first one we have
	\begin{equs}[est:III]
		\int v^\eps(y) \tau_y \psi \tau_y\overline{Y} Z\dif y=\int v^\eps(y) (\tau_y \psi -\psi)\tau_y(\sI(\cR \overline Z)-\overline{\cZ}^{\<3v0m>}) Z\dif y+\psi(v^\eps*\overline{\sI(\cR Z)}Z-\cZ_\eps^{\<3v0m1c>}),
	\end{equs}
	with
	\begin{align*}
		\cZ_\eps^{\<3v0m1c>}\eqdef(v^\eps*\overline{\cZ}^{\<3v0m>})(\prec+\succ)Z+\cZ_\eps^{\<3v0m1cci>},
	\end{align*}
	for $\cZ_\eps^{\<3v0m1cci>}$ defined in \eqref{sto:5}.
	The $\bC^{-\frac12-\kappa}$ norm of the first term on the RHS of \eqref{est:III} is bounded by
	\begin{align*}&\int |v^\eps(y) |\|\tau_y \psi -\psi\|_{\bC^{\frac12+\kappa}}\sup_y\|\tau_y(\sI(\cR \overline Z)-\overline{\cZ}^{\<3v0m>})Z\|_{\bC^{-\frac{1+\kappa}2}}\,\dif y
		\lesssim \|\psi\|_{\bC^{\frac12+2\kappa}}\eps^\kappa \|\mZ_\eps\|^2.
	\end{align*}
	Here in the last step we used Lemma \ref{veps}, \eqref{v1}, and Lemma \ref{lem:para}, \eqref{eq:tauy}, and Propositions \ref{prop:Zn1},  \ref{prop:Zny}, Lemma \ref{comnew} to have
	\begin{align*}
		&	\sup_y\|\tau_y(\sI(\cR \overline Z)-\overline{\cZ}^{\<3v0m>})Z\|_{\bC^{-\frac{1+\kappa}2}}\\\lesssim &\,\sup_y\|\tau_y(\sI(\cR \overline Z)-\overline{\cZ}^{\<3v0m>})(\prec+\succ)Z\|_{\bC^{-\frac{1+\kappa}2}}+	\sup_y\|\tau_y\sI(\cR \overline Z)\circ Z\|_{\bC^{-2\kappa}}+\sup_y\|\cZ_{y}^{\<3v0m1cci>}\|_{\bC^{-\kappa}}
		\\\lesssim&\,(\| Z\|_{\bC^{-\frac{1+\kappa}2}}+\|\overline{\cZ}^{\<3v0m>}\|_{\bC^{\frac12-\kappa}})\|Z\|_{\bC^{-\frac{1+\kappa}2}}+	\sup_y\|\tau_y\sI(\cR \overline Z)\circ Z\|_{\bC^{-2\kappa}}+\sup_y\|\cZ_{y}^{\<3v0m1cci>}\|_{\bC^{-\kappa}}\lesssim  \|\mZ_\eps\|^2.
	\end{align*}
	By the paraproduct estimates Lemma \ref{lem:para}  and \eqref{est:IRZ} we know that the $\bC^{-\frac12-\kappa}$-norm of the difference between the second term on the RHS of \eqref{est:III}  and $\phi (C_1-\cZ^{\<31oc2>})$ can be controlled by the RHS of \eqref{est:dif1}. Hence, we deduce the result.

\end{proof}

\bl It holds that
\begin{align}\label{bd:sto2}
	&\|\big(v^\eps* |Y|^2\big)Y\|_{C_T\bC^{\frac12-\kappa}}\lesssim \|\mZ_\eps\|^3+1,
\end{align}
and
\begin{align*}
	&\|\big(v^\eps*|Y+\psi|^2\big)(Y+\psi) - |-\cZ^{\<30m>}+\phi|^2(-\cZ^{\<30m>}+\phi)\|_{L^\infty}
	\\\lesssim&\,\mM^2_\infty\big(\|\mZ_\eps-\mZ\|+\eps^\kappa\mM
	+\|\psi-\phi\|_{L^\infty}\big),
\end{align*}
where the proportional constants are independent of $\eps$.
\el
\begin{proof}
	Since $Y\to -\cZ^{\<30m>}$ in $C_T\bC^{\frac12-2\kappa}$,  using the paraproduct estimates Lemma \ref{lem:para} and Lemma \ref{veps}, the  bounds follow.
\end{proof}

\bl\label{lem:dif1} It holds that for $t>0$
\begin{align}\label{bd:sto3}
	&\|v^\eps*\Re[Y\overline Z]Y(t)\|_{\bC^{-\frac12-\kappa}}\lesssim (\|\mZ_\eps\|^3+1)(1+t^{-\kappa}),
\end{align}
and
\begin{equs}[est:dif2]
	&\big\|2\big(v^\eps*\Re[(Y+\psi)\overline Z]\big)(Y+\psi) -2\Re[(-\cZ^{\<30m>}+\phi)\overline Z](-\cZ^{\<30m>}+\phi)
	\\&\qquad-(C_1+2C_2)(-\cZ^{\<30m>}+\phi)\big\|_{\bC^{-\frac12-\kappa}}
	\\\lesssim&\,\mM \mM_\infty\big(\|\mZ_\eps-\mZ\|+\eps^{\frac\kappa2}\|\mZ\|\big)(1+t^{-\kappa})
	+\big(\|\psi-\phi\|_{\bC^{\frac12+2\kappa}}\mM_\infty+\|\psi-\phi\|_{L^\infty}\mM\big)\|\mZ\|,
\end{equs}
where the proportional constants are independent of $\eps$ and
\begin{equs}[def:cZphi]
	&2\Re[(-\cZ^{\<30m>}+\phi)\overline Z](-\cZ^{\<30m>}+\phi)
	\\\eqdef&\,2 \Re[\cZ^{\<30m>}\overline Z]\cZ^{\<30m>}-2\Re(\cZ^{\<31oc5>})\phi+2\Re[\phi\overline Z]\phi-(\phi \cZ^{\<31oc5>} +\overline{\phi} \cZ^{\<31oc1>}),
\end{equs}
with
\begin{align*}
	2\Re[\cZ^{\<30m>}\overline Z]\cZ^{\<30m>}\eqdef&\; 2\Re[\cZ^{\<30m>}\prec\overline Z](\prec+\succ)\cZ^{\<30m>}+ 2\Re[\cZ^{\<30m>}\succ\overline Z]\cZ^{\<30m>}+ (\cZ^{\<31oc4>}+\cZ^{\<31oc3>})\cZ^{\<30m>}
	\\&+C(\cZ^{\<30m>}, \overline Z,\cZ^{\<30m>} )+\cZ^{\<30m>}\cZ^{\<31oc4>}+C(\cZ^{\<30mc>},  Z,\cZ^{\<30m>} )+\cZ^{\<30mc>}\cZ^{\<31oc>},
\end{align*}
and
\begin{align*}
	\cZ^{\<31oc5>}\eqdef\cZ^{\<30m>}(\prec+\succ)\overline{Z}+\cZ^{\<31oc4>},\qquad \cZ^{\<31oc1>}\eqdef\cZ^{\<30m>}(\prec+\succ)Z+\cZ^{\<31oc>} .
\end{align*}
\el
\begin{proof}
	We have the following decomposition
	\begin{equs}[eq:dec3]
		&
		\big(v^\eps*\Re[(Y+\psi)\overline Z]\big)(Y+\psi)
		\\=&\,v^\eps*\Re[Y\overline Z]Y+v^\eps*\Re[Y\overline Z]\psi+v^\eps*\Re[\psi\overline Z]\psi+v^\eps*\Re[\psi \overline Z]Y.
	\end{equs}

	\textbf{I.}	We start with the first term on the RHS of \eqref{eq:dec3} and
	by the paraproduct decomposition:
	\begin{align}\label{eq:dec4l}
		(v^\eps*\Re[Y \overline Z])Y=(v^\eps*\Re[Y\prec \overline Z])Y+(v^\eps*\Re[Y \succ\overline Z])Y+(v^\eps*\Re[Y\circ \overline Z])Y.
	\end{align}
	We have for the last term
	\begin{align*}
		&\|(v^\eps*\Re[Y\circ \overline Z])Y+C_2\cZ^{\<30m>}-\frac12(\cZ^{\<31oc4>}+\cZ^{\<31oc3>})\cZ^{\<30m>}\|_{\bC^{-2\kappa}}
		\\\lesssim &\,\|(v^\eps*\Re[\sI(\cR Z)\circ \overline Z])Y+C_2\cZ^{\<30m>}\|_{\bC^{-2\kappa}}+\|(v^\eps*\Re[-\cZ_\eps^{\<3v0m1ci2>}])Y-\frac12(\cZ^{\<31oc4>}+\cZ^{\<31oc3>})\cZ^{\<30m>}\|_{\bC^{-2\kappa}},
	\end{align*}
	which by Propositions \ref{prop:Z}, \ref{prop:Zn},  Lemmas \ref{lem:para}, \ref{veps} and \eqref{est:IRZ} is bounded by the RHS of \eqref{est:dif2}.
	For the second term
	$(v^\eps*\Re[Y \succ\overline Z])Y$ we use Lemma \ref{lem:para}, Proposition \ref{prop:Z} and \eqref{est:IRZ} to have that
	\begin{align*}
		\|(v^\eps*\Re[Y \succ\overline Z])Y-\Re[\cZ^{\<30m>}\succ\overline Z]\cZ^{\<30m>}\|_{\bC^{-3\kappa}}
	\end{align*}
	is bounded by the  RHS of \eqref{est:dif2}.
	For the first term on the RHS of \eqref{eq:dec4l} we  decompose it as
	\begin{align*}
		(v^\eps*\Re[Y\prec \overline Z])Y=(v^\eps*\Re[Y\prec \overline Z])(\prec+\succ)Y+(v^\eps*\Re[Y\prec \overline Z])\circ Y.
	\end{align*}
	By the paraproduct estimates Lemma \ref{lem:para}, Lemma \ref{veps}  and \eqref{est:IRZ} we have for the first term,
	\begin{align*}
		\|(v^\eps*\Re[Y\prec \overline Z])(\prec+\succ)Y-\Re[\cZ^{\<30m>}\prec\overline Z](\prec+\succ)\cZ^{\<30m>}\|_{\bC^{-\frac12-\kappa}}
	\end{align*}
	is bounded by the RHS of \eqref{est:dif2}.
	For the resonant term we have
	\begin{align*}
		&(v^\eps*[Y\prec \overline Z])\circ Y=\int v^\eps(y)(\tau_yY\prec \tau_y \overline Z)\circ Y\,\dif y
		\\=&\int v^\eps(y) C(\tau_yY, \tau_y \overline Z, Y)+\int v^\eps(y) (\tau_yY-Y)(\tau_y \overline Z\circ Y)\,\dif y-Y\cZ_\eps^{\<3v0m1ci10>}+Y\sI(\cR Z)\circ v^\eps*\overline Z,
	\end{align*}
	with $\cZ_\eps^{\<3v0m1ci10>}$ introduced in \eqref{sto:5}, where we also used \eqref{sto:re} and similar calculation as \eqref{sto:re2} to have $\int v^\eps(y) \tau_y \overline Z\circ Y\,\dif y=-\cZ_\eps^{\<3v0m1ci10>}+\sI(\cR Z)\circ v^\eps*\overline Z$.
	Using similar argument as that for  \eqref{vepsc} in the proof of Lemma \ref{lem:conn1},  we know  that the difference between these terms and
	$C(\cZ^{\<30m>},  \overline{Z},\cZ^{\<30m>} )+ \cZ^{\<30m>}\cZ^{\<31oc4>}-C_1\cZ^{\<30m>}$
	in $\bC^{-2\kappa}$ is bounded by the RHS of \eqref{est:dif2}. For $(v^\eps*[\overline Y\prec  Z])\circ Y$ the required bounds follow exactly the same way.
	Thus we obtain that	$\|	v^\eps*\Re[Y\overline Z]Y- \Re[\cZ^{\<30m>}\overline Z]\cZ^{\<30m>}+(\frac12C_1+C_2) \cZ^{\<30m>}
	\|_{\bC^{-\frac12-\kappa}}$
	is bounded by the RHS of \eqref{est:dif2}.
	\eqref{bd:sto3} follows by a similar argument.
	
	\textbf{II.}	For the second term on the RHS of \eqref{eq:dec3},  we have
	\begin{align*}
		&\|v^\eps*\Re[Y\overline Z]\psi+\Re(\cZ^{\<31oc5>})\phi-C_2\phi\|_{\bC^{-\frac12-\kappa}}
		\\\lesssim&\,\|v^\eps*\Re[\sI(\cR Z)\overline Z]\psi-C_2\phi\|_{\bC^{-\frac12-\kappa}}+\|-v^\eps*\Re[\cZ_\eps^{\<3v0m1ci3>}]\psi+\Re(\cZ^{\<31oc5>})\phi\|_{\bC^{-\frac12-\kappa}},
	\end{align*}
	with
	\begin{align*}
		\cZ_\eps^{\<3v0m1ci3>}\eqdef \cZ_\eps^{\<3v0m>}(\prec+\succ)\overline Z+	\cZ_\eps^{\<3v0m1ci2>},
	\end{align*}
	which by Lemmas \ref{lem:para},  \ref{veps},  Proposition \ref{prop:Z} and \eqref{con:RZ2},  \eqref{est:IRZ}, can be bounded by the RHS of \eqref{est:dif2}.
	
	\textbf{III.}	For the third term $v^\eps*\Re[\psi\overline Z]\psi$ on the RHS of \eqref{eq:dec3}
	$$v^\eps*\Re[\psi\overline Z]\psi=v^\eps*\Re[\psi\prec\overline Z]\psi+v^\eps*\Re[\psi\succcurlyeq\overline Z]\psi.$$
	Then by Lemmas \ref{lem:para} and \ref{veps} we  have
	\begin{align*}
		\|v^\eps*\Re[\psi\overline Z]\psi-\Re[\phi\overline Z]\phi\|_{\bC^{-\frac12-\kappa}}
	\end{align*}
	can be bounded by the RHS of \eqref{est:dif2}.
	
	\textbf{IV. }	We have that for the fourth term on the RHS of \eqref{eq:dec3}
	\begin{equs}[v3]
		2v^\eps*\Re[\psi \overline Z]Y=\int v^\eps(y) \tau_y \psi \tau_y\overline Z Y\,\dif y +\int v^\eps(y) \tau_y \overline\psi {\tau_yZ}Y\,\dif y.
	\end{equs}
	Using similar argument as  \textbf{III.} in the proof of Lemma \ref{lem:conn1} we obtain that $\|v^\eps*\Re[\psi \overline Z]Y	+\frac12(\phi \cZ^{\<31oc5>} +\overline{\phi} \cZ^{\<31oc1>})-\frac12 C_1\phi\|_{\bC^{-\frac12-\kappa}}$ is bounded by the RHS of \eqref{est:dif2}.

\end{proof}

\subsection{Proof of Lemma \ref{comnew} and Lemma \ref{boundr1}}\label{sec:comnew}
We first recall the following notations from Appendix \ref{sec:pre}. We will use the dyadic partition of unity  $(\theta_{-1},\theta)$ in the proofs, and refer to Appendix \ref{sec:pre} for its definition. We write $\theta_j=\theta(2^{-j}\cdot)$ for $j\geq0$. Let $\Delta_j$ be the Littlewood-Paley blocks.  Let $\tilde \theta\in C_c^\infty(\mR^3)$ with support in an annulus such that $\tilde\theta \theta=\theta$.  Define $K_j=\mathcal{F}^{-1}\theta_j$, $\tilde{K}_j=\mathcal{F}^{-1}\tilde \theta_j$ with
$\tilde\theta_j=\tilde\theta(2^{-j}\cdot)$.
We recall the following result from \cite{ZZ15}.

\bl\label{lem:sum} (\cite[Lemma 3.10]{ZZ15}) Let $0<l,m<d,l+m-d>0$. Then we have for $k\neq0$
$$\sum_{k_1\in \mathbb{Z}^d\backslash\{0\},|k-k_1|\geq\frac12}\frac{1}{|k_1|^{l}|k-k_1|^{m}}\lesssim \frac{1}{|k|^{l+m-d}}.$$
\el

\begin{proof}[Proof of Lemma \ref{comnew}]

	We use \eqref{eq:fourier} to have for $j\geq0$
	\begin{align*}\Delta_j(\cR f)=&\sum_{k_1}\theta_j(k_1)\widehat{\cR}(k_1)\langle f,e_{k_1}\rangle e_{k_1}=\sum_{k_1}\theta_j(k_1)\tilde \theta_j(k_1)\widehat{\cR}(k_1)\langle f,e_{k_1}\rangle e_{k_1}
		\\=&\,\mathcal{F}^{-1}(\tilde \theta_j \widehat{\cR})*\Delta_j f,
	\end{align*}
	where
	\begin{equs}[def:FvG]
		\widehat{\cR}(k_1)\eqdef\frac1{(2\pi)^{3}}\sum_k\frac1{\la k\ra^2}\Big(\widehat{v^\eps} (k)-\widehat{v^\eps}(k_1-k)\Big).\end{equs}
	We then obtain
	$$\|\Delta_j(\cR f)\|_{L^p}\lesssim \|\mathcal{F}^{-1}(\tilde \theta_j \widehat{\cR})\|_{L^1(\mR^3)}\|\Delta_j f\|_{L^p}.$$
	It remains to bound $\|\mathcal{F}^{-1}(\tilde \theta_j \widehat{\cR})\|_{L^1(\mR^3)}$,  which equals to
	\begin{align*}&\big\|\mathcal{F}^{-1}(\tilde \theta \widehat{\cR}(2^j\cdot))\big\|_{L^1(\mR^3)}\lesssim \big\|(1+|x|^2)\mathcal{F}^{-1}(\tilde \theta \widehat{\cR}(2^j\cdot))\big\|_{L^2(\mR^3)}
		\\\lesssim&\,\big\|\mathcal{F}^{-1}(I-\Delta)(\tilde \theta \widehat{\cR}(2^j\cdot))\big\|_{L^2(\mR^3)}=\big\|(I-\Delta)(\tilde \theta \widehat{\cR}(2^j\cdot))\big\|_{L^2(\mR^3)}.
	\end{align*}
	We have
	\begin{equs}[vgg]
		\widehat{\cR}(2^jk_1)=&\frac1{(2\pi)^{3}}\sum_k \frac1{\la k\ra^2}\Big(\widehat{v^\eps} (k)-\hat{ v^\eps} (2^jk_1-k)-\nabla\widehat{v^\eps} (k)\cdot 2^j k_1\Big),
	\end{equs}
	where  we used the fact that $\nabla\widehat{v^\eps} (k)$ is odd in $k$. Recall that $\tilde \theta$ is support on an annulus.  Subsequently, we employ $(\mathbf{Hv})$ to have for $k_1$ on the support of $\tilde\theta$
	\begin{align*}|\widehat{v^\eps} (k)-\widehat{v^\eps} (2^jk_1-k)-\nabla\widehat{v^\eps} (k)\cdot 2^jk_1|\lesssim&\,(2^j|k_1|)^2\int_0^1\int_0^1 | \nabla^2\widehat{v^\eps}(k-su2^jk_1)|\,\dif u\dif s
		\\\lesssim &\,\eps^2|2^j k_1|^2\int_0^1\int_0^1\frac1{1+\eps^2 |su2^jk_1-k|^{2}}\,\dif u\dif s,
	\end{align*}
	which by Lemma \ref{lem:sum} implies that for $\eta\in(0,1)$ and  for $k_1$ on the support of $\tilde \theta$
	\begin{equs}[vg]
		|\widehat{\cR}(2^jk_1)|\lesssim&\,\eps^{2}|2^jk_1|^2\int_0^1\int_0^1\sum_{k\neq0}\frac1{\la k\ra^2} \frac1{1+\eps^2 |su2^jk_1-k|^{2}}\,\dif u\dif s+\eps^{\eta}|2^jk_1|^{\eta}
		\\	\lesssim&\,\eps^{\eta}|2^jk_1|^2\int_0^1\int_0^1\sum_{k\neq0,|su2^jk_1-k|\geq\frac12} \frac1{\la k\ra^2}\frac1{|su2^jk_1-k|^{2-\eta}}\,\dif u\dif s
		\\&+\eps^{2}|2^jk_1|^2\int_0^1\int_0^1\sum_{k\neq0,|su2^jk_1-k|<\frac12} \frac1{\la k\ra^2}\,\dif u\dif s+\eps^{\eta}|2^jk_1|^{\eta}
		\\\lesssim&\,\eps^{\eta}|2^jk_1|^2\int_0^1\int_0^1\frac1{|su2^jk_1|^{1-\eta}}\,\dif u\dif s+\eps^{\eta}|2^jk_1|^{\eta}\lesssim\eps^{\eta}2^{(1+\eta)j}|k_1|^{1+\eta}.
	\end{equs}
	In the first line, the last term arises from the calculation for $k=0$.   Additionally, we utilized the fact that  when $|k|\geq1,|su2^jk_1-k|<\frac12$, $|su2^jk_1|\simeq |k|$ in the third step.
	Moreover, we have
	\begin{equs}[nablavg]
		\nabla \widehat{\cR}(2^j\cdot)(k_1)=&\frac{2^j}{(2\pi)^3} \sum_k \frac1{\la k\ra^2}(\nabla\widehat{v^\eps} (k-2^jk_1)-\nabla\widehat{v^\eps} (k) ).
	\end{equs}
	Thus using \eqref{nablavg} and similar as \eqref{vg} we obtain on the support of $\tilde\theta$
	\begin{equs}[est:vgd1]
		|\nabla \widehat{\cR}(2^j\cdot)(k_1)|\lesssim&\,2^{2j}\eps^2\int_0^1 \sum_{k\neq0,|2^jk_1s-k|\geq\frac12}\frac1{\la k\ra^2}\frac1{\eps^{2-\eta}|2^jk_1s-k|^{2-\eta}}\,\dif s
		\\&+2^{2j}\eps^2\int_0^1 \sum_{k\neq0,|2^jk_1s-k|<\frac12}\frac1{\la k\ra^2}\,\dif s+\eps 2^j
		\lesssim \eps^{\eta}2^{(1+\eta)j}.
	\end{equs}
	Furthermore, we have
	\begin{equs}
		\nabla^2 \widehat{\cR}(2^j\cdot)(k_1)=&-\frac{2^{2j}}{(2\pi)^3}\sum_k\frac1{\la k\ra^2}\nabla^2\widehat{v^\eps} (2^jk_1-k).
	\end{equs}
	We subsequently  apply Lemma \ref{lem:sum},  \textbf{(Hv)} and similar calculation as \eqref{vg} to have on the support of $\tilde\theta$
	\begin{equs}[est:Fvgd2]
		&|\nabla^2 \widehat{\cR}(2^j\cdot)(k_1)|\lesssim 2^{2j}\eps^2\sum_{k\neq0,|2^jk_1-k|\geq\frac12}\frac1{\la k\ra^2|\eps (2^jk_1-k)|^{2-\eta}}+\eps^\eta 2^{j(1+\eta)}
		\lesssim\eps^{\eta}2^{j(1+\eta)}.
	\end{equs}
	Combining \eqref{vg}, \eqref{est:vgd1} and \eqref{est:Fvgd2} we obtain for $j\geq0$
	$$\|\mathcal{F}^{-1}(\tilde \theta_j \widehat{\cR})\|_{L^1(\mR^3)}\lesssim \eps^{\eta} 2^{j(1+\eta)}.$$
	Moreover, since $\theta_{-1}$ is supported in a ball,  we employ the Bernstein-type lemma along with \eqref{vg} to have
	\begin{equs}[-1est]\|\Delta_{-1}(\cR f)\|_{L^p}^2\lesssim& \,\|\Delta_{-1}(\cR f)\|_{L^2}^2\lesssim \sum_{k_1}|\widehat{\cR}(k_1)|^2\theta_{-1}(k_1)^2|\langle f,e_{k_1}\rangle |^2
		\\\lesssim &\,\eps^{2\eta}\sum_{k_1}\theta_{-1}(k_1)^2|\langle f,e_{k_1}\rangle |^2\lesssim \eps^{2\eta}\|\Delta_{-1}(f)\|_{L^p}^2,
	\end{equs}
	which implies Lemma \ref{comnew}.
\end{proof}

\br We can also prove Lemma \ref{comnew} using the properties of the Green function $G(x-y)$. Here we  use the Fourier analysis-based proof, as we will later rely on some estimates established in the lemma.
\er

\begin{proof}[Proof of Lemma \ref{boundr1}] We have for $j\geq0$
	$$\Delta_j \mathcal R_1 f=\sum_k \widehat{\cR_1}(k)\theta_j(k)\langle f,e_k\rangle e_k=\mathcal{F}^{-1}(\tilde \theta_j \widehat{\cR_1})* \Delta_j f.$$
	We then obtain
	$$\|\Delta_j(\cR_1 f)\|_{L^p}\lesssim \|\mathcal{F}^{-1}(\tilde \theta_j \widehat{\cR_1})\|_{L^1(\mR^3)}\|\Delta_j f\|_{L^p}.$$
	It remains to bound $\|\mathcal{F}^{-1}(\tilde \theta_j \widehat{\cR_1})\|_{L^1(\mR^3)}$,  which equals to
	\begin{align*}&\big\|\mathcal{F}^{-1}(\tilde \theta \widehat{\cR_1}(2^j\cdot))\big\|_{L^1(\mR^3)}\lesssim \big\|(1+|x|^2)\mathcal{F}^{-1}(\tilde \theta \widehat{\cR_1}(2^j\cdot))\big\|_{L^2(\mR^3)}
		\\\lesssim&\,\big\|\mathcal{F}^{-1}(I-\Delta)(\tilde \theta \widehat{\cR_1}(2^j\cdot))\big\|_{L^2(\mR^3)}=\big\|(I-\Delta)(\tilde \theta \widehat{\cR_1}(2^j\cdot))\big\|_{L^2(\mR^3)}.
	\end{align*}
	We use the condition $(\mathbf{Hv})$ to have $$|\widehat{v^\eps}(2^jk+k_1)^2-\widehat{v^\eps}(k_1)^2|\lesssim \eps^{2\kappa} 2^{2j\kappa}|k|^{2\kappa}\Big(\frac1{|\eps (2^jk+k_1)|^{\kappa}}+\frac1{|\eps k_1|^{\kappa}}\Big),$$
	which combined with Lemma \ref{lem:sum} and similar argument as \eqref{vg} implies that
	\begin{align*}|\widehat{\cR_1}(2^jk)|\lesssim &\,\eps^\kappa 2^{2j\kappa} \sum_{k_1\neq0,| 2^jk+k_1|\geq1/2}\frac1{|k_1|^3+1}\frac1{| 2^jk+k_1|^{\kappa}}+\eps^{\kappa } 2^{2j\kappa}
		\lesssim \eps^\kappa 2^{2j\kappa}.
	\end{align*}
	Moreover, the $m$-th ($|m|\leq 2$) derivative of $\widehat{v^\eps}(2^jk+k_1)^2-\widehat{v^\eps}(k_1)^2$ is bounded by
	$$\frac{\eps^{m} 2^{jm}}{|\eps (2^jk+k_1)|^{m}+1}.$$
	Using Lemma \ref{lem:sum} and similar argument as in the proof of Lemma \ref{comnew} we  obtain the result.
\end{proof}

\section{Uniform in $\eps$  estimates and convergence of the measures $\nu^\eps$}\label{sec:uni}
In this section, we prove the convergence of the measure $\nu^\eps$ given in \eqref{e:Phi_eps-measure1} to the $\Phi^4_3$ field $\nu$ from \eqref{e:Phi-measure1}. The first aim of this section is to derive the following uniform in $\eps$ moment bounds for $\nu^\eps$ via the uniform estimates from the dynamic \eqref{eq:mainnew}. In this section we assume that $v$  satisfies $\textbf{(Hv)}$.

\bt\label{uniformmea} Let $\Psi^\eps$ be the stationary solution to equation \eqref{eq:mainnew}. Then it holds that for any $t\geq0$ and $p\geq2$
\begin{align*}\sup_\eps\E\|\Psi^\eps(t)\|^2_{\bC^{-\frac12-\kappa}}+
	\sup_\eps\E\|\Psi^\eps(t)\|^p_{\bB_2^{-\frac12-\kappa}}\lesssim1.
\end{align*}
\et

As $\nu^\eps$ from \eqref{e:Phi_eps-measure1} is the unique invariant measure of the solutions to equation \eqref{eq:mainnew} (c.f. Lemma \ref{lem:zz1} below), Theorem \ref{uniformmea} is enough to derive tightness of $\nu^\eps,\eps>0$ in $\bC^{-\frac12-\kappa}$.
Furthermore, based on Theorem \ref{uniformmea} and Theorem \ref{th:1} we can further identify the limit of $\nu^\eps$  and derive the convergence of $\nu^\eps$ to the $\Phi^4_3$ field $\nu$ and the convergence of the related $n$-point functions.

\bt\label{th:main1} As $\eps\to 0$, $\nu^\eps$ converges weakly to $\nu$ in $\bC^{-\frac12-\kappa}$ with $\nu$ being the $\Phi^4_3$ field given by the unique invariant measure of the solutions to  equation \eqref{eq:mainnew1}. Furthermore, any correlation function $\gamma_n^\eps$ of $\nu^\eps$
converge to the $n$-point correlation function $\gamma_n$ of $\Phi^4_3$ field in $\cS'(\mT^{6n})$.
\et

The main idea of the proof of Theorem \ref{uniformmea} is to establish uniform in $\eps$ moment bounds for the stationary solutions $\Psi^\eps$ through the  dynamics \eqref{eq:mainnew}.  By the decomposition \eqref{eq:decPsi}, equation \eqref{eq:mainnew} can be reduced to equation \eqref{eq:main2new}. Hence, in the sequel we mainly concentrate on  equation \eqref{eq:main2new}. To this end, we rewrite \eqref{eq:main2new} as follows:
\begin{equs}[eq:psi]
	\LL \psi=& \, -v^\eps*|\psi|^2\psi-(Y+\psi)\cZ^{\<2vm>}-2v^\eps*\Re[(Y+\psi) \Wick{\overline Z]Z}
	\\&-6\tilde b^\eps(Z+Y+\psi)+f_2(\psi)+f_1(\psi)+f_0,
	\\\psi(0)=&\,\Psi^\eps(0)-Z(0).
\end{equs}
Here \begin{align*}
	f_0\eqdef&-v^\eps* |Y|^2(Z+Y)-2v^\eps*\Re[Y\overline Z]Y+\mathcal{R}Y
	\\&+(6(\tilde b^\eps-b^\eps)-m+1)(Z+Y),\\
	f_1(\psi)\eqdef &-2v^\eps*\Re[\psi \overline Y]Z-2v^\eps*\Re(\psi \overline Y)Y-v^\eps*|Y|^2\psi\\&-2v^\eps*\Re[Y\overline Z]\psi-2v^\eps*\Re[\psi \overline Z]Y+\mathcal{R}\psi
	\\&+(6(\tilde b^\eps-b^\eps)-m+1)\psi,\\
	f_2(\psi)\eqdef&-Zv^\eps*|\psi|^2-2v^\eps*\Re(\psi \overline Y)\psi-v^\eps*|\psi|^2Y-2v^\eps*\Re[\psi\overline Z]\psi.
\end{align*}
We note that $f_0$ arises from the purely stochastic terms in the third and fourth lines of equation \eqref{eq:main2new}, while $f_1(\psi)$ corresponds to the linear part  of these lines, and $f_2(\psi)$ depends quadratically on $\psi$.
Fix $T>0$. For fixed $\eps>0$ by Theorem \ref{th:global} there exists a unique solution $\psi\in C_T\bC^{-\frac12-\kappa}$ to \eqref{eq:psi} for any $\Psi^\eps(0)\in\bC^{-\frac12-\kappa}$. In the following we will prove global and uniform in $\eps$ estimates for the solution $\psi$.

We adjust $\|\mZ_\eps\|$ from Section \ref{sec:dif} to represent the smallest number greater than both $\|\mZ_\eps\|$ and $\|\mZ-\mZ_\eps\|$ in Section \ref{sec:dif}.
Similar as in Section \ref{sec:dif} we know that for any $r\geq1$
\begin{equs}[bd:mZ1]\sup_\eps\E\|\mZ_\eps\|^r\lesssim1.
\end{equs}
We also use $K(\|\mZ_\eps\|)$ to denote a generic polynomial depending on $\|\mZ_\eps\|$ for the stochastic terms. The coefficients of $K$ are independent of $\eps$ and may change from line to line.

Using \eqref{bd:sto1}, \eqref{bd:sto2}, \eqref{bd:sto3} and \eqref{bd:cRY}, \eqref{dif:esb},  we find that $f_0\in \bC^{-\frac12-2\kappa}$ and for $t>0$
\begin{align}\label{bd:f_0}
	\|f_0(t)\|_{\bC^{-\frac12-2\kappa}}\lesssim (\|\mZ_\eps\|^3+1)(1+t^{-\kappa}).
\end{align}

In the following we decompose $\psi$ into $\psi_l+\psi_h$, as detailed in \eqref{eqv} and \eqref{eqw} below. To this end, we introduce the  localizers in terms of Littlewood--Paley expansions. Let $J\in\mathbb{R}_{+}$. For $f\in\mathcal{S}'(\mathbb{T}^{d})$ we define
\begin{equs}[def:loc]
	\Delta_{> J}f=\sum_{j> J}\Delta_{j}f,\qquad \Delta_{\leq J}f=\sum_{j\leq J}\Delta_{j}f.
\end{equs}
Then, in particular, for $\alpha\leq \beta\leq\gamma$, it holds
\begin{equation}\label{eq:loc}
	\begin{aligned}
		\|\Delta_{> J}f\|_{\bC^{\alpha}}\lesssim 2^{-J(\beta-\alpha)}\|f\|_{\bC^{\beta}},\qquad
		\|\Delta_{\leq J}f\|_{\bC^{\gamma}}\lesssim 2^{J(\gamma-\beta)}\|f\|_{\bC^{\beta}},
	\end{aligned}
\end{equation}
where the proportional constants are independent of $f$.

Now we decompose $\psi=\psi_l+\psi_h$ with $\psi_l$ and $\psi_h$ satisfying
\begin{equs}[eqv]	\LL \psi_h=&\cG_{\prec,> R}(Y+\psi),
\end{equs}
\begin{equs}[eqw]\LL \psi_l=&-v^\eps*|\psi|^2\psi+f_2(\psi)+f_1(\psi)+f_0+\cG_{\prec,\leq R}(Y+\psi)+F_{\succcurlyeq}+\cG_{\succcurlyeq},
\end{equs}
with initial condition $\psi_l(0)=\psi(0), \psi_h(0)=0$ and
\begin{align*}
	\cG_{\succcurlyeq}\eqdef &-\Big(\psi\succcurlyeq\cZ^{\<2vm>}+(\tilde b_1^\eps+\tilde b^\eps_2)(Y+\psi)\Big) -\Big(\int v^\eps(y)\tau_y\psi\succcurlyeq\cZ^{\<2m>}_y\,\dif y+(\tilde b_2^\eps+\tilde b^\eps_3)(Y+\psi)\Big)
	\\&	-\Big(\int v^\eps(y)\tau_y\overline{\psi}\succcurlyeq\cZ^{\<2>}_y\,\dif y+\tilde b_5^\eps(Y+\psi)\Big),
\end{align*}
and
\begin{align*}
	F_\succcurlyeq\eqdef&-\Big(Y\succcurlyeq\cZ^{\<2vm>}+(\tilde b_1^\eps+\tilde b^\eps_2)Z\Big) -\Big(\int v^\eps(y)\tau_yY\succcurlyeq\cZ^{\<2m>}_y\,\dif y+(\tilde b_2^\eps+\tilde b^\eps_3)Z\Big)
	\\&	-\Big(\int v^\eps(y)\tau_y\overline{Y}\succcurlyeq\cZ^{\<2>}_y\,\dif y+\tilde b_5^\eps Z\Big),
\end{align*}
and
\begin{align*}
	\cG_{\prec,> R}(f)\eqdef- f\prec\Delta_{> R}\cZ^{\<2vm>}-\int v^\eps(y)\tau_yf\prec\Delta_{> R}\cZ^{\<2m>}_y\,\dif y
	-\int v^\eps(y)\tau_y\overline{f}\prec\Delta_{> R}\cZ^{\<2>}_y\,\dif y,
\end{align*}
\begin{align*}
	\cG_{\prec,\leq R}(f)\eqdef- f\prec\Delta_{\leq R}\cZ^{\<2vm>}-\int v^\eps(y)\tau_yf\prec\Delta_{\leq R}\cZ^{\<2m>}_y\,\dif y
	-\int v^\eps(y)\tau_y\overline{f}\prec\Delta_{\leq R}\cZ^{\<2>}_y\,\dif y.
\end{align*}
Here $\Delta_{>R}$ and $\Delta_{\leq R}$ are localizers introduced in \eqref{def:loc}, where $R$, (which depends on time,) will be given later and we used \eqref{vepsde}.

In the following we will prove the following coming down from infinity result for the solutions obtained by Theorem \ref{th:global}.

\bt\label{coming} Suppose that $\psi(0)\in \bC^{-\frac12-\kappa}$. For every $t\in(0,T]$, we have
\begin{align}\label{eq:psi:com}
	\|\psi_h(t)\|_{\bB^{\frac1{17}}_4}+\|\psi_l(t)\|_{L^2}\lesssim K(\|\mZ_\eps\|)(1+t^{-\frac12}).\end{align}
Here  the implicit constant is independent of $\eps$ and initial value.
\et

Note that the RHS of \eqref{eq:psi:com} is independent of the initial data $\psi(0)$. In the following we will prove Theorem \ref{coming} under assumption that $\psi(0)\in L^2\cap \bC^{-\frac12-\kappa}$.  Since the solution is continuous w.r.t. the initial data in $\bC^{-\frac12-\kappa}$, as established in Section \ref{first}, we can deduce that \eqref{eq:psi:com} holds for general initial data $\psi_0\in \bC^{-\frac12-\kappa}$ by approximating the initial data (see \cite[Remark 2.4]{MW18}).

The idea to prove Theorem \ref{coming} is to combine $L^2$-energy estimates for the low frequency part of $\psi$ (i.e. $\psi_l$) and smoothing effect of heat operators to establish  uniform in $\eps$ estimates for $\psi$. When  performing the $L^2$-energy estimate for $\psi_l$, we encounter a useful term $\cV^\eps(\psi_l)$ from the nonlinearity, which is defined for $f\in L^2$
\begin{equs}[def:veps]
	\mathcal{V}^\eps(f)\eqdef\int v^\eps(x-y)|f|^2(x)|f|^2(y)\,\dif x\dif y.\end{equs}
We have a useful lower bound for $\cV^\eps(f)$:
\begin{align}\label{lowerbound}\mathcal{V}^\eps(f)=&\sum_k\widehat{v^\eps}(k)|\langle |f|^2,e_k\rangle|^2\gtrsim \|f\|_{L^2}^4,
\end{align}
where  in the last step we used $\widehat{v^\eps}(k)\geq0$, and when $k=0$ in the Fourier transform, we obtain $\|f\|_{L^2}^4$. Using condition $(\mathbf{Hv})$ we can also derive
\begin{align}\label{lowerbound1}
	\|v^\eps*|f|^2\|_{L^2}^2=\sum_k\widehat{v^\eps}(k)^2|\langle |f|^2,e_k\rangle|^2\lesssim \cV^\eps(f).
\end{align}

\br \label{reglobal} For the dynamical $\Phi^4_3$ model, one can use the strong damping term $-|\phi|^2\phi$ and the maximum principle (see \cite{GH18, MW20}) or $L^p$-energy estimates (see \cite{MW18}) to obtain global estimates. However, for $v^\eps * |\psi|^2 \psi$, it appears that only the $L^2$-energy estimate is effective. In \cite{GH18a}, the authors proved a global estimate using only the $L^2$-energy estimate, without relying on Schauder estimates for the dynamical $\Phi^4_3$ model, through a duality argument between the paraproducts $\prec$ and $\circ$. The technique used there may also apply here (see also \cite{SZZ21}), but our proof combines Schauder estimates and the  $L^2$-energy estimate, which simplifies such kind of argument in \cite{AK17, MW18} by introducing localizers. We believe this proof is of independent interest and is more compatible with the analysis in Section \ref{first}, so we have chosen to present it this way.
\er

In the following, we employ  localizer and smoothing effect of heat operator to derive regularity estimates for $\psi_h$, and $\psi^\sharp$ in Sections \ref{sec:epsih}-\ref{sec:epsis}.  We then perform $L^2$-energy estimates for $\psi_l$ in  Section \ref{sec:en} and make use of the dissipation effect from the cubic nonlinearity to derive uniform estimates. The proofs of Theorems \ref{uniformmea}--\ref{coming} are given in Section \ref{proof}.
\subsection{Estimates of $\psi_h$}\label{sec:epsih} In this section we use Duhamel's formula and the smoothing effect of heat operator  to derive a priori estimate for $\psi_h$.
We make use of  the localizers $\Delta_{>R}$ present in the equation for $\psi_h$ in \eqref{eqv}. Namely, by an appropriate choice of $R$ we can always apply \eqref{eq:loc} to get a small constant in front of terms which contain $\psi$. We are therefore able to obtain a bound independent of $\psi_l$ (see \eqref{boundpsih} below).

We use Duhamel's formula to write for $0\leq t\leq T$
\begin{align*}
	\psi_h(t)=&\int_0^tP_{t-s}\cG_{\prec,> R}(Y+\psi)\,\dif s.
\end{align*}
We then use the smoothing effect of heat operator  Lemma \ref{lem:heat} and the paraproduct estimates Lemma \ref{lem:para} and \eqref{eq:tauy} to obtain
\begin{align*}\|\psi_h(t)\|_{\bB^{\frac1{17}}_4}\lesssim& \,K(\|\mZ_\eps\|)+\int_0^t (t-s)^{-1+\frac\kappa2}\|\psi\|_{L^4}\Big(\sup_y\|\Delta_{>R}\cZ^{\<2m>}_y\|_{\bC^{-\frac{33}{17}+\kappa}}\\&+\sup_y\|\Delta_{>R}\cZ^{\<2>}_y\|_{\bC^{-\frac{33}{17}+\kappa}}
	+\|\Delta_{>R} \cZ^{\<2vm>}\|_{\bC^{-\frac{33}{17}+\kappa}}\Big)\,\dif s
	\\\lesssim&\, K(\|\mZ_\eps\|)+\|\mZ_\eps\|\int_0^t (t-s)^{-1+\frac\kappa2}\|\psi(s)\|_{L^4}2^{-R(s)(\frac{16}{17}-2\kappa)}\,\dif s,
\end{align*}
where we used the localizer estimate \eqref{eq:loc} in the last step.

We choose $2^{R(s)(\frac{16}{17}-2\kappa)}=\|\psi(s)\|_{L^4}+1$ and obtain for $t\in[0,T]$
\begin{equs}[boundpsih]\|\psi_h(t)\|_{\bB^{\frac1{17}}_4}\lesssim K(\|\mZ_\eps\|).
\end{equs}
Hence we also derive for $s\in[0,T]$
\begin{equs}[boundr]2^{R(s)(\frac{16}{17}-2\kappa)}\lesssim K(\|\mZ_\eps\|)+ \|\psi_l(s)\|_{L^4}.
\end{equs}
Moreover, we further apply the smoothing effect of heat operator to derive improved regularity estimates for $\psi_h$ in different Besov spaces $\bB^{1-2\kappa}_p$ and $\bB^{\frac12+2\kappa}_p$ in terms of $\psi_l$, which will be used for the estimates of $\psi_l$ below. More precisely, we use Duhamel's formula to express the solution for $0 \leq s \leq t \leq T$
\begin{equs}[duh:psih]
	\psi_h(t)=&\,P_{t-s}\psi_h(s)+\int_s^tP_{t-r}\cG_{\prec,> R}(Y+\psi)\,\dif r
	\\\eqdef&\,\psi_h^{(1)}(t)+\psi_h^{(2)}(t).
\end{equs}
We also use the smoothing effect of heat operator  Lemma \ref{lem:heat} and the paraproduct estimates Lemma \ref{lem:para} to obtain  for $1\leq p\leq 4$
\begin{equs}[bd:psih2]
	\|\psi_h^{(2)}(t)\|_{\bB^{1-2\kappa}_p}\lesssim& \int_s^t (t-r)^{-1+\frac\kappa2}(\|\psi\|_{L^p}+\|\mZ_\eps\|)\|\mZ_\eps\|\,\dif r
	\\\lesssim&\,\|\mZ_\eps\|\int_s^t (t-r)^{-1+\frac\kappa2}\|\psi(r)\|_{L^p}\,\dif r+K(\|\mZ_\eps\|)
	\\\lesssim&\, K(\|\mZ_\eps\|)\Big(1+\int_s^t (t-r)^{-1+\frac\kappa2}\|\psi_l(r)\|_{L^p}\,\dif r\Big),
\end{equs}
where we used \eqref{boundpsih}.
Furthermore we use the smoothing effect of heat operator  Lemma \ref{lem:heat} for $\psi_h^{(1)}$ to have for $1\leq p\leq 4$
\begin{align*}\|\psi_h(t)\|_{\bB^{1-2\kappa}_p}\lesssim&\, (t-s)^{-\frac{8}{17}+\kappa}\|\psi_h(s)\|_{\bB^{\frac1{17}}_p}
	\\&+K(\|\mZ_\eps\|)\Big(1+\int_s^t (t-r)^{-1+\frac\kappa2}\|\psi_l(r)\|_{L^p}\,\dif r\Big),
\end{align*}
which combined with \eqref{boundpsih} implies that for $1\leq p\leq 4$
\begin{equs}[boundpsih1]\|\psi_h(t)\|_{\bB^{1-2\kappa}_p}\lesssim K(\|\mZ_\eps\|)\Big((t-s)^{-\frac8{17}+\kappa}+\int_s^t (t-r)^{-1+\frac\kappa2}\|\psi_l(r)\|_{L^p}\,\dif r\Big).
\end{equs}
Similarly we have
\begin{equs}[psih:j1]\|\psi_h(t)\|_{\bB^{\frac12+2\kappa}_2}\lesssim K(\|\mZ_\eps\|)\Big((t-s)^{-\frac14}+\int_s^t(t-r)^{-\frac34-2\kappa}\|\psi_l(r)\|_{L^2}\,\dif r\Big).
\end{equs}

\subsection{Estimate of $\psi^\sharp$}\label{sec:epsis}
The main aim of this section is to prove a uniform in $\eps$ bound for $\psi^\sharp$, which is defined in \eqref{para}. Our main technique is also the smoothing effect of heat operator Lemma \ref{lem:heat} and the paraproduct estimates Lemma \ref{lem:para} and Lemma \ref{lem:com2}. By \eqref{para}  we have
\begin{align}\label{dec:psisharp}
	\psi^\sharp=&\psi^\sharp_1+\psi^\sharp_2,\end{align}
with
\begin{align}
	\psi^\sharp_1\eqdef&-\int v^\eps(y)\Big([\sI,\tau_y\overline { Y}\prec]\cZ^{\<2>}_y+[\sI,\tau_yY\prec]\cZ^{\<2m>}_y\Big)\,\dif y-[\sI,Y\prec]\cZ^{\<2vm>},\no
	\\\psi^\sharp_2\eqdef&\,\psi-\sI(\cG_{\prec}(Y)) +\psi\prec\cZ^{\<20vm>}+\int v^\eps(y)\Big(\tau_y\psi\prec\cZ^{\<20vm1>}_y+\tau_y\overline{\psi}\prec\cZ^{\<20vm2>}_y\Big)\,\dif y,\label{def:psi2-n}
\end{align}
where $\cG_\prec\eqdef \cG_{\prec,\leq R}+\cG_{\prec,>R}$ and we recall $[\sI,f\prec]g$ denotes the commutator between $\sI$ and $f\prec$ given by
$$[\sI,f\prec]g= \sI(f\prec g)-f\prec \sI(g).$$

In the following we concentrate on proving the following uniform in $\eps$ estimates for $\psi^\sharp$.

\bp \label{boundsharpp}It holds that for any $0\leq s\leq t\leq T$
\begin{equs}[boundsharpf]
	\|\psi^\sharp_1\|_{C_T\bC^{1+3\kappa}}\lesssim K(\|\mZ_\eps\|),
\end{equs}
and
\begin{equs}[eqhighr] \int_s^t \|\psi_2^\sharp\|_{\bB^{1+24\kappa}_2}\,\dif r\lesssim \delta\int_s^t(\mathcal{V}^\eps(\psi_l)+\|\psi_l\|_{H^1}^2)\,\dif r+K(\|\mZ_\eps\|)+\|\psi_l(s)\|_{L^2}^2,\end{equs}
\begin{equs}[boundsharpf1]\int_s^t\|\psi^\sharp_2(r)\|^{\frac32}_{\bB^{1+2\kappa}_2}\,\dif r\lesssim \delta\int_s^t(\mathcal{V}^\eps(\psi_l)+\|\psi_l\|_{H^1}^2)\,\dif r+K(\|\mZ_\eps\|) +\|\psi_l(s)\|_{L^2}^2,
\end{equs}
where $\delta\in(0,1)$ is a small number and the proportional  constant is independent of $\eps$.
\ep

\begin{proof} The first result follows from Lemma \ref{lem:com2} and Proposition \ref{prop:Z}. We then concentrate on the second result. We have
	\begin{align*}\LL \psi^\sharp_2=&\cG+\cG_1-\cG_2,
	\end{align*}
	where \begin{equs}[def:G0]\cG=-v^\eps*|\psi|^2\psi+f_2(\psi)+f_1(\psi)+f_0+F_{\succcurlyeq}+\cG_{\succcurlyeq},
	\end{equs}
	\begin{equs}[def:G1]\cG_1=\LL\psi\prec\cZ^{\<20vm>}+\int v^\eps(y)\Big(\tau_y\LL\psi\prec\cZ^{\<20vm1>}_y+\tau_y\overline{\LL\psi}\prec\cZ_y^{\<20vm2>}\Big)\,\dif y,
	\end{equs}
	\begin{equs}[def:G2]\cG_2=&\,2\nabla\psi\prec\nabla\cZ^{\<20vm>}+2\int v^\eps(y)\Big(\tau_y\nabla\psi\prec\nabla\cZ^{\<20vm1>}_y+\tau_y\overline{\nabla\psi}\prec\nabla\cZ^{\<20vm2>}_y\Big)\,\dif y
		\\&+\psi\prec\cZ^{\<20vm>}+\int v^\eps(y)\Big(\tau_y\psi\prec\cZ^{\<20vm1>}_y+\tau_y\overline{\psi}\prec\cZ^{\<20vm2>}_y\Big)\,\dif y.
	\end{equs}
	We use Duhamel's formula to have for $s\leq r\leq t$
	$$\psi^\sharp_2(r)=P_{r-s}\psi^\sharp_2(s)+\sI_s (\cG+\cG_1-\cG_2)(r),$$
	with $\sI_s(f)=\int_s^\cdot P_{\cdot-r} f\,\dif r$.
	We first use the smoothing effect of heat kernel Lemma \ref{lem:heat} to have
	\begin{align*}
		\int_s^t\|P_{r-s}\psi^\sharp_2(s)\|_{\bB^{1+24\kappa}_2}\,\dif r\lesssim& \int_s^t(r-s)^{-\frac12-12\kappa}\|\psi^\sharp_2(s)\|_{L^2}\,\dif r\lesssim \|\psi^\sharp_2(s)\|_{L^2}
		\\\lesssim&\, \|\psi_l(s)\|_{L^2}\|\mZ_\eps\|+K(\|\mZ_\eps\|)\lesssim  \|\psi_l(s)\|^2_{L^2}+K(\|\mZ_\eps\|).
	\end{align*}
	where we  used the definition of $\psi_2^\sharp$ given in \eqref{def:psi2-n}, \eqref{boundpsih}, and Proposition \ref{prop:Z}, the paraproduct estimates Lemma \ref{lem:para} in the last step.
	
	Using Lemma \ref{boundg}--Lemma \ref{boundg2} below we obtain that $\int_s^t\|\psi_2^\sharp\|_{\bB^{1+24\kappa}_2}\,\dif r$ is bounded by
	$$\delta\int_s^t(\|\psi_2^\sharp\|_{\bB^{1+24\kappa}_2}+\mathcal{V}^\eps(\psi_l)+\|\psi_l\|_{H^1}^2)\,\dif r+K(\|\mZ_\eps\|)+\|\psi_l(s)\|^2_{L^2}.$$
	Then we choose $\delta$ small enough such that $\delta\int_s^t\|\psi_2^\sharp\|_{\bB^{1+24\kappa}_2}\,\dif r$ can be absorbed by the LHS, which implies the second result.
	
	For the third result, we use interpolation to derive it. To this end, we  give a $\bB^{1-2\kappa}_{2}$-norm bound of $\psi_2^\sharp$ using the definition of $\psi_2^\sharp$. More precisely,
	we
	use  the smoothing effect of heat operator Lemma \ref{lem:heat} and the paraproduct estimates Lemma \ref{lem:para} and \eqref{boundpsih}, \eqref{boundpsih1} to estimate each term on the RHS of \eqref{def:psi2-n} and have for any $t\geq s$
	\begin{equs}
		\|\psi_2^\sharp(t)\|_{\bB^{1-2\kappa}_{2}}\lesssim&\, K(\|\mZ_\eps\|)\Big(1+\|\psi_l(t)\|_{L^{2}}+\int_s^t(t-r)^{-1+\frac\kappa2} \|\psi_l\|_{L^{2}}\,\dif r
		\\&\qquad+(t-s)^{-\frac8{17}+\kappa}\Big)+\|\psi_l\|_{\bB^{1-2\kappa}_{2}}.
	\end{equs}
	Thus we apply the interpolation Lemma \ref{lem:interpolation} and Young's inequality to have
	\begin{equs}[bd:psisharpp]\int_s^t\|\psi_2^\sharp(r)\|^{\frac32}_{\bB^{1+2\kappa}_2}\,\dif r\lesssim & \int_s^t\|\psi_2^\sharp(r)\|^{\frac3{13}}_{\bB^{1+24\kappa}_{2}}\|\psi_2^\sharp(r)\|^{\frac{33}{26}}_{\bB^{1-2\kappa}_{2}}\,\dif r
		\\\lesssim &\,\delta\int_s^t\|\psi_2^\sharp(r)\|_{\bB^{1+24\kappa}_{2}}\,\dif r+\int_s^t\|\psi_2^\sharp(r)\|^{\frac{33}{20}}_{\bB^{1-2\kappa}_{2}}\,\dif r
		\\\lesssim &\,\delta\int_s^t\|\psi_2^\sharp(r)\|_{\bB^{1+24\kappa}_{2}}\,\dif r+K(\|\mZ_\eps\|)+\delta\int_s^t(\mathcal{V}^\eps(\psi_l)+\|\psi_l\|_{H^1}^2)\,\dif r,
	\end{equs}
	where we used  \eqref{lowerbound} in the last step. Thus the third result follows.
\end{proof}

\bl\label{boundg} It holds that
\begin{equs}[eqhighr1]\int_s^t \|(\sI_s\cG)(r)\|_{\bB^{1+24\kappa}_2}\,\dif r\lesssim \delta\int_s^t(\|\psi_2^\sharp\|_{\bB^{1+24\kappa}_2}+\mathcal{V}^\eps(\psi_l)+\|\psi_l\|_{H^1}^2)\,\dif r+K(\|\mZ_\eps\|),\end{equs}
where $\delta\in(0,1)$ is a small number and the proportional  constant is independent of $\eps$.
\el	
\begin{proof}

	We start with the  pure stochastic term $f_0+F_{\succcurlyeq}$ on the RHS of \eqref{def:G0}. They stay in $\bC^{-\frac12-2\kappa}$.   We use \eqref{bd:f_0}, Lemma \ref{lem:heat}, Propositions \ref{prop:Z}, \ref{thcomm1}--\ref{prop:Zn1} and the paraproduct estimates Lemma \ref{lem:para} to find that
	\begin{align*}
		&\int_s^t\big\|\sI_s(f_0+F_{\succcurlyeq})(r)\big\|_{\bB^{1+24\kappa}_2}\,\dif r
		\\\lesssim& \int_s^t\int_s^r(r-\tau)^{-\frac34-13\kappa}(\tau^{-\kappa}+1)K(\|\mZ_\eps\|)\,\dif \tau\dif r\lesssim K(\|\mZ_\eps\|).
	\end{align*}

	{\bf{Step I. }} In this step we derive that the contribution from $\cG_\succcurlyeq$ on the RHS of \eqref{def:G0} can be controlled by the RHS of \eqref{eqhighr1}.

	We first consider $\psi\circ \cZ^{\<2vm>}+(\tilde b_1^\eps+\tilde b^\eps_2)(Y+\psi)$.
	Using the Besov embedding Lemma \ref{lem:emb} and the smoothing effect of heat operator Lemma \ref{lem:heat}  we have
	\begin{align*}
		\|\sI_s(\psi\circ \cZ^{\<2vm>}+(\tilde b_1^\eps+\tilde b^\eps_2)(Y+\psi))(r)\|_{\bB^{1+24\kappa}_2}\lesssim&  \int_s^r (r-\tau)^{-\frac12-13\kappa}\|\psi\circ \cZ^{\<2vm>}+(\tilde b_1^\eps+\tilde b^\eps_2)(Y+\psi)\|_{\bB^{-\kappa}_{2}}\,\dif \tau.
	\end{align*}
	Using  Lemma \ref{lem:sto1-n}, \eqref{boundpsih} and take integration w.r.t. $r$ we obtain the following terms:
	\begin{align}\label{bd:est-G1}
		K(\|\mZ^\eps\|)\Big(\int_s^t	\int_s^r (r-\tau)^{-\frac12-13\kappa}(\|\psi^\sharp\|_{\bB^{1+2\kappa}_{2}}+1+\|\psi_l\|_{\bB^{3\kappa}_{2}})\,\dif \tau\dif r\Big).
	\end{align}
	We first take integration w.r.t. $r$ and reduce these terms to the following terms:
	\begin{align*}
		K(\|\mZ_\eps\|)\Big(	\int_s^t\|\psi^\sharp\|_{\bB^{1+2\kappa}_{2}}\,\dif r+\int_s^t\|\psi_l\|_{\bB^{3\kappa}_{2}}\,\dif r+1\Big),
	\end{align*}
	which can be estimated by the RHS of \eqref{eqhighr1} using \eqref{bd:psisharpp} and \eqref{boundsharpf}.

	We then consider $\psi\succ\cZ^{\<2vm>}$. By the paraproduct estimates Lemma \ref{lem:para} we have
	\begin{equs}[est:G1-2]
		\|\sI_s(\psi\succ\cZ^{\<2vm>})(r)\|_{\bB^{1+24\kappa}_2}
		\lesssim& \,\|\mZ_\eps\|\int_s^r (r-\tau)^{-\frac{33}{34}-13\kappa}\|\psi\|_{\bB^{\frac1{17}}_2}\,\dif \tau
		\\\lesssim&\,\|\mZ_\eps\|\int_s^r (r-\tau)^{-\frac{33}{34}-13\kappa}(\|\psi_h\|_{\bB^{\frac1{17}}_2}+\|\psi_l\|_{H^1})\,\dif \tau.
	\end{equs}
	Thus, by \eqref{boundpsih}, and by changing the order of integration, we deduce that  $\int_s^t \|\sI_s(\psi\succ\cZ^{\<2vm>})(r)\|_{\bB^{1+24\kappa}_2}\,\dif r$ is bounded by  the RHS of \eqref{eqhighr}.  Regarding the other two terms in $\cG_{\succcurlyeq}$,  for those involving the paraproduct $\circ$, we employ Lemma \ref{lem:sto1} and \eqref{eq:tauy} to ensure that their $\bB^{1+24\kappa}_2$-norm is controlled by \eqref{bd:est-G1}. Consequently, the same reasoning applies, yielding identical bounds. For the terms involving the  paraproduct $\succ$, the required bounds follow a similar approach as in \eqref{est:G1-2}, utilizing the paraproduct estimates Lemma \ref{lem:para} and \eqref{eq:tauy}.

	{\bf{Step II. }} In this step we derive that the contribution from the first three terms  on the RHS of \eqref{def:G0} can be controlled by the RHS of \eqref{eqhighr1}.
	
	\textbf{II.1} We start with $-v^\eps*|\psi|^2\psi$. We also use the Besov embedding Lemma \ref{lem:emb} and  the smoothing effect of heat operator Lemma \ref{lem:heat} to have
	\begin{align*}&
		\|\sI_s(v^\eps*|\psi|^2\psi)(r)\|_{\bB^{1+24\kappa}_2}
		\lesssim \|\sI_s(v^\eps*|\psi|^2\psi)(r)\|_{\bB^{\frac74+24\kappa}_{\frac43}}
		\\\lesssim& \int_s^r (r-\tau)^{-\frac78-12\kappa}\|v^\eps*|\psi|^2\psi\|_{L^{\frac43}}\,\dif \tau
		\lesssim\int_s^r (r-\tau)^{-\frac78-12\kappa}\|v^\eps*|\psi|^2\|_{L^2}\|\psi\|_{L^4}\,\dif \tau.
	\end{align*}
	We then separate $\psi=\psi_l+\psi_h$ and use the Sobolev embedding $H^{\frac34}\subset L^4$ and the interpolation Lemma \ref{lem:interpolation},  \eqref{boundpsih}, \eqref{lowerbound}, \eqref{lowerbound1} to have it bounded by
	\begin{align*}
		&\int_s^r(r-\tau)^{-\frac78-12\kappa}\big(\|v^\eps*|\psi_l|^2\|_{L^2}+\|\psi_l\|_{L^4}\|\psi_h\|_{L^4}+\|\psi_h\|_{L^4}^2\big)\big(\|\psi_l\|^{\frac14}_{L^2}\|\psi_l\|^{\frac34}_{H^1}+K(\|\mZ_\eps\|)\big)\,\dif r
		\\\lesssim&\int_s^r(r-\tau)^{-\frac78-12\kappa}\Big(\mathcal{V}^\eps(\psi_l)^{\frac9{16}}(\|\psi_l\|^{\frac34}_{H^1}+1)+K(\|\mZ_
		\eps\|)\|\psi_l\|^{\frac12}_{L^2}\|\psi_l\|^{\frac32}_{H^1}\Big)\,\dif r+K(\|\mZ_\eps\|)
		\\\lesssim&\int_s^r (r-\tau)^{-\frac78-12\kappa}\delta[\mathcal{V}^\eps(\psi_l)+\|\psi_l\|^{2}_{H^1}]\,\dif \tau+K(\|\mZ_\eps\|).
	\end{align*}
	Thus taking integral w.r.t. $r$ we obtain that
	$
	\int_s^t \|\sI_s(v^\eps*|\psi|^2\psi)(r)\|_{\bB^{1+24\kappa}_2}\,\dif r$
	is bounded by the RHS of \eqref{eqhighr}.

	\textbf{II.2}
	We then consider $f_2(\psi)$.  We also have two type of terms from $f_2(\psi)$ given by
	$v^\eps*|\psi|^2\cZ$ and $v^\eps*\Re[\psi\overline \cZ]\psi$ with $\|\cZ\|_{C_T\bC^{-\frac12-\kappa}}\lesssim \|\mZ_\eps\|$.
	For the first type $v^\eps*|\psi|^2\cZ$, we use the paraproduct decomposition to have
	\begin{align}\label{dec:par}
		v^\eps*|\psi|^2\cZ=(v^\eps*|\psi|^2)\prec \cZ+(v^\eps*|\psi|^2)\succcurlyeq \cZ.\end{align}
	For the first term from the RHS of \eqref{dec:par} by using Lemma \ref{lem:heat} and Lemma \ref{lem:para} we have
	\begin{equs}[est:G1-3]
		\|\sI_s(v^\eps*|\psi|^2\prec \cZ)(r)\|_{\bB^{1+24\kappa}_2}\lesssim K(\|\mZ_\eps\|) \int_s^r (r-\tau)^{-\frac34-13\kappa}\|\psi\|_{L^4}^2\,\dif \tau.\end{equs}
	We then take integral w.r.t. $r$ and obtain
	\begin{align*}
		\int_s^t	\|\sI_s(v^\eps*|\psi|^2\prec \cZ)(r)\|_{\bB^{1+24\kappa}_2}\,\dif r\lesssim K(\|\mZ_\eps\|)\Big(1+\int_s^t\|\psi\|_{L^4}^2\,\dif r\Big).
	\end{align*}
	We then decompose $\psi=\psi_l+\psi_h$ and apply \eqref{boundpsih}, the Sobolev embedding $H^{\frac34}\subset L^4$ and \eqref{lowerbound} to obtain that $\int_s^t \|\sI_s(v^\eps*|\psi|^2\prec \cZ)(r)\|_{\bB^{1+\kappa}_2}\,\dif r$ is bounded by  the RHS of \eqref{eqhighr1}. For the contribution from the second term of \eqref{dec:par} we use the Besov embedding Lemma \ref{lem:emb}, followed by Lemma \ref{lem:heat} and Lemma \ref{lem:para} to have
	\begin{align*}&
		\|\sI_s(v^\eps*|\psi|^2\succcurlyeq \cZ)(r)\|_{\bB^{1+24\kappa}_2}
		\\\lesssim&\, \|\sI_s(v^\eps*|\psi|^2\succcurlyeq \cZ)(r)\|_{\bB^{\frac74+24\kappa}_{\frac43}}\lesssim K(\|\mZ_\eps\|)\int_s^r (r-\tau)^{-\frac78-12\kappa}\||\psi|^2\|_{\bB^{\frac12+2\kappa}_{\frac43}}\,\dif \tau.
	\end{align*}
	We also first take integration w.r.t. $r$ and by Lemma \ref{lem:para} have
	\begin{align*}&
		\int_s^t\|\sI_s(v^\eps*|\psi|^2
		\succcurlyeq \cZ)(r)\|_{\bB^{1+24\kappa}_2}\,\dif r
		\lesssim \|\mZ_\eps\|\int_s^t\||\psi|^2\|_{\bB^{\frac12+2\kappa}_{\frac43}}\,\dif r
		\lesssim \|\mZ_\eps\|\int_s^t \|\psi\|_{\bB^{\frac12+2\kappa}_{2}}\|\psi\|_{L^4}\,\dif r.
	\end{align*}
	We  also decompose  $\psi=\psi_l+\psi_h$ and use \eqref{boundpsih}, Lemma \ref{lem:interpolation} and the Sobolev embedding $H^{\frac34}\subset L^4$ to have it controlled by
	\begin{align*}
		&\|\mZ_\eps\|\int_s^t ( \|\psi_h\|_{\bB^{\frac12+2\kappa}_{2}}+\|\psi_l\|_{H^1}^{\frac12+2\kappa}\|\psi_l\|_{L^2}^{\frac12-2\kappa})(\|\psi_l\|_{H^1}^{\frac34}\|\psi_l\|_{L^2}^{\frac14}+K(\|\mZ_\eps\|))\,\dif r
		\\\lesssim&\int_s^t \Big(\|\psi_h\|_{\bB^{\frac12+2\kappa}_{2}}^2+K(\|\mZ_\eps\|)(\|\psi_l\|_{H^1}^{\frac32}\|\psi_l\|_{L^2}^{\frac12}
		+\|\psi_l\|_{H^1}^{\frac54+2\kappa}\|\psi_l\|_{L^2}^{\frac34-2\kappa})\Big)\,\dif r+K(\|\mZ_\eps\|).
	\end{align*}
	We then use Young's inequality and \eqref{lowerbound}, \eqref{psih:j1} to have it bounded by the RHS of \eqref{eqhighr1}.
	
	On the other hand, for the second type term $v^\eps*\Re[\psi\overline \cZ]\psi$ from $f_2(\psi)$ we use the paraproduct decomposition to have
	\begin{equs}[f2]
		v^\eps*\Re[\psi\overline \cZ]\psi=&\,v^\eps*\Re[\psi\prec \overline \cZ]\psi+v^\eps*\Re[\psi\succcurlyeq \overline \cZ]\psi
		\\=&\,v^\eps*\Re[\psi\prec \overline \cZ](\prec+\succ)\psi+v^\eps*\Re[\psi\prec \overline \cZ]\circ\psi+v^\eps*\Re[\psi\succcurlyeq \overline \cZ]\psi.
	\end{equs}
	Using Lemma \ref{lem:para} we easily find that the contribution from	the first term on the RHS of \eqref{f2} can be controlled by the RHS of \eqref{est:G1-3}, 	which can be bounded the same way as $(v^\eps*|\psi|^2)\prec \cZ$.  For the last two terms on the RHS of \eqref{f2} we also use the paraproduct estimates Lemma \ref{lem:para} and Lemma \ref{lem:heat} to find
	\begin{align*}
		\int_s^t\|\sI_s(v^\eps*\Re[\psi\prec \overline \cZ]\circ\psi+v^\eps*\Re[\psi\succcurlyeq \overline \cZ]\psi)(r)\|_{\bB^{1+24\kappa}_2}\,\dif r\lesssim  \|\mZ_\eps\|\int_s^t \|\psi\|_{\bB^{\frac12+2\kappa}_{2}}\|\psi\|_{L^4}\,\dif r,
	\end{align*}
	which can be bounded exactly the same way as that for $v^\eps*|\psi|^2\succcurlyeq \cZ$.

	\textbf{II.3} We then consider $ f_1(\psi)$. We note that we have three type of terms:
	\begin{enumerate}
		\item $v^\eps*\Re[\psi\overline Y]Z$ and $v^\eps*\Re[\psi\overline Z]Y$;
		\item $\cR\psi$ and $(6(\tilde b^\eps-b^\eps)-m+1)\psi$;
		\item  $\psi \cZ$ with $\|\cZ(t)\|_{\bC^{-\frac12-\kappa}}\lesssim \|\mZ_\eps\|(t^{-\kappa}+1)$ and $v^\eps*\Re(\psi \overline{Y})Y$.
	\end{enumerate}
	For the first type, we recall that the $\bC^{-\frac12-\kappa}$-norm of
	the purely stochastic terms $\tau_y\overline YZ$, $\tau_yYZ$, $\tau_y\overline ZY$ and $\tau_y ZY$ on the RHS of \eqref{v2},   \eqref{v3} can by bounded by $\|\mZ_\eps\|^2$ using the paraproduct estimates Lemma \ref{lem:para}, Propositions \ref{prop:Z}, \ref{prop:Zny}  and Propositions \ref{prop:Zn}, \ref{prop:Zn1}. Hence, we have by \eqref{eq:tauy}
	\begin{align*}
		\|v^\eps*\Re[\psi\overline Y]Z(t)\|_{\bB^{-\frac12-\kappa}_2}+\|v^\eps*\Re[\psi\overline Z]Y\|_{\bB^{-\frac12-\kappa}_2}\lesssim \|\psi\|_{\bB^{\frac12+2\kappa}_2}K(\|\mZ_\eps\|).
	\end{align*}
	We then use Lemma \ref{comnew} and \eqref{dif:esb} for the second type, and the paraproduct estimates from Lemma \ref{lem:para} for the third type to deduce
	\begin{align}\label{bd:f1com} \|f_1(\psi)(t)\|_{\bB^{-\frac12-\kappa}_2}\lesssim \|\psi\|_{\bB^{\frac12+2\kappa}_2}K(\|\mZ_\eps\|)(t^{-\kappa}+1),
	\end{align}
	which implies that
	\begin{align*}&
		\|\sI_s(f_1(\psi))(r)\|_{\bB^{1+24\kappa}_2}
		\lesssim K(\|\mZ_\eps\|)\int_s^r (r-\tau)^{-\frac34-13\kappa}(\tau^{-\kappa}+1)\|\psi\|_{\bB^{\frac12+2\kappa}_2}\,\dif \tau.
	\end{align*}
	Using \eqref{psih:j1} and  changing the order of integral as before, we obtain that $\int_s^t \|\sI_s(f_1(\psi))(r)\|_{\bB^{1+24\kappa}_2}\,\dif r$ is bounded by
	$$K(\|\mZ_\eps\|)\int_s^tr^{-\kappa} \|\psi\|_{\bB^{\frac12+2\kappa}_2}\,\dif r\lesssim \int_s^t \|\psi\|^2_{\bB^{\frac12+2\kappa}_2}\,\dif r+K(\|\mZ_\eps\|),$$
	which, by \eqref{psih:j1}, is bounded by the RHS of \eqref{eqhighr1}.

	Combining the above calculations, the result follows.
\end{proof}

\bl\label{boundg1} For $\cG_1$ in \eqref{def:G1}, it holds that
\begin{equs}[eqhighr2]\int_s^t \|(\sI_s\cG_1)(r)\|_{\bB^{1+24\kappa}_2}\,\dif r\lesssim \delta\int_s^t(\|\psi_2^\sharp\|_{\bB^{1+24\kappa}_2}+\mathcal{V}^\eps(\psi_l)+\|\psi_l\|_{H^1}^2)\,\dif r+K(\|\mZ_\eps\|),\end{equs}
where $\delta\in(0,1)$ is a small number and the proportional  constant is independent of $\eps$.
\el	
\begin{proof} From \eqref{def:G1} and $\LL\psi = \cG+\cG_\prec(Y+\psi)$ we know
	\begin{equs}\cG_1=&( \cG+\cG_\prec(Y+\psi))\prec\cZ^{\<20vm>}
		\\&+\int v^\eps(y)\Big(\tau_y( \cG+G_\prec(Y+\psi))\prec\cZ^{\<20vm1>}_y+\tau_y\overline{( \cG+\cG_\prec(Y+\psi))}\prec\cZ^{\<20vm2>}_y\Big)\,\dif y.
	\end{equs}
	Since the regularity of the terms $\cG\prec\cZ^{\<20vm>}$ and $\int v^\eps(y)\tau_y \cG\prec\cZ^{\<20vm1>}_y\,\dif y$, $\int v^\eps(y) \tau_y\overline{ \cG}\prec\cZ^{\<20vm2>}_y\,\dif y$, which involve $\cG$, is enhanced compared with $\cG$ due to the paraproduct estimates provided by Lemma \ref{lem:para}, they can be bounded in the same manner as demonstrated in the proof of Lemma \ref{boundg}.
	For the other terms involving $\cG_\prec(Y+\psi)$ we
	use the paraproduct estimates Lemma \ref{lem:para} and \eqref{eq:tauy} to have
	\begin{equs}[est:G-4]
		\|\cG_\prec(Y+\psi)\|_{\bB^{-1-\kappa}_2}\lesssim K(\|\mZ_\eps\|)\|Y+\psi\|_{L^2}.
	\end{equs}
	Hence, we have by \eqref{boundpsih}
	\begin{align*}\|\cG_\prec(Y+\psi)\prec\cZ^{\<20vm>}\|_{\bB^{-2\kappa}_2}\lesssim K(\|\mZ_\eps\|)\|Y+\psi\|_{L^2}\lesssim K(\|\mZ_\eps\|)(1+\|\psi_l\|_{L^2}),
	\end{align*}
	which implies that
	$$ \|(\sI_s(\cG_\prec(Y+\psi)\prec\cZ^{\<20vm>}))(r)\|_{\bB^{1+24\kappa}_2}\,\dif r\lesssim \int_s^r(r-\tau)^{-\frac12-13\kappa}K(\|\mZ_\eps\|)(1+\|\psi_l\|_{L^2})\,\dif \tau.$$
	Taking integral w.r.t $r$ we obtain that it is bounded by the RHS of \eqref{eqhighr2}. For the other two terms involving $\cG_\prec(Y+\psi)$ we use \eqref{est:G-4}, \eqref{eq:tauy}, and Proposition \ref{prop:Zny} to have the same bound and the result follows.
\end{proof}

\bl\label{boundg2} For $\cG_2$ in \eqref{def:G2}, it holds that
\begin{equs}[eqhighr3]\int_s^t \|(\sI_s\cG_2)(r)\|_{\bB^{1+24\kappa}_2}\,\dif r\lesssim \delta\int_s^t(\mathcal{V}^\eps(\psi_l)+\|\psi_l\|_{H^1}^2)\,\dif r+K(\|\mZ_\eps\|),\end{equs}
where $\delta\in(0,1)$ is a small number and the proportional  constant is independent of $\eps$.
\el	
\begin{proof}	Recall \eqref{def:G2} and we use the paraproduct estimates Lemma \ref{lem:para} and \eqref{eq:tauy} to have
	$$\|\cG_2\|_{\bB^{-3\kappa}_2}\lesssim \|\psi\|_{\bB_2^{1-2\kappa}}K(\|\mZ_\eps\|),$$
	which implies that
	$$ \|\sI_s(\cG_2)(r)\|_{\bB^{1+24\kappa}_2}\,\dif r\lesssim \int_s^r(r-\tau)^{-\frac12-14\kappa}K(\|\mZ_\eps\|)(\|\psi_h\|_{\bB_2^{1-2\kappa}}+\|\psi_l\|_{H^1})\,\dif \tau.$$
	Taking integral w.r.t $r$ and using \eqref{boundpsih1} we obtain it is bounded by the RHS of \eqref{eqhighr3}.
\end{proof}

\subsection{Energy estimates for $\psi_l$}\label{sec:en}
In this section, we undertake $L^2$-energy estimates for $\psi_l$. Our objective is to derive energy estimates for $\psi_l$ that are uniform in $\eps$, specifically Proposition \ref{energy}, by leveraging the dissipative impact of the nonlinearity. It's worth noting that when we compute the $L^2$
-inner product of $\psi_l$ on both sides of equation \eqref{eqw}, the cubic nonlinearity yields $\cV^\eps(\psi_l)$. Thanks to the lower bounds provided by \eqref{lowerbound} and \eqref{lowerbound1}, $\cV^\eps(\psi_l)$
exhibits a dissipative effect, enabling us to gain uniform control over the other nonlinear terms. However, the dissipative effect from $\cV^\eps(\psi_l)$ is not as potent as $\|\phi\|_{L^4}^4$
in the context of energy estimates for the classical $\Phi^4$
model (refer to \cite{GH18a} for details). To obtain a finer control over the nonlinearity, we must exploit the specific structure of $\cV^\eps$.

Consider the energy estimates for $\psi_l$, we first have the following result.

\bl\label{lem:energy} It holds that
\begin{equs}
	&\frac12\frac{\dif}{\dif t}\| \psi_l\|_{L^2}^2+\|\nabla\psi_l\|_{L^2}^2+\|\psi_l\|_{L^2}^2+\cV^\eps(\psi_l)
	= \Theta+\Xi,
\end{equs}
with
\begin{align*}
	\Theta\eqdef& -\langle v^\eps*|\psi_l+\psi_h|^2\psi_h+v^\eps*(2\Re(\psi_l\overline{\psi_h})+|\psi_h|^2) \psi_l,{\psi_l}\rangle,
	\\
	\Xi\eqdef &\langle \cG_{\succcurlyeq}+F_{\succcurlyeq},{\psi_l}\rangle+\langle f_2(\psi)+f_1(\psi)+f_0,{\psi_l}\rangle
	+ \langle \cG_{\prec,\leq R}(Y+\psi),{\psi_l}\rangle.
\end{align*}
\el
\begin{proof}
	The result follows by taking $L^2$-inner product on both sides of equation \eqref{eqw}. We also use $\langle v^\eps*|\psi_l|^2 \psi_l,\psi_l\rangle=\cV^\eps(\psi_l)$.
\end{proof}

The aim of this section is to prove the following uniform in $\eps$ estimate of $\psi_l$.

\bp \label{energy}It holds that for $0\leq s\leq t\leq T$
\begin{equs}[mest:psil]
	&\| \psi_l(t)\|_{L^2}^2+\int_s^t\|\nabla\psi_l\|_{L^2}^2\,\dif r+\int_s^t\cV^\eps(\psi_l)\,\dif r
	\lesssim K(\|\mZ_\eps\|)+\|\psi_l(s)\|_{L^2}^2,
\end{equs}
with the proportional constant independent of $\eps$.
\ep
\begin{proof} The idea of the proof is to bound $\int_s^t (\Theta+\Xi)\,\dif r$ by
	$$\delta\Big(\int_s^t\|\nabla\psi_l\|_{L^2}^2\,\dif r+\int_s^t\cV^\eps(\psi_l)\,\dif r\Big)+C_\delta\Big(K(\|\mZ_\eps\|)+\|\psi_l(s)\|_{L^2}^2\Big),$$
	for $\delta\in (0,1)$ and $C_\delta$ independent of $\eps$, which follows by  Lemma \ref{lem:Theta}-- Lemma \ref{lem:Xil} and \eqref{bd:f0p} below. The result then follows by choosing $\delta$ small and applying Lemma \ref{lem:energy}.
\end{proof}

To derive Proposition \ref{energy} we start with the control of $\Theta$ in the following lemma.

\bl\label{lem:Theta} It holds that
\begin{align*}
	|\Theta|\lesssim\delta\cV^\eps(\psi_l)+\delta
	\|\psi_l\|_{H^{1}}^{2}+K(\|\mZ_\eps\|),
\end{align*}
for  $\delta\in(0,1)$, where the proportional constant is independent of $\eps$.
\el
\begin{proof}
	For $\Theta$ we use \eqref{boundpsih}, \eqref{lowerbound1} and the Besov embedding Lemma \ref{lem:emb} to have
	\begin{align*}
		|\Theta|\lesssim& \,(\|v^\eps* |\psi_l|^2\|_{L^2}+\|v^\eps* |\psi_h|^2\|_{L^2})\|\psi_h\|_{L^4}\|\psi_l\|_{L^4}+\|v^\eps* |\psi_l|^2\|_{L^2}\|\psi_h\|_{L^4}^2
		\\\lesssim&\, K(\|\mZ_\eps\|)\Big(\cV^\eps(\psi_l)^{\frac12}+K(\|\mZ_\eps\|)\Big)\Big(
		\|\psi_l\|_{H^{1}}^{\frac34}\|\psi_l\|_{L^2}^{\frac14}+1\Big).
	\end{align*}
	Here we used $\langle v^\eps*f,g\rangle=\langle f,v^\eps*g\rangle$.
	We then use \eqref{lowerbound} to bound $$\|\psi_l\|_{L^2}^{\frac14}\lesssim \cV^\eps(\psi_l)^{\frac1{16}},$$
	which combined with Young's inequality implies the result.
\end{proof}

We then continue with the more complicated term $\Xi$. We will consider each term separately.

Using \eqref{bd:f_0}, we have for $0<\delta <1$
\begin{align}\label{bd:f0p}
	|	\langle f_0,\psi_l\rangle|\lesssim \|\psi_l\|_{H^{\frac12+3\kappa}}K(\|\mZ_\eps\|)(t^{-\kappa}+1)\lesssim \delta \|\psi_l\|_{H^1}^2+K(\|\mZ_\eps\|)(t^{-2\kappa}+1).
\end{align}

We then consider $\langle (v^\eps *|\psi_l|^2)\psi_l, Z\rangle$ from $f_2(\psi)$, which follows  essentially the same argument as in the proof of \cite[Lemma 6.1]{OOT24}. Here we use the fact that $v\geq0$.

\bl\label{cubic}Suppose that $\cZ\in C_T\bC^{-\frac12-\frac\kappa2}$ with $\|\cZ\|_{\bC^{-\frac12-\frac\kappa2}}\lesssim K(\|\mZ_\eps\|)$.  It holds that for $\delta>0$
\begin{equs}[eqcubic]
	|\langle (v^\eps *|f|^2)f, \cZ\rangle |+|\langle (v^\eps *|f|^2)\overline{f}, \cZ\rangle|\lesssim \delta(\mathcal{V}^\eps(f)+\|f\|_{H^1}^2)+K(\|\mZ_\eps\|),\end{equs}
where the implicit constant is independent of $\eps$.
\el
\begin{proof} We only consider the first term, as the second term follows exactly the same way. We use \cite[Lemma A.5]{SSZZ20} to have
	\begin{align}\label{eq:gcubic}
		|\langle g,\cZ\rangle |\lesssim \big(\|\nabla g\|_{L^1}^{\frac12+\kappa}\|g\|_{L^1}^{\frac12-\kappa}+\|g\|_{L^1}\big)\|\cZ\|_{\bC^{-\frac12-\kappa}}.
	\end{align}
	We take $g=(v^\eps *|f|^2)f$ and use \eqref{lowerbound}, \eqref{lowerbound1} to have
	\begin{align*}
		\|(v^\eps *|f|^2)f\|_{L^1}\leq \|v^\eps *|f|^2\|_{L^2}\|f\|_{L^2}\leq \cV^\eps(f)^{\frac34}.
	\end{align*}
	On the other hand, we use \eqref{lowerbound1} to have
	\begin{align*}
		\|\nabla g\|_{L^1}\leq& \int v^\eps*|f|^2 |\nabla f|\,\dif x+\int v^\eps*(|\nabla f| |f|) |f|\,\dif x
		\\\lesssim &\, \cV^\eps(f)^{\frac12}\|f\|_{H^1}+\||f|v^\eps*|f|\|_{L^2}\|f\|_{H^1}
		\lesssim\cV^\eps(f)^{\frac12}\|f\|_{H^1},
	\end{align*}
	where we used integration by parts in the second step and used H\"older's inequality to have
	$$\int|f|^2(v^\eps*|f|)^2\,\dif x\leq \int|f|^2(v^\eps*|f|^2)\,\dif x=\cV^\eps(f)$$ in the last step.
	Substituting the above two estimates into \eqref{eq:gcubic} we obtain
	\begin{equs}[est8:cubic]
		|\langle (v^\eps *|f|^2)f, \cZ\rangle |\lesssim \Big[\cV^\eps(f)^{\frac34}+\cV^\eps(f)^{\frac34(\frac12-\kappa)+\frac12(\frac12+\kappa)}\|f\|_{H^1}^{\frac12+\kappa}\Big]\|\cZ\|_{\bC^{-\frac12-\kappa}}.
	\end{equs}
	Applying Young's inequality, we derive the result.
\end{proof}

\bl\label{lem:f2} It holds that for $\delta>0$ and $0\leq s\leq t\leq T$
\begin{align}\label{eq:f2pl}
	\int_s^t	|\langle f_2(\psi),\psi_l\rangle|\,\dif r\lesssim \delta\int_s^t(\mathcal{V}^\eps(\psi_l)+\|\psi_l\|_{H^1}^2)\,\dif r+K(\|\mZ_\eps\|),
\end{align}
with the proportional constant independent of $\eps$.
\el

\begin{proof} We have the following decomposition:
	\begin{equs}[eest:f2]
		\langle f_2(\psi),\psi_l\rangle=&\,\langle v^\eps*|\psi_l|^2\cZ,\psi_l\rangle+2\langle v^\eps*|\psi_l|^2, \Re(\cZ_1\psi_l)\rangle+\sum_{i=1}^4J_i,
	\end{equs}
	with $\cZ=-Z-Y$, $\cZ_1=-\overline{Z}-\overline{Y}\in C_T\bC^{-\frac12-\frac\kappa2}$ and
	\begin{align*}
		J_1\eqdef&\, 2\langle v^\eps*\Re(\psi_l\overline{\psi}_h)\cZ,{\psi_l}\rangle,\qquad J_2\eqdef \,\langle v^\eps*|\psi_h|^2\cZ,{\psi_l}\rangle,
		\\ J_3\eqdef&\, 2\langle v^\eps*\Re(\psi \cZ_1)\psi_h,{\psi_l}\rangle,\qquad J_4\eqdef\, 2\langle v^\eps*\Re(\psi_h \cZ_1)\psi_l,{\psi_l}\rangle.
	\end{align*}
	We know that
	\begin{align*}
		\|\cZ\|_{C_T\bC^{-\frac12-\frac\kappa2}}+\|\cZ_1\|_{C_T\bC^{-\frac12-\frac\kappa2}}\lesssim \|\mZ_\eps\|.
	\end{align*}
	Concerning the first two terms on the RHS of \eqref{eest:f2} we use Lemma \ref{cubic} to derive the desired estimates. It remains to consider $J_i$, $i=1,\dots,4$.
	
	\textbf{I.} For $J_1$ we use the paraproduct decomposition to have
	\begin{align*}
		\frac12J_1=	\langle (v^\eps*\Re(\psi_l\overline{\psi}_h))\prec\cZ,{\psi_l}\rangle+	\langle (v^\eps*\Re(\psi_l\overline{\psi}_h))\succcurlyeq\cZ,{\psi_l}\rangle.
	\end{align*}
	By Lemma \ref{lem:multi}, the paraproduct estimates Lemma \ref{lem:para} and \eqref{boundpsih} we have
	\begin{align*}
		&|\langle (v^\eps*\Re(\psi_l\overline{\psi}_h))\prec\cZ,\psi_l\rangle|\lesssim \|\psi_l\|_{L^4}\|\psi_h\|_{L^4}\|\psi_l\|_{H^{\frac12+\kappa}}\|\mZ_\eps\|
		\\\lesssim&\, \delta\|\psi_l\|_{H^1}^2+\|\psi_l\|_{L^2}^2K(\|\mZ_\eps\|)\lesssim \delta\|\psi_l\|_{H^1}^2+\delta \cV^{\eps}(\psi_l)+K(\|\mZ_\eps\|),
	\end{align*}
	where we used the interpolation Lemma \ref{lem:interpolation}, Young's inequality in the second step and \eqref{lowerbound} in the last step.
	
	For the remaining part we need a more delicate estimate for $\psi_h$. More precisely, we recall the decomposition in \eqref{duh:psih}.
	In the following we estimate $\psi^{(1)}_h$ and $\psi^{(2)}_h$ separately.
	By the paraproduct estimates Lemma \ref{lem:para}, Lemma \ref{lem:multi} and \eqref{boundpsih} we have
	\begin{equs}[bd:psih1]
		&|\langle (v^\eps*\Re(\psi_l\overline{\psi^{(1)}_h}))\succcurlyeq\cZ,\psi_l\rangle|\lesssim \|\psi_l\|_{L^4}\|\psi_l\overline{\psi^{(1)}_h}\|_{\bB^{\frac12+2\kappa}_{\frac43}}\|\mZ_\eps\|
		\\\lesssim&\, \|\psi_l\|_{L^4}\Big(\|\psi_l\|_{\bB^{\frac12+2\kappa}_{2}}\|\psi^{(1)}_h\|_{L^4}+\|\psi_l\|_{L^2}\|\psi^{(1)}_h\|_{\bB^{\frac12+2\kappa}_{4}}\Big)\|\mZ_\eps\|
		\\\lesssim&\, \|\psi_l\|_{H^1}^{\frac54+2\kappa}\|\psi_l\|_{L^2}^{\frac34-2\kappa}K(\|\mZ_\eps\|)+\|\psi_l\|_{L^4}\|\psi_l\|_{L^2}(r-s)^{-\frac{15}{68}-\kappa}\|\psi_h(s)\|_{B^{\frac1{17}}_4}\|\mZ_\eps\|,
	\end{equs}
	where we also used the Sobolev embedding $H^{\frac34}\subset L^4$, the interpolation Lemma \ref{lem:interpolation} and Lemma \ref{lem:heat} in the last step. By applying Young's inequality, the first term
	on the RHS of \eqref{bd:psih1} can be bounded by $\delta(\mathcal{V}^\eps(\psi_l)+\|\psi_l\|_{H^1}^2)+K(\|\mZ_\eps\|)$.	Using \eqref{boundpsih} and Young's inequality with exponents $(\frac83,\frac{16}5,\frac{16}5)$ we obtain that the second term on the RHS of \eqref{bd:psih1} can be bounded by
	\begin{align*}
		&\|\psi_l\|_{H^1}^{\frac34}\|\psi_l\|_{L^2}^{\frac54}(r-s)^{-\frac{15}{68}-\kappa}
		\lesssim\delta(\mathcal{V}^\eps(\psi_l)+\|\psi_l\|_{H^1}^2)+(r-s)^{-\frac{12}{17}-4
			\kappa}K(\|\mZ_\eps\|).
	\end{align*}
	For the term including $\psi^{(2)}_h$ we also use the paraproduct estimates Lemma \ref{lem:para}, Lemma \ref{lem:multi}, the Sobolev embedding $H^{\frac34}\subset L^4$ and the interpolation Lemma \ref{lem:interpolation} to have
	\begin{equs}[eq:psi2e]
		&|\langle (v^\eps*\Re(\psi_l\overline{\psi^{(2)}_h}))\succcurlyeq\cZ,\psi_l\rangle|\lesssim \|\psi_l\|_{L^4}\|\psi_l\overline{\psi^{(2)}_h}\|_{\bB^{\frac12+2\kappa}_{\frac43}}\|\mZ_\eps\|
		\\\lesssim&\, \|\psi_l\|_{L^4}\|\mZ_\eps\|\Big(\|\psi_l\|_{\bB^{\frac12+2\kappa}_{2}}\|{\psi^{(2)}_h}\|_{L^4}+\|\psi_l\|_{L^3}\|{\psi^{(2)}_h}\|_{\bB^{\frac12+2\kappa}_{\frac{12}5}}\Big)
		\\\lesssim&\, \|\psi_l\|_{H^1}^{\frac54+2\kappa}\|\psi_l\|_{L^2}^{\frac34-2\kappa}K(\|\mZ_\eps\|)+\|\psi_l\|_{L^4}\|\psi_l\|_{L^3}\|{\psi^{(2)}_h}\|_{L^4}^{1-\sigma}\|{\psi^{(2)}_h}\|_{\bB^{1-2\kappa}_2}^{\sigma}\|\mZ_\eps\|.
	\end{equs}
	Here $\sigma=\frac{1+4\kappa}{2-4\kappa}$ and we used \eqref{boundpsih}.
	The first term on the RHS can be bounded the same as the first term in \eqref{bd:psih1}.	
	Using \eqref{bd:psih2}, \eqref{boundpsih} and the Sobolev embedding $H^{\frac12}\subset L^3, H^{\frac34}\subset L^4$, we obtain that the second term on the RHS of \eqref{eq:psi2e} can be bounded by
	\begin{align*}
		&\|\psi_l\|_{H^1}^{\frac54}\|\psi_l\|_{L^2}^{\frac34}\Big(1+\int_s^r(r-\tau)^{-1+\frac\kappa2}\|\psi_l(\tau)\|_{L^2}\,\dif \tau\Big)^{\sigma}K(\|\mZ_\eps\|)
		\\\lesssim&\,\delta(\mathcal{V}^\eps(\psi_l)+\|\psi_l\|_{H^1}^2)+K(\|\mZ_\eps\|)\Big(1+\int_s^r(r-\tau)^{-1+\frac\kappa2}\|\psi_l(r)\|_{L^2}^{6\sigma}\,\dif \tau\Big),
	\end{align*}
	where we use \eqref{lowerbound} and  Young's inequality with exponents $(\frac85,\frac{16}3, \frac{16}3)$.
	Taking integration w.r.t. $r$ we use $6\sigma <4$ to obtain
	\begin{align*}
		&\Big|\int_s^t\langle (v^\eps*\Re(\psi_l\overline{\psi}_h))\succcurlyeq\cZ,\psi_l\rangle \,\dif r\Big|\lesssim\delta\int_s^t(\mathcal{V}^\eps(\psi_l)+\|\psi_l\|_{H^1}^2)\,\dif r+K(\|\mZ_\eps\|).
	\end{align*}
	
	\textbf{II.} For $J_2$ we also use the paraproduct decomposition to have
	\begin{align*}
		J_2=	\langle (v^\eps*|\psi_h|^2)\prec\cZ,\psi_l\rangle+	\langle (v^\eps*|\psi_h|^2)\succcurlyeq\cZ,\psi_l\rangle.
	\end{align*}
	By the paraproduct estimates Lemma \ref{lem:para} and \eqref{boundpsih}, \eqref{lowerbound}, the interpolation Lemma \ref{lem:interpolation} we have
	\begin{align*}
		&|\langle (v^\eps*|\psi_h|^2)\prec\cZ,\psi_l\rangle|\lesssim \|\psi_h\|_{L^4}^2\|\psi_l\|_{H^{\frac12+\kappa}}\|\mZ_\eps\|
		\\\lesssim&\, \delta\|\psi_l\|_{H^1}^2+\delta \cV^{\eps}(\psi_l)+K(\|\mZ_\eps\|),
	\end{align*}
	and
	\begin{align*}
		&|\langle (v^\eps*|\psi_h|^2)\succcurlyeq\cZ,\psi_l\rangle|\lesssim \|\psi_l\|_{L^4}\||\psi_h|^2\|_{\bB^{\frac12+\kappa}_{\frac43}}\|\mZ_\eps\|\lesssim \|\psi_l\|_{L^4}\|\psi_h\|_{\bB^{\frac12+\kappa}_2}\|\psi_h\|_{L^4}\|\mZ_\eps\|
		\\\lesssim&\,\|\psi_l\|_{L^4}\|\psi_h\|_{\bB^{\frac12+\kappa}_2}K(\|\mZ_\eps\|)
		\lesssim\|\psi_l\|^2_{L^4}+\|\psi_h\|^2_{\bB^{\frac12+\kappa}_2}K(\|\mZ_\eps\|).
	\end{align*}
	The first term can be bounded by the Sobolev embedding $H^{\frac34}\subset L^4$ and Young's inequality.
	By \eqref{psih:j1} we have
	\begin{equs}[j1]\|\psi_h(r)\|^2_{\bB^{\frac12+2\kappa}_2}\lesssim K(\|\mZ_\eps\|)\Big((r-s)^{-\frac12}+\int_s^r(r-\tau)^{-\frac34-2\kappa}\|\psi_l\|_{L^2}^2\,\dif \tau\Big),
	\end{equs}
	which implies that after integration w.r.t. $r$, $\int_s^tJ_2\,\dif r$ can be bounded by the RHS of \eqref{eq:f2pl}.
	
	\textbf{III. }For $J_3$ we use $\psi=\psi_l+\psi_h$ to have
	\begin{align*}
		&	\frac12J_3=\langle v^\eps*(\overline{\psi_l}\psi_h ),\Re(\cZ_1\psi)\rangle
		\\=&\,\langle v^\eps*(\overline{\psi_l}\psi_h ),\Re(\cZ_1\psi_l)\rangle+\langle v^\eps*(\overline{\psi_l}\psi_h ),\Re(\cZ_1\psi_h)\rangle.
	\end{align*}
	The first term can be estimated exactly the same way as  $J_1$. For the second term we focus on $\langle v^\eps*(\overline{\psi_l}\psi_h ),\cZ_1\psi_h\rangle$; the other terms can be handled similarly. Using the paraproduct decomposition, we obtain
	\begin{align*}
		\langle v^\eps*(\overline{\psi_l}\psi_h ),\cZ_1\psi_h\rangle=\langle (v^\eps*(\overline{\psi_l}\psi_h) )\prec\overline\cZ_1,\psi_h\rangle+\langle (v^\eps*(\overline{\psi_l}\psi_h ))\succcurlyeq\overline\cZ_1,\psi_h\rangle.
	\end{align*}
	We use Lemma \ref{lem:para} and Lemma \ref{lem:multi}, \eqref{boundpsih} to have
	\begin{align*}
		|\langle (v^\eps*(\overline{\psi_l}\psi_h ))\prec\overline\cZ_1,\psi_h\rangle|\lesssim \|\psi_l\|_{L^4}\|\psi_h \|_{L^4}\|\psi_h\|_{\bB^{\frac12+\kappa}_2}\|\mZ_\eps\|\lesssim\|\psi_l\|_{L^4}\|\psi_h\|_{\bB^{\frac12+\kappa}_2}K(\|\mZ_\eps\|),
	\end{align*}
	which can be estimated similarly as the second term in $J_2$. Similarly we use Lemma \ref{lem:para} and \eqref{boundpsih} to have
	\begin{align*}
		&|	\langle (v^\eps*(\psi_l\psi_h ))\succcurlyeq\cZ_1,\psi_h\rangle|\lesssim \|\psi_h\|_{L^4}\|\psi_l\psi_h\|_{\bB^{\frac12+\kappa}_{\frac43}}\|\mZ_\eps\|
		\\\lesssim&\, K(\|\mZ_\eps\|)\Big(\|\psi_h\|_{\bB^{\frac12+\kappa}_2}\|\psi_l\|_{L^4}+\|\psi_l\|_{\bB^{\frac12+\kappa}_2}\|\psi_h\|_{L^4}\Big),
	\end{align*}
	which can be estimated as $J_2$ by using Young's inequality and \eqref{j1}.
	
	\textbf{IV.} For $J_4$ we focus on $\langle  v^\eps*|\psi_l|^2\overline\cZ_1,\psi_h\rangle$ and the other part follows similarly. We use the paraproduct decomposition to obtain
	\begin{align*}
		\langle  v^\eps*|\psi_l|^2\overline\cZ_1,\psi_h\rangle=\langle  (v^\eps*|\psi_l|^2)\prec\overline \cZ_1,\psi_h\rangle+\langle ( v^\eps*|\psi_l|^2)\succcurlyeq\overline\cZ_1,\psi_h\rangle.
	\end{align*}
	For the first term we use Lemma \ref{lem:para} and \eqref{lowerbound1} to have
	\begin{align*}
		&	|\langle ( v^\eps*|\psi_l|^2)\prec\overline \cZ_1,\psi_h\rangle|\lesssim \|v^\eps*|\psi_l|^2\|_{L^2}\|\psi_h\|_{\bB^{\frac12+\kappa}_2}\|\mZ_\eps\|
		\\\lesssim&\, \mathcal{V}^\eps(\psi_l)^{\frac12}\|\psi_h\|_{\bB^{\frac12+\kappa}_2}\|\mZ_\eps\|
		\lesssim \delta\mathcal{V}^\eps(\psi_l)+\|\psi_h\|^2_{\bB^{\frac12+\kappa}_2}K(\|\mZ_\eps\|),
	\end{align*}
	which can be estimated by using \eqref{j1}. For the second term we use Lemma \ref{lem:para}, Lemma \ref{lem:multi} and the Sobolev embedding $H^{\frac34}\subset L^4$ and the interpolation Lemma \ref{lem:interpolation} to have
	\begin{align*}
		|\langle ( v^\eps*|\psi_l|^2)\succcurlyeq \cZ_1,\psi_h\rangle|\lesssim \|\psi_l\|_{\bB^{\frac12+\kappa}_2}\|\psi_l\|_{L^4}
		\|\psi_h\|_{L^4}\|\mZ_\eps\|\lesssim \delta\|\psi_l\|_{H^1}^2+\|\psi_l\|_{L^2}^2K(\|\mZ_\eps\|),
	\end{align*}
	which by \eqref{lowerbound} implies the desired bound.
	
\end{proof}

\bl It holds that for $\delta>0$ and $0\leq s\leq t\leq T$
\begin{equs}[est:YcZ2c]
	&\int_s^t|	\langle \cG_{\succcurlyeq}+F_{\succcurlyeq},\psi_l\rangle|\,\dif r
	\lesssim \delta\int_s^t(\mathcal{V}^\eps(\psi_l)+\|\psi_l\|_{H^1}^2)\,\dif r+K(\|\mZ_\eps\|)+\|\psi_l(s)\|_{L^2}^2,
\end{equs}
with the proportional constant independent of $\eps$.
\el
\begin{proof}

	By Propositions \ref{prop:Z}, \ref{thcomm1}  and  \ref{prop:Zn1} and \eqref{eq:tauy}  we have
	\begin{align*}
		&|	\langle  F_{\succcurlyeq},\psi_l\rangle|
		\lesssim K(\|\mZ_\eps\|)\|\psi_l\|_{H^{\frac12+3\kappa}}(r^{-\kappa}+1),
	\end{align*}
	which implies the desired bound for it by Lemma \ref{lem:interpolation} and Young's inequality.
	For the term involving the paraproduct $\succ$ from $\cG_\succcurlyeq$ we use the paraproduct estimates Lemma \ref{lem:para} and Proposition  \ref{prop:Z} to have
	\begin{equs}[j2]
		|	\langle \psi\succ\cZ^{\<2vm>},\psi_l\rangle|
		\lesssim&\, K(\|\mZ_\eps\|)(\|\psi_l\|_{\bB_2^{\frac12+\kappa}}+\|\psi_h\|_{\bB_2^{\frac12+\kappa}})\|\psi_l\|_{\bB_2^{\frac12+\kappa}}
		\\\lesssim&\,\delta\|\psi_l\|_{H^1}^2+K(\|\mZ_\eps\|)\|\psi_l\|_{L^2}^2+\|\psi_h\|^2_{\bB_2^{\frac12+\kappa}},
	\end{equs}
	which can be bounded  using \eqref{j1}. For $\int v^\eps(y)\tau_y\psi\succ\cZ^{\<2m>}_y\,\dif y$ and
	$\int v^\eps(y)\tau_y\overline{\psi}\succ\cZ^{\<2>}_y\,\dif y$, we use \eqref{eq:tauy}, Proposition \ref{prop:Zny} and exactly the same argument to conclude the desired estimates.
	
	Moreover, by  Lemma \ref{lem:sto1-n}  we have
	\begin{align*}
		&|\langle\psi\circ\cZ^{\<2vm>}+(\tilde b_1^\eps+\tilde b^\eps_2)(Y+\psi),\psi_l\rangle|
		\\\lesssim & K(\|\mZ_\eps\|)\Big((\|\psi\|_{\bB^{3\kappa}_2}+1)\|\psi_l\|_{H^{2\kappa}}+\|\psi^\sharp\|_{\bB^{1+2\kappa}_2}\|\psi_l\|_{H^{2\kappa}}\Big)
		\\\lesssim&K(\|\mZ_\eps\|)+\|\psi^\sharp\|^{\frac32}_{\bB^{1+2\kappa}_2}+\delta(\|\psi_l\|_{H^1}^{2}+\|\psi_l\|^4_{L^2}),
	\end{align*}
	where we use \eqref{boundpsih} and Young's inequality  in the last step.
	By integrating with respect to time and applying  Proposition \ref{boundsharpp}, the  desired estimate follows for $\psi\circ\cZ^{\<2vm>}$. The desired bounds for the terms involving $\int v^\eps(y)\tau_y\psi\circ\cZ^{\<2m>}_y\,\dif y$ and
	$\int v^\eps(y)\tau_y\overline{\psi}\circ\cZ^{\<2>}_y\,\dif y$ follow the same arguments by Lemma \ref{lem:sto1}.
\end{proof}

\bl It holds that for $\delta>0$ and $0\leq s\leq t\leq T$
\begin{align*}
	\int_s^t|\langle f_1(\psi),\psi_l\rangle|\,\dif r\lesssim\delta\int_s^t(\mathcal{V}^\eps(\psi_l)+\|\psi_l\|_{H^1}^2)\,\dif s+K(\|\mZ_\eps\|),
\end{align*}
with the proportional constant independent of $\eps$.
\el
\begin{proof}
	Using \eqref{bd:f1com} we know that
	\begin{equs}[est:Df1]
		|\langle f_1,\psi_l\rangle|\lesssim&\,(1+r^{-\kappa}) \|\psi\|_{\bB^{\frac12+2\kappa}_2}\|\psi_l\|_{\bB^{\frac12+2\kappa}_2}K(\|\mZ_\eps\|)
		\\\lesssim&\, \|\psi_h\|^2_{\bB^{\frac12+2\kappa}_2}+(1+r^{-2\kappa})\|\psi_l\|^2_{\bB^{\frac12+2\kappa}_2}K(\|\mZ_\eps\|).
	\end{equs}
	We then use \eqref{j1} for $\|\psi_h\|^2_{\bB^{\frac12+2\kappa}_2}$ and the interpolation Lemma \ref{lem:interpolation}, Young's inequality for $\|\psi_l\|^2_{\bB^{\frac12+2\kappa}_2}$ to obtain the desired estimates.

	Combining these estimates the result follows.
\end{proof}

\bl\label{lem:Xil} It holds that for $\delta>0$
\begin{align*}
	&|	\langle \cG_{\prec,\leq R}+F_{\prec,\leq R},\psi_l\rangle|
	\lesssim\delta(\mathcal{V}^\eps(\psi_l)+\|\psi_l\|_{H^1}^2)+K(\|\mZ_\eps\|),
\end{align*}
with the proportional constant independent of $\eps$.
\el
\begin{proof}
	We use the paraproduct estimates Lemma \ref{lem:para}, Lemma \ref{lem:multi}, Propositions \ref{prop:Z}, \ref{prop:Zny}, \eqref{eq:tauy} and the localizer estimate \eqref{eq:loc}  to have
	\begin{align*}
		&	|\langle \cG_{\prec,\leq R}+F_{\prec,\leq R},\psi_l\rangle|
		\\\lesssim&\, (K(\|\mZ_\eps\|)+\|\psi\|_{L^2})\Big(\sup_y\|\Delta_{\leq R}\cZ^{\<2m>}_y\|_{\bC^\kappa}+\sup_y\|\Delta_{\leq R}\cZ^{\<2>}_y\|_{\bC^\kappa}+\|\Delta_{\leq R}\cZ^{\<2vm>}\|_{\bC^\kappa}\Big)\|\psi_l\|_{L^2}
		\\\lesssim&\,K(\|\mZ_\eps\|)(1+\|\psi\|_{L^2})2^{R(r)(1+2\kappa)}\|\psi_l\|_{L^2}.
	\end{align*}
	
	We then use  \eqref{boundr} to have it bounded by
	\begin{align*}
		&(1+\|\psi_l\|^2_{L^2})(\|\psi_l\|_{L^4}+1)^{\frac{1+2\kappa}{\frac{16}{17}-2\kappa}}K(\|\mZ_\eps\|)
		\\\lesssim&\, (1+\|\psi_l\|^2_{L^2})(\|\psi_l\|_{H^1}^{\frac34}\|\psi_l\|_{L^2}^{\frac14}+1)^{\frac{1+2\kappa}{\frac{16}{17}-2\kappa}}K(\|\mZ_\eps\|).
	\end{align*}
	We consequently use Young's inequality and \eqref{lowerbound} to derive the desired result.
\end{proof}

\subsection{Proof of Theorems \ref{uniformmea}--\ref{coming}}\label{proof} In this section we first give the proof of Theorem \ref{coming} based on the uniform estimates derived in Proposition \ref{energy}.

\begin{proof}[Proof of Theorem \ref{coming}] The bound for $\psi_h$ follows from \eqref{boundpsih}. The result for $\psi_l$ follows from the same argument as in the proof of \cite[Theorem 7.1]{MW18}. More precisely, let $F(s)=\|\psi_l(s)\|_{L^2}^2+1$, and by Proposition \ref{energy} we obtain
	$$\int_s^t F(r)^2\,\dif r\leq CK(\|\mZ_\eps\|) F(s).$$
	By the construction of solutions in Theorem \ref{th:global}, $F$ is continuous. Using \cite[Lemma 7.3]{MW18} we obtain  a sequence $0=t_0<t_1<t_2<...<t_N=T$, such that for every $n\in \{0,...,N-1\}$
	$$F(t_n)\lesssim K(\|\mZ_\eps\|) t_{n+1}^{-1}.$$
	The proportional constant is uniform in $\eps$. Hence, for any $t\in[0,T]$ we can find $n\in \{0,...,N-1\}$ such that $t\in(t_n,t_{n+1}]$.  We then apply  Proposition \ref{energy} to obtain
	$$\|\psi_l(t)\|_{L^2}^2\lesssim K(\|\mZ_\eps\|)\Big(\|\psi_l(t_n)\|_{L^2}^2+1\Big)\lesssim K(\|\mZ_\eps\|)\Big( t^{-1}_{n+1}+1\Big)\lesssim K(\|\mZ_\eps\|)(t^{-1}+1).$$
	Thus the result follows.
\end{proof}

It will be convenient to have a stationary coupling of the linear and nonlinear dynamics \eqref{eq:lin} and \eqref{eq:mainnew}, which is stated in the following lemma.

\bl\label{lem:zz1} There exists a  stationary process $(\Psi^\eps, Z)$ such that the components $\Psi^\eps, Z$ are stationary  solutions to \eqref{eq:mainnew} and \eqref{eq:lin} respectively.
\el

The proof follows essentially the same steps as in \cite[Lemma 5.7]{SSZZ20} and \cite[Lemma 4.2]{SZZ21} and we put it in Appendix \ref{app:pro}.

\begin{proof}[Proof of Theorem \ref{uniformmea}]  We take $(\Psi^\eps,Z)$ to be the stationary process given in Lemma \ref{lem:zz1}. As $\Psi^\eps$ is stationary solution, we only need to prove the result for $t=1$. We then set
	\begin{align}\label{def:psi4}
		\psi\eqdef \Psi^\eps -Z+\cZ^{\<3v0m>}-\sI(\cR(Z)).
	\end{align}
	We then know $\psi$ satisfies equation \eqref{eq:psi} with $\psi(0)\in \bC^{-\frac12-\kappa}$. We can also decompose $\psi=\psi_l+\psi_h$ with $\psi_l$, $\psi_h$ satisfying \eqref{eqw} and \eqref{eqv}, respectively, and then apply the uniform in $\eps$ estimate derived in Section \ref{sec:epsih}--Section \ref{sec:en} to $\psi$ and $\psi_l, \psi_h$.
	Using stationary of $(\Psi^\eps,Z)$ and \eqref{boundpsih} we obtain
	\begin{equs}[bd:psi0]
		&\E\|\psi(0)\|_{L^2}^2=\E\|(\Psi^\eps-Z)(0)\|_{L^2}^2=\E\|(\Psi^\eps-Z)(1)\|_{L^2}^2
		\\\lesssim&\,\E\|\psi_l(1)\|_{L^2}^2+\E\|\psi_h(1)\|_{L^2}^2+\E\|\cZ^{\<3v0m>}(1)\|_{L^2}^2+\E\|\sI(\cR(Z))(1)\|_{L^2}^2\lesssim1.
	\end{equs}
	Here we used Theorem \ref{coming} and \eqref{bd:mZ1} in the last step and the proportional constant is independent of $\eps$.
	
	On the other hand, taking expectation on the both sides of \eqref{mest:psil} in  Proposition \ref{energy} with $s=0$ we use \eqref{bd:psi0} to obtain
	\begin{equs}\label{bd:psil-1}
		\E\int_0^1\|\psi_l\|_{H^1}^2\,\dif s +\E\int_0^1\mathcal{V}^\eps(\psi_l)\,\dif s\lesssim \E\|\psi(0)\|_{L^2}^2+1\lesssim 1.
	\end{equs}
	Thus we use Lemma \ref{lem:emb} to obtain
	\begin{align*}
		&\E\|\Psi^\eps(1)\|_{\bC^{-\frac12-\kappa}}^2\lesssim\E\|(\Psi^\eps-Z)(1)\|^2_{\bC^{-\frac12-\kappa}}+\E\|Z(1)\|_{\bC^{-\frac12-\kappa}}^2
		\\\lesssim&\,\E\int_0^1\|(\Psi^\eps-Z)(s)\|^2_{\bC^{-\frac12-\kappa}}\,\dif s+1
		\\\lesssim &\,\E\int_0^1 \|\psi_h\|^2_{\bB^{1-2\kappa}_4}\,\dif s+\E\int_0^1\|\psi_l\|^2_{H^1}\,\dif s+1\lesssim1,\end{align*}
	where we used stationary of $(\Psi^\eps, Z)$ in the second step and the proportional constant is independent of $\eps$.
	Here we used \eqref{bd:psil-1}
	and \eqref{boundpsih1} to obtain
	\begin{align*}\E\int_0^1 \|\psi_h\|^2_{\bB^{1-2\kappa}_4}\,\dif s\lesssim& \,1+\E\Big( K(\|\mZ_\eps\|)\int_0^1\int_0^s(s-r)^{-1+\frac\kappa2}\|\psi_l(r)\|_{L^4}^2\,\dif r \dif s\Big)
		\\\lesssim &\,1+\E\int_0^1\|\psi_l\|_{H^1}^2\,\dif s+\E\int_0^1\cV^\eps(\psi_l)\,\dif s\lesssim1.
	\end{align*}
	Thus the first result follows.
	
	Using Theorem \ref{coming} and \eqref{bd:mZ1} we obtain for any $p\geq2$
	$$\sup_\eps(\E\|\psi^\eps_l(1)\|_{L^2}^p+\E\|\psi^\eps_h(1)\|_{L^2}^p)\lesssim 1.$$
	We then use \eqref{def:psi4} to have
	\begin{align*}
		&\E\|\Psi^\eps(1)\|_{\bB^{-\frac12-\kappa}_2}^p\lesssim\E\|(\Psi^\eps-Z)(1)\|_{\bB^{-\frac12-\kappa}_2}^p+\E\|Z(1)\|_{\bB^{-\frac12-\kappa}_2}^p
		\\\lesssim&\,\E\|\psi^\eps_l(1)\|_{L^2}^p+\E\|\psi^\eps_h(1)\|_{L^2}^p+1
		\lesssim1.\end{align*}
	Thus the second result follows.
\end{proof}

We now give the proof of Theorem \ref{th:main1}.

\begin{proof}[Proof of Theorem \ref{th:main1}]
	Using Theorem \ref{uniformmea}, we know that $\nu^\eps$, for $\eps > 0$, is tight in $\bC^{-\frac12-\kappa}$. Suppose that a subsequence, still denoted by $\nu^\eps$ for simplicity, converges weakly to $\tilde \nu$ in $\bC^{-\frac12-\kappa}$. We begin with the unique solutions $\Psi^\eps$ to equation \eqref{eq:mainnew1} with initial distribution $\nu^\eps$, and the unique solutions $\Phi$ to \eqref{eq:mainnew1} with initial distribution $\tilde \nu$. By Theorem \ref{th:1}, we know that $\Psi^\eps$ converges to $\Phi$ in $C([0,T];\bC^{-\frac12-\kappa})$ in probability. Since $\nu^\eps$ is an invariant measure for \eqref{eq:mainnew} by Theorem \ref{th:global}, we conclude that $\Psi^\eps$ is a stationary solution to equation \eqref{eq:mainnew}. Hence, $\Phi$ is a stationary solution to equation \eqref{eq:mainnew1}, and therefore, $\tilde \nu$ is an invariant measure for equation \eqref{eq:mainnew1}. Given that the invariant measure for equation \eqref{eq:mainnew1} is unique by Theorem \ref{th:phi43} and is given by $\nu$, we deduce that $\tilde \nu = \nu$. Consequently, the entire sequence $\nu^\eps$ converges weakly to $\nu$ in $\bC^{-\frac12-\kappa}$. Furthermore, by Theorem \ref{uniformmea}, we obtain
	$$
	\sup_\eps \int \|\Psi\|^p_{\bB_2^{-\frac12-\kappa}} \nu^\eps(\dif \Psi) \lesssim 1,
	$$
	for any $p \geq 1$, which implies the convergence of $p$-point correlation functions by uniform integrability.
\end{proof}

\section{Stochastic calculations}\label{sec:sto}

In this section, we prove the stochastic estimates required in Section \ref{first}, following the notations from \cite[Section 9]{GP17}. We express the complex-valued white noise through its spatial Fourier transform. More precisely, let  $E=\mathbb{Z}^3$ and let $B(\cdot,k)=\langle \xi,e_k\rangle$ and $\overline{B}(\cdot,k)=\langle \overline{\xi},e_k\rangle$ for $e_k(x)=(2\pi)^{-\frac{3}{2}}e^{\imath x\cdot k},x\in\mathbb{T}^3$, $k\in\mZ^3$. Note that $\overline{B}(\cdot,k)$ is not the conjugate of $B(\cdot,k)$, but rather we have
$\overline{B}(\cdot,k)=\overline{B(\cdot,-k)}$ with $\overline{B(\cdot,-k)}$ being the conjugate of $B(\cdot,-k)$.   We view $B(\cdot,k)$ as a Gaussian noise on $\mathbb{R}\times E$ with covariance given by
\begin{equs}[eq:Ito]
	\E\Big(\int_{\mathbb{R}\times E}f(\eta) B (\dif\eta) \int_{\mathbb{R}\times E}g(\eta') \overline{B} (\dif \eta')\Big) =& \, 2\int_{\mathbb{R}\times E}
	g(\eta_1)f(\eta_{-1})\,\dif \eta_1,\\
	\E\Big(\int_{\mathbb{R}\times E}f(\eta) B (\dif\eta) \int_{\mathbb{R}\times E}g(\eta') {B} (\dif \eta')\Big) =&\,\E\Big(\int_{\mathbb{R}\times E}f(\eta) \overline{B} (\dif\eta) \int_{\mathbb{R}\times E}g(\eta') \overline{B} (\dif \eta')\Big)= 0,
\end{equs}
where $\eta_a=(s_a,k_a)$, $s_{-a}=s_a, k_{-a}=-k_a$ and the measure $\dif\eta_a=\dif s_a\dif k_a$ is the product of the Lebesgue measure $\dif s_a$ on $\mathbb{R}$ and of the counting measure $\dif k_a$ on $E$.

We write
\begin{align*}
	Z(t)=&\int_{\mR\times E}e_kP_{t-s}B(\dif \eta),\qquad \langle Z(t),e_k\rangle=\int_{\mR} p_{t-s}(k)\,\dif B(s,k),
	\\
	\overline Z(t)=&\int_{\mR\times E}e_kP_{t-s}\overline{B}(\dif \eta),\qquad \langle \overline{Z}(t),e_{k}\rangle=\int_{\mR} p_{t-s}(k)\,\dif \overline{B}(s,k),
\end{align*}
with $p_{t}(k)=e^{-\la k\ra^2t}\1_{t\geq0}$. Hence, we use \eqref{eq:Ito} to have for $t\geq\sigma$
\begin{equs}[cor:Z]
	\E \langle Z(t),e_k\rangle\langle \overline{Z}(\sigma),e_{k'}\rangle=&\,2\1_{\{k+k'=0\}}\int_{-\infty}^\sigma p_{t-s}(k)p_{\sigma-s}(k)\,\dif s=\frac{\1_{\{k+k'=0\}}}{\la k\ra^2}e^{-\la k\ra^2(t-\sigma)},\\ \E \langle Z(t),e_k\rangle\langle Z(\sigma),e_{k'}\rangle=&\,0,\qquad \E \langle \overline Z(t),e_k\rangle\langle \overline Z(\sigma),e_{k'}\rangle=0.
\end{equs}
We will meet the random fields written as
\begin{align*}
	\cZ^\tau=\int_{(\mR\times E)^{n+m}}H(t,x,\eta)\prod_{i=1}^nB(\dif \eta_i)\prod_{j=n+1}^{n+m}\overline{B}(\dif \eta_j),
\end{align*}
which are easier to handle when decomposed into different chaos. We now introduce the following notations $k_{[1...n]}=\sum_{i=1}^nk_i$, $\eta_{1...n}=(\eta_1,...,\eta_n)\in(\mathbb{R}\times E)^n$, $\dif\eta_{1...n}=\dif\eta_1\dots\dif\eta_n$.  Denote by
$$\int_{(\mathbb{R}\times E)^{n+m}}f(\eta_{1\dots n+m})B(\dif\eta_{1\dots n})\diamond \overline{B}(\dif\eta_{n+1\dots n+m})$$
a generic element of the $n+m$-th chaos of $B$ on $\mathbb{R}\times E$. Since $B$ is complex valued white noise, we need to consider the conjugate part here. We refer to \cite[Appendix A]{HIN17} for more details.  By \cite[Section 9.2]{GP17} and \cite[Appendix A]{HIN17} we know that
\begin{align}\label{bd:mL2}
	\E\Big(\Big|\int_{(\mathbb{R}\times E)^{n+m}}f(\eta_{1...n+m})B(\dif\eta_{1...n})\diamond \overline{B}(\dif\eta_{n+1\dots n+m})\Big|^2\Big)\lesssim_{n,m} \int_{(\mathbb{R}\times E)^{n+m}}|f(\eta_{1...n+m})|^2\dif\eta_{1...n+m}.\end{align}
Hence for bounding the variance of the chaos it is enough to bound the $L^2$ norm of the unsymmetrized kernels.

We introduce the following notation: for $k_1,k_2\in \mathbb{Z}^3$
$$\psi_\prec(k_1,k_2)\eqdef\sum_{j\geq-1}\sum_{i<j-1}\theta_i(k_1)\theta_j(k_2),\quad\psi_\circ(k_1,k_2)\eqdef\sum_{|i-j|\leq1}\theta_i(k_1)\theta_j(k_2),$$
and $\psi_{\succcurlyeq}=1-\psi_{\prec}$.

In the following $\kappa>0$ is any small number.

\subsection{Proof of Propositions \ref{prop:Zn} and \ref{prop:Zn1}}\label{proof2}
In this subsection, we perform stochastic calculations for the random fields introduced in Propositions \ref{prop:Zn} and \ref{prop:Zn1}.

\begin{proof}[Proof of Proposition \ref{prop:Zn}]
	We recall that for $k_1\in \mZ^3$ $$\langle \sI(\cR Z),e_{k_1}\rangle=\int_0^t\widehat{\cR}(k_1)e^{-(t-s)\la k_1\ra^2}\langle Z(s), e_{k_1}\rangle\,\dif s,$$
	with
	$\widehat{\cR}(k)=\frac1{(2\pi)^{3}} \sum_k\frac1{\la k\ra^2}\Big(\widehat{v^\eps}(k)-\widehat{v^\eps}(k_1-k)\Big)$.

	We first consider $(v^\eps*\overline{\sI(\cR Z)})\circ Z$.   By chaos decomposition we have
	\begin{align*}
		(v^\eps*\overline{\sI(\cR Z)})\circ Z=c_1^\eps+I^2,
	\end{align*}
	with $I^2$ in the second chaos and $c_1^\eps$ in the zeroth chaos.

	\noindent\textbf{Terms in the zeroth chaos}: By direct calculation and \eqref{cor:Z} we find
	\begin{equs}[defc1]c_1^\eps =&\,\frac1{(2\pi)^{3}}\sum_{k_1}\widehat{\cR}(k_1)\widehat{v^\eps}(k_1)\frac{1-e^{-2t\la k_1\ra^2}}{2\la k_1\ra^4}\psi_\circ(k_1,k_1)
		\\=&\,\frac1{(2\pi)^{6}}\sum_{k,k_1}\frac{\widehat{v^\eps}(k_1)(\widehat{v^\eps}(k)-\widehat{v^\eps}(k_1-k))(1-e^{-2t\la k_1\ra^2})}{2\la k_1\ra^4\la k\ra^2}.
	\end{equs}
	Since $\widehat{v^\eps}(k)=\hat{v}(\eps k)$, we use change of variable to have
	\begin{align*}
		c_1^\eps=&\,\eps^6\sum_{k,k_1\in\mathbb \eps \mZ^3}\frac{\hat v(k_1)(\hat v(k)-\hat v(k_1-k))(1-e^{-2t(|k_1|^2\eps^{-2}+1)})}{2\cdot(2\pi)^{6}(|k_1|^2+\eps^2)^2(|k|^2+\eps^2)}=\sum_{i=1}^2L_i,
	\end{align*}
	with
	\begin{align*}
		L_1\eqdef &\,\eps^6\sum_{k,k_1\in\mathbb \eps \mZ^3,|k_1|\leq1}\frac{\hat v(k_1)(\hat v(k)-\hat v(k_1-k))}{2\cdot(2\pi)^{6}(|k_1|^2+\eps^2)^2(|k|^2+\eps^2)}(1-e^{-2t(|k_1|^2\eps^{-2}+1)}),
		\\
		L_2\eqdef&\,\eps^6\sum_{k,k_1\in\mathbb \eps \mZ^3,|k_1|>1 }\frac{\hat v(k_1)(\hat v(k)-\hat v(k_1-k))}{2\cdot (2\pi)^6(|k_1|^2+\eps^2)^2(|k|^2+\eps^2)}(1-e^{-2t(|k_1|^2\eps^{-2}+1)}).
	\end{align*}
	Recall that $C_1$, given in \eqref{eq:counterterms}, can be decomposed into $C_{11} + C_{12}$, where $C_{11}$ and $C_{12}$ correspond to the first and second terms on the RHS of the definition of $C_1$ in  \eqref{eq:counterterms}, respectively.
	It is easy to see that
	$$|L_2- C_{12}|
	\lesssim\eps^\kappa(1+t^{-\frac\kappa2}).$$
	For $L_1$ we use $\nabla\hat v(k)=-\nabla\hat v(-k)$ to have
	\begin{align*}	
		L_1=&\,\eps^6\sum_{k,k_1\in\mathbb \eps \mZ^3,|k_1|\leq1}\frac{\hat v(k_1)(\hat v(k)-\hat v(k_1-k)-k_1\cdot \nabla \hat v(k))}{2\cdot (2\pi)^{6}(|k_1|^2+\eps^2)^2(|k|^2+\eps^2)}(1-e^{-2t(|k_1|^2\eps^{-2}+1)}).
	\end{align*}
	We first consider the term excluding $e^{-2t(|k_1|^2\eps^{-2}+1)}$ and compare it with $C_{11}$. We employ Taylor expansion
	and interpolation for the numerator from $L_1$ and $C_{11}$ to
	have it
	bounded by
	\begin{equs}[convegc1]&\eps^\kappa\int_{|x|\leq1}\frac1{|x|^4|y|^2}|x|^2\int_0^1\int_0^1\frac1{(1+|y-sux|^2)^{1-\kappa}}\,\dif s\dif u\dif x\dif y
		\\\lesssim&\,\eps^\kappa\int_0^1\int_0^1\int_{|x|\leq1,|y|\leq1}\frac1{|x|^2|y|^2}\,\dif x\dif y\dif s\dif u\\&+\eps^\kappa\int_0^1\int_0^1\int_{|x|\leq1}\frac1{|x|^2}\Big(\int_{|y|>1}\frac1{|y|^4}\,\dif y\Big)^{\frac12}\Big(\int_{|y|>1}\frac1{(1+|y-sux|^2)^{2-2\kappa}}\,\dif y\Big)^{\frac12}\dif x\dif s\dif u
		\\\lesssim&\, \eps^\kappa
		.
	\end{equs}
	For the term involving $e^{-2t(|k_1|^2\eps^{-2}+1)}$ the approach is similar, since we can bound it with an additional factor $\eps^\kappa t^{-\frac\kappa2}|k_1|^{-\kappa}.$

	Thus we obtain
	\begin{equs}[constant]|c_1^\eps- C_1|\lesssim \eps^{\kappa}(1+t^{-\frac\kappa2}).
	\end{equs}

	\noindent\textbf{Terms in the second chaos}: By direct calculation we have
	\begin{align*}&
		I^2_t=(2\pi)^{-3}\sum_{k_1,k_2}\psi_\circ(k_1,k_2)\widehat{v^\eps}(k_1)\int_0^te^{-(t-s)\la k_1\ra^2}\widehat{\cR}(k_1)\Wick{\langle  \overline Z_s,e_{k_1}\rangle \langle Z_t, e_{k_2}\rangle} \,\dif s e^{\imath k_{[12]}\cdot x},
	\end{align*}
	with $\Wick{\langle  \overline Z_s,e_{k_1}\rangle \langle Z_t, e_{k_2}\rangle}$ denoting  Wick product of $\langle  \overline Z_s,e_{k_1}\rangle \langle Z_t, e_{k_2}\rangle$ given by
	$$\langle  \overline Z_s,e_{k_1}\rangle \langle Z_t, e_{k_2}\rangle-\E[\langle  \overline Z_s,e_{k_1}\rangle \langle Z_t, e_{k_2}\rangle].$$
	Thus we  use \eqref{vg} with $\eta=\kappa$, \eqref{cor:Z}, \eqref{bd:mL2} and Lemma \ref{lem:sum} to obtain
	\begin{align*}&
		\E|\Delta_q I^2|^2\lesssim\sum_{k_1,k_2}\frac{\psi_\circ(k_1,k_2) \eps^{2\kappa}\theta_q(k_{[12]}) }{\la k_1\ra^{4-2\kappa}\la k_2\ra^2}\lesssim \eps^{2\kappa}2^{2q\kappa},
	\end{align*}
	where we used $i\sim j$ in $\psi_\circ$ to have $|k_1|\backsimeq |k_2|$ in the last step.
	Since $	\Delta_q[(v^\eps*\overline{\sI(\cR Z)})\circ Z-c_1^\eps]$ is a random variable with a finite chaos decomposition,  we use Gaussian hypercontractivity to have for $p\geq2$
	\begin{align*}
		\E|	\Delta_q((v^\eps*\overline{\sI(\cR Z)})\circ Z-c_1^\eps)|^{p}\lesssim \eps^{p\kappa}2^{qp\kappa},
	\end{align*}
	which implies
	\begin{align*}
		\sup_q2^{-pq\kappa }\E\|	\Delta_q[v^\eps*\overline{\sI(\cR Z)}\circ Z-c_1^\eps]\|_{L^p}^{p}\lesssim \eps^{p\kappa}.
	\end{align*}
	Using the Besov embedding Lemma \ref{lem:emb}  we choose $p$ large enough to obtain
	\begin{align*}
		\E\|(v^\eps*\overline{\sI(\cR Z)})\circ Z\|_{\bC^{-2\kappa}}^{p}\lesssim \E\|	(v^\eps*\overline{\sI(\cR Z)})\circ Z\|_{\bB^{-\frac{3\kappa}2}_{p,p}}^{p}\lesssim \eps^{p\kappa}.
	\end{align*}
	Furthermore, the second bound in \eqref{con:RZ} follows by considering the time difference  $$|1-e^{-|t_1-t_2||k|^2}|\lesssim (|t_1-t_2| |k|^{2})\wedge 1,$$
	and Kolmogorov's continuity criterion.
	
	We now consider the remaining terms.
	For  $(v^\eps*\overline{Z})\circ \sI(\cR Z)$ the required bounds and convergence follow the same line.
	For  $(v^\eps*\sI(\cR Z))\circ Z$ and $ (v^\eps*Z)\circ \sI(\cR Z)$, the zeroth chaos part vanishes, and the terms in the second chaos are bounded in a similar manner.
	
	In the case of $\Re[\overline{Z} \circ \sI(\cR Z)]$, the terms in the second chaos are similarly bounded, and we focus only on the terms in the zeroth chaos, given by $c_2^\eps$.
	
	\noindent\textbf{Terms in the zeroth chaos}: By direct calculation
	\begin{equs}[defc2]c_2^\eps=&\sum_{k_1}\widehat{\cR}(k_1)\frac{1-e^{-2t\la k_1\ra^2}}{2\cdot(2\pi)^{3} \la k_1\ra^4}\psi_\circ(k_1,k_1)
		\\=&\sum_{k,k_1}\frac{(\widehat{v^\eps}(k)-\widehat{v^\eps}(k_1-k))(1-e^{-2t\la k_1\ra^2})}{2\cdot (2\pi)^6\la k_1\ra^4\la k\ra^2}.
	\end{equs}
	Note that $c_2^\eps$ is the same as $c_1^\eps$, with $\widehat{v^\eps}(k_1)$ replaced by 1. Therefore, we can apply a similar decomposition and follow the same steps as in the calculation for \eqref{constant} to obtain
	$$|c_2^\eps- C_2|\lesssim \eps^{\kappa}(1+t^{-\frac\kappa2}).$$
	
	Using the Besov embedding Lemma \ref{lem:emb}, Gaussian hypercontractivity, the bound in \eqref{con:RZ2} follows.
\end{proof}

\begin{proof}[Proof of Proposition \ref{prop:Zn1}] We classify these calculations into the following two categories:
	\begin{align*}
		\textbf{I. }  \sI(\cR Z)\circ\cZ^{\<2vm>},\quad \int v^\eps(y)\tau_y\sI(\cR Z)\circ\cZ^{\<2m>}_y\,\dif y,\qquad  \int v^\eps(y)\tau_y\sI(\cR \overline Z)\circ\cZ^{\<2>}_y\,\dif y,
	\end{align*}
	and
	\begin{align*}\textbf{II. } \tau_y\overline{\sI(\cR Z)}\circ Z, \quad\tau_y\sI(\cR Z)\circ Z,\quad {\sI(\cR Z)}\circ\tau_y\overline Z,\quad \sI(\cR Z)\circ \tau_yZ.\end{align*}
	
	\textbf{I.}
	We start with $\sI(\cR Z)\circ\cZ^{\<2vm>}$, which by chaos decomposition is given by
	$$\sI(\cR Z)\circ\cZ^{\<2vm>}=I_t^3+I_t^1, $$
	with
	\begin{align*}
		I_t^3=\int H(\eta_{123})B(\dif \eta_{12})\diamond\overline{B}(\dif \eta_3),\qquad I_t^1=\int H(\eta_{12(-1)})B(\dif \eta_{2}),
	\end{align*}
	for
	\begin{align*}
		H\eqdef\psi_\circ(k_1,k_{[23]})\widehat{v^\eps}(k_{[23]}) \int_0^t\widehat{\cR}(k_1)e^{-(t-s)(|k_1|^2+1)}p_{s-s_1}(k_1)p_{t-s_2}(k_2)p_{t-s_3}(k_3)\,\dif s\prod_{i=1}^3e_{k_i}.
	\end{align*}
	
	\noindent\textbf{Terms in the third chaos}: Using \eqref{vg} and \eqref{cor:Z}, \eqref{bd:mL2} we obtain by Lemma \ref{lem:sum}
	\begin{align*}&
		\E|\Delta_q I_t^3|^2\lesssim \eps^{\kappa}\sum_{k_1,k_2,k_3} \psi_\circ(k_1,k_{[23]}) \frac{ \theta_q(k_{[123]}) }{\la k_1\ra^{4-\kappa}\la k_2\ra^2\la k_3\ra^2}\lesssim \eps^{\kappa} 2^{q(1+\kappa)},
	\end{align*}
	where we used $|k_{[23]}|\backsimeq |k_1|$ in the last step.

	\noindent\textbf{Terms in the first chaos}: We have
	\begin{align*}
		I_t^1=&\,\frac1{(2\pi)^{3}}\sum_{k_1,k_2} \psi_\circ(k_1,k_2-k_1)\frac{\widehat{\cR}(k_1)}{2\la k_1\ra^4} \widehat{v^\eps}(k_2-k_1)
		\langle Z_t,e_{k_2}\rangle e_{k_2}(1-e^{-2t\la k_1\ra^2})
		\\=&\,L_1+c_1^\eps Z,
	\end{align*}
	with
	$c_1^\eps$  given by \eqref{defc1} and
	\begin{align*}
		L_1\eqdef&\, (2\pi)^{-3}\sum_{k_1,k_2} \frac{\widehat{\cR}(k_1)}{2\la k_1\ra^4}[\psi_\circ(k_1,k_2-k_1) \widehat{v^\eps}(k_2-k_1)-\psi_\circ(k_1,k_1) \widehat{v^\eps}(k_1)]
		\\&\qquad\qquad\times \langle Z_t,e_{k_2}\rangle e_{k_2}(1-e^{-2t\la k_1\ra^2}).
	\end{align*}
	We then claim that $L_1$ converge to zero in $C_T\bC^{-\frac12-\kappa}$. In fact, we use interpolation to
	have
	\begin{align*}
		&|\psi_\circ(k_1,k_2-k_1) \widehat{v^\eps}(k_2-k_1)-\psi_\circ(k_1,k_1) \widehat{v^\eps}(k_1)|
		\\\lesssim&\,\psi_\circ(k_1,k_2-k_1)(\eps |k_2|)^{\frac{3\kappa}4}|\widehat{v^\eps}(k_2-k_1)-\widehat{v^\eps}(k_1)|^{1-\kappa}+(|k_1|^{-1}|k_2|)^{\frac{3\kappa}4}
		\\\lesssim&\,(|k_1|^{-1}|k_2|)^{\frac{3\kappa}4}+ \psi_\circ(k_1,k_2-k_1)|k_2|^{\frac{3\kappa}4}(|k_2-k_1|^{-\frac{3\kappa}4}+|k_1|^{-\frac{3\kappa}4})
		\\\lesssim&\,(|k_1|^{-1}|k_2|)^{\frac{3\kappa}4},
	\end{align*}
	which combined with \eqref{vg}  implies
	\begin{equs}\label{vg1}\E|\Delta_qL_1|^2\lesssim& \,\eps^{\kappa}\sum_{k_2}\frac{\theta_q(k_2)^2}{\la k_2\ra^2}\bigg(\sum_{k_1}\frac{(|k_1|^{-1}|k_2|)^{\frac{3\kappa}4}}{1+|k_1|^{3-\frac\kappa2}}\bigg)^2
		\\\lesssim& \,\eps^{\kappa}\sum_{k_2}\frac{\theta_q(k_2)^2}{|k_2|^{2-\frac{3\kappa}2}+1}\lesssim \eps^\kappa2^{q(1+\frac{3\kappa}2)}.
	\end{equs}
	Using the Besov embedding Lemma \ref{lem:emb}, Gaussian hypercontractivity, we have
	\begin{align*}
		\sI(\cR Z)\circ\cZ^{\<2vm>}-c_1^\eps Z\to0\text{ in } C_T\bC^{-\frac12-\kappa}.
	\end{align*}
	In the case of $ \int v^\eps(y)\tau_y\sI(\cR Z)\circ\cZ^{\<2m>}_y\,\dif y$, we have the same chaos decomposition, with $\widehat{v^\eps}(k_{[23]})$ in $H$ replaced by $\widehat{v^\eps}(k_{[13]})$, which gives $c_2^\eps Z$ as the renormalization counterterm. The convergence of $\int v^\eps(y)\tau_y\sI(\cR Z)\circ\cZ^{\<2m>}_y\,\dif y - c_2^\eps Z$ in $C_T \bC^{-\frac12-\kappa}$ follows the same reasoning.
	For $ \int v^\eps(y)\tau_y\sI(\cR \overline Z)\circ\cZ^{\<2>}_y\,\dif y$, we have two terms in the first chaos, arising from $k_1 + k_2 = 0$ and $k_1 + k_3 = 0$, which yield $c_1^\eps Z$ and $c_2^\eps  Z$, respectively. Thus, the convergence of $ \int v^\eps(y)\tau_y\sI(\cR \overline Z)\circ\cZ^{\<2>}_y\,\dif y - c_1^\eps Z - c_2^\eps  Z$ follows the same approach.
	
	\textbf{II.} For the second type, we employ a calculation akin to $\Re[\overline Z\circ \sI(\cR Z)]$ from the proof of Proposition \ref{prop:Zn}.  The $\sup_y$ bound follows from  $|e^{\imath k \cdot y_1}-e^{\imath k\cdot y_2}|\lesssim (|k||y_1-y_2|)\wedge 1$ and Kolmogorov's continuity criterion. In fact the calculation for the second order chaos follows the same line. The zeroth chaos components of $\tau_y\sI(\cR Z)\circ Z$ and $\sI(\cR Z)\circ \tau_yZ$ vanish, while  the zeroth chaos part of $\tau_y\overline{\sI(\cR Z)}\circ Z$,  ${\sI(\cR Z)}\circ\tau_y\overline Z$ is given by
	$$c_2^\eps(y)=\sum_{k_1}\widehat{\cR}(k_1)\frac{1-e^{-2t\la k_1\ra^2}}{2\cdot(2\pi)^{3} \la k_1\ra^4}\psi_\circ(k_1,k_1)e^{-\imath k_1\cdot y}.$$
	We bound $e^{-\imath k_1\cdot y}$   by $1$ and use a similar  calculation as in \eqref{convegc1} to have
	$\sup_y |c_2^\eps(y)|\lesssim1$. The result then follows.
\end{proof}

\subsection{Proof of Proposition \ref{thcomm1}}\label{proof1}

We focus on the first term $\int v^\eps(y)\tau_y\cZ_\eps^{\<3v0m>}\circ \cZ^{\<2m>}_y\,\dif y$, as the other term follows the same line.
We also have chaos decomposition
\begin{align*}
	\int v^\eps(y)\tau_y\cZ_\eps^{\<3v0m>}\circ \cZ^{\<2m>}_y\,\dif y= I_t^5+I_t^3+I_t^1,
\end{align*}
with
\begin{align*}
	I_t^5=\int H(\eta_{12345})B(\dif \eta_{135})\diamond\overline{B}(\dif \eta_{24}),
\end{align*}
and $I_t^i, i=1,3, $ given below, where
\begin{align*}
	H(\eta_{1\dots5})\eqdef&\, \psi_\circ(k_{[123]},k_{[45]})\widehat{v^\eps}(k_{[1234]})\int_0^t	p_{t-s}(k_{[123]})\widehat{v^\eps}(k_{[12]})
	\prod_{i=1}^3p_{s-s_i}(k_i)\,\dif s\prod_{j=4}^5p_{t-s_j}(k_j)\prod_{i=1}^5e_{k_i}.
\end{align*}
In the following we will solely demonstrate a uniform in   $\eps$ bound for $\int v^\eps(y)\tau_y\cZ_\eps^{\<3v0m>}\circ \cZ^{\<2m>}_y\,\dif y$ and the required result for the difference of $\int v^\eps(y)\tau_y\cZ_\eps^{\<3v0m>}\circ \cZ^{\<2m>}_y\,\dif y$ and $\cZ^{\<32oc>}$  can be derived analogously by applying
$$|\widehat{v^\eps}(k)-1|\lesssim \eps^\kappa |k|^\kappa.$$

\noindent\textbf{Terms in the fifth chaos}:
Using  Lemma \ref{lem:sum} and \eqref{cor:Z}, \eqref{bd:mL2} we have
$$\aligned \E|\Delta_q I_t^5|^2\lesssim&\sum_{k_1,k_2,k_3,k_4,k_5} \theta(2^{-q}k_{[12345]})\prod_{i=1}^5\frac{1}{\la k_i\ra^2}
\frac{\psi_\circ(k_{[123]},k_{[45]})}{\la k_{[123]}\ra^2(\la k_{[123]}\ra^2+\la k_1\ra^2+\la k_2\ra^2)}\\
\lesssim& \sum_{k_{[123]},k_{[45]}}\theta(2^{-q}k_{[12345]}) \frac{1}{\la k_{[123]}\ra^{3-\kappa}\la k_{[45]}\ra^{2+\kappa}}\lesssim 2^{q},\endaligned$$
where we used $|k_{[45]}|\simeq |k_{[123]}|$ in the second step.

\noindent\textbf{Terms in the third chaos}:
We have the following decomposition:
\begin{align*}
	I_t^{3}=\sum_{j=1}^{3}I_{t}^{3j},
\end{align*}
with
\begin{align*}
	I_t^{31}\eqdef&\int H(t,x,\eta_{123(-3)5})\,\dif \eta_{3}B(\dif\eta_{15})\diamond\overline{B}(\dif \eta_2), \qquad  I_t^{32}\eqdef\int H(t,x,\eta_{1234(-2)})\,\dif \eta_{2}B(\dif\eta_{13})\diamond\overline{B}(\dif \eta_{4}),
	\\I_t^{33}\eqdef&\int H(t,x,\eta_{123(-1)5})\,\dif \eta_{1}B(\dif\eta_{35})\diamond\overline{B}(\dif \eta_2).
\end{align*}
We use Lemma \ref{lem:sum} to have
$$\aligned \E|\Delta_q I_t^{31}|^2
\lesssim
&\sum_{k_1,k_2,k_5} \frac{\theta(2^{-q}k_{[125]})}{\la k_1\ra^2\la k_2\ra^2\la k_5\ra^2}
\Big(\sum_{k_3}\frac{1}{\la k_3\ra^2(\la k_{[123]}\ra^2+\la k_3\ra^2)}\Big)^2
\\
\lesssim& \sum_{k_{[12]},k_5} \frac{\theta(2^{-q}k_{[125]})}{\la k_{[12]}\ra^{3-\frac{\kappa}2}\la k_5\ra^2}\,\dif k_{[12]5}\lesssim 2^{q(1+\frac{\kappa}2)}.\endaligned$$
For $I_t^{32}$, we simply replace $k_5$ with $k_4$ and change the roles of $k_2$ and $k_3$ in the above estimate and obtain the same bounds for $\E|\Delta_q I_t^{32}|^2$. Similarly, for $I_t^{33}$, we exchange the roles of $k_1$ and $k_3$ in the estimate and derive the same bounds for $\E|\Delta_q I_t^{33}|^2 $.

\noindent\textbf{Terms in the first chaos}:
We have the following decomposition:
\begin{align*}
	I_t^{1}=\sum_{j=1}^{2}I_{t}^{1j},
\end{align*}
with
\begin{align*}
	I_t^{11}\eqdef&\int H(t,x,\eta_{123(-1)(-2)})\,\dif \eta_{12}B(\dif\eta_{3}), \qquad  I_t^{12}\eqdef\int H(t,x,\eta_{123(-3)(-2)})\,\dif \eta_{23}B(\dif\eta_{1}).
\end{align*}
The convergence of the terms $I_t^{1j}, j=1,2,$  necessitates renormalization.  More precisely we prove the probabilistic bounds for the following terms:
\begin{align*}
	I_t^{11}-\tilde b_2^\eps Z, \qquad I_t^{12}-\tilde b_3^\eps Z.
\end{align*}
Using $\int |p_{t-s_3}(k_3)-p_{s-s_3}(k_3)|^2\,\dif s\lesssim \frac1{|k_3|^{2-\frac\kappa2}}|t-s|^{\frac\kappa4},$
we derive
$$\aligned \E|\Delta_q(I_t^{11}-\tilde b_2^\eps Z)|^2
\lesssim &\,\sum_{k_3}\frac{\theta(2^{-q}k_{3})}{|k_3|^{2-\frac\kappa2}}\bigg(\sum_{k_1,k_2}\frac{1}{\la k_1\ra^2\la k_2\ra^2(\la k_{[123]}\ra^2+\la k_1\ra^2+\la k_2\ra^2)^{1+\frac\kappa8}}\bigg)^2
\\&+\sum_{k_3}  \frac{\theta(2^{-q}k_{3})}{\la k_3\ra^2}\bigg(\sum_{k_1,k_2}\frac{\widehat{v^\eps}(k_{[12]})h_1(k_1,k_2,k_3)}{\la k_1\ra^2\la k_2\ra^2}\bigg)^2,\endaligned$$
where
\begin{align*}
	h_1(k_1,k_2,k_3)\eqdef\widehat{v^\eps}(k_{[23]})L(k_{[123]})\psi_\circ(k_{[123]},k_{[12]})-\widehat{v^\eps}(k_{2})L(k_{[12]})\psi_\circ(k_{[12]},k_{[12]}),
\end{align*}
and
$$L(k)\eqdef\frac{1-e^{-t(\la k\ra^2+\la k_1\ra^2+\la k_2\ra^2)}}{\la k\ra^2+\la k_1\ra^2+\la k_2\ra^2}.$$
The term $\widehat{v^\eps}(k_{2})L(k_{[12]})\psi_\circ(k_{[12]},k_{[12]})$ comes from  the renormalization counterterm $\tilde b_2^\eps Z$.
We rewrite $h_1(k_1,k_2,k_3)$ as
\begin{align*}
	&\big(\widehat{v^\eps}(k_{[23]})-\widehat{v^\eps}(k_{2})\big)L(k_{[123]})\psi_\circ(k_{[123]},k_{[12]})+\widehat{v^\eps}(k_{2})\big(L(k_{[123]})-L(k_{[12]})\big)\psi_\circ(k_{[123]},k_{[12]})
	\\&+\big(\psi_\circ(k_{[123]},k_{[12]})-\psi_\circ(k_{[12]},k_{[12]})\big)L(k_{[12]})\widehat{v^\eps}(k_{2}),
\end{align*}
which by interpolation is bounded
\begin{align*}&\eps^{\frac{\kappa}2} |k_3|^{\frac{\kappa}4}\Big(\frac1{(\eps|k_{[23]}|)^{\frac\kappa4}}+\frac1{(\eps|k_{2}|)^{\frac\kappa4}}\Big)\frac1{|k_{[123]}|^2+\la k_1\ra^2+\la k_2\ra^2} \\&+|k_3|^{\frac{\kappa}4}\frac1{(|k_1|+|k_2|)^{2+\frac{\kappa}4}+1}+\Big(\frac{|k_3|}{|k_{[12]}|}\Big)^{\frac\kappa4} L(k_{[12]}).
\end{align*}
We then use Lemma \ref{lem:sum} to have that $\E|\Delta_q(I_t^{11}-\tilde b_2^\eps Z)|^2$ is bounded by
$$\sum_{k_3} \frac{\theta(2^{-q}k_{3})}{\la k_3\ra^{2-\frac{\kappa}2}}\lesssim 2^{q(1+\frac{\kappa}2)}.$$
For  $I_t^{12}-\tilde b_3^\eps Z$ similarly we have
$$\aligned\E|\Delta_q(I_t^{12}-\tilde b_3^\eps Z)|^2
\lesssim&\,
\sum_{k_1}\frac{\theta(2^{-q}k_{1})}{|k_1|^{2-\frac\kappa2}}\bigg(\sum_{k_3,k_2}\frac{1}{\la k_3\ra^2\la k_2\ra^2(\la k_{[123]}\ra^2+\la k_3\ra^2+\la k_2\ra^2)^{1+\frac\kappa8}}\bigg)^2
\\&+\sum_{k_1}  \frac{\theta(2^{-q}k_{1})}{\la k_1\ra^2}\bigg(\sum_{k_2,k_3}\frac{h_2(k_1,k_2,k_3)}{\la k_2\ra^2\la k_3\ra^2}\bigg)^2,\endaligned$$
where
\begin{align*}
	h_2(k_1,k_2,k_3)\eqdef&\,\widehat{v^\eps}(k_{[12]})^2L_1(k_{[123]})\psi_\circ(k_{[123]},k_{[23]})-\widehat{v^\eps}(k_{2})^2L_1(k_{[23]})\psi_\circ(k_{[23]},k_{[23]}),
\end{align*}
and
$$L_1(k)\eqdef\frac{1-e^{-t(\la k\ra^2+\la k_2\ra^2+\la k_3\ra^2)}}{\la k\ra^2+\la k_2\ra^2+\la k_3\ra^2}.$$
Here $\widehat{v^\eps}(k_{2})^2L_1(k_{[23]})$ arises from   the renormalization counterterms $\tilde b_3^\eps Z$.
We write $h_2$ as
\begin{align*}
	&(\widehat{v^\eps}(k_{[12]})^2-\widehat{v^\eps}(k_{2})^2)L_1(k_{[123]})\psi_\circ(k_{[123]},k_{[23]})+\widehat{v^\eps}(k_{2})^2(L_1(k_{[123]})-L_1(k_{[23]}))\psi_\circ(k_{[123]},k_{[23]})
	\\&+(\psi_\circ(k_{[123]},k_{[23]})-\psi_\circ(k_{[23]},k_{[23]}))L_1(k_{[23]})\widehat{v^\eps}(k_{2})^2.
\end{align*}
We subsequently apply analogous calculation to those used for  $\E|\Delta_q(I_t^{11}-\tilde b_2^\eps Z)|^2$ to derive the same bound for $\E|\Delta_q(I_t^{12}-\tilde b_3^\eps Z)|^2$.

Combining the above probabilistic bounds and  using Gaussian hypercontractivity and Kolmogorov's continuity criterion and the Besov embedding Lemma \ref{lem:emb}, we obtain the result in Proposition \ref{thcomm1}.

\section{Bounds on classical partition functions via variational approach }\label{sec:part}

In this section we employ the variational approach developed in \cite{BG18}  to derive uniform bounds on the partition functions of $\nu^\eps$ and its finite-dimensional approximations. As a byproduct, we can also derive tightness of $\nu^\eps, \eps>0$ in $\bC^{-\frac12-\kappa}, \kappa>0$ (see Theorem \ref{th:partition-revised} below). The key technique relies on the Bou\'{e}-Dupuis variational formula (Lemma \ref{lem:B-D}), which allows us to reformulate the partition functions in terms of a stochastic control problem. This transformation reduces the problem to deriving analytic estimates similar to the $L^2$-energy estimates in Section \ref{sec:en}. However, unlike the approach in Section \ref{sec:uni}, we cannot directly employ Schauder estimates for the high-frequency components. Instead, we exploit carefully chosen commutators and cancelations to obtain uniform bounds.
Furthermore, due to the weak dissipation arising from the convolution in the nonlinearity, we introduce additional decompositions to derive uniform analytic estimates (see Lemma \ref{lem:E0} below). In this section we assume that $v$  satisfies $\textbf{(Hv)}$.

We begin by introducing a cut-off function:  Let $\chi:\mR^+\to[0,1]$ be a smooth, non-increasing function with $\chi(y)=1$ for $0\leq y\leq 1/2$ and $\chi(y)=0$ for $y\geq1$. We also set for $t\geq0$
$$\chi_t(k)\eqdef \chi\Big(\frac{\la k\ra}{ t+1}\Big),\qquad \sigma_t(k)\eqdef \Big(\frac{\d}{\d t}\chi_t^2(k)\Big)^{\frac12}. $$
We then consider the projection operator $P_N=\chi_N(\nabla)=\mathcal{F}^{-1}(\chi_N\mathcal{F})$ for $N\in\mathbb{N}$ and the finite-dimensional approximation of $\nu^\eps$ given by
\begin{equs}\label{eq:def-mu-eps-N}
\text{}	\d \mu_N^\eps(\Psi)\eqdef \frac1{\sZ_{\eps,N}}\exp\Big(-\cD(P_N\Psi)\Big)\d \mu_0(\Psi),\qquad \sZ_{\eps,N} \eqdef\int \exp\Big(-\cD(P_N\Psi)\Big)\d \mu_0(\Psi),
\end{equs}
Here we recall that
$$ \cD[u] = \frac12\int :\!{|u(x)|^2}\!: v^\eps(x-y) :\! {|u(y)|^2}\!: \d x \d y  - {(a^\eps-6 b^\eps-m+1)\int_{\mathbb T^3} :\!{|u(x)|^2}\!: \d x}.
$$ The main aim of this section is to prove the following result.
\bt\label{th:partition-revised}
For $f:\bC^{-\frac12-\kappa}\to \mR$ with at most linear growth, it holds that
\begin{equs}[mom:muN-eps-revised]
\sup_{\eps\in(0,1),N\geq \eps^{-2-\kappa}}\int\exp(f({P_N}\Psi))\d \mu_N^\eps(\Psi)<\infty.
\end{equs}
Furthermore, for any finite dimensional projection $\widetilde P$ and $c_0, c_1\in\mR$,  it holds that
\begin{equs}[mom:mutildeP-revised]
	\Big|\log\int \exp\Big(-\cD(P_N\Psi) +c_0\la \widetilde P\Psi,\Psi\ra + c_1 \int \Wick{|P_N\Psi|^2} \Big)\,\dif \mu_0(\Psi) - \log \sZ_{\eps,N} \Big|\lesssim C_{\widetilde P}+1,
\end{equs}
with $ C_{\widetilde P}$ being a constant depending on $\widetilde P$.
\et

To prove this result, we modify $\mu_N^\eps$ to the following probability measure $\tilde \mu^\eps_N$:
\begin{align*}
	\d\tilde{ \mu}_N^\eps(\Psi)\eqdef \frac1{\widetilde\sZ_{\eps,N}}\exp\Big(-\cV^{\eps,N}(P_N\Psi)\Big)\d \mu_0(\Psi),\qquad \widetilde\sZ_{\eps,N}\eqdef\int \exp\Big(-\cV^{\eps,N}(P_N\Psi)\Big)\d \mu_0(\Psi),
\end{align*}
where we set the renormalized potential energy given by
\begin{equs}
	\cV^{\eps,N}(\Psi)\eqdef \frac12\int (v^\eps*\Wick{|\Psi|^2})\Wick{|\Psi|^2}\d x-(a^{\eps,N}-6b^{\eps,N}-m+1)\int\Wick{|\Psi|^2}\d x-c^{\eps,N},
\end{equs}
with the renormalization constant $a^{\eps,N}, b^{\eps,N}$ defined in
\begin{equs}[def:re-con-8]
	a^{\eps,N}\eqdef&\sum_{k\in\mZ^3}\frac{\chi_N(k)^2\widehat{v^\eps}(k)}{(2\pi)^3\la k\ra^2},
	\\ 6b^{\eps,N}\eqdef&\sum_{k_1,k_2\in\mZ^3}\frac{\widehat{v^\eps}(k_1)^2+\widehat{v^\eps}(k_2)\widehat{v^\eps}(k_1)}{(2\pi)^6\la k_1\ra^2\la k_2\ra^2\la k_1+k_2\ra^2}\chi_N(k_1)^2\chi_N(k_2)^2\chi_N(k_1+k_2)^2,
\end{equs}
and $c^{\eps,N}$ given in \eqref{def:cepsT} below.

The difference between $\tilde \mu_N^\eps$ and $\mu_N^\eps$ lies in the renormalization constants in $\cD[u]$ and $\cV^{\eps,N}$: the constants $a^{\eps}$ and $b^{\eps}$ are replaced by $a^{\eps,N}$ and $b^{\eps,N}$, which are more suitable renormalization constants for $P_N \Psi$.
In the following, we first establish uniform bounds under the measure $\tilde \mu_N^\eps$. These bounds will then be used to prove Theorem \ref{th:partition-revised}.
It is straightforward to observe that $\widetilde\sZ_{\eps,N}$ depends on the choice of $c^{\eps,N}$. In what follows, we will select an appropriate $c^{\eps,N}$, as defined in \eqref{def:cepsT} below, to ensure that $\widetilde\sZ_{\eps,N} \backsimeq 1$. More precisely, we obtain the following result.

\bt\label{th:partition} It holds that
\begin{equs}[bd:par]
	\sup_{N\in\mathbb{N},\eps\in(0,1)}|\log \widetilde\sZ_{\eps,N}|\leq C, \qquad \widetilde\sZ_{\eps,N}\backsimeq 1,
\end{equs}
with the implicit constant independent of $\eps, N$. Furthermore, for $f:\bC^{-\frac12-\kappa}\to \mR$ with at most linear growth, it holds that
\begin{equs}[mom:muN-eps]
	\sup_{N\in\mathbb{N},\eps\in(0,1)}\int\exp(f({P_N}\Psi))\d \tilde\mu_N^\eps(\Psi)<\infty.
\end{equs}
\et

The key tool for proving Theorem \ref{th:partition} is the  Bou\'e--Dupuis variational formula (see Lemma \ref{lem:B-D} below). To this end, let  $B$ be a complex valued $L^2$-cylindrical Wiener process, i.e. $B=B_1+iB_2$ with $B_i, i=1,2$, being independent $L^2$-cylindrical Wiener process. We set $(\cF_t)_{t\geq0}$ being the normal filtration generated by $B$.
We choose $W$ being the Gaussian process given by
\begin{align*}
	W(t)\eqdef \sum_n e_n W^{(n)}(t),\qquad W^{(n)}(t)\eqdef \int_0^t\frac{\sigma_s(n)}{\sqrt 2\la n\ra}\d B^{(n)}(s),
\end{align*}
with $B^{(n)}=\la B,e_n\ra$. The covariance of $W(N)$ is given by $(I-\Delta)^{-1}P_N^2$.
Compared to \cite{Bri, BG18}, in our case,  $W$ is complex valued and is given by
\begin{align*}
	W(t)=\int_0^t J_s\d B(s), \qquad \la J_tf,e_n\ra\eqdef \frac{\sigma_t(n)}{\sqrt 2\la n \ra}\la f,e_n\ra.
\end{align*}
We also define $$I_t(u)\eqdef \int_0^tJ_su_s\d s.$$
Let $\mathbf{H}_a$  denote the space of $(\cF_t)_{t\geq0}$-progressively measurable processes $u:\Omega\times [0,\infty)\times \mT^3\to\mR$, which belongs to $L^2([0,\infty)\times \mT^3)$.

\bl\label{lem:B-D} For  $1<p,q<\infty$ with $\frac1p+\frac1q=1$, let $F:C([0,\infty);C^\infty(\mT^3))\to \mR$ be  a measurable function. Suppose that
\begin{align*}
	\E[|F(W)|^p]<\infty,\qquad \E[e^{-q F(W)}]<\infty.
\end{align*}
It holds that
\begin{equs}[eq:8-B-D]
	-\log\E\Big[e^{-F(W)}\Big]=\inf_{u\in \mathbf{H}_a}\E\Big[F(W+I(u))+\frac12\int_0^\infty\|u_s\|_{L^2}^2\d s\Big].
\end{equs}

\el
Compared to the classical variational principle in \eqref{eq:zr-rel}, here we replace the relative entropy in \eqref{eq:zr-rel} with $\frac12\int_0^\infty\|u_s\|_{L^2}^2\d s$. It is important to note that the law of $W(\infty)$ is given by $\mu_0$. A key step in proving Theorems \ref{th:partition-revised} and \ref{th:partition} is to apply Lemma \ref{lem:B-D} in order to transform the exponential integral under $\mu_0$
in these theorems into the stochastic control problem appearing on the RHS of \eqref{eq:8-B-D}.

Since we are working with the finite-dimensional approximation $\mu^\eps_N$, we introduce the following cut-off:
\begin{align*}
	W^N(t)\eqdef P_{N}W(t), \qquad J_t^N\eqdef P_NJ_t, \qquad I_t^N(u)\eqdef P_NI_t(u)=\int_0^tJ_s^Nu_s\d s.
\end{align*}
We note that
\begin{equs}[eq:8-Wt]
	W^N(t)=W^N(2N+2),\qquad J_t^N=0,\qquad I_t^N=I_{2N+2}^N,\qquad  \text{for  } t\geq 2N+2,\end{equs}
and
\begin{align*}
	I_\infty^N(u)=I_{2N+2}^N(u)=\int_0^{2N+2}J_s^Nu_s\d s.
\end{align*}
We also introduce the following suitable regularization for $t\geq0$
\begin{align*}
	I_t^{N,\flat}(u)=\widetilde\chi_tI_t^{N}(u)=\widetilde\chi_tI_s^N(u),\qquad s\geq t,
\end{align*}
for $\widetilde\chi_t=\widetilde {\chi}(\frac{\la\nabla\ra}{t+1})=\mathcal{F}^{-1}(\widetilde{\chi}(\frac{\la\cdot\ra}{t+1})\mathcal{F})$ with $\widetilde{\chi}$ being smooth function on $\mR^+$ such that
\begin{align*}
	\widetilde{\chi}(y)=1 \text{ for }0\leq y\leq \frac14, \qquad \widetilde{\chi}(y)=0, \text{ for } y\geq \frac13.
\end{align*}

\subsection{Stochastic objects}\label{sec:sto-W}
To prove Theorem \ref{th:partition} we choose $F$ in Lemma \ref{lem:B-D} as $\cV^{\eps,N}$, and we need to calculate $\E[\cV^{\eps,N}(W^N+I^N(u))]$. This calculation yields several stochastic objects arising from the binomial expansion of $\cV^{\eps,N}(W^N+I^N(u))$. In this section, we introduce the stochastic objects that will be utilized later.

Before proceeding we  introduce the following renormalization constants
\begin{equs}[eq:ren-con-bi]
	\begin{array}{lllllll}
		b_1^{\eps,N}&\eqdef&\sum_{k_1,k_2}\tilde b^{\eps,N}(k_1,k_2)\widehat{v^\eps}(k_1+k_2)^2,& b_2^{\eps,N}&\eqdef&\sum_{k_1,k_2}\tilde{b}^{\eps,N}(k_1,k_2)\widehat{v^\eps}(k_1+k_2)\widehat{v^\eps}(k_1),\\
		b_3^{\eps,N}&\eqdef&\sum_{k_1,k_2}\tilde{b}^{\eps,N}(k_1,k_2)\widehat{v^\eps}(k_1)^2,&
		b_4^{\eps,N}&\eqdef&\sum_{k_1,k_2}\tilde{b}^{\eps,N}(k_1,k_2)\widehat{v^\eps}(k_2)\widehat{v^\eps}(k_1),
	\end{array}
\end{equs}
for
\begin{align*}
	\tilde{b}^{\eps,N}(t,k_1,k_2)\eqdef&\,\frac{\chi_N(k_1)^2\chi_N(k_2)^2\chi_N(k_1+k_2)^2}{(2\pi)^6\la k_1\ra^2\la k_2\ra^2\la k_1+k_2\ra^2}\int_0^t \chi_s(k_1)^2\chi_s(k_2)^2\sigma_s(k_1+k_2)^2\,\dif s,\\\tilde{b}^{\eps,N}(k_1,k_2)\eqdef&\, \tilde{b}^{\eps,N}(\infty,k_1,k_2).
\end{align*}
Similarly we can define $b_i^{\eps,N}(t)$ with $\tilde{b}^{\eps,N}(k_1,k_2)$ in \eqref{eq:ren-con-bi} replaced by $\tilde{b}^{\eps,N}(t,k_1,k_2)$.

We first use symmetry to compare the renormalization constant $b^{\eps, N}$ introduced in \eqref{def:re-con-8} with $b_i^{\eps,N}$. Additionally, we compare $b^{\eps,N}$ with $b^\eps$ in \eqref{eq:counterterms} in the following results.

\bl\label{lem:b} It holds that
\begin{equs}[eq:re-b]
	6b^{\eps,N}=b_1^{\eps,N}+2b^{\eps,N}_2+2b^{\eps,N}_3+b^{\eps,N}_4,\end{equs}
and $6b^{\eps,N}\to 6b^{\eps}$, $N\to\infty$.
\el
\begin{proof}
	Recall $F$ defined in \eqref{def-F}.
	Using symmetry we obtain
	\begin{align*}
		&\text{ RHS of \eqref{eq:re-b}}=\,\sum_{k_1,k_2} F(k_1,k_2)\tilde{b}^{\eps,N}(k_1,k_2)
		\\=&\,\frac13\sum_{k_1,k_2}F(k_1,k_2)\frac{\chi_N(k_1)^2\chi_N(k_2)^2\chi_N(k_1+k_2)^2}{(2\pi)^6\la k_1\ra^2\la k_2\ra^2\la k_1+k_2\ra^2}\int_0^\infty \partial_s( \chi_s(k_1)^2\chi_s(k_2)^2\chi_s(k_1+k_2)^2)\,\dif s
		\\=&\,\frac13\sum_{k_1,k_2}F(k_1,k_2)\frac{\chi_N(k_1)^2\chi_N(k_2)^2\chi_N(k_1+k_2)^2}{(2\pi)^6\la k_1\ra^2\la k_2\ra^2\la k_1+k_2\ra^2}
		=6b^{\eps,N}.
	\end{align*}
	Hence, the first result follows. Letting $N\to\infty$ we  obtain $6b^{\eps,N}\to 6b^\eps$.

\end{proof}

Based on Lemma \ref{lem:b} we can also define $b^{\eps,N}(t), t\geq0$ given by
\begin{equs}[def:bepsNt]
	6b^{\eps,N}(t)\eqdef b_1^{\eps,N}(t)+2b^{\eps,N}_2(t)+2b^{\eps,N}_3(t)+b^{\eps,N}_4(t).
\end{equs}
We have $b^{\eps,N}=b^{\eps,N}(\infty)$. By direct calculation we have for $t\geq0$
\begin{equs}[bd8:b]
	|b^{\eps,N}|\lesssim \log N,\qquad |b^{\eps,N}(t)|\lesssim \log (t+1),
\end{equs}
with the proportional constant independent of $\eps$, $N$ and $t$.

In what follows, we introduce the stochastic objects as in Section \ref{sec:ren}, where $W$ plays the same role as $Z$. We also adopt the graph notations introduced in Section \ref{sec:ren}. Additionally, we write $W_t^N=W^N(t)$ and define
\begin{align*}
	\begin{array}{llllllll}
		\cW^{\<2m>}_{N}(t)&\eqdef& |W^N_t|^2-a_t^N, & \cW^{\<2vm>}_{\eps,N}(t)&\eqdef&
		v_\eps* \cW_N^{\<2m>}(t),
		\\
		\cW_{\eps,N}^{\<3vm>}(t)	&\eqdef& \cW^{\<2vm>}_{\eps,N}(t) W^N_t-\cM^N_t W^N_t,
		\\\cW^{\<2>}_{y,N}(t)&\eqdef& \tau_yW^N_tW^N_t,& \cW^{\<2m>}_{y,N}(t)&\eqdef& \tau_y\overline{W^N_t}W^N_t-\E(\tau_y\overline{W^N_t}W^N_t),
	\end{array}
\end{align*}
where
$
a_t^N\eqdef \frac1{(2\pi)^3}\sum_{k\in\mZ^3}\frac{\chi_t^N(k)^2}{\la k\ra^2}$ for $\chi_t^N=\chi_t\chi_N$, and
\begin{align*}
	\la{\cM_t^Nf},e_k\ra\eqdef \frac1{(2\pi)^3}\sum_{k_1\in\mZ^3}\widehat{v^\eps}(k+k_1)\frac{\chi_t^N(k_1)^2}{\la k_1\ra^2}\la f,e_k\ra.
\end{align*}
As mentioned in Remark \ref{C1C2e}, to establish a connection with the quantum many-body problem, we need to transform the usual Wick product defined by $\cM^N$ to a mass renormalization $a^{\eps,N}$.  To achieve this, as in \eqref{eq:mainn}, we introduce
\begin{align*}
	\cR_t^N f=a^{\eps,N}_tf-\cM_t^Nf,\qquad a^{\eps,N}_t\eqdef \sum_{k\in\mZ^3}\frac{\chi_t^N(k)^2\widehat{v^\eps}(k)}{\la k\ra^2(2\pi)^3}.
\end{align*}
Here $\cR^N$ is the finite dimensional cut-off of the operator $\cR$ introduced in \eqref{def:cR}.
Then it is easy to see
$a^{\eps,N}=a^{\eps,N}_\infty$ with $a^{\eps,N}$ given in \eqref{def:re-con-8}.
Using the same argument as Lemma \ref{comnew} we also have the following result.
\bl\label{comnew1} Let $f\in \bB_p^\alpha,\alpha\in\mathbb{R},p\in[1,\infty]$. Then for $\eta\in (0,1)$
\begin{align*}
	\|\cR_t^N f\|_{\bB_p^{\alpha-1-\eta}}\lesssim \eps^{\eta}\|f\|_{\bB_p^{\alpha}}.
\end{align*}
Here the implicit constant is independent of $\eps,t,N$.
\el
We also define
$\cY_t^N\eqdef \cW_{\eps,N}^{\<3vm>}(t)-\cR_\infty^N W^N_t,$
which is an $(\cF_t)_{t\geq0}$-martingale. We set
$$\mY_t^N\eqdef 2I_t^N(J^N\cY^N).$$
We can view $\mY_t^N$ as a counterpart of $Y_t$  introduced in \eqref{def:Y} for stochastic quantization.
Additionally, we introduce the following stochastic terms using the renormalization constants defined in \eqref{eq:ren-con-bi}:
\begin{equs}[sto-sec:8-W0]	\cW^{\<3v02vm-W1>}_{\eps,N}(\infty)\eqdef&\,{\mY_\infty^N}\circ \cW^{\<2vm>}_{\eps,N}(\infty)-(b_1^{\eps,N}+b_2^{\eps,N})W^N_\infty,\qquad \cW^{\<3v02vm-W>}_{\eps,N}\eqdef\overline{\cW}^{\<3v02vm-W1>}_{\eps,N}.
	\\\cW_{\eps,N}^{\<3v0m1cci1>}(\infty)\eqdef&\int v^\eps(y)\cW^{\<2>}_{y,N}(\infty)\circ \tau_y\overline{\mY_\infty ^N}\,\dif y-(b_3^{\eps,N}+b_4^{\eps,N})W^N_\infty,\\
	\cW^{\<3v0m1ci1>}_{\eps,N}(\infty)\eqdef&\int v^\eps(y)\cW^{\<2m>}_{y,N}(\infty)\circ \tau_y{\mY_\infty ^N}\,\dif y-(b_2^{\eps,N}+b_3^{\eps,N})W_\infty^N.
\end{equs}
When deriving suitable moment bounds for the terms in \eqref{sto-sec:8-W0} involving $\cR^N_\infty W^N_\infty$, we encounter constants $C_i^{\eps,N}$, which play the same role as $b_i^{\eps,N}$ and correspond to $c_i^\eps$ in Proposition \ref{prop:Zn1}. As in the proof of Proposition \ref{prop:Zn}, these constants satisfy $C_i^{\eps,N} \simeq 1$, with the proportionality constant independent of both $N$ and $\eps$. Therefore, there is no need to subtract them.

Furthermore, we introduce the following $\eps,N$-dependent stochastic objects:
\begin{equs}[sto:W-n]
	\begin{array}{lllllll}
		\cW_{\eps,N}^{\<3v0m1ci-Wn>}(\infty)&\eqdef&(v^\eps*\mY_\infty^N)\circ W^N_\infty,& \cW_{\eps,N}^{\<3v0m1cci-Wn>}(\infty)&\eqdef&(v^\eps*\overline{\mY_\infty^N})\circ W^N_\infty,\\
		\cW_{\eps,N}^{\<3v0m1ci2-W>}(\infty)&\eqdef&\mY_\infty^N\circ \overline{W^N_\infty},& \cW_{\eps,N}^{\<3v0m1cci2-W>}(\infty)&\eqdef&\overline{\mY_\infty^N}\circ W^N_\infty,
	\end{array}
\end{equs}
and
\begin{equs}[sto-sec:8-W]
	\cW_{\eps,N}^{\<22vm-W>}(t)\eqdef&\,\, 2[(J_t^N)^2 \cW^{\<2vm>}_{\eps,N}(t)]\circ  \cW^{\<2vm>}_{\eps,N}(t)-\dot{b}_1^{\eps,N}(t).
	\end{equs}
Here $\dot{b}_1^{\eps,N}(t)\eqdef\sum_{k_1,k_2}\dot{\tilde{b}}_t^{N}(k_1,k_2)\widehat{v^\eps}(k_1+k_2)^2$ and
 we also introduce
\begin{equs}[def:dotbi]
	\begin{array}{llllll}
			\dot{b}_2^{\eps,N}(t)&\eqdef&\sum_{k_1,k_2}\dot{\tilde{b}}_t^{N}(k_1,k_2)\widehat{v^\eps}(k_1+k_2)\widehat{v^\eps}(k_1),&\dot{b}_3^{\eps,N}(t)&\eqdef&\sum_{k_1,k_2}\dot{\tilde{b}}_t^{N}(k_1,k_2)\widehat{v^\eps}(k_1)^2,\\
		\dot{b}_4^{\eps,N}(t)&\eqdef&\sum_{k_1,k_2}\dot{\tilde{b}}_t^{N}(k_1,k_2)\widehat{v^\eps}(k_2)\widehat{v^\eps}(k_1),&&&
	\end{array}
\end{equs}
with $$\dot{\tilde{b}}_t^{N}(k_1,k_2)\eqdef\frac{\chi_N(k_1)^2\chi_N(k_2)^2\chi_N(k_1+k_2)^2}{(2\pi)^6\la k_1\ra^2\la k_2\ra^2\la k_1+k_2\ra^2} \chi_t(k_1)^2\chi_t(k_2)^2\sigma_t(k_1+k_2)^2.$$
We also set
\begin{equs}[def:dotbepsN]
	6\dot{b}^{\eps,N}\eqdef \dot{b}_1^{\eps,N}+2\dot{b}_2^{\eps,N}+2\dot{b}_3^{\eps,N}+\dot{b}_4^{\eps,N}.
\end{equs}
Additionally, we introduce the following $y$-dependent stochastic objects, which can be regarded as the counterparts  to the stochastic objects \eqref{sto:ny2} in Section \ref{sec:ren}:
\begin{equs}
	\cW_{y,\eps,N}^{\<3v0m1ci-W>}(\infty)&\eqdef&(\tau_y\mY_\infty^N)\circ W^N_\infty,\qquad \cW_{y,\eps,N}^{\<3v0m1cci-W>}(\infty)&\eqdef&(\tau_y\overline{\mY^N_\infty})\circ W^N_\infty,
\end{equs}
and
\begin{equs}[sto-sec:8-Wy]
	\cW_{y,\eps,N}^{\<22vm-Wy>}(t)\eqdef&\,\, 2[(J_t^N)^2 \cW^{\<2vm>}_{\eps,N}(t)]\circ {\overline{\cW}^{\<2>}_{y,N}(t)},
	\\\cW_{y,\eps,N}^{\<22vm-Wy1>}(t)\eqdef&\,\, 2[(J_t^N)^2 \cW^{\<2vm>}_{\eps,N}(t)]\circ \overline{\cW}^{\<2m>}_{y,N}(t)-\dot{b}_2^{\eps,N}(t,y),
	\\\cW_{y,y_1,N}^{\<22ccvm-Wy>}(t)\eqdef&\,\, 2[(J_t^N)^2{\cW^{\<2>}_{y,N}(t)}]\circ \overline{\cW}^{\<2>}_{y_1,N}(t)-\dot{b}_3^{N}(t,y,y_1)-\dot{b}_4^{N}(t,y,y_1),\\
	\cW_{y,y_1, N}^{\<22ccvm-Wy1>}(t)\eqdef&\,\, 2[(J_t^N)^2{\cW^{\<2m>}_{y,N}(t)}]\circ \overline{\cW}^{\<2>}_{y_1,N}(t),\\
	\cW_{y,y_1,N}^{\<22ccvm-Wy2>}(t)\eqdef&\,\, 2[(J_t^N)^2{\cW^{\<2m>}_{y,N}(t)}]\circ \overline{\cW}^{\<2m>}_{y_1,N}(t)-\dot{b}_3^{N}(t,y,y_1).\end{equs}
Here
\begin{align*}
	\dot{b}_2^{\eps,N}(t,y)\eqdef\sum_{k_1,k_2}\dot{\tilde{b}}_t^{N}(k_1,k_2)\widehat{v^\eps}(k_1+k_2)e^{-\imath k_1\cdot y},
\end{align*}
\begin{equs}[def:dotbiN]
	\dot{b}_3^{N}(t,y,y_1)\eqdef&\sum_{k_1,k_2}\dot{\tilde{b}}_t^{N}(k_1,k_2)e^{-\imath k_1\cdot y}e^{-\imath k_1\cdot y_1},
	\qquad\dot{b}_4^{N}(t,y,y_1)\eqdef\sum_{k_1,k_2}\dot{\tilde{b}}_t^{N}(k_1,k_2)e^{-\imath k_1\cdot y}e^{-\imath k_2\cdot y_1}.\end{equs}

\bp\label{prop:W} It holds that for $p\geq1$, $\kappa>0$
\begin{align*}
	\sup_{N}\E\|W^N\|_{C\bC^{-\frac12-\kappa}}^p+\sup_{\eps,N}\E\|\cW^{\<2vm>}_{\eps,N}\|^p_{C\bC^{-1-\frac\kappa2}}+	\sup_{\eps,N}\E\|\mY^N\|_{C\bC^{\frac12-\kappa}}^p\lesssim1,
\end{align*}
and
\begin{align*}
	\sup_{\eps,N}\E\|\cW_{\eps,N}^{\<22vm-W>}\|_{L^1\bC^{-\kappa}}^p+	\sup_{\eps,N}\E \|\cW^{\tau_\eps}_{\eps,N}(\infty)\|_{\bC^{-\frac12-\kappa}}^p+\sup_{\eps,N}\E\|\cW^{\tau_\eps^1}_{\eps,N}(\infty)\|_{\bC^{-\kappa}}^p\lesssim 1,
\end{align*}
for  $\cW^{\tau_\eps}_{\eps,N}\in\{\cW^{\<3v02vm-W>}_{\eps,N},\cW_{\eps,N}^{\<3v0m1cci1>},\cW^{\<3v0m1ci1>}_{\eps,N}\}$ and $\cW^{\tau_\eps^1}_{\eps,N}\in \{	\cW_{\eps,N}^{\<3v0m1ci-Wn>},$
$\cW_{\eps,N}^{\<3v0m1cci-Wn>},
\cW_{\eps,N}^{\<3v0m1ci2-W>},  \cW_{\eps,N}^{\<3v0m1cci2-W>}\}$. Furthermore, for $y$-dependent stochastic objects,  it holds that
\begin{align*}
	\sup_{\eps,N}\E\sup_y\|\cW^{\tau}_{y,N}\|^p_{C\bC^{-1-\frac\kappa2}}+\E \sup_y \|\cW^{\tau_1^\eps}_{y,\eps, N}(\infty)\|_{\bC^{-\kappa}}^p\lesssim1,
\end{align*}
for $\cW^{\tau}_{y,N}\in\{\cW^{\<2>}_{y,N}, \cW^{\<2m>}_{y,N}\}$ and $\cW^{\tau_1^\eps}_{y,\eps,N}\in 	\{\cW_{y,\eps,N}^{\<3v0m1ci-W>}, \cW_{y,\eps, N}^{\<3v0m1cci-W>}\}$,
and
\begin{align*}
	\E \Big(\int_0^\infty\sup_y\|\cW_{y,\eps,N}^\tau\|_{\bC^{-\kappa}}\d t\Big)^p+\E\Big(\int_0^\infty\sup_{y,y_1}\|\cW_{y,y_1,N}^{\tau_1}\|_{\bC^{-\kappa}}\d t\Big)^p\lesssim1,
\end{align*}
for $\cW_{y,\eps,N}^\tau\in\{\cW_{y,\eps,N}^{\<22vm-Wy>}, \cW_{y,\eps,N}^{\<22vm-Wy1>}\}$ and $\cW_{y,y_1,N}^{\tau_1}\in\{\cW_{y,y_1,N}^{\<22ccvm-Wy>}, 	\cW_{y,y_1,N}^{\<22ccvm-Wy1>}, 	\cW_{y,y_1,N}^{\<22ccvm-Wy2>}\}$.
\ep
\begin{proof}
	The proof follows a similar argument by employing chaos expansion and probabilistic calculations, as in Propositions \ref{prop:Z}--\ref{prop:Zny} and Section \ref{sec:sto}, with the modification that $ W^N(t) $ is a martingale.
	It is easy to verify that for any $\delta\in[0,1]$, $k\in\mZ^3$
	\begin{equs}[bd:sigma]
		\sigma_t(k)=\bigg(-2\chi_t(k)\chi_t'(k)\frac{\la k\ra}{(t+1)^2}\bigg)^{1/2}\lesssim \frac{\la k\ra^{\delta/2}}{(t+1)^{(1+\delta)/2}},\end{equs}
	which implies  uniform-in-time estimates. For further details, we refer to the calculations in \cite[Section 6]{BG18} and \cite[Section 2]{Bri}.
\end{proof}

In the following we fix $\kappa>0$ small enough and we use $\|\mW\|$ to denote the smallest number bigger than $1$ and
$$\|\cW^{\<2vm>}_{\eps,N}\|_{C\bC^{-1-\frac\kappa2}},\qquad \sup_y\|\cW^{\tau}_{y,N}\|_{C\bC^{-1-\frac\kappa2}},\qquad \|\mY^N\|_{C\bC^{\frac12-\kappa}},\qquad \|W^N\|_{C\bC^{-\frac12-\kappa}},$$
and
$$\|\cW_{\eps,N}^{\<22vm-W>}\|_{L^1\bC^{-\kappa}},\qquad \|\cW^{\tau_\eps}_{\eps,N}(\infty)\|_{\bC^{-\frac12-\kappa}},\qquad \|\cW^{\tau_\eps^1}_{\eps,N}(\infty)\|_{\bC^{-\kappa}},$$
as well as	 $$\int_0^\infty\sup_y\|\cW_{y,\eps,N}^\tau\|_{\bC^{-\kappa}}\d t,\qquad \int_0^\infty\sup_{y,y_1}\|\cW_{y,y_1,N}^{\tau_1}\|_{\bC^{-\kappa}}\d t,\qquad \sup_y \|\cW^{\tau_1^\eps}_{y,\eps, N}(\infty)\|_{\bC^{-\kappa}},$$ for $\cW^{\tau}_{y,N}$, $\cW^{\tau_{\eps}}_{\eps,N}$,  $\cW^{\tau^1_\eps}_{\eps,N}$ and $\cW_{y,\eps,N}^\tau$, $\cW_{y,y_1,N}^{\tau_1}$, $\cW^{\tau_1^\eps}_{y,\eps, N}$ in Proposition \ref{prop:W}. We also use $K(\|\mW\|)$ to denote the polynomials of $\|\mW\|$, which may change from line to line.

\subsection{Decomposition}\label{sec:dec8}

In this section we decompose the RHS of \eqref{eq:8-B-D}, with $F$ given by $\cV^{\eps,N}-f$ for $f$ in Theorem \ref{th:partition}, in terms of the stochastic objects introduced in Section \ref{sec:sto-W} via the binomial formula. We set
\begin{align*}
	\sV^{\eps,N}(u)\eqdef	\cV^{\eps,N}(W_\infty^N+I_\infty^N(u))+\frac12\int_0^\infty\|u\|_{L^2}^2\,\dif t-f(W_\infty^N+I_\infty^N(u)).
\end{align*}
We also choose
\begin{equs}[def:cepsT]
	c^{\eps,N}=&\,\,\frac12\E \int\cW^{\<2vm>}_{\eps,N}(\infty)\cW^{\<2m>}_{N}(\infty)\d x-2\int_0^\infty\E\|J_t^N\cY_t^N\|_{L^2}^2\,\dif t+	\E\int \cW^{\<2vm>}_{\eps,N}(\infty)|\mY_\infty^N|^2\,\dif x
	\\&+2\E\int v^\eps*\Re(W_\infty^N\overline{\mY_\infty^N})\Re(W_\infty^N\overline{\mY_\infty^N})\,\dif x-(6b^{\eps,N}+m-1)\cdot2\E\Re\int W_\infty^N \overline{\mY_\infty^N}\d x
	\\&-2\E\Re \int v^\eps*|\mY_\infty^N|^2 W_\infty^N\overline{\mY_\infty^N}\,\dif x-a^{\eps,N}\E\|\mY_\infty^N\|_{L^2}^2.
\end{equs}
We then have the following decomposition proposition.

\bp\label{prop:dec} Define $l^N_t(u)\in \mathbf{H}_a$ given by
\begin{equs}[def:ltT1]
	l_t^N(u)\eqdef&\,\,w_t+\sJ_t\eqdef u_t+2J_t^N\cY_t^N+\sJ_t,
	\\
	\sJ_t\eqdef&\,\,2J_t^N\big(	 \cW^{\<2vm>}_{\eps,N}(t)\succ{I_t^{N,\flat}(u)}\big)
	+2\int v^\eps(y)J_t^N\big(\cW^{\<2>}_{y,N}(t)\succ\tau_y\overline{I^{N,\flat}_t(u)}\big)\d y
	\\&+2\int v^\eps(y)J_t^N\big(\cW^{\<2m>}_{y,N}(t)\succ\tau_y{I^{N,\flat}_t(u)}\big)\d y,
\end{equs}
for $t\geq0$. Then for $c^{\eps,N}$ in \eqref{def:cepsT}, it holds that
\begin{equs}[dec:sV]
	\E [\sV^{\eps,N}(u)]=-\E[f(W_\infty^N+I_\infty^N(u))]+\sum_{i=0}^6\E[\cE_i]+\frac12\E\cV^\eps(I_\infty^N(u))+\frac12\int_0^\infty\E\|l^N_t(u)\|_{L^2}^2\,\dif t,
\end{equs}
with $\cV^\eps$ introduced in \eqref{def:veps}   and
\begin{align*}
	\cE_0\eqdef &\,\,2\Re \int v^\eps*|I_\infty^N(w)-\mY_\infty^N|^2 W_\infty^N\overline{I_\infty^N(w)}\,\dif x
	-2\Re \int v^\eps*|I_\infty^N(w)|^2 W_\infty^N\overline{\mY_\infty^N}\,\dif x\\&+4\int v^\eps*\Re(I_\infty^N(w)\overline{\mY_\infty^N}) \Re(W_\infty^N\overline{\mY_\infty^N})\,\dif x,
	\\
	\cE_1\eqdef&-2\Re	\int (\cW^{\<2vm>}_{\eps,N}(\infty)\prec\overline{\mY_\infty^N})I_\infty^N(w)\,\dif x -2\Re	\int \cW^{\<3v02vm-W>}_{\eps,N}(\infty)I_\infty^N(w)\,\dif x
	\\&-2\Re
	\int v^\eps(y)(\cW^{\<2>}_{y,N}(\infty)\prec\tau_y\overline{\mY_\infty^N})\overline{I_\infty^N(w)}\,\dif x\d y
	-2\Re\int \cW_{\eps,N}^{\<3v0m1cci1>}(\infty)\overline{I_\infty^N(w)}\,\dif x\d y
	\\&-2\Re
	\int v^\eps(y)(\cW^{\<2m>}_{y,N}(\infty)\prec\tau_y{\mY_\infty^N})\overline{I_\infty^N(w)}\,\dif x\d y
	-2\Re\int \cW_{\eps,N}^{\<3v0m1ci1>}(\infty)\overline{I_\infty^N(w)}\,\dif x\d y,
	\\
	\cE_2\eqdef&\,\,	\Re D(\overline{I_\infty^N(w)},\cW^{\<2vm>}_{\eps,N} ,I_\infty^N(w))
	+\Re\int v^\eps(y)	D(\overline{I_\infty^N(w)},\cW^{\<2>}_{y,N}(\infty),\tau_y\overline{I_\infty^N(w)})\,\dif y
	\\&+\Re\int v^\eps(y)	D(\overline{I_\infty^N(w)},\cW^{\<2m>}_{y,N}(\infty),\tau_y{I_\infty^N(w)})\,\dif y
	\\&+\Re\int (\cW^{\<2vm>}_{\eps,N}\prec\overline{I_\infty^N(w)})I_\infty^N(w)\,\dif x+\Re\int v^\eps(y)(\cW^{\<2>}_{y,N}(\infty)\prec\tau_y\overline{I_\infty^N(w)})\overline{I_\infty^N(w)}\,\dif x \d y
	\\&+\Re\int v^\eps(y)(\cW^{\<2m>}_{y,N}(\infty)\prec\tau_y{I_\infty^N(w)})\overline{I_\infty^N(w)}\,\dif x \d y,
	\\
	\cE_3\eqdef&-\la \cR^N_\infty I_\infty^N(w),I_\infty^N(w)\ra+2\Re \la \cR^N_\infty I_\infty^N(w),\mY^N_\infty\ra
	+	(m-1)(	2\Re\la W^{N}_\infty,I_\infty^N(w)\ra+\|I_\infty^N(u)\|^2_{L^2}),\\
	\cE_4\eqdef&\,\,2\Re	\int \big(\cW^{\<2vm>}_{\eps,N}(2N+2)\succ(\overline{I_{2N+2}^N(u)}-\overline{I_{2N+2}^{N,\flat}(u)})\big)I_{2N+2}^N(w)\,\dif x
	\\&+2\Re\int v^\eps(y)\big(\cW^{\<2>}_{y,N}({2N+2})\succ\tau_y(\overline{I^N_{2N+2}(u)}-\overline{I^{N,\flat}_{2N+2}(u)})\big)\overline{I_{2N+2}^N(w)}\,\dif x\d y
	\\&+2\Re\int v^\eps(y)\big(\cW^{\<2m>}_{y,N}({2N+2})\succ\tau_y({I^N_{2N+2}(u)}-{I^{N,\flat}_{2N+2}(u)})\big)\overline{I_{2N+2}^N(w)}\,\dif x\d y
	\\&+2\Re\int_0^{2N+2}\int \big(\cW^{\<2vm>}_{\eps,N}(t)\succ\overline{\p_t{I}^{N,\flat}_t(u)}\big){I_t^N(w)}\,\dif x\d y\d t
	\\&	+2\Re\int_0^{2N+2}\int v^\eps(y)\big(\cW^{\<2>}_{y,N}(t)\succ\tau_y\overline{\p_t{I}^{N,\flat}_t(u)}\big)\overline{I_t^N(w)}\,\dif x\d y\d t
	\\&	+2\Re\int_0^{2N+2}\int v^\eps(y)\big(\cW^{\<2m>}_{y,N}(t)\succ\tau_y{\p_t{I}^{N,\flat}_t(u)}\big)\overline{I_t^N(w)}\,\dif x\d y\d t,
	\\
	\cE_5\eqdef& \,\,12b^{\eps,N}\Re\la I^{N,\flat}_{2N+2}(u),I^N_{2N+2}(u)-I^{N,\flat}_{2N+2}(u)\ra+6b^{\eps,N}\|I^N_{2N+2}(u)-I^{N,\flat}_{2N+2}(u)\|_{L^2}^2\\&+
	12\Re\int_0^{2N+2}b^{\eps,N}(t)\la I^{N,\flat}_t(u),\p_t{I}^{N,\flat}_t(u)\ra\d t,
	\\\cE_6\eqdef& -\frac12\int_0^{{2N+2}}\la \sJ_t,\sJ_t\ra\,\dif t+	\int_0^{{2N+2}}6\dot{b}^{\eps,N}\|I^{N,\flat}_t(u)\|_{L^2}^2\d t,
\end{align*}
where $b^{\eps,N}, b^{\eps,N}(t)$ and $\dot{b}^{\eps,N}$ are introduced in \eqref{def:re-con-8}, \eqref{def:bepsNt} and \eqref{def:dotbepsN}, respectively.
\ep
\begin{proof} We begin by applying the binomial formula to obtain the following decomposition:
	\begin{equs}[dec:8-cV]
		\sV^{\eps,N}(u)
		=&\,\,2\Re\la \cY_\infty^N,I_\infty^N(u)\ra
		+\frac12\cV^\eps(I_\infty^N(u))+\frac12\int_0^\infty\|u\|_{L^2}^2\,\dif t-f(W_\infty^N+I_\infty^N(u))
		\\&+\frac12\int \cW^{\<2vm>}_{\eps,N}(\infty)\cW^{\<2m>}_{N}(\infty)\,\dif x-(a^{\eps,N}-6b^{\eps,N}-m+1)\int\cW^{\<2m>}_N(\infty)\d x-c^{\eps,N}
		\\&+2\Re \int v^\eps*|I_\infty^N(u)|^2 (W_\infty^N\overline{I_\infty^N(u)})\,\dif x\
		\\&+\la \cW^{\<2vm>}_{\eps,N}(\infty),|I_\infty^N(u)|^2\ra+(6b^{\eps,N}+m-1)(	2\Re\la W^{N}_\infty,I_\infty^N(u)\ra+\|I_\infty^N(u)\|^2_{L^2})
		\\&+2\int v^\eps*\Re(W_\infty^N\overline{I_\infty^N(u)})\Re(W_\infty^N\overline{I_\infty^N(u)})\,\dif x-\la\cM^N_\infty {I_\infty^N(u)},{I_\infty^N(u)}\ra-\la \cR^N_\infty I_\infty^N(u),I_\infty^N(u)\ra.
	\end{equs}
	Here we absorb $2a^{\eps,N}\Re\la  W_\infty^N, I^N_\infty(u)\ra$ into $2\Re\la \cY_\infty^N,I_\infty^N(u)\ra$.
	In the last line we write
	$$ a^{\eps,N}\la {I_\infty^N(u)},{I_\infty^N(u)}\ra=\la\cM^N_\infty {I_\infty^N(u)},{I_\infty^N(u)}\ra+\la \cR^N_\infty I_\infty^N(u),I_\infty^N(u)\ra,$$
	where
	$\la\cM^N_\infty {I_\infty^N(u)},{I_\infty^N(u)}\ra$ corresponds to the Wick renormalization in $v^\eps*\Re(W_\infty^N\overline{I_\infty^N(u)})\Re(W_\infty^N\overline{I_\infty^N(u)})$. Additionally, we incorporate the expectation of the first term in the second line into $c^{\eps, N}$, while the expectation of the second term in the second line is zero.

	\textbf{Step 1.}
	We  further decompose  the terms in the first line  of the RHS of \eqref{dec:8-cV}. By applying It\^o's formula and \eqref{eq:8-Wt} to
	$\Re\la \cY_\infty^N,I_\infty^N(u)\ra$ we obtain
	\begin{align*}
		\Re\la \cY_\infty^N,I_\infty^N(u)\ra=\Re\int_0^{2N+2}\la \cY_t^N, J_t^Nu_t\ra\,\dif t+\Re\int_0^{2N+2}\la  I_t^N(u),\dif \cY_t^N\ra,
	\end{align*}
	with the second term being a martingale. Recall
	$
	u_t= w_t-2J_t^N\cY_t^N,
	$
	and substitute it into $2\Re\la \cY_\infty^N,I_\infty^N(u)\ra+\frac12\int_0^\infty\|u\|_{L^2}^2\,\dif t$, which can then be expressed as
	\begin{equs}[dec:8-1]
		&-2\int_0^{\infty}\|J_t^N\cY_t^N\|_{L^2}^2\,\dif t +\frac12\int_0^\infty\|w\|_{L^2}^2\,\dif t+\text{martingale}.
	\end{equs}
	We  incorporate $-2\int_0^\infty\E\|J_t^N\cY_t^N\|_{L^2}^2\,\dif t$ into $c^{\eps,N}$. We will use $\frac12\int_0^\infty\|w\|_{L^2}^2\,\dif t$ in Step 3.
	
	\textbf{Step 2.} Next, we analyze  the contribution from the last two lines of \eqref{dec:8-cV}. We will prove that they are equal to
	\begin{align*}
		\sum_{i=1}^3\cE_i+\cE+\cQ^{\eps,N}+\text{martingales},
	\end{align*}
	with
	\begin{align*}
		\cE\eqdef&\,\,2\Re	\int (\cW^{\<2vm>}_{\eps,N}(\infty)\succ\overline{I_\infty^N(u)})I_\infty^N(w)\,\dif x+2\Re\int v^\eps(y)(\cW^{\<2>}_{y,N}(\infty)\succ\tau_y\overline{I^N_\infty(u)})\overline{I_\infty^N(w)}\,\dif x\d y
		\\&+2\Re\int v^\eps(y)(\cW^{\<2m>}_{y,N}(\infty)\succ\tau_y{I^N_\infty(u)})\overline{I_\infty^N(w)}\,\dif x\d y+6b^{\eps,N}	\|I_\infty^N(u)\|^2_{L^2},
	\end{align*}
	and $\cQ^{\eps,N}$ being random fields independent of $u$. We put $\E \cQ^{\eps,N}$ into $c^{\eps,N}$.
	
	We first use $I_\infty^N(u)=I_\infty^N(w)-\mY_\infty^N$ to write
	\begin{equs}[dec:bN]
		6b^{\eps,N}\cdot	2\Re\la W^{N}_\infty,I_\infty^N(u)\ra=6b^{\eps,N}\cdot	2\Re\la W^{N}_\infty,I_\infty^N(w)\ra-6b^{\eps,N}\cdot2\Re\la W^{N}_\infty,\mY_\infty^N\ra.
	\end{equs}
	The first term on the RHS of \eqref{dec:bN} serves as a renormalization counter-term below. We also incorporate  the last term in \eqref{dec:bN} into $\cQ^{\eps,N}$.
	
	We then consider $\la \cW^{\<2vm>}_{\eps,N}(\infty),|I_\infty^N(u)|^2\ra$ in the fourth line of \eqref{dec:8-cV} and use	$I_\infty^N(u)=I_\infty^N(w)-\mY_\infty^N$ to replace $I_\infty^N(u)$
	\begin{align*}
		&\la \cW^{\<2vm>}_{\eps,N}(\infty),|I_\infty^N(u)|^2\ra
		=
		\la\cW^{\<2vm>}_{\eps,N}(\infty),|\mY_\infty^N|^2\ra-2\Re	\la\cW^{\<2vm>}_{\eps,N}(\infty),\overline{\mY_\infty^N}I_\infty^N(w)\ra+
		\la \cW^{\<2vm>}_{\eps,N}(\infty),|I_\infty^N(w)|^2\ra.
	\end{align*}
	We also include  the first term on the RHS as part of $\cQ^{\eps,N}$.
	Using the paraproduct decomposition for the last two terms, we obtain
	\begin{equs}[dec:8-W1]
		&-2\Re	\la\cW^{\<2vm>}_{\eps,N}(\infty),\overline{\mY_\infty^N}I_\infty^N(w)\ra+
		\la \cW^{\<2vm>}_{\eps,N}(\infty),|I_\infty^N(w)|^2\ra
		\\=&\,\,2\Re	\int (\cW^{\<2vm>}_{\eps,N}(\infty)\succ\overline{I_\infty^N(u)})I_\infty^N(w)\,\dif x	-2\Re	\int (\cW^{\<2vm>}_{\eps,N}(\infty)\circ\overline{\mY_\infty^N})I_\infty^N(w)\,\dif x
		\\&-2\Re	\int (\cW^{\<2vm>}_{\eps,N}(\infty)\prec\overline{\mY_\infty^N})I_\infty^N(w)\,\dif x+\Re\int (\cW^{\<2vm>}_{\eps,N}(\infty)\prec\overline{I_\infty^N(w)})I_\infty^N(w)\,\dif x
		\\&+\Re\int (\cW^{\<2vm>}_{\eps,N}(\infty)\circ\overline{I_\infty^N(w)})I_\infty^N(w)\,\dif x-\Re	\int (\cW^{\<2vm>}_{\eps,N}(\infty)\succ\overline{I_\infty^N(w)})I_\infty^N(w)\,\dif x.
	\end{equs}
	We put the first term on the RHS of \eqref{dec:8-W1} into $\cE$. By Lemma \ref{lem:com3}, the last line can be addressed using commutator estimates and can be expressed as
	\begin{align*}
		\Re D(\overline{I_\infty^N(w)},\cW^{\<2vm>}_{\eps,N}(\infty) ,I_\infty^N(w)),
	\end{align*}
	which can be included into $\cE_2$.
	We include the first term in the third line of \eqref{dec:8-W1} into $\cE_1$ and the second term into $\cE_2$. For the second term in the second line of  \eqref{dec:8-W1}, we use the renormalization from $(b^{\eps,N}_1+b^{\eps,N}_2)\cdot2\Re\int W_\infty^N \overline{I^N_\infty(w)}$, which comes from the first term on the RHS of \eqref{dec:bN} using \eqref{eq:re-b},  to express it as
	$-2\Re	\int \cW^{\<3v02vm-W>}_{\eps,N}I_\infty^N(w)\,\dif x$, which can be included into $\cE_1$. Here $\cW^{\<3v02vm-W>}_{\eps,N}$ is the stochastic objects defined in \eqref{sto-sec:8-W0}.

	We then consider $2\int v^\eps*\Re(W_\infty^N\overline{I_\infty^N(u)})\Re(W_\infty^N\overline{I_\infty^N(u)})\,\dif x$ from the last line of  \eqref{dec:8-cV} and write it as
	\begin{align*}
		&2\int v^\eps*\Re(W_\infty^N(\overline{I_\infty^N(w)}-\overline{\mY_\infty^N}))\Re(W_\infty^N(\overline{I_\infty^N(w)}-\overline{\mY_\infty^N}))\,\dif x
		\\=&\,\,2\int v^\eps*\Re(W_\infty^N\overline{\mY_\infty^N})\Re(W_\infty^N\overline{\mY_\infty^N})\,\dif x+2\int v^\eps*\Re(W_\infty^N\overline{I_\infty^N(w)})\Re(W_\infty^N\overline{I_\infty^N(w)})\,\dif x
		\\&-4\int v^\eps*\Re(W_\infty^N\overline{\mY_\infty^N})\Re(W_\infty^N\overline{I_\infty^N(w)})\,\dif x.
	\end{align*}
	We also include the first term on the RHS in $\cQ^{\eps,N}$ and write the last two terms as $2\Re\cL_1+2\Re \cL_2$ with
	\begin{align*}
		\cL_1\eqdef&\,\,\frac12\int v^\eps*(W_\infty^N\overline{I_\infty^N(w)})W_\infty^N\overline{I_\infty^N(w)}\,\dif x-\int v^\eps*(W_\infty^N\overline{\mY_\infty^N})W_\infty^N\overline{I_\infty^N(w)}\,\dif x,\\
		\cL_2\eqdef&\,\,\frac12\int v^\eps*(\overline{W_\infty^N}{I_\infty^N(w)}){W_\infty^N}\overline{I_\infty^N(w)}\,\dif x
		-\int v^\eps*(\overline{W_\infty^N}\mY_\infty^N){W_\infty^N}\overline{I_\infty^N(w)}\,\dif x.
	\end{align*}
	We use the paraproduct decomposition to decompose $\cL_1$ as
	\begin{equs}[dec:8-W2]
		&\int v^\eps(y)(\cW^{\<2>}_{y,N}(\infty)\succ\tau_y\overline{I^N_\infty(u)})\overline{I_\infty^N(w)}\,\dif x\d y-
		\int v^\eps(y)(\cW^{\<2>}_{y,N}(\infty)\circ\tau_y\overline{\mY_\infty^N})\overline{I_\infty^N(w)}\,\dif x\d y
		\\
		&+\frac12\int v^\eps(y)(\cW^{\<2>}_{y,N}(\infty)\prec\tau_y\overline{I_\infty^N(w)})\overline{I_\infty^N(w)}\,\dif x \d y-
		\int v^\eps(y)(\cW^{\<2>}_{y,N}(\infty)\prec\tau_y\overline{\mY_\infty^N})\overline{I_\infty^N(w)}\,\dif x\d y
		\\&
		-\frac12\int v^\eps(y)(\cW^{\<2>}_{y,N}(\infty)\succ\tau_y\overline{I_\infty^N(w)})\overline{I_\infty^N(w)}\,\dif x \d y
		+\frac12\int v^\eps(y)(\cW^{\<2>}_{y,N}(\infty)\circ\tau_y\overline{I_\infty^N(w)})\overline{I_\infty^N(w)}\,\dif x \d y.
	\end{equs}
	We put the first term into $\cE$. For the last line of \eqref{dec:8-W2} we use $v^\eps(y)=v^\eps(-y)$ and change of variable to write the first term   as
	\begin{align*}
		-	\frac12\int v^\eps(y)(\cW^{\<2>}_{y,N}(\infty)\succ\overline{I_\infty^N(w)})\tau_y\overline{I_\infty^N(w)}\,\dif x \d y,
	\end{align*}
	which cancels with the last term. Consequently, Lemma \ref{lem:com3} implies that the last line of \eqref{dec:8-W2} yields
	\begin{align*}
		\frac12\int v^\eps(y)	D(\overline{I_\infty^N(w)},\cW^{\<2>}_{y,N}(\infty),\tau_y\overline{I_\infty^N(w)})\,\dif y,
	\end{align*}
	which can be included in $\cE_2$. We also place the first term in the second line of \eqref{dec:8-W2} into $\cE_2$, and the second term in the second line of \eqref{dec:8-W2} into $\cE_1$. For the second term in \eqref{dec:8-W2}, we apply the renormalization from $(b^{\eps,N}_3 + b^{\eps,N}_4)\int W^N_\infty \overline{I^N_\infty(w)}$, which also stems from the first term on the RHS of \eqref{dec:bN}, using \eqref{eq:re-b}. This allows us to express it as $ \int \cW_{\eps,N}^{\<3v0m1cci1>}(\infty)\overline{I_\infty^N(w)}\,\dif x $, which can then be placed into $\cE_1$.
	We also have similar decomposition as in \eqref{dec:8-W2} for $\overline\cL_1$.

	For $2\Re\cL_2$, we apply the renormalization from $-\la\cM^N_\infty {I_\infty^N(u)}, {I_\infty^N(u)}\ra$, which appears in the last line of \eqref{dec:8-cV}, and express $2\Re\cL_2 - \la\cM^N_\infty {I_\infty^N(u)}, {I_\infty^N(u)}\ra$ as
	\begin{align*}
		&	\Re\int v^\eps(y)\cW^{\<2m>}_{y,N}(\infty)\tau_y{I_\infty^N(w)}\overline{I_\infty^N(w)}\,\dif x\d y
		-2\Re\int v^\eps(y)\cW^{\<2m>}_{y,N}(\infty)\tau_y\mY_\infty^N\overline{I_\infty^N(w)}\,\dif x\d y
		\\&- \la\cM^N_\infty \mY^N_\infty,\mY^N_\infty\ra.
	\end{align*}
	The last term can be included in $\cQ^{\eps,N}$. For the first two terms, we have a similar decomposition as in \eqref{dec:8-W2}, with the last line of \eqref{dec:8-W2} modified to
	\begin{align*}
		&-\frac12\int v^\eps(y)(\cW^{\<2m>}_{y,N}(\infty)\succ\tau_y{I_\infty^N(w)})\overline{I_\infty^N(w)}\,\dif x \d y
		+\frac12\int v^\eps(y)(\cW^{\<2m>}_{y,N}(\infty)\circ\tau_y{I_\infty^N(w)})\overline{I_\infty^N(w)}\,\dif x \d y
		\\&-\frac12\int v^\eps(y)(\overline{\cW}^{\<2m>}_{y,N}(\infty)\succ\tau_y\overline{I_\infty^N(w)}){I_\infty^N(w)}\,\dif x \d y
		+\frac12\int v^\eps(y)(\overline{\cW}^{\<2m>}_{y,N}(\infty)\circ\tau_y\overline{I_\infty^N(w)}){I_\infty^N(w)}\,\dif x \d y.
	\end{align*}
	We also use $v^\eps(y) = v^\eps(-y)$ and a change of variables, leading to cancellations between the first and last terms, as well as between the second and third terms. This results in
	\begin{align*}
		\Re\int v^\eps(y)	D(\overline{I_\infty^N(w)},\cW^{\<2m>}_{y,N}(\infty),\tau_y{I_\infty^N(w)})\,\dif y.
	\end{align*}
	This can be included in $\cE_2$. Similarly, the other corresponding terms as in  \eqref{dec:8-W2} are placed in $\cE$, $\cE_1$, and $\cE_2$.

	We consider $-\la \cR^N_\infty I_\infty^N(u),I_\infty^N(u)\ra$ from the last line of \eqref{dec:8-cV} and write it as
	\begin{align*}
		-\la \cR^N_\infty I_\infty^N(w),I_\infty^N(w)\ra+2\Re \la \cR^N_\infty I_\infty^N(w),\mY^N_\infty\ra-\la \cR^N_\infty \mY^N_\infty,\mY^N_\infty\ra.
	\end{align*}
	We place the first two terms into $\cE_3$ and the last term into $\cQ^{\eps,N}$.
	We also consider
	$(m-1)(	2\Re\la W^{N}_\infty,I_\infty^N(u)\ra+\|I_\infty^N(u)\|^2_{L^2})$ from the fourth line of \eqref{dec:8-cV} and write it as
		\begin{align*}
			(m-1)(	-2\Re\la W^{N}_\infty,\mY^N_\infty\ra+2\Re\la W^{N}_\infty,I_\infty^N(w)\ra+\|I_\infty^N(u)\|^2_{L^2}).\end{align*}
			We then put the first term into $\cQ^{\eps,N}$ and the last two terms into $\cE_3$.

	\textbf{Step 3.}
	We further decompose the term in $\cE$, and in this step, we conclude
	\begin{align*}
		\cE=\sum_{i=4}^6\cE_i+\text{martingale}.
	\end{align*}
	We only decompose the second term in $\cE$, as the other terms can be handled similarly. Using \eqref{eq:8-Wt}, we express the second term in $\cE$ as
	\begin{align*}
		&2\Re	\int v^\eps(y)\big(\cW^{\<2>}_{y,N}({2N+2})\succ\tau_y\overline{I^{N,\flat}_{2N+2}(u)}\big)\overline{I_{2N+2}^N(w)}\,\dif x\d y
		\\&+2\Re\int v^\eps(y)\big(\cW^{\<2>}_{y,N}({2N+2})\succ\tau_y(\overline{I^N_{2N+2}(u)}-\overline{I^{N,\flat}_{2N+2}(u)})\big)\overline{I_{2N+2}^N(w)}\,\dif x\d y.
	\end{align*}
	We place the second term into $\cE_4$. For the first term, we apply It\^o's formula to express it as
	\begin{align*}
		&	2\Re\int_0^{2N+2}\int v^\eps(y)(\cW^{\<2>}_{y,N}(t)\succ\tau_y\overline{I^{N,\flat}_t(u)})\overline{\p_t{I}_t^N(w)}\,\dif x\d y\d t
		\\&+2\Re\int_0^{2N+2}\int v^\eps(y)(\cW^{\<2>}_{y,N}(t)\succ\tau_y\overline{\p_t{I}^{N,\flat}_t(u)})\overline{I_t^N(w)}\,\dif x\d y\d t+\text{martingale}.
	\end{align*}
	The second term can also be placed in $\cE_4$. Similarly, we have contributions from the other terms in $\cE$, which, when combined with the above first term, can be written as
	\begin{align*}
		\sH\eqdef&\,\,
		2\Re\int_0^{2N+2}	\int (\cW^{\<2vm>}_{\eps,N}(t)\succ\overline{I_t^{N,\flat}(u)})\p_t{I}_t^N(w)\,\dif x\d t
		\\&+	2\Re \int_0^{2N+2}\int v^\eps(y)(\cW^{\<2>}_{y,N}(t)\succ\tau_y\overline{I^{N,\flat}_t(u)})\overline{\p_t{I}_t^N(w)}\,\dif x\d y\d t
		\\&+2\Re\int_0^{2N+2}\int v^\eps(y)(\cW^{\<2m>}_{y,N}(t)\succ\tau_y{I^{N,\flat}_t(u)})\overline{\p_t{I}_t^N(w)}\,\dif x\d y\d t.
	\end{align*}
	We consider the sum of $\sH$ and $\frac12\int_0^\infty\|w_t\|_{L^2}^2\d t$ from Step 1, and use $\p_t{I}_t^N=J_t^N$, $w_t=l_t^N-\sJ_t$,  \eqref{eq:8-Wt}, \eqref{def:ltT1} and the self-adjointness of $J_t^N$ to obtain
	\begin{align*}
		&\sH+\frac12\int_0^\infty\|w_t\|_{L^2}^2\d t
		=-\frac12\int_0^{2N+2}\la \sJ_t,\sJ_t\ra\,\dif t+\frac12\int_0^\infty\|l_t^N(u)\|_{L^2}^2\d t.
	\end{align*}
	We place  $-\frac12\int_0^{2N+2}\la \sJ_t,\sJ_t\ra\,\dif t$ into $\cE_6$. For the estimates  below for $-\frac12\int_0^{\infty}\la \sJ_t,\sJ_t\ra\,\dif t$, we need renormalization counter-terms from $6b^{\eps,N}\|I^N_\infty(u)\|_{L^2}^2=6b^{\eps,N}\|I^N_{2N+2}(u)\|_{L^2}^2$, which is decomposed into
	\begin{align*}
		&6b^{\eps,N}\|I^{N,\flat}_{2N+2}(u)\|_{L^2}^2+12b^{\eps,N}\Re\la I^{N,\flat}_{2N+2}(u),I^N_{2N+2}(u)-I^{N,\flat}_{2N+2}(u)\ra\\&+6b^{\eps,N}\|I^N_{2N+2}(u)-I^{N,\flat}_{2N+2}(u)\|_{L^2}^2.
	\end{align*}
	We then express the first term as
	\begin{equs}[eq:dotb]
		\int_0^{2N+2}6\dot{b}^{\eps,N}\|I^{N,\flat}_t(u)\|_{L^2}^2\d t+12\Re\int_0^{2N+2}b^{\eps,N}(t)\la I^{N,\flat}_t(u),\p_t{I}^{N,\flat}_t(u)\ra\d t.
	\end{equs}
	The first term of \eqref{eq:dotb} is used for renormalization of $-\frac12\int_0^{2N+2}\la \sJ_t,\sJ_t\ra\,\dif t$, which is placed into $\cE_6$,  while the remaining terms can be collected into $\cE_5$.

	\textbf{Step 4.} We finally consider the first term in the third line of \eqref{dec:8-cV}: we replace $I_\infty^N(u)$ by $I_\infty^N(w)-\mY_\infty^N$ and have
	\begin{equs}[dec-8:third]
		&	2\Re \int v^\eps*|I_\infty^N(u)|^2 (W_\infty^N\overline{I_\infty^N(u)})\,\dif x=2\Re \int v^\eps*|I_\infty^N(w)-\mY_\infty^N|^2 W_\infty^N(\overline{I_\infty^N(w)}-\overline{\mY_\infty^N})\,\dif x
		\\=&-2\Re \int v^\eps*|\mY_\infty^N|^2 W_\infty^N\overline{\mY_\infty^N}\,\dif x+\cE_0.
	\end{equs}
	We place the expectation of the first term in $c^{\eps,N}$.
\end{proof}

\subsection{Proof of Theorems \ref{th:partition-revised} and \ref{th:partition}}
In this section, we present the proof of Theorems \ref{th:partition-revised} and \ref{th:partition} using analytic estimates. From the decomposition of the terms in \eqref{dec:sV}, we obtain the following coercive terms:
\begin{align*}
	\frac12\cV^\eps(I_\infty^N(u))+\frac12\int_0^\infty\|l_t^N(u)\|_{L^2}^2\d t,
\end{align*}
which are used to control the remaining terms in \eqref{dec:sV}.

We recall several useful estimates from \cite[Lemma 2]{BG18}: for $\alpha\in\mR$,  $T\in(0,\infty]$
\begin{equs}[est:I]
	\sup_{t\in[0,T]}\|I^N_t(v)\|_{H^{\alpha+1}}^2\lesssim \int_0^T\|v\|_{H^\alpha}^2\d t.
\end{equs}
It is straightforward to observe that for $T\in(0,\infty]$, $\alpha\in \mR, p,q\in[1,\infty]$
\begin{equs}[est:flat]
	\sup_{t\in[0,T]}\|I_t^{N,\flat}(u)\|_{\bB^\alpha_{p,q}}\lesssim 	\|I_T^{N}(u)\|_{\bB^\alpha_{p,q}},
\end{equs}
with the proportional constants independent of $\eps, N$. Moreover, using \eqref{bd:sigma} we have that for $\alpha\in\mR$, $\delta\in [0,1]$
\begin{equs}[est:jt]
	\|J_t^N u\|_{H^\alpha}\leq \|J_t u\|_{H^\alpha}\lesssim \frac{1}{(t+1)^{(1+\delta)/2}}\|u\|_{H^{\alpha-1+\delta/2}}.
\end{equs}

In the following we write $l_t^N=l_t^N(u)$ for simplicity. Additionally,  we  fix $\kappa>0$ as a small parameter in Proposition \ref{prop:W}.

\bl\label{lem:Iw} It holds that
\begin{align*}
	\E\sup_t	\|I_t^N(w)\|_{H^{1-\kappa}}^2
	\lesssim \E\int_0^\infty\|l_t^N\|_{L^2}^2\d t+\E\cV^\eps(I_\infty^N(u))+1,
\end{align*}
with the proportional constant independent of $\eps, N$.
\el
\begin{proof}
	We apply \eqref{est:I}, \eqref{est:jt}, and the paraproduct estimates from Lemma \ref{lem:para} to obtain
	\begin{equs}[est:Iw-n]
		\E\sup_t	\|I_t^N(w)\|_{H^{1-\kappa}}^2	\lesssim&\,\, \E\int_0^\infty\|l_t^N\|_{L^2}^2\,\d t+\E\int_0^\infty\|\sJ_t\|_{H^{-\kappa}}^2\,\d t
		\\\lesssim &\,\,\E\int_0^\infty\|l_t^N\|_{L^2}^2\,\d t+\E\int_0^\infty( t+1)^{-1-\frac\kappa2} \|\cW(t)\|_{\bC^{-1-\frac\kappa2}}^2\|I_t^{N,\flat}(u)\|_{L^2}^2\,\d t,
	\end{equs}
	where
	\begin{equs}[def:cW]
		\|\cW(t)\|_{\bC^{-1-\frac\kappa2}}\eqdef\|\cW^{\<2vm>}_{\eps,N}(t)\|_{\bC^{-1-\frac\kappa2}}+\sup_y\|\cW^{\<2>}_{y,N}(t)\|_{\bC^{-1-\frac\kappa2}}+\sup_y\|\cW^{\<2m>}_{y,N}(t)\|_{\bC^{-1-\frac\kappa2}} .\end{equs}
	We then use \eqref{est:flat}, \eqref{lowerbound} and Proposition \ref{prop:W} to have
	\begin{align*}
		\E\sup_t	\|I_t^N(w)\|_{H^{1-\kappa}}^2	\lesssim &\,\E\int_0^\infty\|l_t^N\|_{L^2}^2\,\d t+\E\|I_\infty^N(u)\|_{L^2}^4+1
		\\\lesssim&\, \E\int_0^\infty\|l_t^N\|_{L^2}^2\,\d t+\E\cV^\eps(I_\infty^N(u))+1.
	\end{align*}
\end{proof}

Next, we consider $\cE_1+\cE_2$. In the following we also use $\delta\in (0,1)$ to denote a small constant.

\bl\label{lem:E1E2} It holds that for any $\delta>0$
\begin{align*}
	\E |\cE_1|+\E|\cE_2|\lesssim1+\delta \E\int_0^\infty\|l_t^N\|_{L^2}^2\,\d t+\delta\E\cV^\eps(I_\infty^N(u)),
\end{align*}
with the proportional constant independent of $\eps, N$.
\el
\begin{proof} We start with $\cE_1$, focusing only on the terms involving $\cW^{\<2vm>}_{\eps,N}$, as the other terms can be handled similarly. Using Proposition \ref{prop:W}, Lemma \ref{lem:multi} and the paraproduct estimates Lemma \ref{lem:para}, we obtain
	\begin{align*}
		&\E\Big|\Re	\int (\cW^{\<2vm>}_{\eps,N}(\infty)\prec\overline{\mY_\infty^N}+\cW^{\<3v02vm-W>}_{\eps,N}(\infty))I_\infty^N(w)\,\dif x\Big|
		\\
		\lesssim&
		\,\E\Big(\|\cW^{\<2vm>}_{\eps,N}(\infty)\|_{\bC^{-1-\kappa}}\|\mY_\infty^N\|_{\bC^{\frac12-\kappa}}+\|\cW^{\<3v02vm-W>}_{\eps,N}(\infty)\|_{\bC^{-\frac12-\kappa}}\Big) \|I_\infty^N(w)\|_{H^{\frac12+3\kappa}}
		\\\lesssim&\,1+\delta \E\|I_\infty^N(w)\|_{H^{1-\kappa}}^2,
	\end{align*}
	which implies the desired bound by using Lemma \ref{lem:Iw}.
	
	Next, we consider the terms in $\cE_2$. For the terms that depend quadratically on $I_\infty^N(w)$, we focus only on those involving $\cW^{\<2vm>}_{\eps,N}(\infty)$, as the estimates for the remaining terms can be derived in a similar manner. We apply the paraproduct estimates from Lemma \ref{lem:para}, the interpolation results from Lemma \ref{lem:interpolation}, \eqref{lowerbound}, and Proposition \ref{prop:W} to obtain
	\begin{equs}[est:8-WI2]
		&\E\Big|\Re\int (\cW^{\<2vm>}_{\eps,N}(\infty)\prec\overline{I_\infty^N(w)})I_\infty^N(w)\,\dif x\Big|	\lesssim
		\E\|\cW^{\<2vm>}_{\eps,N}(\infty)\|_{\bC^{-1-\kappa}}\|I_\infty^N(w)\|_{H^{\frac12+\kappa}}^2
		\\
		\lesssim&\,\E\|\cW^{\<2vm>}_{\eps,N}(\infty)\|_{\bC^{-1-\kappa}} \|I_\infty^N(w)\|_{H^{1-\kappa}}^{2\gamma}\|I_\infty^N(w)\|_{L^2}^{2(1-\gamma)}
		\\\lesssim&\,1+\delta \E\|I_\infty^N(w)\|_{H^{1-\kappa}}^2+\delta \E \|I^N_\infty(u)\|_{L^2}^4\lesssim 1+\delta \E\|I_\infty^N(w)\|_{H^{1-\kappa}}^2+\delta\E\cV^\eps(I_\infty^N(u)).
	\end{equs}
	Here $\gamma=\frac{1/2+\kappa}{1-\kappa}$ and we used Young's inequality and $I^N_\infty(w)=I^N_\infty(u)-\mY^N_\infty$ in the third step.
	For the terms involving the opeartor $D$,  we  use Lemma \ref{lem:com3} to have
	\begin{align*}
		&\E|	D(\overline{I_\infty^N(w)},\cW^{\<2vm>}_{\eps,N}(\infty),I_\infty^N(w))|+\E\Big|\int v^\eps(y)	D(\overline{I_\infty^N(w)},\cW^{\<2>}_{y,N}(\infty),\tau_y\overline{I_\infty^N(w)})\,\dif y\Big|
		\\\lesssim &\,\E\Big(\sup_y\|\cW^{\<2>}_{y,N}(\infty)\|_{\bC^{-1-\kappa}}+\|\cW^{\<2vm>}_{\eps,N}(\infty)\|_{\bC^{-1-\kappa}}\Big)\|I_\infty^N(w)\|_{H^{\frac12+\kappa}}^2,
	\end{align*}
	which can be bounded as in \eqref{est:8-WI2}. The other term in $\cE_2$ involving the operator $D$  can be estimated exactly the same way.

\end{proof}

In the following we consider the contribution from $\cE_0$. Before proceeding we first prove the following lemma for $\cV^\eps$ given in \eqref{def:veps}.

\bl\label{lem:vlvu} It holds that
\begin{align*}
	\cV^\eps(I^N_\infty(l^N))+\cV^\eps(I^N_\infty(w))\lesssim \cV^\eps(I_\infty^N(u))(1+K(\|\mW\|))+K(\|\mW\|),
\end{align*}
with the proportional constant independent of $\eps, N$.
\el
\begin{proof} We use the definition of $l^N$ from \eqref{def:ltT1} to have
	\begin{align*}
		\cV^\eps(I^N_\infty(l^N))=&\,\cV^\eps(I^N_\infty(u)+\mY^N_\infty+I_\infty^N(\sJ))
		\\\lesssim&\,\cV^\eps(I^N_\infty(u))+\int v^\eps*|I^N_\infty(u)|^2(|\mY^N_\infty|^2+|I_\infty^N(\sJ)|^2)\d x
		\\&+\int v^\eps*(|\mY^N_\infty|^2+|I_\infty^N(\sJ)|^2)(|\mY^N_\infty|^2+|I_\infty^N(\sJ)|^2)\d x.
	\end{align*}
	We then use H\"older's inequality and \eqref{lowerbound1}
	to have
	\begin{align*}
		\cV^\eps(I^N_\infty(l^N))\lesssim&\,\cV^\eps(I^N_\infty(u))+\|v^\eps*|I^N_\infty(u)|^2\|_{L^2}(\|I_\infty^N(\sJ)\|_{L^4}^2+\|\mY_\infty^N\|_{\bC^{\frac12-\kappa}}^2)
		\\&+\|\mY_\infty^N\|_{\bC^{\frac12-\kappa}}^4+\|I_\infty^N(\sJ)\|_{L^4}^4
		\\\lesssim&\,\cV^\eps(I^N_\infty(u))+\|I_\infty^N(\sJ)\|_{L^4}^4+K(\|\mW\|).
	\end{align*}
	It remains to estimate $\|I_\infty^N(\sJ)\|_{L^4}^4$. We use Sobolev embedding $H^{\frac34}\subset L^4$, along with \eqref{est:I}--\eqref{est:jt}, \eqref{lowerbound}, the same calculation as in \eqref{est:Iw-n}, and the paraproduct estimates Lemma \ref{lem:para}, to obtain
	\begin{equs}[bd8:IsJ]
		&\|I_\infty^N(\sJ)\|_{L^4}^4\lesssim \|I^N_\infty(\sJ)\|_{H^{1-\kappa}}^4
		\\\lesssim&\Big(\int_0^\infty (t+1)^{-1-\frac\kappa2}\|\cW(t)\|_{\bC^{-1-\frac\kappa2}}^2\d t\Big)^2\|I_\infty^N(u)\|_{L^2}^4
		\lesssim\cV^\eps(I_\infty^N(u))K(\|\mW\|),
	\end{equs}
	where $\|\cW(t)\|_{\bC^{-1-\frac\kappa2}}$ is defined in \eqref{def:cW}.
	Consequently, we obtain the desired bound for $\cV^\eps(I^N_\infty(l^N))$.
	Similarly  we achieve the same bounds for $\cV^\eps(I^N_\infty(w))$.
\end{proof}

\bl\label{lem:E0} It holds that for any $\delta>0$
\begin{align*}
	\E |\cE_0|\lesssim1+\delta \E\int_0^\infty\|l_t^N\|_{L^2}^2\,\d t+\delta\E\cV^\eps(I_\infty^N(u)),
\end{align*}
with the proportional constant independent of $\eps, N$.
\el
\begin{proof}
	Since  the second and the third terms in $\cE_0$ depend linearly and quadratically on $I_\infty^N(w)$, which can be estimated similarly as in the proof of Lemma \ref{lem:E1E2} and using the stochastic objects from \eqref{sto:W-n}.  We omit the details of them. Here we  note that for the third term in $\cE_0$, the stochastic objects  can be handled similarly to \eqref{bd:sto3} via the paraproduct estimates.  	We now focus on the first term in $\cE_0$:
	\begin{equs}[est-8:third]
		&	2\Re \int v^\eps*|I_\infty^N(w)-\mY_\infty^N|^2 W_\infty^N\overline{I_\infty^N(w)}\,\dif x
		\\&=2 \int v^\eps*\bigg(|I_\infty^N(w)|^2-2\Re (I_\infty^N(w)\overline{\mY_\infty^N})+|\mY_\infty^N|^2\bigg) \Re(W_\infty^N\overline{I_\infty^N(w)})\,\dif x.
	\end{equs}
	The most difficult term is the following two terms:
	\begin{equs}[dec:cub-1]
		\int v^\eps*|I_\infty^N(w)|^2 W_\infty^N\overline{I_\infty^N(w)}\,\dif x,\quad \int v^\eps*|I_\infty^N(w)|^2 \overline{W_\infty^N}{I_\infty^N(w)}\,\dif x.
	\end{equs} Our proof primarily focuses on estimate the first  term and the second term follows similarly. The remaining terms depend linearly and quadratically on $I^N_\infty(w)$ and can be estimated in a manner similar to the proof of Lemma \ref{lem:E1E2}. We will omit the details for brevity. We also mention that the stochastic objects in the third term of \eqref{est-8:third} can be handled similarly as in \eqref{bd:sto1} via the paraproduct estimates.
	
	In contrast to the energy estimates in Lemma \ref{cubic}, we do not have the $H^1$-norm estimate of $I_\infty^N(w)$, and further decomposition of $I_\infty^N(w)$ is required. We recall
	\begin{align*}
		I_\infty^N(w)=I_\infty^N(l^N)-I_\infty^N(\sJ),
	\end{align*}
	which substituted into the first term in \eqref{dec:cub-1} implies the following decomposition:
	\begin{equs}[dec:8-cubic]
		&\int v^\eps*|I_\infty^N(w)|^2 W_\infty^N\overline{I_\infty^N(w)}\,\dif x=\sum_{i=1}^5R_i.
	\end{equs}
	with
	\begin{equs}
		R_1\eqdef&\,	\la v^\eps*|I_\infty^N(l^N)|^2{I_\infty^N(l^N)}, W_\infty^N\ra,\qquad R_2\eqdef-\la v^\eps*|I_\infty^N(l^N)|^2{I_\infty^N(\sJ)}, W_\infty^N\ra,
		\\
		R_3\eqdef&\,\la v^\eps*|I_\infty^N(\sJ)|^2{I_\infty^N(w)}, W_\infty^N\ra,\qquad
		R_4\eqdef-2	\la v^\eps*\Re(I_\infty^N(l^N) \overline{I_\infty^N(\sJ)}){I_\infty^N(l^N)}, W_\infty^N\ra,\\
		R_5\eqdef&\, 2\la v^\eps*\Re(I_\infty^N(l^N) \overline{I_\infty^N(\sJ)}){I_\infty^N(\sJ)}, W_\infty^N\ra.
	\end{equs}
	
	We consider each term on the RHS of \eqref{dec:8-cubic}:
	
	\textbf{I.} For $R_1$, we use  \eqref{est8:cubic} to have
	\begin{align*}
		|R_1|\lesssim& \Big[\cV^\eps(I_\infty^N(l))^{\frac34}+\cV^\eps(I_\infty^N(l))^{\frac34(\frac12-\kappa)+\frac12(\frac12+\kappa)}\|I_\infty^N(l)\|_{H^1}^{\frac12+\kappa}\Big]\|W_\infty^N\|_{\bC^{-\frac12-\kappa}}
		\\\lesssim&\,\delta\cV^\eps(I_\infty^N(u))+K(\|\mW\|)+\delta \int_0^\infty\|l_t^N\|_{L^2}^2\d t,
	\end{align*}
	where we used \eqref{est:I}, Lemma \ref{lem:vlvu} and Young's inequality in the last step.
	
	\textbf{II.} For the second term $R_2$  we use the localizer defined in \eqref{def:loc} to  decompose $W_\infty^N=\Delta_{>L} W_\infty^N+\Delta_{\leq L}W_\infty^N$. We apply Lemma \ref{lem:para}, interpolation Lemma \ref{lem:interpolation}, Besov embedding Lemma \ref{lem:emb} to have
	\begin{equs}
		\|v^\eps*(f^2)g\|_{\bB^{1-3\kappa}_{1,1}}\lesssim \|f^2\|_{\bB^{1-\kappa}_{\frac43,2}}\|g\|_{H^{1-2\kappa}}\lesssim \|f\|_{H^1}\|f\|_{H^{\frac34}}\|g\|_{H^{1-2\kappa}}\lesssim \|f\|_{H^1}^{\frac74}\|f\|_{L^2}^{\frac14}\|g\|_{H^{1-2\kappa}},
	\end{equs}
	which combined with Lemma \ref{lem:multi} implies that
	\begin{equs}[est:8-I2]
		&|	\la v^\eps*|I_\infty^N(l^N)|^2{I_\infty^N(\sJ)}, \Delta_{>L}W_\infty^N\ra|
		\\\lesssim&\,
		\|I^N_\infty(l^N)\|_{H^1}^{\frac74}\|I^N_\infty(l^N)\|_{L^2}^{\frac14}
		\|I^N_\infty(\sJ)\|_{H^{1-2\kappa}} \|\Delta_{>L}W_\infty^N\|_{\bC^{-1+3\kappa}}
		\\\lesssim&\,\|I^N_\infty(l^N)\|_{H^1}^{\frac74}\|I^N_\infty(l^N)\|_{L^2}^{\frac14}2^{(-\frac12+4\kappa)L}\|I^N_\infty(u)\|_{L^2}K(\|\mW\|)
		\\\lesssim&\,\|I^N_\infty(l^N)\|_{H^1}^{\frac74}\cV^\eps(I^N_\infty(u))^{\frac5{16}}2^{(-\frac12+4\kappa)L}K(\|\mW\|).
	\end{equs}
	Here in the second step we used \eqref{bd8:IsJ} and \eqref{eq:loc} and in the last step we used \eqref{lowerbound} and Lemma \ref{lem:vlvu}.
	For the part evolving $\Delta_{\leq L}W^N_\infty$ we also use \eqref{eq:loc}, \eqref{lowerbound1}, \eqref{bd8:IsJ} and Lemma  \ref{lem:vlvu} to have
	\begin{equs}[est:8-I2-1]
		&|	\la v^\eps*|I_\infty^N(l^N)|^2{I_\infty^N(\sJ)}, \Delta_{\leq L}W_\infty^N\ra|\lesssim\|v^\eps*|I_\infty^N(l^N)|^2\|_{L^2}\|I^N_\infty(\sJ)\|_{L^2} \|\Delta_{\leq L}W_\infty^N\|_{L^\infty}
		\\\lesssim&\,\|v^\eps*|I_\infty^N(l^N)|^2\|_{L^2}2^{(\frac12+2\kappa)L}\|I^N_\infty(u)\|_{L^2}K(\|\mW\|)
		\lesssim\cV^\eps(I^N_\infty(u))^{\frac34}2^{(\frac12+2\kappa)L}K(\|\mW\|).
	\end{equs}
	Combining \eqref{est:8-I2}, \eqref{est:8-I2-1}, we choose
	\begin{align*}
		2^{(-\frac12+4\kappa)L}=\|I^N_\infty(l^N)\|_{H^1}^{-\frac78}\cV^\eps(I^N_\infty(u))^{\frac7{32}},
	\end{align*}
	to have
	\begin{equs}[est:8-8.9]
		|	R_2|\lesssim&\,\Big(\|I^N_\infty(l^N)\|_{H^1}^{\frac78}\cV^\eps(I^N_\infty(u))^{\frac{17}{32}}+\cV^\eps(I^N_\infty(u))^{\frac34-\frac{7}{32}a}\|I_\infty^N(l^N)\|_{H^1}^{\frac78a}\Big)K(\|\mW\|)
		\\\lesssim&\,\delta\|I_\infty^N(l^N)\|_{H^1}^2+\delta \cV^\eps(I^N_\infty(u))+K(\|\mW\|)
		\\\lesssim&\,\delta\int_0^\infty\|l^N\|_{L^2}^2\,\dif t+\delta \cV^\eps(I^N_\infty(u))+K(\|\mW\|),
	\end{equs}
	for $a=\frac{\frac12+2\kappa}{\frac12-4\kappa}$. Here in the last step we used \eqref{est:I}.
	
	Similarly for $R_4$ we have
	\begin{align*}
		R_4=&-2\la v^\eps*\Re(I_\infty^N(l^N) \overline{I_\infty^N(\sJ)}){I_\infty^N(l^N)}, \Delta_{>L}W_\infty^N\ra-2\la v^\eps*\Re(I_\infty^N(l^N) \overline{I_\infty^N(\sJ)}){I_\infty^N(l^N)}, \Delta_{\leq L}W_\infty^N\ra
		\\:=&\,R_{4,>}+R_{4,\leq}.
	\end{align*}
	For $R_{4,>}$ we use Lemma \ref{lem:para} to have
	\begin{equs}[bd:8-fgf]
		&\|v^\eps*(fg)f\|_{\bB^{1-3\kappa}_{1,1}}\lesssim \|fg\|_{\bB^{1-2\kappa}_{\frac43,2}}\|f\|_{L^4}+\|fg\|_{L^2}\|f\|_{H^1}
		\\
		\lesssim&\, \|f\|_{H^1}\|f\|_{H^{\frac34}}\|g\|_{H^{1-2\kappa}}\lesssim \|f\|_{H^1}^{\frac74}\|f\|_{L^2}^{\frac14}\|g\|_{H^{1-2\kappa}},
	\end{equs}
	which implies exactly the same bounds as \eqref{est:8-I2} for $R_{4,>}$. For $R_{4,\leq}$ we use  H\"older's inequality to have
	\begin{align*}
		\|v^\eps*(fg)f\|_{L^1}\leq \|v^\eps*(fg)\|_{L^2}\|f\|_{L^2}\leq \|v^\eps*|f|^2\|_{L^2}^{\frac12}\|v^\eps*|g|^2\|_{L^2}^{\frac12}\|f\|_{L^2}\lesssim \|v^\eps*|f|^2\|_{L^2}^{\frac12}\|g\|_{L^4}\|f\|_{L^2},
	\end{align*}
	which combined with Lemma \ref{lem:multi} and \eqref{eq:loc} implies that
	\begin{align*}
		|	R_{4,\leq}|
		\lesssim&\,\|v^\eps*\Re(I_\infty^N(l^N) \overline{I_\infty^N(\sJ)}){I_\infty^N(l^N)}\|_{L^1} \|\Delta_{\leq L}W_\infty^N\|_{L^\infty}
		\\\lesssim&\,\|v^\eps*|I_\infty^N(l^N)|^2\|_{L^2}^{\frac12}\|I_\infty^N(\sJ)\|_{L^4}2^{(\frac12+2\kappa)L}\|I^N_\infty(l^N)\|_{L^2}\|\mW\|.
	\end{align*}
	We then use Lemma \ref{lem:vlvu}, \eqref{lowerbound1} and \eqref{bd8:IsJ} to have it bounded by
	\begin{align*}
		\cV^\eps(I^N_\infty(u))^{\frac34}2^{(\frac12+2\kappa)L}K(\|\mW\|).
	\end{align*}
	Thus we obtain the same estimates for $|R_4|$ as in \eqref{est:8-8.9}.
	
	\textbf{III.}
	It remains to consider
	$R_3$ and
	$R_5$.
	By Lemma \ref{lem:para} we obtain
	\begin{align*}
		\|v^\eps*(fg)h\|_{\bB^{\frac12+\kappa}_{1,1}}\lesssim\|fg\|_{L^2}\|h\|_{H^{\frac34}}+\|fg\|_{\bB^{\frac12+\kappa}_{\frac43,1}}\|h\|_{L^4}\lesssim \|f\|_{H^{\frac34}}\|g\|_{H^{\frac34}}\|h\|_{H^{\frac34}},
	\end{align*}
	which combined with \eqref{bd8:IsJ} implies that
	\begin{align*}
		&|R_3|+|R_5|\lesssim\|I_\infty^N(\sJ)\|_{H^{1-2\kappa}}^2\big(\|I_\infty^N(w)\|_{H^{\frac34}}+\|I_\infty^N(l^N)\|_{H^{\frac34}}\big)\|W_\infty^N\|_{\bC^{-\frac12-\kappa}}
		\\\lesssim&\,\|I^N_\infty(u)\|_{L^2}^2\big(\|I_\infty^N(w)\|_{H^{\frac34}}+\|I^N_\infty(u)\|_{L^2}\big)K(\|\mW\|)
		\\\lesssim&\,\|I^N_\infty(u)\|_{L^2}^2\big(\|I_\infty^N(w)\|_{H^{1-\kappa}}^{\theta}\|I_\infty^N(w)\|_{L^2}^{1-\theta}+\|I^N_\infty(u)\|_{L^2}\big)K(\|\mW\|)
		\\\lesssim&\, \delta\cV^\eps(I_\infty^N(u))+\delta \|I^N_\infty(w)\|_{H^{1-\kappa}}^2+K(\|\mW\|).
	\end{align*}
	Here $\theta=\frac3{4(1-\kappa)}$ and we used Lemma \ref{lem:vlvu}, \eqref{lowerbound} and Young's inequality in the last step.
	Using Lemma \ref{lem:Iw} and Proposition \ref{prop:W} the result follows.
\end{proof}

	\bl It holds that
	\begin{align*}
		\E|\cE_3|\lesssim 1+\delta \E\int_0^\infty\|l_t^N\|_{L^2}^2\,\d t+\delta\E\cV^\eps(I_\infty^N(u)),
	\end{align*}
	with the proportional constant independent of $\eps, N$.
	\el
	\begin{proof} We use Lemma \ref{comnew1} and Lemma \ref{lem:multi} to have
		\begin{align*}
			|\la \cR^N_\infty I_\infty^N(w),I_\infty^N(w)\ra|+|\la \cR^N_\infty I_\infty^N(w),\mY^N_\infty\ra|
			\lesssim&\, \|I^N_\infty(w)\|_{H^{\frac12+\kappa}}^2+\|I^N_\infty(w)\|_{H^{\frac12+2\kappa}}\|\mY_\infty^N\|_{\bC^{\frac12-\kappa}}
			\\\lesssim&\, \delta \|I^N_\infty(w)\|_{H^{1-\kappa}}^2+\delta \cV^\eps(I^N_\infty(u))+K(\|\mW\|) ,
		\end{align*}
		where we used interpolation Lemma \ref{lem:interpolation}, Young's inequality, \eqref{lowerbound} and Lemma \ref{lem:vlvu} in the last step.
		We also have
		\begin{align*}
				|2\Re\la W^{N}_\infty,I_\infty^N(w)\ra|+\|I_\infty^N(u)\|^2_{L^2}\lesssim K(\|\mW\|)+\delta \|I^N_\infty(w)\|_{H^{1-\kappa}}^2+\delta \cV^\eps(I^N_\infty(u)).
		\end{align*}
		The result then follows by Lemma \ref{lem:Iw} and Proposition \ref{prop:W}.
	\end{proof}
	
	\bl It holds that
	\begin{align*}
		\E|\cE_4|\lesssim 1+\delta \E\int_0^\infty\|l_t^N\|_{L^2}^2\,\d t+\delta\E\cV^\eps(I_\infty^N(u)),
	\end{align*}
	with the proportional constant independent of $\eps, N$.
	\el
	\begin{proof}
		We concentrate on the terms evolving $\cW^{\<2>}_{y,N}(\infty)$, while the remaining terms can be estimated in the same manner. There are two types of terms in $\cE_4$. We denote the first type term as
		\begin{align*}
			\cE_{41}\eqdef\Re\int v^\eps(y)\big(\cW^{\<2>}_{y,N}({2N+2})\succ\tau_y(\overline{I^N_{2N+2}(u)}-\overline{I^{N,\flat}_{2N+2}(u)})\big)\overline{I_{2N+2}^N(w)}\,\dif x\d y,
		\end{align*}
		and  set the second type as
		\begin{align*}
			\cE_{42}\eqdef \Re\int_0^{2N+2}\int v^\eps(y)(\cW^{\<2>}_{y,N}(t)\succ\tau_y\overline{\p_t{I}^{N,\flat}_t(u)})\overline{I_t^N(w)}\,\dif x\d y\d t.
		\end{align*}
		For $\cE_{41}$, we use the fact that the spectral support of $\cW^{\<2>}_{y,N}$ is contained within a ball of  radius $2N$, while the spectral support of $I^N_{2N+2}(u)-I^{N,\flat}_{2N+2}(u)$  lies within an annulus with an inner radius of $\frac{N}4$, to obtain for $\gamma\geq0$
		\begin{align*}
			\|\cW^{\<2>}_{y,N}(2N+2)\|_{\bC^{-1+\frac{3\kappa}2}}\lesssim&\, \|\cW^{\<2>}_{y,N}(2N+2)\|_{\bC^{-1-\frac\kappa2}}N^{2\kappa},
		\end{align*}
		\begin{equs}[eq:difIN]
			N^{\gamma}\|I^N_{2N+2}(u)-I^{N,\flat}_{2N+2}(u)\|_{L^2}\lesssim&\, \|I^N_{2N+2}(u)-I^{N,\flat}_{2N+2}(u)\|_{H^{\gamma}},
		\end{equs}
		which combined with
		the paraproduct estimates Lemma \ref{lem:para} implies that
		\begin{align*}
			|\cE_{41}|\lesssim&\int v^\eps(y)\|\cW^{\<2>}_{y,N}({2N+2})\|_{\bC^{-1+\frac{3\kappa}2}}\d y\|I^N_{2N+2}(u)-I^{N,\flat}_{2N+2}(u)\|_{L^2}\|I_{2N+2}^N(w)\|_{H^{1-\kappa}}
			\\\lesssim&\,\|\mW\|\|I^N_{2N+2}(u)-I^{N,\flat}_{2N+2}(u)\|_{H^{2\kappa}}\|I_{2N+2}^N(w)\|_{H^{1-\kappa}}.
		\end{align*}
		We apply interpolation Lemma \ref{lem:interpolation} and \eqref{est:flat} to have
		\begin{align*}
			\|I^N_{2N+2}(u)-I^{N,\flat}_{2N+2}(u)\|_{H^{2\kappa}}\lesssim&\, \|I^N_{2N+2}(u)\|_{H^{\frac12-\kappa}}^{\frac{2\kappa}{\frac12-\kappa}}\|I^N_{2N+2}(u)\|_{L^2}^{\frac{\frac12-3\kappa}{\frac12-\kappa}}
			\\\lesssim&\,\Big(\|\mY^N_{2N+2}\|_{H^{\frac12-\kappa}}^{\frac{2\kappa}{\frac12-\kappa}}+\|I^N_{2N+2}(w)\|_{H^{1-\kappa}}^{\frac{2\kappa}{\frac12-\kappa}}\Big)\|I^N_{2N+2}(u)\|_{L^2}^{\frac{\frac12-3\kappa}{\frac12-\kappa}}.
		\end{align*}
		Consequently, applying Young's inequality, we obtain
		\begin{align*}
			|\cE_{41}|\lesssim& \,K(\|\mW\|)+\delta\|I^N_{2N+2}(w)\|_{H^{1-\kappa}}^2+\delta \cV^\eps(I^N_{2N+2}(u)).
		\end{align*}
		Using Lemma \ref{lem:Iw} the desired bound   for the first type term $\cE_{41}$ follows.
		
		For the second type  terms in $\cE_4$, since $\p_t{I}^{N,\flat}_t(u)=-\widetilde \chi'(\frac{\nabla}{t+1})\frac{\nabla}{(t+1)^2}I^N_\infty(u)$ is spectrally supported in an annulus with an inner radius $\frac{t+1}4$ and an outer radius $\frac{t+1}3$, we obtain
		\begin{equs}[bd:8-dotIt]
			\|\p_t{I}^{N,\flat}_t(u)\|_{\bB^{s'}_{p,q}}\lesssim \|I^{N}_\infty(u)\|_{\bB^{s}_{p,q}}\frac1{ (t+1)^{1+s-s'}},\qquad s,s'\in\mR, p\in (1,\infty),q\in[1,\infty],
		\end{equs}
		which combined with the paraproduct estimates Lemma \ref{lem:para} and Lemma \ref{lem:vlvu} implies that
		\begin{align*}
			|\cE_{42}|\lesssim&\int_0^{2N+2} \sup_y\|\cW^{\<2>}_{y,N}(t)\|_{\bC^{-1+3\kappa/2}}\|\p_t{I}^{N,\flat}_t(u)\|_{L^2}\|I_t^N(w)\|_{H^{1-\kappa}}\d t
			\\\lesssim&\int_0^{2N+2}\frac1{ ( t+1)^{1+\kappa}} \sup_y\|\cW^{\<2>}_{y,N}(t)\|_{\bC^{-1-\kappa/2}}\|I_t^N(w)\|_{H^{1-\kappa}}\d t\|I^{N}_\infty(u)\|_{H^{3\kappa}}
			\\\lesssim&\, K(\|\mW\|)+\delta \sup_t\|I_t^N(w)\|_{H^{1-\kappa}}^2+\delta \cV^\eps(I_\infty^N(u)).
		\end{align*}
		In the second step, we used the spectral property of $\cW^{\<2>}_{y,N}$ to deduce $$\|\cW^{\<2>}_{y,N}(t)\|_{\bC^{-1+\frac{3\kappa}2}}\lesssim (t+1)^{2\kappa}\|\cW^{\<2>}_{y,N}(t)\|_{\bC^{-1-\frac\kappa2}}.$$ Additionally, we applied Young's inequality, interpolation Lemma \ref{lem:interpolation} and $I_\infty^N(u)=I_\infty^N(w)-\mY_\infty^N$ in the last step. Using Lemma \ref{lem:Iw} and Proposition \ref{prop:W} the desired bound   for the second type terms in $\cE_4$ follows.
	\end{proof}

	\bl\label{lem:E5} It holds that
	\begin{align*}
		\E|\cE_5|\lesssim 1+\delta \E\int_0^\infty\|l_t^N\|_{L^2}^2\,\d t+\delta\E\cV^\eps(I_\infty^N(u)),
	\end{align*}
	with the proportional constant independent of $\eps, N$.
	\el
	\begin{proof}
		We use \eqref{eq:difIN}
		and \eqref{bd8:b}, \eqref{eq:8-Wt} to  conclude that the terms in the first line of $\cE_5$ can be controlled by
		\begin{align*}
			&|\log N|\Big(\|I^{N,\flat}_{2N+2}(u)\|_{L^2}\|I^N_{2N+2}(u)-I^{N,\flat}_{2N+2}(u)\|_{L^2}+\|I^N_{2N+2}(u)-I^{N,\flat}_{2N+2}(u)\|_{L^2}^2\Big)
			\\\lesssim&\,|\log N| N^{-\frac14}\|I^N_{2N+2}(u)\|_{L^2}\|I^N_{2N+2}(u)-I^{N,\flat}_{2N+2}(u)\|_{H^{\frac14}}\lesssim\|I^N_{\infty}(u)\|_{L^2}\|I^N_{\infty}(u)\|_{H^{\frac14}}
			\\\lesssim&\, K(\|\mW\|)+\delta \|I_{\infty}^N(w)\|_{H^{1-\kappa}}^2+\delta \cV^\eps(I_{\infty}^N(u)),
		\end{align*}
		where we used Proposition \ref{prop:W}, Young's inequality and $I_{\infty}^N(u)=I_{\infty}^N(w)-\mY_{\infty}^N$ in the last step.
		
		For the last term in $\cE_5$  we use \eqref{bd:8-dotIt} and \eqref{bd8:b} to have
		\begin{align*}
			&\E\Big|\Re\int_0^{2N+2}b^{\eps,N}(t)\la I^{N,\flat}_t(u),\p_t{I}^{N,\flat}_t(u)\ra\d t\Big|
			\lesssim \E\int_0^{2N+2}|b^{\eps,N}(t)|\|I^{N,\flat}_t(u)\|_{L^2}\|\p_t{I}^{N,\flat}_t(u)\|_{L^2}\d t
			\\\lesssim&\,\E\int_0^{2N+2}|\log (t+1)|\|I^{N}_\infty(u)\|_{L^2}\|I^N_\infty(u)\|_{H^{\frac14}}(t+1)^{-1-\frac14}\d t\lesssim 1+\delta \E\|I_\infty^N(w)\|_{H^{1-\kappa}}^2+\delta \E\cV^\eps(I_\infty^N(u)),
		\end{align*}
		where we also applied  Proposition \ref{prop:W}, Young's inequality and $I_\infty^N(u)=I_\infty^N(w)-\mY_\infty^N$ in the last step. Using Lemma \ref{lem:Iw} the desired bound  follows.
	\end{proof}

	\bl\label{lem:E6} It holds that
	\begin{align*}
		\E |\cE_6|\lesssim 1+\delta \E\int_0^\infty\|l_t^N\|_{L^2}^2\d t+\delta\E\cV^\eps(I_\infty^N(u)).
	\end{align*}
	\el
	\begin{proof}
		We start with the term in $\frac12\int_0^{2N+2}\la \sJ_t,\sJ_t \ra\d t$ involving two $\cW^{\<2vm>}_{\eps,N}$ and use the commutators from Lemma \ref{lem:com4} and $J_t^N$ is a symmetric operator to express it as follows:
		\begin{equs}[eq:JtT]
			&2\int_0^{2N+2}\Big\la J_t^N\Big(	 \cW^{\<2vm>}_{\eps,N}\succ{I_t^{N,\flat}(u)}\Big),J_t^N\Big(	 \cW^{\<2vm>}_{\eps,N}\succ{I_t^{N,\flat}(u)}\Big)\Big\ra\,\d t
			\\=&\,
			2\int_0^{2N+2}\int
			[(J_t^N)^2(\cW^{\<2vm>}_{\eps,N})\circ	\cW^{\<2vm>}_{\eps,N}] |I_t^{N,\flat}(u)|^2\,\d x\d t
			+2\int_0^{2N+2}\sR_t(\cW^{\<2vm>}_{\eps,N},\cW^{\<2vm>}_{\eps,N},I_t^{N,\flat}(u),I_t^{N,\flat}(u))\,\d t.
		\end{equs}
		Using \eqref{sto-sec:8-W}, we combine the first term with $\int_0^{2N+2}\dot{b}_1^{\eps,N}|I_t^{N,\flat}(u)|^2\,\d x\d t$ and express it as
		$$\int_0^{2N+2}\int
		\cW_{\eps,N}^{\<22vm-W>}(t) |I_t^{N,\flat}(u)|^2\,\d x\d t,$$  which, by \eqref{est:flat}, can be controlled by
		\begin{equs}[est:E6-W1]
			\int_0^{2N+2}\|\cW_{\eps,N}^{\<22vm-W>}(t) \|_{\bC^{-\kappa}}\d t\|I_\infty^N(u)\|_{L^2}\|I_\infty^N(u)\|_{H^{2\kappa}}
			\lesssim K(\|\mW\|)
			+\delta\cV^\eps(I_\infty^N(u))+\delta\|I_\infty^N(w)\|_{H^{1-\kappa}}^2.
		\end{equs}
		By applying Proposition \ref{prop:W} and Lemma \ref{lem:Iw} we derive the desired bound for it.
		For the second term on the RHS of \eqref{eq:JtT}, we  use  Lemma \ref{lem:com4} and \eqref{est:flat} to have it controlled by
		\begin{equs}[est:E6-W2]
			&\int_0^{2N+2}( t+1)^{-1-\kappa}\|\cW^{\<2vm>}_{\eps,N}\|_{\bC^{-1-\kappa}}^2\|I^N_\infty(u)\|^2_{H^{\frac12-\kappa}}\d t
			\\\lesssim&\,K(\|\mW\|)+K(\|\mW\|)\|I^N_\infty(w)\|_{H^{1-\kappa}}^{2\gamma}\|I^N_\infty(w)\|_{L^2}^{2(1-\gamma)}	\lesssim  K(\|\mW\|)+\delta \|I_\infty^N(w)\|_{H^{1-\kappa}}^2+\delta \cV^\eps(I_\infty^N(u)),
		\end{equs}
		for $\gamma=\frac{\frac12-\kappa}{1-\kappa}$.
		Here we decomposed $I^N_\infty(u)=I^N_\infty(w)-\mY^N_\infty$, and in the final step, we applied Lemma \ref{lem:vlvu} and Young's inequality. For the remaining terms, we focus on the following one:
		\begin{align*}
			\cE_{60}\eqdef&\,2\int_0^{2N+2}\Big\langle\int v^\eps(y)J_t^N\big(\cW^{\<2>}_{y,N}(t)\succ\tau_y\overline{I^{N,\flat}_t(u)}\big)\d y,\int v^\eps(y)J_t^N\big({\cW}^{\<2>}_{y,N}(t)\succ\tau_y\overline{I^{N,\flat}_t(u)}\big)\d y\Big\rangle\d t
			\\&-\int_0^{2N+2}(\dot{b}_3^{\eps,N}+\dot{b}_4^{\eps,N})\|I_t^{N,\flat}(u)\|_{L^2}^2\,\dif t,
		\end{align*}
		and other terms follow in a similar manner. Using \eqref{sto-sec:8-Wy} we obtain
		\begin{align*}
			\cE_{60}=\sum_{i=1}^4\cE_{6i},
		\end{align*}
		with
		\begin{align*}
			\cE_{61}\eqdef&\,2\int_0^{2N+2}\int v^\eps(y)v^\eps(y_1)\sR_t(\cW^{\<2>}_{y,N}(t), \overline{\cW}^{\<2>}_{y_1,N}(t),\tau_y\overline{I^{N,\flat}_t(u)}, \tau_{y_1}{I^{N,\flat}_t(u)})\,\d y\d y_1\d t,\\
			\cE_{62}\eqdef&
			\int_0^{2N+2}\int v^\eps(y)v^\eps(y_1)\cW_{y,y_1,N}^{\<22ccvm-Wy>}(t)\tau_y\overline{I^{N,\flat}_t(u)}\tau_{y_1}{I^{N,\flat}_t(u)} \,\d x\d y\d y_1\d t,
			\\\cE_{63}\eqdef&\int_0^{2N+2}\int v^\eps(y)v^\eps(y_1)(\dot{b}_3^{N}+\dot{b}_4^{N})(t,y,y_1)\tau_y\overline{I^{N,\flat}_t(u)} \tau_{y_1}{I^{N,\flat}_t(u)}\,\d x\d y\d y_1\d t
			\\&-\int_0^{2N+2}(\dot{b}_3^{\eps,N}+\dot{b}_4^{\eps,N})\|I_t^{N,\flat}(u)\|_{L^2}^2\,\dif t,
		\end{align*}
		for $\dot{b}_i^{N}(t,y,y_1)$ and $\dot{b}_i^{\eps,N}$ in \eqref{def:dotbi}, \eqref{def:dotbiN}.
		It is straightforward to see that	$\cE_{61}$  can be estimated similarly as in \eqref{est:E6-W2}. For $\cE_{63}$ we can rewrite it as follows:
		\begin{align*}\int_0^{2N+2} \sum_{k,k_1,k_2}& \dot{\tilde{b}}^{N}_t(k_1,k_2)\bigg(\widehat{v^\eps}(k_1+k)\widehat{v^\eps}(k_1-k)+\widehat{v^\eps}(k_1-k)\widehat{v^\eps}(k_2+k)\\&-\widehat{v^\eps}(k_1)^2-\widehat{v^\eps}(k_1)\widehat{v^\eps}(k_2)\bigg)
			|\langle I^{N,\flat}_t(u),e_k\rangle|^2\,\dif t,
		\end{align*}
		which by \eqref{bd:sigma} is bounded by
		\begin{align*}&\int_0^{2N+2} \sum_{k_1,k_2,k}\frac{(t+1)^{-1-\kappa}}{\la k_1\ra^2 \la k_2 \ra^2 \la k_1+k_2\ra^{2-\kappa}}|k|^{2\kappa}\bigg(\frac1{|k_1+k|^{2\kappa}}+\frac1{|k_1|^{2\kappa}}\bigg)|\langle I^{N,\flat}_t(u),e_k\rangle|^2\,\dif t
			\\\lesssim&\int_0^{2N+2}(t+1)^{-1-\kappa} \| I^{N,\flat}_t(u)\|_{H^\kappa}^2\,\dif t
			\lesssim\|I^{N}_{\infty}{(u)\|_{H^\kappa}^2}.
		\end{align*}
		Here we used similar argument as the proof of Lemma \ref{comnew}.
		Thus $\cE_{63}$ can be bounded by a similar argument as \eqref{est:E6-W1}.
		It remains to consider $\cE_{62}$, and we apply \eqref{est:flat}, \eqref{eq:tauy} along with a similar calculation to \eqref{est:E6-W1} to obtain
		\begin{align*}
			\cE_{62}\lesssim&\int_0^{2N+2}\sup_{y,y_1}\|\cW_{y,y_1,N}^{\<22ccvm-Wy>}\|_{\bC^{-\kappa}}\d t\|I_\infty^N(u)\|_{L^2}\|I_\infty^N(u)\|_{H^{2\kappa}}
			\\\lesssim&\,  K(\|\mW\|)+\delta \|I_\infty^N(w)\|_{H^{1-\kappa}}^2+\delta \cV^\eps(I_\infty^N(u)),
		\end{align*}
		which by Proposition \ref{prop:W} and Lemma \ref{lem:Iw} implies the desired bound.
	\end{proof}
	
	\begin{proof}[Proof of Theorem \ref{th:partition}]
		Using Lemma \ref{lem:B-D} we obtain
		\begin{align*}
			-\log \widetilde\sZ_{\eps,N}=\inf_{u\in\mathbf{H}_a} \E[\sV^{\eps,N}(u)],
		\end{align*}
		with $f=0$ in $\sV^{\eps,N}(u)$.
		Applying Proposition \ref{prop:dec} we have
		\begin{align*}
			\E\sV^{\eps,N}(u)=&\,\,\frac12\E\int_0^\infty\|l_t^N\|_{L^2}^2\,\dif t+\frac12\E\cV^\eps(I^N_\infty(u))+\sum_{i=0}^6\E[\cE_i].
		\end{align*}
		By Lemma \ref{lem:E1E2}--Lemma \ref{lem:E6} we have
		\begin{align*}
			\Big|\sum_{i=0}^6\E[\cE_i]\Big|\leq K+\frac\delta2\Big(\E\int_0^\infty\|l_t^N\|_{L^2}^2\,\dif t+\E\cV^\eps(I^N_\infty(u))\Big),
		\end{align*}
		for $0<\delta<1$ and some constant $K$, which is independent of $\eps,N$.
		We then obtain
		\begin{equs}[par:lower]
			-K\leq -K+\frac12(1-\delta)\Big(\E\int_0^\infty\|l_t^N\|_{L^2}^2\,\dif t+\E\cV^\eps(I^N_\infty(u))\Big)\leq \E\sV^{\eps,N}(u),
		\end{equs}
		and
		\begin{equs}[bd:wf-8]
			\E\sV^{\eps,N}(u)\leq K+\frac12(1+\delta)\Big(\E\int_0^\infty\|l_t^N\|_{L^2}^2\,\dif t+\E\cV^\eps(I^N_\infty(u))\Big).
		\end{equs}
		for $l_t^N=l_t^N(u)$ given by \eqref{def:ltT1}. Using \eqref{par:lower} we get the following lower bound
		\begin{equs}[eq:par:lower]
			-K	\leq-\log\widetilde\sZ_{\eps,N}.
		\end{equs}
		Hence, it remains to choose $u$ such that the RHS of \eqref{bd:wf-8}  is uniformly bounded in $\eps, N$. To this end, we choose $u$ to satisfy
		\begin{equs}[choice:u]
			u_t=&-2J^N_t\cY^N_t
			-2J_t^N\big(	 \Delta_{>L}\cW^{\<2vm>}_{\eps,N}(t)\succ{I_t^{N,\flat}(u)}\big)
			-2\int v^\eps(y)J_t^N\big(\Delta_{>L}\cW^{\<2>}_{y,N}(t)\succ\tau_y\overline{I^{N,\flat}_t(u)}\big)\d y
			\\&-2\int v^\eps(y)J_t^N\big(\Delta_{>L}\cW^{\<2m>}_{y,N}(t)\succ\tau_y{I^{N,\flat}_t(u)}\big)\d y,
		\end{equs}
		for the localizer $\Delta_{>L}$ introduced in \eqref{def:loc} and some $L>0$ given below. The existence and uniqueness of $u$ follow by usual fixed point arguments. Additionally, we use
		 the paraproduct estimates from Lemma \ref{lem:para}, and \eqref{est:flat}, \eqref{est:vx} to obtain
		\begin{align*}
			\|I_t^N(u)\|_{\bC^\kappa}\lesssim&\,\|\mY_t^N\|_{\bC^{\kappa}}+ \int_0^t\Big[\Big\|(J_s^N)^2\Big(	 (\Delta_{>L}\cW^{\<2vm>}_{\eps,N}(t)\succ{I_s^{N,\flat}(u)})\Big)\Big\|_{\bC^{\kappa}}
			\\&+\Big\|(J_s^N)^2\Big(\int v^\eps(y)(\Delta_{>L}\cW^{\<2>}_{y,N}(s)\succ\tau_y\overline{I^{N,\flat}_s(u)})\d y\Big)\Big\|_{\bC^{\kappa}}
			\\&+\Big\|(J_s^N)^2\Big(\int v^\eps(y)(\Delta_{>L}\cW^{\<2m>}_{y,N}(s)\succ\tau_y{I^{N,\flat}_s(u)})\d y\Big)\Big\|_{\bC^{\kappa}}\Big]\d s
			\\\leq&\, \|\mY_t^N\|_{\bC^{\kappa}}+C\int_0^t\frac1{(s+1)^{1+\frac18}}\Big[\|\Delta_{>L}\cW^{\<2vm>}_{\eps,N}(s)\|_{\bC^{-1-3\kappa}}+\sup_y\|\Delta_{>L}\cW^{\<2>}_{y,N}(s)\|_{\bC^{-1-3\kappa}}
			\\&\qquad+\sup_y\|\Delta_{>L}\cW^{\<2m>}_{y,N}(s)\|_{\bC^{-1-3\kappa}}\Big]\d s\|I_t^N(u)\|_{\bC^\kappa}.
		\end{align*}
		We apply the localizer estimates \eqref{eq:loc} to have
		\begin{align*}
			&\|\Delta_{>L}\cW^{\<2vm>}_{\eps,N}(t)\|_{\bC^{-1-3\kappa}}+\sup_y\|\Delta_{>L}\cW^{\<2>}_{y,N}(t)\|_{\bC^{-1-3\kappa}}+\sup_y\|\Delta_{>L}\cW^{\<2m>}_{y,N}(t)\|_{\bC^{-1-3\kappa}}
			\\\lesssim& \,2^{-2\kappa L}\Big(\|\cW^{\<2vm>}_{\eps,N}(t)\|_{\bC^{-1-\kappa}}+\sup_y\|\cW^{\<2>}_{y,N}(t)\|_{\bC^{-1-\kappa}}+\sup_y\|\cW^{\<2m>}_{y,N}(t)\|_{\bC^{-1-\kappa}}\Big).
		\end{align*}
		We then choose $L=L(s)$ given by
		\begin{align*}
			2^{2\kappa L}\backsimeq 8C(\|\cW^{\<2vm>}_{\eps,N}(s)\|_{\bC^{-1-\kappa}}+\sup_y\|\cW^{\<2>}_{y,N}(s)\|_{\bC^{-1-\kappa}}+\sup_y\|\cW^{\<2m>}_{y,N}(s)\|_{\bC^{-1-\kappa}}).
		\end{align*}
		such that
		\begin{align*}
			C\Big[\|\Delta_{>L}\cW^{\<2vm>}_{\eps,N}(s)\|_{\bC^{-1-3\kappa}}+\sup_y\|\Delta_{>L}\cW^{\<2>}_{y,N}(s)\|_{\bC^{-1-3\kappa}}+\sup_y\|\Delta_{>L}\cW^{\<2m>}_{y,N}(s)\|_{\bC^{-1-3\kappa}}\Big]\leq1/16.
		\end{align*}
		Hence, we obtain
		\begin{align*} \|I_t^N(u)\|_{\bC^\kappa}\leq&\, C\|\mY_t^N\|_{\bC^{\kappa}}+\frac1{16}\int_0^t\frac1{( s+1)^{1+\frac18}}	\|I_t^N(u)\|_{\bC^\kappa}\d s
			\\\leq&\,C\|\mY_t^N\|_{\bC^{\kappa}}+\frac12\|I_t^N(u)\|_{\bC^\kappa},
		\end{align*}
		which implies that
		\begin{align*}
			\|I_t^N(u)\|_{\bC^\kappa}\lesssim K(\|\mW\|).
		\end{align*}
		Consequently, we apply Proposition \ref{prop:W} to obtain that for $p\geq1$
		\begin{equs}[eq:INin-4]
			\sup_{N,\eps}\E \|I_\infty^N(u)\|_{\bC^\kappa}^p\lesssim \E K(\|\mW\|)\lesssim 1,
		\end{equs}
		which yields that
		\begin{align*}
			\sup_{N,\eps}\E\cV^\eps(I^N_\infty(u))\lesssim \sup_{N,\eps}\E\|I_\infty^N(u)\|_{\bC^\kappa}^4\lesssim \E K(\|\mW\|)\lesssim 1.
		\end{align*}
		Combining \eqref{def:ltT1} and \eqref{choice:u} we obtain
		\begin{align*}
			l_t^N=&\,2	J_t^N\Big(	 (\Delta_{\leq L}\cW^{\<2vm>}_{\eps,N}(t)\succ{I_t^{N,\flat}(u)})\Big)
			+2J_t^N\Big(\int v^\eps(y)(\Delta_{\leq L}\cW^{\<2>}_{y,N}(t)\succ\tau_y\overline{I^{N,\flat}_t(u)})\d y\Big)
			\\&+2J_t^N\Big(\int v^\eps(y)(\Delta_{\leq L}\cW^{\<2m>}_{y,N}(t)\succ\tau_y{I^{N,\flat}_t(u)})\d y\Big).
		\end{align*}
		Using the localizer estimates \eqref{eq:loc} we obtain
		\begin{align*}
			\|\Delta_{\leq L}\cW^{\<2vm>}_{\eps,N}(t)\|_{\bC^{-1+\kappa}}+\sup_y\|\Delta_{\leq L}\cW^{\<2>}_{y,N}(t)\|_{\bC^{-1+\kappa}}+\sup_y\|\Delta_{\leq L}\cW^{\<2m>}_{y,N}(t)\|_{\bC^{-1+\kappa}}\lesssim K(\|\mW\|).
		\end{align*}
		By applying the paraproduct estimates Lemma \ref{lem:para} and \eqref{est:jt} we obtain
		\begin{align*}
			\E\int_0^\infty\|l^N_t\|_{L^2}^2\d t\lesssim& \,\E\int_0^\infty (t+1)^{-1-\kappa}	\Big(\|\Delta_{\leq L}\cW^{\<2vm>}_{\eps,N}(t)\|_{\bC^{-1+\kappa}}+\sup_y\|\Delta_{\leq L}\cW^{\<2>}_{y,N}(t)\|_{\bC^{-1+\kappa}}
			\\&+\sup_y\|\Delta_{\leq L}\cW^{\<2m>}_{y,N}(t)\|_{\bC^{-1+\kappa}}\Big)^2\d t\|I^N_\infty(u)\|_{L^2}^2,
		\end{align*}
		which combined with \eqref{eq:INin-4} implies that
		\begin{align*}
			\sup_{N,\eps}\E\int_0^\infty\|l^N_t\|_{L^2}^2\d t\lesssim1.
		\end{align*}
		Hence, we obtain a uniform upper bound of $-\log \widetilde\sZ_{\eps,N}$, which combined with \eqref{eq:par:lower} implies \eqref{bd:par}.
		For \eqref{mom:muN-eps} we use Lemma \ref{lem:B-D} to have
		\begin{equs}[feq]
			-\log\int\exp(f(\Psi))\d \tilde\mu_N^\eps(\Psi)=\log \widetilde\sZ_{\eps,N}+\inf_{u\in\mathbf{H}_a} \E[\sV^{\eps,N}(u)].
		\end{equs}
		The first term on the RHS has been controlled. For the second term, we also use Besov embedding Lemma \ref{lem:emb} to have
		\begin{align*}
			&|f(W_\infty^N+I^N_\infty(u))|\lesssim 1+\|W_\infty^N\|_{\bC^{-\frac12-\kappa}}+\|I^N_\infty(u)\|_{\bC^{-\frac12-\kappa}}
			\\\lesssim&\,1+ \|W_\infty^N\|_{\bC^{-\frac12-\kappa}}+\|\mY^N_\infty\|_{\bC^{\frac12-\kappa}}+\|I^N_\infty(w)\|_{\bC^{-\frac12-\kappa}}
			\lesssim\|\mW\|+\|I^N_\infty(w)\|_{H^{1-\kappa}}.
		\end{align*}
		Using Lemma \ref{lem:Iw} and exactly the same arguments as above, \eqref{mom:muN-eps} follows.
	\end{proof}

\begin{proof}[Proof of Theorem \ref{th:partition-revised}]
We begin by proving \eqref{mom:muN-eps-revised}. Note that the difference in renormalization
constants requires us to add the following extra terms into  $f(W_\infty^N+I_\infty^N(u))$  in the proof of Theorem \ref{th:partition}:
	\begin{align*}
		[(a^\eps-a^{\eps,N})-(6b^\eps-6b^{\eps,N})]\int \Wick{|W_\infty^N+I_\infty^N(u)|^2}\d x,
	\end{align*}
	where $a^\eps, b^\eps$ are given in \eqref{eq:counterterms}.
We then apply the paraproduct estimates from Lemma~\ref{lem:para} together with Proposition~\ref{prop:W} to obtain for $\delta>0$
	\begin{equs}[west]
		&\Big|\E\int \Wick{|W_\infty^N+I_\infty^N(u)|^2}\d x\Big|=\Big|2\E\Re \la W_\infty^N,I_\infty^N(u)\ra+\E\|I_\infty^N(u)\|^2_{L^2}\Big|
		\\\lesssim &\,|\E\Re\la W_\infty^N,{\mY^N_\infty}\ra|+|\E\Re \la W_\infty^N,{I^N_\infty(w)}\ra|+\E\|I_\infty^N(u)\|^2_{L^2}
		\\\lesssim&\,1+\delta \E\|I^N_\infty(w)\|_{H^{1-\kappa}}^2+\delta \E\cV^\eps(I^N_\infty(u)).
	\end{equs}
We claim that we can prove that for $N\geq\eps^{-2-\kappa}$
\begin{equs}[dif:abc]
	|a^\eps-a^{\eps,N}|+|b^\eps-b^{\eps,N}|\lesssim1.\end{equs}
Thus we obtain for $N\geq\eps^{-2-\kappa}$ and $\delta>0$
\begin{equs}[eqab] &\E\Big|[(a^\eps-a^{\eps,N})-(6b^\eps-6b^{\eps,N})]\int \Wick{|W_\infty^N+I_\infty^N(u)|^2}\d x\Big|
\\\lesssim&
\,1+\delta \E\|I^N_\infty(w)\|_{H^{1-\kappa}}^2+\delta \E\cV^\eps(I^N_\infty(u)),
\end{equs}
which  from the proof of Theorem \ref{th:partition} implies that for $N\geq\eps^{-2-\kappa}$
\begin{equs}[partitionre]|\log \sZ_{\eps,N}+c^{\eps,N}|\lesssim 1.\end{equs}
Similar as \eqref{feq} we obtain
\begin{equs}
			&-\log\int\exp(f(\Psi))\d \mu_N^\eps(\Psi)
\\=&\log \sZ_{\eps,N}+c^{\eps,N}+\inf_{u\in\mathbf{H}_a} \E[\sV^{\eps,N}(u)-[(a^\eps-a^{\eps,N})-(6b^\eps-6b^{\eps,N})]\int \Wick{|W_\infty^N+I_\infty^N(u)|^2}\d x],
		\end{equs}
the absolute value of which is bounded  by similar argument as in the proof of Theorem \ref{th:partition} and using \eqref{eqab} and \eqref{partitionre}.

It remains to prove \eqref{dif:abc}. By direct calculation and using $(\mathbf{Hv})$, we have
	\begin{equs}[dif:aN]
		|a^\eps-a^{\eps,N}|=\Big|\sum_{k\in\mZ^3}\frac{(\chi_N(k)^2-1)\widehat{v^\eps}(k)}{(2\pi)^3\la k\ra^2}\Big|\lesssim \sum_{|k|\gesim N}\frac{\widehat{v^\eps}(k)}{\la k\ra^2}\lesssim \eps^{-2-\kappa}N^{-1}.
	\end{equs}
	Similarly we obtain
	\begin{equs}[dif:bN]
		|b^\eps-b^{\eps,N}|\lesssim \eps^{-\frac12-\kappa}N^{-\frac12}.
	\end{equs}
Combining \eqref{dif:aN}--\eqref{dif:bN}, we obtain that	\eqref{dif:abc} and \eqref{mom:muN-eps-revised} hold due to $N\geq \eps^{-2-\kappa}$.

To prove \eqref{mom:mutildeP-revised}, we only need to modify
		$f(W_\infty^N+I^N_\infty(u))$ to
		$$c_0\la \widetilde P (W_\infty^N+I^N_\infty(u)), W_\infty^N+I^N_\infty(u)\ra+ c_1\int \Wick{|W_\infty^N+I_\infty^N(u)|^2}\d x.$$ We then use
		\begin{align*}
		 \Big|\E \la \widetilde P (W_\infty^N+I^N_\infty(u)), W_\infty^N+I^N_\infty(u)\ra\Big|
		 \lesssim C_{\widetilde P}+\delta \cV^\eps(I_\infty^N(u)),
		\end{align*}
and similar argument as \eqref{west}
		to derive \eqref{mom:mutildeP-revised} similarly.

\end{proof}


\section{A priori estimates for quantum Gibbs state}\label{sec:gibbs}

Now we begin our analysis of the quantum Gibbs state. In this section, we establish several a priori estimates for the quantum Gibbs state, including 

\begin{itemize}

\item First-order correlation estimates, derived from the variational principle and energy functional coercivity; and
\item Second-order correlation estimates, which exploit the specific structure of the Gibbs state.

\end{itemize}
These estimates form the foundation for our asymptotic approximation of the quantum free energy in the following section. While our overall strategy aligns with the framework of \cite[Sections 5--8]{LewNamRou-21}, we significantly simplify the analysis to cleanly track the $\eps$-dependence of the singular potential, enabled by a new abstract correlation inequality from \cite{DeuNamNap-25}. As a consequence, we obtain the quantitative results needed for the treatment of the singular potential.

\subsection{Second quantization formalism.} For every $f\in L^2(\bT^3)$, we  define the annihilation operator $a(f)$ and the creation operator $a^*(f)$ on Fock space $\gF=\gF(L^2(\bT^3))$ by
\begin{equs}[eq:annihilationOperator]
    (a(f) \psi)(x_1,...,x_{n-1}) &= \sqrt{n} \int_{\bT^3} \overline{f(x)}  \psi(x_1,...,x_{n-1},x) \d x,\\
        (a^*(f) \psi)(x_1,...,x_{n+1}) &= \frac1{\sqrt{n+1}} \sum_{j=1}^{n+1} f(x_j) \psi(x_1,...,x_{j-1}, x_{j+1}, ..., x_{n+1}),
\end{equs}
and extend to $\gF$ by linearity. These operators are related to the operators $a_x,a^*_x$ in \eqref{eq:many body hamil-x-intro} by
$$
a(f)= \int_{\bT^3} \overline{f(x)} a_x \d x , \quad a^{*}(f) = \int_{\bT^3} {f(x)} a_x^* \d x,\quad \forall f\in L^2(\bT^3).
$$
In particular, for every $k\in \bZ^3$, we denote $a_k=a(e_k)$, $a_k^*=a^*(e_k)$ with function $e_k(x) = (2\pi)^{-3/2} e^{\imath k \cdot x}$. They satisfy
\begin{equs}[eq:CCR]
\text{}[a_k,a_j]= 0= [a^*_k,a^*_j], \quad [a_k,a^*_j] =  \delta_{k,j},\quad \forall k,j\in \bZ^3.
\end{equs}
which are similar to \eqref{eq:CCR-x}.

If $A$ is a self-adjoint operator on $\gH$, then its second quantization on Fock space can be written as
$$
\d\Gamma (A) = 0\oplus \bigoplus_{n=1}^\infty(\sum_{\ell=1}^n A_\ell)=\sum_{j,k\in\mZ^3}\langle e_j, A e_k\rangle a_j^* a_k = \int_{\bT^3} a^*_x A a_x \d x.
$$
For example, the number operator is $\cN=\dG(1)$. More generally, if $A_n$ is a self-adjoint operator on $\gH^n$, we can write its second quantization $\bA_n$ defined in \eqref{def:mAn} as
\begin{align*}
\bA_n =\frac{1}{n!}\sum_{k_1,\cdots, k_n, j_1,\cdots, j_n \in \bZ^3} \langle e_{k_1} \otimes_s \cdots e_{k_n}, A_n e_{j_1} \otimes_s \cdots e_{j_n} \rangle a_{j_1}^* \cdots a_{j_n}^* a_{k_n}\dots a_{k_1}.
\end{align*}

We will always denote by $h=-\Delta+1$ the one-body operator on $L^2(\bT^3)$. Thus the Gaussian quantum state in \eqref{eq:GFF-quantum} is $\Gamma_0=\cZ_0^{-1}e^{-\lambda \dG(h)}$. Moreover, with our choice of the chemical potential $\vartheta$ in \eqref{def:gamma}, the interacting Hamiltonian $\mH_\lambda$ in \eqref{eq:bH-lambda} can be written as $\lambda \dG(h)+\bW- E_\lambda$ with
\begin{equs}[eq:bW]
	\bW  &= \frac{\lambda^2}{2(2\pi)^3}  (\cN-N_0)^2 + \sum_{k\ne 0} \frac{\lambda^2}{2} \widehat{v^\eps}(k) |\d\Gamma(e_k)|^2  - \lambda \vartheta^\eps (\cN-N_0),\\
\vartheta^\eps&= a^\eps - 6b^\eps - m +1 -  \mathfrak{e}_{\lambda},\quad m=m_0 - 2C_1 - 2C_2,\nn\\
	\mathfrak{e}_{\lambda}&= \lambda \rho_0 - \frac{\zeta\left(\frac{3}2\right)}{(4\pi)^{3/2}} \lambda^{-1/2} {-} C_0  -\frac\lambda2 v^\eps(0) = O(\lambda^{1/2}),
\end{equs}
and
\begin{equs}[def:elamdba]
N_0 = (2\pi)^3 \rho_0 ,\quad E_{\lambda} =     \frac{\lambda^2}{2} (2\pi)^{3} \rho_{0}^2 +  (2\pi)^3 \lambda \rho_0 {\vartheta^\eps},
\end{equs}
where $\rho_0$, $C_0$ and $a^\eps, b^\eps, C_1, C_2$ are introduced in \eqref{eq:rho0}, \eqref{def:C0} and \eqref{eq:counterterms}, respectively.
Moreover, we used the notation $|A|^2=A^*A$ with $A=\dG(e_k)$ the second quantization of the multiplication operator $e_k(x)=(2\pi)^{-3/2} e^{\imath k\cdot x}$. Since the constant $E_\lambda$ plays no role in the Gibbs state $\Gamma_\lambda$ (although it changes the partition function), from now on we will take
\begin{equs}[eq:Gibbs-state-def-new]
\Gamma_{\lambda} = \cZ_{\lambda}^{-1} e^{-\bH_\lambda},\quad \cZ_{\lambda} = \Tr e^{-\bH_\lambda},\quad \bH_\lambda= \lambda \dG(h)+\bW
\end{equs}
as definition, where the renormalized interaction $\bW$ is given in \eqref{eq:bW}. We can also write $\mH_\lambda$ as
\begin{align*}
	\bH_\lambda &= \lambda \int_{\bT^3}  a_x^* (-\Delta_x +m) a_x  \d x+ \frac{\lambda^2}{2} \iint_{\bT^3\times \bT^3}  v^\eps(x-y) (a_x^* a_x - \rho_{0})  (a_y^* a_y -\rho_{0}) \d x \d y \nn \\
	&\quad -  (a^\eps - 6b^\eps {-}	\mathfrak{e}_{\lambda}) \lambda \int (a_x^* a_x -\rho_0)  \d x ,
\end{align*}
which is a quantum analogue of the energy functional in \eqref{e:Phi_eps-measure1}.

\subsection{Variational principle and first-order a-priori estimates.} Now we collect a-priori estimates for the Gibbs state $\Gamma_\lambda$. Our starting point is the Gibbs variational principle, which asserts that the interacting Gibbs state $\Gamma_\lambda=\cZ_\lambda^{-1}e^{-\bH_\lambda}$ in \eqref{eq:Gibbs-state-def-new} is the {\em unique minimizer} for the variational problem
\begin{equs} [eq:rel-energy]
	-\log \frac{\cZ_\lambda}{\cZ_0} = \min_{\substack{\Gamma\ge 0\\ \Tr \Gamma=1}} \Big\{ \cH(\Gamma,\Gamma_{0}) +  \Tr\left[\bW \Gamma\right] \Big\}.
\end{equs}
with the relative entropy $
\cH(\Gamma,\Gamma_{0}) = \Tr [\Gamma (\log \Gamma-\log \Gamma_0)] \ge 0
$. This leads to the following a-priori estimate.

\begin{lemma}[First-order a-priori estimates]\label{lem:partition} The relative partition function satisfies
	\begin{equs}[eq:part bound]
		\left| \log  \frac{\cZ_\lambda}{\cZ_{0} } \right|  \lesim \eps^{-2}.
	\end{equs}
	Consequently, we have the following a-priori estimates on the Gibbs state $\Gamma_\lambda:$
	\begin{equs}[eq:a-priori-estimate-1]
		\cH(\Gamma_\lambda,\Gamma_{0})+ \lambda^2 \Tr [(\cN-N_0)^2  \Gamma_\lambda] \lesim  \eps^{-2},
	\end{equs}
	Moreover, we have the Hilbert--Schmidt estimate on the one-body density matrix:
	\begin{equs}[eq:HS-input]
		\lambda \Big\|\sqrt{h}\big(\Gamma_\lambda^{(1)}-\Gamma^{(1)}_0\big)\sqrt{h}\Big\|_{\rm HS} \lesim \eps^{-2} .
	\end{equs}
\end{lemma}

\begin{proof} From \eqref{eq:rel-energy}, by using $\Gamma_0$ as a trial state, we have the upper bound
	\begin{equs}[eq:free-energy-simple-1a]
		-\log \frac{\cZ_\lambda}{\cZ_0} \le \Tr [\bW \Gamma_0] = \frac{\lambda^2}{2{(2\pi)^3}}  \Tr[ (\cN-N_0)^2 \Gamma_0] + \frac{\lambda^2}{2} \sum_{k\ne 0}  \widehat{v^\eps}(k) \Tr \Big[ |\d\Gamma(e_k)| ^2 \Gamma_0 \Big] .
	\end{equs}
	Here the expectation of $\vartheta^\eps (\cN-N_0)$ in $\bW$ against $\Gamma_0$ is $0$. From  \cite[Eq. (5.43)]{LewNamRou-21} we have
	\begin{equs}[eq:free-energy-simple-1b]
		\lambda^2 \Tr[ (\cN-N_0)^2 \Gamma_0] \lesim \Tr [h^{-2}] \lesim 1.
	\end{equs}
	The second term on the right-hand side of \eqref{eq:free-energy-simple-1a} is called the exchange energy, which can be computed explicitly. Note that if $k\ne 0$, then $\Tr [\dG(e_k)\Gamma_0]=0$ since $\Gamma_0$ preserves the total momentum while $\dG(e_k)=\frac1{(2\pi)^{3/2}}\sum_{p\in \bZ^3} a^*_{p+k}a_p$ does not. Therefore, from the variance computation in  \cite[Lemma 5.11]{LewNamRou-21}, we have
	$$
	\frac{\lambda^2}{2} \sum_{k\ne 0}  \widehat{v^\eps}(k) \Tr \Big[ |\d\Gamma(e_k)| ^2 \Gamma_0 \Big] = \frac{\lambda^2}{2} \sum_{k\ne 0} \widehat{v^\eps}(k)\Big( \Tr [e_k^* \Gamma_0^{(1)}e_k \Gamma_0^{1}]+\frac1{(2\pi)^3}\tr(\Gamma_0^{(1)})\Big).
	$$
	By the operator inequality
	$$\lambda \Gamma_0^{(1)}=\frac{\lambda}{e^{\lambda h}-1} \le h^{-1}$$
	and the cyclicity of the trace we can bound
	\begin{align*}
		\lambda^2 \Tr [e_k^* \Gamma_0^{(1)}e_k \Gamma_0^{(1)}] &=  \lambda^2 \Tr \Big[\sqrt{ \Gamma_0^{(1)}} e_k^* \Gamma_0^{(1)}e_k \sqrt{\Gamma_0^{(1)}}\Big]
		\le \lambda \Tr \Big[\sqrt{ \Gamma_0^{(1)}} e_k^* h^{-1}e_k \sqrt{\Gamma_0^{(1)}}\Big] \\
		&=\lambda \Tr \Big[h^{-1/2} e_k \Gamma_0^{(1)}  e_k^* h^{-1/2}  \Big] \le \Tr \Big[h^{-1/2} e_k h^{-1} e_k^* h^{-1/2}  \Big]\\
		&= \Tr \Big[ e_k^* h^{-1} e_k h^{-1}  \Big].
	\end{align*}
	Therefore,
	\begin{equs}[eq:free-energy-simple-1c]
		\frac{\lambda^2}{2}  \sum_{k\ne 0}  \widehat{v^\eps}(k) \Tr \Big[  |\d\Gamma(e_k)| ^2 \Gamma_0\Big] &\le \frac{1}{2}  \sum_{k\ne 0}  \widehat{v^\eps}(k)  \Tr \Big[ e_k^* h^{-1} e_k h^{-1}  \Big] +\frac{\lambda^2}{2} \sum_{k\ne 0} \widehat{v^\eps}(k)\frac1{(2\pi)^3}\tr(\Gamma_0^{(1)}).\nn\\
		&\lesssim  \sum_{k\ne 0}  \widehat{v^\eps}(k) (\widehat{G^2}(k)+\lambda^{1/2-}) \lesssim \sum_{k\ne 0}  \widehat{v^\eps}(k)\Big(\frac1{1+|k|}+\lambda^{1/2-}\Big)\lesim \eps^{-2}.
	\end{equs}
	Here $G$ is the Green function and in the last second estimate we used Lemma \ref{lem:sum}.

 From \eqref{eq:free-energy-simple-1a}, \eqref{eq:free-energy-simple-1b} and \eqref{eq:free-energy-simple-1c}, we obtain the upper bound
	\begin{equs}[eq:free-energy-simple-1d]
		-\log \frac{\cZ_\lambda}{\cZ_0} \le \Tr [\bW \Gamma_0] \lesim \eps^{-2}.
	\end{equs}
	
	On the other hand, the Cauchy--Schwarz inequality implies for some $C>0$
	\begin{equs}[eq:free-energy-simple-21]
		\frac{\lambda^2}{2(2\pi)^3}  (\cN-N_0)^2  -  \lambda \vartheta^\eps (\cN-N_0) \ge  \frac{\lambda^2}{4(2\pi)^3}  (\cN-N_0)^2   -  C (\vartheta^\eps)^2.
	\end{equs}
	Therefore,
	\begin{equs}[eq:free-energy-simple-2]
		&-\log \frac{\cZ_\lambda}{\cZ_0}= \cH(\Gamma_\lambda,\Gamma_0) +{\rm Tr}[\bW \Gamma_\lambda]
		\nn\\
		&\ge \cH(\Gamma_\lambda,\Gamma_0) + \sum_{k\ne 0} \frac{\lambda^2}{2}  \widehat{v^\eps}(k)
		{\rm Tr}[|\d\Gamma(e_k)|^2 \Gamma_\lambda]  + \frac{\lambda^2}{4(2\pi)^3} {\rm Tr}[ (\cN-N_0)^2 \Gamma_\lambda] - C (\vartheta^\eps)^2.
	\end{equs}
	
	Since $\cH(\Gamma_\lambda,\Gamma_0) \ge 0$ and $\widehat{ v^\eps} \ge 0$, the right-hand side of \eqref{eq:free-energy-simple-2} is bounded from below by $ - C (\vartheta^\eps)^2= O(\eps^{-{2}})$. Thus we get the lower bound
	\begin{equs}[eq:free-energy-simple-1e]
		-\log \frac{\cZ_\lambda}{\cZ_0} \gesim - \eps^{-2},
	\end{equs}
	thus concluding the proof of \eqref{eq:part bound}.
	
	From \eqref{eq:free-energy-simple-2} and \eqref{eq:free-energy-simple-1d},  we also obtain \eqref{eq:a-priori-estimate-1}.  From the relative entropy estimate $\cH(\Gamma_\lambda,\Gamma_0) \lesim \eps^{-2}$ and the  inequality from \cite[Theorem 6.1]{LewNamRou-21}:
	\begin{equs}[eq:estim_HS_relative_entropy]
		\lambda^2 \tr\Big[ \left|\sqrt{h}\big(\Gamma^{(1)}-\Gamma^{(1)}_0\big)\sqrt{h}\right|^2\Big]\leq 4\,\cH(\Gamma,\Gamma_0)\left(\sqrt2 +\sqrt{\cH(\Gamma,\Gamma_0)}\right)^2,
	\end{equs}
	we obtain \eqref{eq:HS-input}. Note that there is no $\lambda$ in the statement of \cite[Theorem 6.1]{LewNamRou-21}, but applying this abstract with $h\mapsto \lambda h$, which is due to our choice $\Gamma_0=\cZ_0^{-1} e^{-\lambda \dG(h)}$, then we have the factor $\lambda^2$ in \eqref{eq:estim_HS_relative_entropy}.
\end{proof}


\subsection{Second order correlation estimates for high momenta}

In principle, the estimates in \eqref{eq:a-priori-estimate-1} hold for any ``good approximate minimizer" of the free energy. But to control the contribution from high momenta, we will need to use specific properties of the Gibbs state $\Gamma_\lambda$, which do not hold for any approximate minimizer of the free energy. The main result of this subsection is the following key correlation estimates, which allow us to reduce the quantum problem to a finite dimensional setting.

\begin{theorem}[Second order correlation estimates] \label{thm:correlation-Gamma} Let $N\in \mN$ and $\lambda^{-1/4} \ge N \ge \eps^{-8}$. Assume  $ \1(h\le (N+1)^2/C) \le  P_N \le  \1(h\le (N+1)^2)$ and $[P_N,h]=0$.  For all $k\in \bZ^3$, we denote $e_k=e_k^+ +e_k^-$ with $e_k^-= P_N e_k P_N$, where $e_k$ is the multiplication operator with $(2\pi)^{-3/2}e^{\imath k\cdot x}$. Then
	\begin{equs}
		[eq:variance bound]
		\lambda^2  \left\langle \left| \dG (e_k ^+) - \langle \dG (e_k ^+) \rangle_0  \right| ^2 \right\rangle_\lambda
		\lesim  \eps^{-2}N^{-1/4} + \eps^{-1}  (\lambda^{3/2}   |k|^2 + \lambda^{3/2}   |k| {N}) + \eps^{-5}\lambda,
	\end{equs} 
	and
	\begin{equs}
		[eq:variance bound-2]
		\lambda^2  \left\langle \left| \dG (e_k ^-) - \langle \dG (e_k ^-) \rangle_0  \right| ^2 \right\rangle_\lambda
		\lesim \eps^{-4} + \eps^{-5} \lambda^{3/2}   {N^2} + \eps^{-9}\lambda.
	\end{equs}
	Here $\langle\cdot \rangle_0$ and $\langle\cdot \rangle_\lambda$ denote the expectation against $\Gamma_\lambda$ and $\Gamma_0$ in the Fock space $\gF$, respectively.   
\end{theorem}

{{Note that here we do not assume that $P_N$ is a projection. This flexibility allows us to apply the bounds on the classical field theory established in Section \ref{sec:part}.}}

This result is an extension of \cite[Theorem 8.1]{LewNamRou-21} which is concerned with the case $\eps \ssim 1, P_N=\1(h\le (N+1)^2)$ and which  is the most challenging part in \cite{LewNamRou-21}. The bound \eqref{eq:variance bound-2} was not  included in \cite{LewNamRou-21} and we add it here to improve the condition on $v$ later. In the proof below, we will use the following recent result from \cite[Theorem 2]{DeuNamNap-25}, which is an improvement of \cite[Theorem 7.1]{LewNamRou-21}.

\begin{theorem}[Second order correlation inequality] \label{thm:correlation-intro} Let $A$ be a self-adjoint operator on a separable Hilbert space such that $\Tr [e^{-sA} ] <+\infty$ holds for all $s> 0$. Let $B$ be a symmetric operator such that $B$ is $A$-relatively bounded with a relative bound strictly smaller than $1$. We also assume that the perturbed Gibbs states
	\begin{equs}[eq:Gibbs-t-intro]
		G_t = \frac{\exp(-A+tB)}{\Tr [ \exp(-A+tB) ]} , \quad t\in [-1,1]
	\end{equs}
	satisfies
	\begin{equs}[eq:CRI-condition-intro]
		\sup_{t\in [-1,1]}|\Tr (B G_t )| \le a.
	\end{equs}
	Then we have
	\begin{equs}[eq:StahlA2-intro]
		\Tr [ B^2 G_0 ] \leq a e^{a} + \frac{1}{4} \Tr[ [B, [A, B]] G_0].
	\end{equs}
\end{theorem}

\begin{proof}[Proof of Theorem \ref{thm:correlation-Gamma}] Let us denote $P=P_N$ for short. Since $e_k=(2\pi)^{-3/2} e^{\imath k\cdot x}= (2\pi)^{-3/2} (\cos (k\cdot x) + \imath \sin (k\cdot x))$, it suffices to prove \eqref{eq:variance bound} with  $e_k^+$ replaced by $f_k^+= f_k  - P f_k P$ where $f_k(x) \in \{\cos (k\cdot x),\sin (k\cdot x)\}$ is real-valued. To apply Theorem \ref{thm:correlation-intro}, we consider the perturbed Gibbs states
	\begin{equs}[eq:Gibbs-t]
		\Gamma_{\lambda,t} = \cZ_{\lambda,t}^{-1} e^{-\bH_\lambda+t\mB} ,\quad \Gamma_{0,t} = \cZ_{0,t}^{-1} e^{-\lambda \dG(h)+t\mB} ,\quad  t\in [-1,1]
	\end{equs}
with
$$\cZ_{\lambda,t}=\tr[ e^{-\bH_\lambda+t\mB}] ,\quad \cZ_{0,t}=\tr[ e^{-\lambda \dG(h)+t\mB} ],$$
	$$\mB=\frac{1}{4}\lambda (  \d\Gamma(f_k^+) - \langle \dG (f_k ^+)\rangle_0 ). 
	$$
	Note that $k\ne 0$, $\langle \dG (f_k^+)\rangle_0=\langle \dG (f_k)\rangle_0 = \langle \dG (Pf_kP)\rangle_0=0$ due to the fact that $\Gamma_0$ preserves the total momentum (we have the same identity for $e_k$ and the function $f_k$ is a linear combination of $e_k$ and $e_{-k}$).

Note that the constant $\langle \dG (f_k ^+)\rangle_0$ in $\mB$ does not change the Gibbs states $\Gamma_{\lambda,t}$ and $\Gamma_{0,t}$, but it affects the partition functions $\cZ_{\lambda,t}$ and $\cZ_{0,t}$.  Hence, equivalently we can write
$$
		\Gamma_{\lambda,t} = e^{-\frac{1}{4}\lambda t \langle \dG (f_k ^+)\rangle_0}\cZ_{\lambda,t}^{-1} e^{-\lambda \dG(h_t) -\bW} ,\quad \Gamma_{0,t} = e^{-\frac{1}{4}\lambda t \langle \dG (f_k ^+)\rangle_0} \cZ_{0,t}^{-1} e^{-\lambda \dG(h_t)} ,\quad  t\in [-1,1],
$$
with
$$
h_t = h - \frac{t}{4} f_k^+  = -\Delta + 1 - \frac{t}{4} f_k^+, \quad t\in [-1,1].
$$

For all $t\in [-1,1]$, since $\|f_k^+\|\le 2$, we have $\|h_t-h\|\le 1/2$, and consequently $h \gesim h_t \gesim h$,  where and in the following $\|\cdot\|$ means the operator norm. In the following, we will use the notation $Q=\1-P$. Thus
\begin{equs}[eq:Q-properties]
 \text{} [Q,h]=0, \quad 0\le Q \le \1(h\ge (N+1)^2/C), \quad \| Q h^{-1}\|_{\rm HS}^2 \lesim  \sum_{ |k|^2 \ge (N+1)^2/C} \frac{1}{(| k|^2+1)^2} \lesim  N^{-1}.
\end{equs}

	{\bf Step 1.} Let us mimic the proof of Lemma \ref{lem:partition} to show that
	\begin{equs}[eq:HS-input-per]
		\lambda \Big\|\sqrt{h}\big(\Gamma_{\lambda,t}^{(1)}-\Gamma^{(1)}_{0,t}\big)\sqrt{h}\Big\|_{\rm HS} \lesim \eps^{-2}.
	\end{equs}
	Note that $\Gamma_{\lambda,t}$ is the unique minimizer for the following variational problem, which is similar to \eqref{eq:rel-energy},
\begin{equs}[eq:rel-energy-t]
	-\log \frac{\cZ_{\lambda,t}}{\cZ_{0,t}} = \min_{\substack{\Gamma\ge 0\\ \Tr \Gamma=1}} \Big\{ \cH(\Gamma,\Gamma_{0,t}) +  \Tr\left[\bW \Gamma\right] \Big\}.
\end{equs}
Using $\Gamma_{0,t}$ as a trial state for the variational principle of $\Gamma_{\lambda,t}$ similar to \eqref{eq:rel-energy}, we have  the upper bound
	\begin{equs}[eq:free-energy-simple-1a-per]
		-\log \frac{\cZ_{\lambda,t}}{\cZ_{0,t}} &\le  \frac{\lambda^2}{2(2\pi)^3}  {\rm Tr}[ (\cN-N_0)^2 \Gamma_{0,t}] + \frac{\lambda^2}{2} \sum_{k\ne 0}  \widehat{v^\eps}(k) {\rm Tr} \Big[ |\d\Gamma(e_k)| ^2 \Gamma_{0,t} \Big] \nn\\
		&\quad -\lambda \vartheta^\eps {\rm Tr}[ (\cN-N_0) \Gamma_{0,t}].
	\end{equs}
	The term $\vartheta^\eps \Tr[ (\cN-N_0) \Gamma_{0,t}]$ can be estimated by \cite[Lemma 5.4]{DeuNamNap-25}  as
	\begin{equs}[eq:N0t-N0]
		\lambda \left| {\rm Tr}[ \cN \Gamma_{0,t}] - {\rm Tr}[ \cN \Gamma_{0}] \right| &=\left|  {\rm Tr} \left[ \frac{\lambda}{e^{\lambda h_t}-1} - \frac{\lambda}{e^{\lambda h}-1} \right] \right| \nn\\
		&\lesim \| h_t- h \| ( {\rm Tr} [h^{-2}] + \Tr [h_t^{-2}] ) \lesim 1.
	\end{equs}
	(In \cite[Lemma 5.4]{DeuNamNap-25}, the trace is taken over the subspace $\{e_0\}^\bot$ of $L^2(\bT^3)$ since their lemma only assumes that $h,h_t\gesim -\Delta$. However, in our case here, we have $h,h_t\ge -\Delta + 1/2$, and hence the trace can be taken over the whole $L^2(\bT^3)$, by the same proof. ) Thus $-\lambda \vartheta^\eps \Tr[ (\cN-N_0) \Gamma_{0,t}]\lesim |\vartheta^\eps| \lesim \eps^{-1}$. We also have
	$$\lambda^2 \Tr[ (\cN- \Tr[ \cN \Gamma_{0,t}] )^2 \Gamma_{0,t}] \lesim  \Tr [ h_t^{-2}]  \lesim  1$$  similarly to \eqref{eq:free-energy-simple-1b}. Combining this with \eqref{eq:N0t-N0} and using the Cauchy--Schwarz inequality we can bound
	$$
	\lambda^2  \Tr[ (\cN- N_0 )^2 \Gamma_{0,t}] \le 2 \lambda^2  \Tr[ (\cN- \Tr[ \cN \Gamma_{0,t}]   )^2 \Gamma_{0,t}]   + 2\lambda^2 \left( \Tr[ \cN \Gamma_{0,t}] - \Tr[ \cN \Gamma_{0}] \right)^2 \lesim 1.
	$$
	Moreover, by using the operator inequality
	$$\lambda\Gamma_{0,t}^{(1)} =\frac{\lambda}{e^{\lambda h_t}-1} \le h_t^{-1} \le 2h^{-1},$$
	we can proceed as in \eqref{eq:free-energy-simple-1c} and obtain
	\begin{equs}[eq:free-energy-simple-1c-per]
		\frac{\lambda^2}{2}  \sum_{k\ne 0}  \widehat{v^\eps}(k) \Tr \Big[  |\d\Gamma(e_k)| ^2 \Gamma_{0,t}\Big] \lesim  \sum_{k\ne 0}  \widehat{v^\eps}(k) \Big( \Tr \Big[ e_k^* h^{-1} e_k h^{-1}  \Big]+\lambda^2\tr(\Gamma_0^{(1)})\Big) \lesim \eps^{-2}.
	\end{equs}
	Thus \eqref{eq:free-energy-simple-1a-per} gives the upper bound
	\begin{equs}[eq:free-energy-simple-1a-per-conclusion]
		-\log \frac{\cZ_{\lambda,t}}{\cZ_{0,t}} \lesim  \eps^{-2}.
	\end{equs}
	The lower bound $-\log \frac{\cZ_{\lambda,t}}{\cZ_{0,t}} \gesim - \eps^{-2}$ can be obtained exactly as in the proof of Lemma \ref{lem:partition}.  In this way, we also obtain the relative entropy estimate $\cH(\Gamma_{\lambda,t},\Gamma_{0,t})\lesim \eps^{-2}$, which implies
	\begin{equs}[eq:HS-input-per-0]
		\lambda \Big\|\sqrt{h_t}\big(\Gamma_{\lambda,t}^{(1)}-\Gamma^{(1)}_{0,t}\big)\sqrt{h_t}\Big\|_{\rm HS} \lesim \eps^{-2}
	\end{equs}
	by  \cite[Theorem 6.1]{LewNamRou-21}. Since $h_t \ge \frac{1}{2} h$, the operator $h^{1/2} h_t^{-1/2}$ is bounded, and hence we can replace $\sqrt{h_t}$ in \eqref{eq:HS-input-per-0} by $\sqrt{h}$ and obtain \eqref{eq:HS-input-per}.

	{\bf Step 2.} Using \eqref{eq:HS-input-per} and  \cite[Lemma 6.3]{LewNamRou-21}, we have
	\begin{equs}[eq:first-moment-input-new]
		|\Tr [\mB \Gamma_{\lambda,t}]| &= \frac14\lambda |\Tr [f_k^+ \Gamma_{\lambda,t}^{(1)}] - \Tr [f_k^+ \Gamma_{0}^{(1)}] | \le{\frac14} \lambda |\Tr [f_k^+ ( \Gamma_{\lambda,t}^{(1)} - \Gamma_{0,t}^{(1)} )  ] | + {\frac14}\lambda | \Tr [f_k^+ ( \Gamma_{0,t}^{(1)} - \Gamma_{0}^{(1)}] | \nn\\
		&\lesim \lambda \Big\|\sqrt{h}\big(\Gamma_{\lambda,t}^{(1)}-\Gamma^{(1)}_{0,t}\big)\sqrt{h}\Big\|_{\rm HS} \Big\| h^{-1/2} f_k^+ h^{-1/2} \Big\|_{\rm HS} + \Tr [f_k^+ h^{-1}f_k^+ h^{-1}]\nn\\
		&\lesim \eps^{-2} \| Q h^{-1}\|_{\rm HS}^{1/2} \lesim  \eps^{-2} N^{-1/4},
			\end{equs}
	for all $t\in [-1,1]$. Here we decomposed $f_k^+ = f_k - P f_k P = P f_k Q + Q f_k P + Q f_k Q$ and used \eqref{eq:Q-properties}.

	{\bf Step 3:} From the first moment estimate in Step 2, Theorem \ref{thm:correlation-intro} gives us
	\begin{equs}[eq:second-moment-apply-1]
		\Tr [\mB^2 \Gamma_\lambda] \lesim \eps^{-2} N^{-1/4} + \Tr \Big[ [\mB,[\bH_\lambda,\mB]] \Gamma_\lambda\Big].
	\end{equs}
In our choice of $N$ and $\eps$, we have $0<a = C \eps^{-2} N^{-1/4}\lesim 1$, and hence $a e^{a} \ssim a$.

	Next, let us estimate the double commutator
	$$[\mB,[\mB,\bH_\lambda]]=\frac{\lambda^3}{16} [[\dGamma(f_k^+), [\dGamma(f_k^+), \dGamma(h)]] +  \frac{\lambda^2}{16}  [[\dGamma(f_k^+), [\dGamma(f_k^+), \bW]].$$
	For the kinetic term, note that $[\d\Gamma(X),\d\Gamma(Y)]=\d\Gamma([X,Y])$ and
	\begin{equs}[eq:fk-fk-h]
		\pm [f_k^+, [f_k^+, h]] &=  \pm [ f_k - Pf_k P, [f_k,h] - P[f_k,h] P]\nn\\
		&\le \|[f_k, [f_k, h]] \| + 2 \| [f_k,h] Pf_k P\| +  2 \| f_k P[f_k,h] P\| + 2 \| Pf_k P^2 [f_k,h] P\| \nn\\
		& \lesim   |k|^2 + |k|N.
	\end{equs}
	Here we used $\|f_k\|\lesim 1$, $\|[f_k,[f_k,h]]\|\lesssim\| \nabla f_k\|_{L^\infty}^2\lesim |k|^2$, and $[f_k,h]=\Delta f_k + 2 \nabla f_k \cdot \nabla$ which gives
	$$\|[f_k,h]P\| \le \|\Delta f_k\|_{L^\infty} + 2 \|\nabla f_k \|_{L^\infty} \| \nabla P \| \lesim |k|^2 + |k|N.$$

	Therefore,
	\begin{equs}[eq:double-com-kinetic]
		\pm \lambda^3 \Tr \Big[ [[\dGamma(f_k^+), [\dGamma(f_k^+), \dGamma(h)]] \Gamma_\lambda\Big]  &\lesim \lambda^3 ( |k|^2 + |k|N) \Tr [\cN \Gamma_\lambda] \nn\\
		&\lesssim \eps^{-1} \lambda^{3/2}  ( |k|^2 + |k|N).
	\end{equs}
	Here we used $\Tr [\cN \Gamma_\lambda] \lesim \eps^{-1}  \lambda^{-3/2}$, which comes from the fact that
	$$\Tr [(\cN-N_0) \Gamma_\lambda] \le  \Tr [(\cN-N_0)^2 \Gamma_\lambda]^{1/2} \lesim \eps^{-1}\lambda^{-1} $$
	due to \eqref{eq:a-priori-estimate-1} and $N_0 \le \lambda^{-3/2}$.

	For the interaction term, using \eqref{eq:bW} we have
	\begin{align*}
		\pm \lambda^2  [[\dGamma(f_k^+), [\dGamma(f_k^+), \bW]] &= \pm  \frac{\lambda^4 }{2} \sum_{\ell\ne 0}  \widehat{v^\eps} (\ell) [[\dGamma(f_k^+), [\dGamma(f_k^+), |\dGamma(e_\ell)|^2]]   \nn\\
		&\lesim \lambda^4  \sum_{\ell\ne 0}  |\widehat{v^\eps} (\ell)| \cN^2 \lesim \lambda^4 \eps^{-3} \cN^2.
	\end{align*}
	Here for the first inequality we used again $[\d\Gamma(X),\d\Gamma(Y)]=\d\Gamma([X,Y])$ and the simple bound $\pm \dG(\mathcal{O}) \le \|O\| \cN$.

	Combining with the bound $\Tr [(\cN-N_0)^2 \Gamma_\lambda] \lesim \eps^{-2}\lambda^{-2}$ from  \eqref{eq:a-priori-estimate-1} we get
	\begin{equs}[eq:double-com-interaction]
		\pm \lambda^2 \Tr \Big[  \big[\dGamma(f_k^+), [\dGamma(f_k^+), \bW]\big] \Gamma_\lambda\Big] \lesssim \lambda \eps^{-5}.
	\end{equs}
	
	In summary, from \eqref{eq:second-moment-apply-1}, \eqref{eq:double-com-kinetic} and \eqref{eq:double-com-interaction} we have
	\begin{equs}[eq:second-moment-apply-2]
		\Tr [\mB^2 \Gamma_\lambda] \lesim  \eps^{-2}N^{-1/4} + \eps^{-1} \lambda^{3/2}    ( |k|^2 +|k|N) + \eps^{-5}\lambda,
	\end{equs}
	which implies the desired inequality \eqref{eq:variance bound}.
	
	For the rough bound \eqref{eq:variance bound-2}, we will use Theorem \ref{thm:correlation-intro} with
	$$A=\bH_\lambda,\quad B=\widetilde \mB=\eps^2 \lambda (  \d\Gamma(f_k^-)- \langle \dG (f_k ^-)\rangle_0 ). 
	$$
	Here $f_k^-=Pf_kP$.  Note that we have an extra factor $\eps^2$ in the definition of $\widetilde \mB$ to make sure that it is of order $1$. Then proceeding exactly as in \eqref{eq:HS-input-per} we also have
	\begin{equs}[eq:HS-input-per-t]
		\lambda \Big\|\sqrt{h}\big(\widetilde \Gamma_{\lambda,t}^{(1)}-\widetilde \Gamma^{(1)}_{0,t}\big)\sqrt{h}\Big\|_{\rm HS} \lesim \eps^{-2}
	\end{equs}
	with
	\begin{equs}[eq:Gibbs-t-t]
		\widetilde \Gamma_{0,t} = \frac{\exp(-\lambda \dG(h-{\eps^2 t}f_k^-) )}{\Tr [ \exp(-\lambda \dG(h- {\eps^2 t}f_k^-) ) ]}  ,\quad \widetilde \Gamma_{\lambda,t} = \frac{\exp(-\mH_\lambda+t\widetilde \mB)}{\Tr [ \exp(-\mH_\lambda+t\widetilde \mB) ]} , \quad t\in [-1,1].
	\end{equs}
	Now instead of \eqref{eq:first-moment-input-new} and \eqref{eq:double-com-kinetic}, we have
	\begin{equs}[eq:first-moment-input-new-t]
		|\Tr [\widetilde \mB \Gamma_{\lambda,t}]| &\lesim \eps^{2}\lambda \Big\|\sqrt{h}\big(\widetilde \Gamma_{\lambda,t}^{(1)}-\widetilde \Gamma^{(1)}_{0,t}\big)\sqrt{h}\Big\|_{\rm HS} \Big\| h^{-1/2} f_k^- h^{-1/2} \Big\|_{\rm HS} +\eps^{2}  \Tr [f_k^- h^{-1}f_k^- h^{-1}]\nn\\
		&\lesim  \| P h^{-1}\|_{\rm HS}   \lesim  1,
	\end{equs}
	and
	\begin{equs}[eq:double-com-kinetic-t]
		\pm \lambda^3 \Tr \Big[ [[\dGamma(f_k^-), [\dGamma(f_k^-), \dGamma(h)]] \Gamma_\lambda\Big]  &\lesim \lambda^3 N^2 \Tr [\cN \Gamma_\lambda] \lesssim \eps^{-1} N^2 \lambda^{3/2},
	\end{equs}
	since we can replace \eqref{eq:fk-fk-h} by
	\begin{equs}[eq:fk-fk-h-t]
		\pm [f_k^-, [f_k^-, h]] {\lesssim} \|f_k\|^2_{L^\infty} \|Ph\| \le { (1+N)^2}.
	\end{equs}
	The bound \eqref{eq:double-com-interaction} still holds true with $f_k^+$ replaced by $f_k^-$. Therefore, we have
	the following replacement for \eqref{eq:second-moment-apply-2}:
\begin{equs}[eq:second-moment-apply-2-t]
		\Tr [\widetilde \mB^2 \Gamma_\lambda] \lesim  1 + {\eps^{-1}} \lambda^{3/2}  { N^2} + {\eps^{-5}}\lambda,
	\end{equs}
	which implies \eqref{eq:variance bound-2}.
\end{proof}

\section{Quantum free energy}\label{sec:CV-energy}

Thanks to the estimates in the previous section, we can now prove give a rigorous comparison between the quantum free energy and its classical analogue.   
The main result of this section is the following.

\begin{theorem}[From quantum to classical free energy]\label{thm:local-W-P} Let $\eps \ge \lambda^{\eta}$,  and $\lambda^{-1/16} \ge N \ge \eps^{- 1/ (16\eta)}$  for a sufficiently small parameter $\eta>0$. Let $\Gamma_\lambda=\cZ_\lambda^{-1}e^{-\bH_\lambda}$ be defined  in \eqref{eq:Gibbs-state-def-new} and let $\Gamma_{0}=\cZ_0^{-1}e^{-\lambda \dG(h)}$ be the non-interacting Gibbs state.  Then for any self-adjoint operator $P_N$ satisfying $\1(h\le (N+1)^2/C) \le P_N \le \1(h\le (N+1)^2)$ and $[P_N,h]=0$, we have
	\begin{equs}[eq:partition-ul-1]
		-\log \frac{\cZ_\lambda}{\cZ_0} = -\log \left(  \int  e^{-\cD(P_N u)} \; \d\mu_{0}(u) \right) + O(\eps^{-8}{N^{-1/8}}) + O(\lambda^{\delta})
	\end{equs}
	for a constant $\delta>0$ depending only on $\delta_0>0$ in \eqref{eq:def-v}.
\end{theorem}

Our proof strategy closely follows that of \cite[Sections 9 and 10]{LewNamRou-21}: we restrict the variational problem \eqref{eq:rel-energy} to low momenta and then apply a semiclassical approximation. The key difference arises when handling the singular potential $v^\eps$: while a direct adaptation of  the semiclassical analysis in \cite{LewNamRou-21} would require the strong restriction $\eps \ge |\log \lambda|^{-\eta}$ for $0<\eta<1/2$ (see Remark \ref{rmk:eps-log-eta}), we will improve semiclassical estimates yielding error bounds valid for $\eps$ depending only polynomially on $\lambda$. In particular, we will add suitable mass terms to the quantum problem before comparing with the classical variational principle, which helps to simplify the analysis in the classical model,  and also derive a new pointwise inequality between the de Finetti measure $\mu^\lambda_{P,0}$ of the non-interacting Gibbs measure $\Gamma_0$ and the cylindrical projection $\mu_{0,P}$ of the Gaussian free field $\mu_0$, which allows a clean comparison to the classical partition function at the end.

First, we recall the localization method on Fock space and introduce the quantum de Finetti measure on the localized space in Section \ref{sec:localization}, which provides a natural link to classical field theory. Then, we discuss the free energy lower bound  in Section \ref{sec:low ener} and the free energy upper bound  in Section \ref{sec:upp ener}.

\subsection{Fock-space localization and de Finetti measure}\label{sec:localization}
Let us recall the standard {\em localization method} in Fock space.  Let $P$ be an orthogonal projection on $\gH$ and let $Q=\1-P$. We have the unitary equivalence
\begin{equs}[eq:factorization_Fock_space]
	\gF =\gF (\gH)=\gF (P\gH\oplus Q\gH )  \approx \gF(P\gH)\otimes\gF(Q\gH),
\end{equs}
namely there is a unitary
\begin{equs}[eq:Fock factor]
	\cU : \gF ( P\gH \oplus Q\gH) \mapsto \gF(P\gH) \otimes \gF (Q\gH)
\end{equs}
satisfying
\begin{equs}[eq:loc creation]
	\cU\cU^* = \1,\quad  \cU a ^*( f ) \cU ^*  = a ^*(Pf) \otimes \1 + \1 \otimes a^*(Q f)
\end{equs}
and a similar formula for annihilation operators. Consequently, for any state $\Gamma$ on $\gF$ and any orthogonal projector $P$, we define its \emph{localization} $\Gamma_{P}$ as a state on $\gF$ obtained by taking the partial trace over $\gF(Q\gH)$:
$$\Gamma_{P}\eqdef\tr_{\gF(Q\gH)}\left[ \cU \Gamma \cU^* \right].$$
The density matrices of $\Gamma_P$ can be shown to be equal to
\begin{equs}[eq:GammaV-k]
	(\Gamma_P)^{(k)}=P^{\otimes k}\Gamma^{(k)}P^{\otimes k},\qquad \forall k\geq1.
\end{equs}

Now we consider the finite dimensional space $P\gH$ and the associated Fock space $\gF (P\gH)$. Note that we have the resolution of identity
\begin{equs}[eq:resolution_coherent2]
	\1_{\gF(P\gH)} = \pi^{-\Tr(P)} \int_{P\gH} |W(u)\rangle \langle W(u)| \d u,
\end{equs}
where $\d u$ is the usual Lebesgue measure on $P\gH \approx \C^{\Tr(P)}$ and
\begin{equs}[eq:coherent state]
	W(u)\eqdef \exp(a^{*}(u)-a(u))  |0\rangle = e^{-\norm{u}^2/2} \exp\left(a^{*}(u)\right)  |0\rangle = e^{-\norm{u}^2/2} \bigoplus_{n=0}^\infty \frac{u^{\otimes n}} {\sqrt{n!}}
\end{equs}
is the coherent state on $\gF (P\gH)$, with $|0\rangle$ the vacuum in $\gF (P\gH)$ and $\|u\|$ the $L^2$-norm of $u\in P\gH$.

We have the following quantitative version of  the \emph{quantum de Finetti theorem}  \cite[Lemma 6.2 and Remark 6.4]{LewNamRou-15}.

\begin{theorem}[Quantitative quantum de Finetti]\label{thm:quant deF}
	For any state $\Gamma$ on $\gF$, using the coherent states in \eqref{eq:coherent state} we define the \emph{lower symbol} of $\Gamma$ on $P\gH$ at scale $\lambda$ by
	\begin{equs}[eq:Husimi]
		\d\mu_{P,\Gamma}^{\lambda}(u)\eqdef(\lambda\pi)^{-\Tr(P)}\Big\langle W(u/\sqrt{\lambda}),\Gamma_P W(u/\sqrt{\lambda})\Big\rangle_{\gF(P\gH)} \d u.
	\end{equs}
	Then for all $k\in\N$, we have
	\begin{equs}[eq:Chiribella]
		\int_{P\gH}|u^{\otimes k}\ra\la u^{\otimes k}|\;\d\mu^\lambda_{P,\Gamma}(u) = k!\lambda^k\Gamma^{(k)}_P + k! \lambda ^k \sum_{\ell = 0} ^{k-1} {k \choose \ell} \Gamma^{(\ell)} \otimes_s \one_{\otimes_s ^{k-\ell} P\gH}.
	\end{equs}
	Thus, with $d=\Tr [P]$,
	\begin{equs}[eq:quantitative]
		\Tr \left| k!\lambda^k\Gamma^{(k)}_P-\int_{P\gH}|u^{\otimes k}\ra\la u^{\otimes k}|\;\d\mu^\lambda_{P,\Gamma}(u) \right| \leq \lambda^k \sum_{\ell=0}^{k-1}{k\choose \ell}^2  \frac{(k-\ell +d-1)!}{(d-1)!}\tr \left[ \cN^{\ell}\Gamma_P\right].
	\end{equs}
\end{theorem}

The following Berezin-Lieb type inequality links the relative entropy of two quantum states to the classical entropy of their de Finetti measures \cite[Theorem 7.1]{LewNamRou-15}.

\begin{theorem}[Relative entropy: quantum to classical] \label{thm:rel-entropy}
	Let $\Gamma$ and $\Gamma'$ be two states on $\gF$. Let $\mu_{P,\Gamma}^\lambda$ and $\mu_{P,\Gamma'}^\lambda$ be the lower symbols defined in \eqref{eq:Husimi}.  Then we have
	\begin{equs}[eq:Berezin-Lieb]
		\cH(\Gamma,\Gamma')\geq \cH(\Gamma_P,\Gamma'_P)\geq \cHcl(\mu_{P,\Gamma}^\lambda,\mu_{P,\Gamma'}^\lambda).
	\end{equs}
\end{theorem}

\subsection{Free energy lower bound}\label{sec:low ener}

Now we are ready to relate the quantum free energy from the variational principle \eqref{eq:rel-energy} to the classical analogue in low momenta. Let $P_N$ as in Theorem \ref{thm:local-W-P}. In the following we will apply the quantum de Finetti Theorem \ref{thm:quant deF} with the projection $P = \1(h\le (N+1)^2)$. Note that $P = P_N P=P_N$.

By decomposing $e_k=e_k^+ +e_k^-$ with $e_k^- = P_N e_k P_N$, then by Theorem \ref{thm:correlation-Gamma} and the Cauchy--Schwarz inequality, we can bound
\begin{equs}[eq:fe-lb-2a]
	&\Big| \lambda^2 \left\langle \big|\dG(e_k) - \langle \dG(e_k) \rangle_0\big|^2 \right\rangle_{\lambda} -\lambda^2 \left\langle \big|\dG(e_k^-) - \langle \dG(e_k^-) \rangle_0\big|^2 \right\rangle_{\lambda} \Big| \nn \\
	&\lesim \lambda^2 \left\langle \big|\dG(e_k^+) - \langle \dG(e_k^+) \rangle_0\big|^2 \right\rangle_{\lambda} +\lambda^2 \left\langle \big|\dG(e_k^+) - \langle \dG(e_k^+) \rangle_0\big|^2 \right\rangle_{\lambda}^{\frac12}\left\langle \big|\dG(e_k^-) - \langle \dG(e_k^-) \rangle_0\big|^2 \right\rangle_{\lambda}^{\frac12} \nn\\
	&\lesim \eps^{-5}    (   N^{-1/4} + \lambda^{\frac32}   |k|^2 + \lambda^{\frac32}   |k|N  + \lambda) + \eps^{-5}    (   N^{-1/4} + \lambda^{\frac32}   |k|^2 + \lambda^{\frac32}   |k|N  + \lambda)^{\frac12},
\end{equs}
and
\begin{equs}[eq:fe-lb-2b]
	\lambda^2 \left\langle \big|\dG(e_k^-) - \langle \dG(e_k^-) \rangle_0\big|^2 \right\rangle_{\lambda} \lesim \eps^{-4} + \eps^{-{5}} \lambda^{3/2}   N^2 + {\eps^{-9}}\lambda \lesim \eps^{-5}.
\end{equs}
Summing over $k\in \bZ^3$ we have
\begin{equs}[eq:fe-lb-2]
	&\frac{\lambda^2}{2} \sum_{k\in \bZ^3}  \widehat{ v^\eps} (k) \left\langle \big|\dG(e_k) - \langle \dG(e_k) \rangle_0\big|^2 \right\rangle_{\lambda} - \frac{\lambda^2}{2} \sum_{k\in \bZ^3}  \widehat{ v^\eps} (k) \left\langle \big|\dG(e_k^-) - \langle \dG(e_k^-) \rangle_0\big|^2 \right\rangle_{\lambda}  \nn \\
	&\ge - \eps^{-5}  \sum_{k\in \bZ^3}  \widehat{v^\eps} (k) \min \Big\{ 1, (   N^{-1/4} + \lambda^{\frac32}   |k|^2 + \lambda^{\frac32}   |k|N  + \lambda)^{1/2} \Big\} \gesim   -C( \eps^{-8}{N^{-1/8} }+ \lambda^{\delta})
\end{equs}
with some small constant $\delta>0$ depending only on $\delta_0>0$ in \eqref{eq:def-v}.

Following the calculation from \cite[(9.7)-(9.11)]{LewNamRou-21}, we observe that
\begin{equs} [eq:ek-Gamma0]
	\lambda  \langle \dG( e_k^-)\rangle_0  &=\lambda  \Tr \left( e_k^- \Gamma_0^{(1)} \right) = \lambda \Tr   \left( e_k^- \frac{1}{e^{\lambda h}-1} \right)  = \Tr  \left(  e_k^- h^{-1} \right) +\lambda \Tr   \left( e_k^- \left(\frac{1}{e^{\lambda h}-1} -\frac{1}{\lambda h}\right) \right) \nn\\
	&= \int\left\langle u, e_k^-  u \right\rangle \d\mu_0(u) + O(\lambda N^{3}).
\end{equs}
Here we used $\| (e^{\lambda h}-1)^{-1} -(\lambda h)^{-1}\| \lesim 1$ and $K=\Tr P  \lesim N^{3}$. The bound \eqref{eq:ek-Gamma0} also shows that $\lambda  \Tr \left( e_k^- \Gamma_0^{(1)} \right) = O(N)$, which then implies that  $\lambda  \Tr \left( e_k^- \Gamma_\lambda^{(1)} \right) \lesim \eps^{-2}N^2$ 
by using \eqref{eq:HS-input}  as in \eqref{eq:first-moment-input-new}.
Combining with $(e_k^-)^*=P_N (e_k)^* P_N = P_N e_{-k} P_N = e^-_{-k}$ and $P_N P=P_N$, we find that
\begin{equs}[eq:part-local-W-1aaa]
	\text{} & \frac{\lambda^2}{2}\left\langle \left| \dG (e_k^-) - \left\langle \dG (e_k^-) \right\rangle_{0} \right| ^2 \right\rangle_\lambda
	\geq  \lambda^2 \Tr \left( e_{-k}^- \otimes {e_k^-}\Gamma_\lambda^{(2)} \right)
	- \lambda  \Re\Big[ \Tr \left( e_{-k}^- \Gamma_\lambda^{(1)} \right)  \int\left\langle u, e_k^-  u \right\rangle \d\mu_0(u) \Big] \nn\\
	& \quad + \frac{1}{2} \left|\int \left\langle u, e_k^-  u \right\rangle \d\mu_0(u)\right|^2 -C  \lambda {\eps^{-2}} N^5 \nn\\
	&=  \lambda^2 \Tr \left( P_N e_{-k}{P_N} \otimes P_N {e_k}{P_N} (\Gamma_\lambda)_P^{(2)} \right)
	- \lambda  \Re\Big[ \Tr \left( e_{-k}^-  (\Gamma_\lambda)_P^{(1)} \right)  \int\left\langle P_N u, e_k  P_N u \right\rangle \d\mu_0(u)\Big]  \nn\\
	& \quad +  \frac{1}{2} \left|\int  \left\langle P_N u, e_k P_N u \right\rangle \d\mu_0(u)\right|^2 -C  \lambda {\eps^{-2}} N^5.
\end{equs}
Inserting \eqref{eq:part-local-W-1aaa} in \eqref{eq:fe-lb-2} and using $\sum_k \widehat v^\eps(k) \lesim  \eps^{-3}$, we find that
\begin{align*}
	&\frac{\lambda^2}{2} \sum_{k\in \bZ^3}  \widehat{ v^\eps} (k) \left\langle \big|\dG(e_k)- \langle \dG(e_k) \rangle_0\big|^2 \right\rangle_{\lambda} \nn\\
	&\ge \lambda^2 \sum_{k\in \bZ^3} \widehat{v^\eps} (k)  \Tr \left( P_N e_{-k} P_N \otimes P_N {e_k}P_N(\Gamma_\lambda)_P^{(2)} \right)
	\\&-   \lambda  \sum_{k\in \bZ^3} \widehat{v^\eps} (k) \Re \Big[\Tr \left( e_{-k}^-  (\Gamma_\lambda)_P^{(1)} \right)  \int\left\langle P_N u, e_k  P_N u \right\rangle \d\mu_0(u)\Big] \\
	&\quad + \frac{1}{2}  \sum_{k\in \bZ^3}  \widehat{v^\eps} (k)  \left|\int  \left\langle P_N u, e_k  P_N u \right\rangle \d\mu_0(u)\right|^2 - C\Big(\lambda^\delta+\eps^{-8}N^{-1/8}\Big).	\end{align*}
Next, we represent $(\Gamma_\lambda)_P^{(1)}$ and $(\Gamma_\lambda)_P^{(2)}$ by de Finetti measure $\mu^\lambda_{P,\lambda}=\mu^\lambda_{P,\Gamma_\lambda}$. Here recall that $\Tr[(\Gamma_\lambda)_P]=1$ and
\begin{equs}[eq:mass-PGlambda]
	\text{} \Tr[(\Gamma_\lambda)_P^{(1)}]&= \Tr[P \Gamma_\lambda^{(1)} ] =   \Tr[P  \Gamma_0^{(1)}] + \Tr[P ( \Gamma_\lambda^{(1)} - \Gamma_0^{(1)}) P] \nn\\
	& \le \lambda^{-1} \Tr ( P h^{-1}) + \| P h^{-1/2}\|_{\rm HS}^2 \| \sqrt{h} ( \Gamma_\lambda^{(1)} - \Gamma_0^{(1)}) \sqrt{h}\|_{\rm HS} \lesssim \lambda^{-1}\eps^{-2}N.
\end{equs}
Here we used \eqref{eq:HS-input}  and $\| P h^{-1/2}\|_{\rm HS}^2= \Tr ( P h^{-1}) \lesssim N$ since $P=\1(h\le (N+1)^2)$.
Therefore, using  \eqref{eq:quantitative} with $k=1$ and $k=2$, we use $\Tr[P]\lesim N^3$ to have
\begin{equs}[eq:deF-app-k1]
	\Tr \left| \lambda (\Gamma_\lambda)_P^{(1)}-\int_{P\gH}|u\ra\la u|\;\d\mu^\lambda_{P,\lambda}(u) \right| &\lesim \lambda\tr[P] =\lambda N^3,
	\\
	\Tr \left| \lambda^2(\Gamma_\lambda)_P^{(2)}-\frac{1}{2}\int_{P\gH}|u^{\otimes 2}\ra\la u^{\otimes 2}|\;\d\mu^\lambda_{P,\lambda}(u) \right| &\lesim \lambda^2 \Big(\tr[P]^2 +\tr[P]\Tr [\cN (\Gamma_\lambda)_P]\Big)
	\\&\lesim \lambda^2(N^6+N^3\lambda^{-1}\eps^{-2}N)\lesim \lambda N^4\eps^{-2}. 
\end{equs}
Using also $\sum_k \widehat{v^\eps}(k) e_{-k}(x) e_k(y)= v^\eps(x-y)$, $\sum_k \widehat{v^\eps}(k) \lesim  \eps^{-3}$, $\Tr[P]\lesim N^3$ and $\Tr [P\otimes P]\lesim N^6$, we get
\begin{align*}
	&\frac{\lambda^2}{2} \sum_{k\in \bZ^3}  \widehat{ v^\eps} (k) \left\langle \big|\dG(e_k)- \langle \dG(e_k) \rangle_0\big|^2 \right\rangle_{\lambda} \\
	&\ge \frac{1}{2} \int_{P\gH} \sum_{k\in \bZ^3}  \widehat{ v^\eps} (k)  \langle P_N u, e_{-k} P_N u\rangle \langle P_N u, e_k P_N u\rangle \d\mu^\lambda_{P,\lambda}(u) \nn\\
	&\quad 		- \sum_{k\in \bZ^3}  \widehat{ v^\eps} (k)   \Re \Big[\int_{P\gH} \langle P_N u,  e_{-k} P_N u\rangle   \d\mu^\lambda_{P,\lambda}(u)    \int\left\langle P_N u, e_k  P_N u \right\rangle \d\mu_0(u)\Big]  \nn\\
	& \quad +  \frac{1}{2} \sum_{k\in \bZ^3}  \widehat{ v^\eps} (k)    \left|\int_{P\gH} \left\langle P_N u, e_k  P_N u \right\rangle \d\mu_0(u)\right|^2 - {C\lambda^\delta-C \eps^{-8} N^{-1/8}}  \nn\\
	&\ge \int_{P \gH} \left( \frac12\int :\!{|P_N u(x)|^2}\!: v^\eps(x-y) :\! {|P_N u(y)|^2}\!: \d x \d y \right) \d \mu^\lambda_{P,\lambda} - C\lambda^{\delta}{-C \eps^{-8} N^{-1/8}}.
\end{align*}
This gives us the first term in $\cD[P_N u]$ in \eqref{eq:cDu}. Concerning the contribution of $\lambda \vartheta^\eps(\cN-N_0)$,  we have
\begin{equs}
	\lambda | \Tr [(1-P_N^2) ( \Gamma_\lambda^{(1)} - \Gamma_0^{(1)})] | \le \lambda  \| \sqrt{h} ( \Gamma_\lambda^{(1)} - \Gamma_0^{(1)}) \sqrt{h}\|_{\rm HS} \|(1-P_N^2) h^{-1}\|_{\rm HS} \lesim \eps^{-2} N^{-1/2}
\end{equs}
by \eqref{eq:HS-input}.
Using \eqref{eq:a-priori-estimate-1}, $\lambda \mathfrak{e}_{\lambda} \Tr[(\cN-N_0)\Gamma_\lambda]= O(\eps^{-1}\lambda^{1/2})$ where $\mathfrak{e}_{\lambda}$ is defined in \eqref{eq:bW}. Hence,
\begin{equs}[eq:fe-lb-1]
	\vartheta^\eps \lambda  \left\langle ( \cN - N_0 )  \right\rangle_\lambda &\le (a^\eps - 6b^\eps +1 -m)  \lambda \Tr [{P_N^2} (\Gamma_\lambda)_P^{(1)} -{P_N^2} (\Gamma_0)_P^{(1)}] +  C\eps^{-3}N^{-1/2} \nn\\
	&\le  (a^\eps-6 b^\eps-m+1) \int \left( \int : |P_N u(x)|^2 : \d x \right) \d \mu^\lambda_{P,\lambda} (u) + C\lambda^{\delta} +  C\eps^{-3}N^{-\frac12} ,
\end{equs}
which gives the second term in $\cD[P_N u]$ in \eqref{eq:cDu}. 
Here we used again \eqref{eq:ek-Gamma0} and \eqref{eq:deF-app-k1}.

Thus in summary, for the interaction term we have
\begin{equs}[eq:local-W-P]
	\Tr [\bW\Gamma_\lambda] \ge \int_{P\gH} \cD [P_N u] \d\mu^\lambda_{P,\lambda}(u) -C\lambda^{\delta} - C\eps^{-8}N^{-1/8} , 
\end{equs}
with $\cD[u]$ in \eqref{eq:cDu}.

\begin{remark}\label{rmk:eps-log-eta} %
	Following the approach in \cite{LewNamRou-21}, from \eqref{eq:local-W-P}, the Berezin-Lieb inequality \eqref{eq:Berezin-Lieb}, and the classical variational principle \eqref{eq:zr-rel} we have that
	\begin{equs}[eq:fe-lb-trunc-1]
		\text{}&-\log \frac{\cZ_\lambda}{\cZ_0} = \cH(\Gamma_\lambda,\Gamma_0) +  \Tr [\bW\Gamma_\lambda]\\
		&\ge \cHcl(\mu^\lambda_{P,\lambda}, \mu^\lambda_{P,0})+\int_{P\gH}  \cD [P_N u] \d\mu^\lambda_{P,\lambda}(u) - C (\eps^{-8}N^{-1/8}+\lambda^{\delta})
		\nn\\
		&\ge -\log \left( \int_{P\gH} e^{-\cD[P_N u]} \; \d\mu^\lambda_{P,0}(u)\right) - C(\eps^{-8}N^{-1/8} + \lambda^{\delta})
	\end{equs}
	where $\mu^\lambda_{P,0}$ is the de Finetti measure of the non-interacting Gibbs measure $\Gamma_0$. The measure $\mu^\lambda_{P,0}$ is not the same with the the cylindrical projection $\mu_{0,P}$ of the Gaussian free field $\mu_0$, but if $\eps \ge |\log \lambda|^{-\eta}$ for a constant  $0<\eta<1/2$, then they can be compared using the simple bound for some $C>0$
	\begin{equs}[eq:e-cDu]
		0 \le  e^{-\cD[P_N u]}  \le e^{C \eps^{-2}}
	\end{equs}
	and the $L^1$-estimate $\norm{\mu_{P,0}^\lambda-\mu_{0,P}}_{L^1(P\gH)}\leq 2\tr[h^{-2}]\lambda  N^6$ from \cite[Lemma 9.3]{LewNamRou-21}. Thus in this case   \eqref{eq:fe-lb-trunc-1} gives
	\begin{equs}[eq:partition-lwb-1-easy]
		-\log \frac{\cZ_\lambda}{\cZ_0} &\ge -\log \left(  \int_{P\gH} e^{-\cD[P_N u]} \; \d\mu_{0,P}(u) + e^{C \eps^{-2}} \lambda  N^6 \right)- C
			{(\eps^{-8}N^{-1/8} + \lambda^{\delta})} \nn\\
		&\ge -\log \left(  \int_{P\gH} e^{-\cD[P_N u]} \; \d\mu_{0,P}(u) \right)  -  C{(\eps^{-8}N^{-1/8} + \lambda^{\delta})e^{C\eps^{-2}}}.
	\end{equs}
	In the last estimate we also used the lower bound
	\begin{equs}[eq:z-rough-lower]
		\int_{P\gH} e^{-\cD[P_Nu]} \; \d\mu_{0,P}(u) \gesim e^{-C\eps^{-2}},
	\end{equs}
	which follows from
	\begin{align*}\int \cD(u)\d\mu_{0,P}(u)=\frac12\int v^\eps(x-y) G_P(x-y)^2 \d x \d y\simeq \eps^{-2}.
	\end{align*}
	and Jensen's inequality. Here $G_P$ is the Green function projected on $P L^2(\mathbb{T}^3)$.
	However, the above approach is insufficient if we only assume that $\eps$ is bounded from below polynomially on $\lambda$. To deal with the more general case, we will use a new argument below.
\end{remark}

To avoid the problem in Remark \ref{rmk:eps-log-eta}, we will compare $\mu^\lambda_{P,0}$ and $\mu_{0,P}$ differently. Our new observation is  the following pointwise upper bound
\begin{equs}[eq:mu0P-pointwise]
	\mu^\lambda_{P,0} (u) \le e^{C\lambda N^4 \|u\|^2}\mu_{0,P} (u).
\end{equs}
To see this, let us compute $\mu^\lambda_{P,0} (u)$ directly using the definition of $\mu^\lambda_{P,0}$ in \eqref{eq:Husimi}, the factorized structure
$$
\Gamma_{0,P}=  \bigotimes_{j=1}^K \left( e^{-\lambda \lambda_j a^*(e_j) a(e_j)} (1-e^{-\lambda \lambda_j}) \right),
$$
with $K=\Tr P \sim N^3$ and  $h=\sum_j \lambda_j |e_j\rangle \langle e_j|$, and the explicit form of the coherent state in \eqref{eq:coherent state}. We have
\begin{align*}
	{\rm d}\mu^\lambda_{P,0} (u) &= (\lambda \pi)^{-K} \langle W(u/\sqrt{\lambda}), \Gamma_{0,P} W(u/\sqrt{\lambda})\rangle \d u\\
	&= \bigotimes_{j=1}^K  \left[ (\lambda \pi)^{-1} \langle W(\alpha_j e_j /\sqrt{\lambda}), \Gamma_{0,P} W(\alpha_j e_j /\sqrt{\lambda})\rangle \right] \d \alpha_j \\
	&= \bigotimes_{j=1}^K  \left[ (\lambda \pi)^{-1} (1-e^{-\lambda \lambda_j}){e^{ -\frac{|\alpha_j|^2}{\lambda}}} \sum_{n=0}^\infty  \left\langle \frac{(\alpha_j e_j/\sqrt{\lambda})^{\otimes n} }{\sqrt{n!}}, e^{-\lambda \lambda_j a^*(e_j) a(e_j)}   \frac{(\alpha_j e_j/\sqrt{\lambda})^{\otimes n} }{\sqrt{n!}} \right\rangle \right] \d \alpha_j \\
	&= \bigotimes_{j=1}^K  \left[ (\lambda \pi)^{-1} (1-e^{-\lambda \lambda_j}) e^{ -\frac{|\alpha_j|^2}{\lambda}}\sum_{n=0}^\infty \frac{|\alpha_j|^{2n}}{\lambda^n (n!)} e^{-\lambda \lambda_j n }\right] \d \alpha_j \\
	&= \bigotimes_{j=1}^K  \left[ (\lambda \pi)^{-1} (1-e^{-\lambda \lambda_j}) \exp \left( -\frac{|\alpha_j|^2}{\lambda} (1 - e^{-\lambda \lambda_j } ) \right)  \right] \d \alpha_j,
\end{align*}
with the complex variable $\alpha_j = \langle e_j, u\rangle$. Then using for some $C>0$
$$
\lambda \lambda_j \ge 1-e^{-\lambda \lambda_j} \ge \lambda \lambda_j - C (\lambda \lambda_j)^2,
$$
with $\lambda_j \le (N+1)^2  \ll \lambda^{-1}$, we obtain \eqref{eq:mu0P-pointwise}:
\begin{align*}
	{\rm d}\mu^\lambda_{P,0} (u) &\le  \bigotimes_{j=1}^K  \left[ \frac{\lambda_j}{\pi}  \exp \left( -\frac{|\alpha_j|^2}{\lambda} ( \lambda \lambda_j - C (\lambda \lambda_j)^2  ) \right)  \right] \d \alpha_j \\
	&= \bigotimes_{j=1}^K  \left[ \frac{\lambda_j}{\pi} e^{-\lambda_j |\alpha_j|^2} e^{C \lambda \lambda_j^2 |\alpha_j|^2}  \right] \d \alpha_j \leq e^{C \lambda N^4 \|u\|^2} \d\mu_{0,P}(u).
\end{align*}

From  the Berezin-Lieb inequality \eqref{eq:Berezin-Lieb} and the pointwise estimate \eqref{eq:mu0P-pointwise}, we can bound
\begin{equs}[eq:Brezin-Lieb-app-new]
	\text{} \cH(\Gamma_\lambda,\Gamma_0) & \ge \cH_{\rm cl}(\mu_{P,\lambda}^\lambda, \mu_{P,0}^\lambda) \ge \cH_{\rm cl}(\mu_{P,\lambda}^\lambda, \mu_{0,P}) - C \lambda \Lambda^2 \int \|u\|^2 \d\mu_{P,\lambda}^\lambda (u) \nn \\
	&\ge \cH_{\rm cl}(\mu_{P,\lambda}^\lambda, \mu_{0,P}) - C \lambda N^5 \eps^{-2}.
\end{equs}
Here in the last estimate we used
\begin{equs}[eq:local-W-P-mod]
	\text{}	 \int_{P\gH} \|u\|^2  \d\mu^\lambda_{P,\lambda}(u) \le C \lambda N^3 +    \lambda \Tr [ (\Gamma_\lambda)_P^{(1)}]   \lesssim  \eps^{-2} N,
\end{equs}
which follows from \eqref{eq:deF-app-k1} and \eqref{eq:mass-PGlambda}.

In summary, combining \eqref{eq:Brezin-Lieb-app-new}, \eqref{eq:local-W-P} we conclude that
\begin{equs}[eq:partition-lwb-1]
	\text{}&-\log \frac{\cZ_\lambda}{\cZ_0} = \cH(\Gamma_\lambda,\Gamma_0) +  \Tr [\bW\Gamma_\lambda]\\
	&\ge \cHcl(\mu^\lambda_{P,\lambda}, \mu_{0,P})+\int_{P\gH} \cD [P_N u] \d\mu^\lambda_{P,\lambda}(u) - C (\eps^{-8}N^{-1/8}+\lambda^{\delta})
	\nn\\
	&=  \cHcl(\mu^\lambda_{P,\lambda}, \tilde{\mu}^{\eps}_N) -\log \left( \int_{P\gH} e^{-\cD[P_N u] }  \; \d\mu_{0,P}(u)\right) - C(\eps^{-8}N^{-1/8} + \lambda^{\delta})
\end{equs}
with $\tilde{\mu}^{\eps}_N=\mu_N^\eps\circ P^{-1}$ for the measure $\mu^\eps_N$ defined in \eqref{eq:def-mu-eps-N}. By dropping $\cHcl(\mu^\lambda_{P,\lambda}, \tilde\mu^{\eps}_N) \ge 0$ and using $\int_{P\gH} e^{-\cD[P_N u] }  \; \d\mu_{0,P}(u)=\int e^{-\cD[P_N u] }  \; \d\mu_{0}(u)$, we obtain the desired free energy lower bound.

\subsection{Free energy upper bound}\label{sec:upp ener}

We use again the variational principle \eqref{eq:rel-energy}, and denote
$P = \1(h\le \Lambda),\quad Q=\1-P.$  Following the strategy in  \cite[Proposition 10.1]{LewNamRou-21}, we use the unitary $\cU$ in \eqref{eq:Fock factor} to define the trial state
\begin{equs}[eq:trial state]
\Gammat= \cU^* \Big(  \Gamma_{\lambda,P} \otimes (\Gamma_0)_Q \Big) \cU,
\end{equs}
where $(\Gamma_0)_Q$ is the $Q$-localization of the Gaussian state $\Gamma_0$, and $\Gamma_{\lambda,P}$ is the interacting Gibbs state in $\gF(P\gH)$:
\begin{equs}[eq:Gibbs localized]
\Gamma_{\lambda,P}= \frac{e^{-\lambda  \dGamma(Ph)-  \bW_P}}{\Tr_{\gF(P\gH)} e^{-\lambda  \dGamma(Ph)- \bW_P}}
\end{equs}
with $\bW_P$ the localized interaction
\begin{equs}[eq:WP]
\bW_P = \frac {\lambda^2} 2 \sum_k \widehat{v^\eps} (k) \left|\dG (e_k^-) - \left\langle \dG (e_k^-) \right\rangle_0 \right| ^2  - \lambda \vartheta^\eps (\dG({{P_N^2}}) -\langle \dG({{P_N^2}})\rangle_0) - \lambda^{4/3} \dG(P), 
\end{equs}
where $e_k^-=P_N e_k P_N$. 
Note that the expectation $\langle\dG (P_N e_k P_N ) \rangle_0$ in $\Gamma_0$ is the same as that in $(\Gamma_0)_P$. In general, $\Gamma_{\lambda,P}$ is different from the state $(\Gamma_\lambda)_P$ obtained by $P$-localizing the full interacting Gibbs state $\Gamma_\lambda$. Moreover, unlike \cite{LewNamRou-21}, here we also include the mass term $\lambda^{4/3} \dG(P)$, which helps to have a better control of error terms arising from the semiclassical approximation.

Under the definition, we can compute explicitly (c.f. \cite[Eq. (10.6)-(10.7)]{LewNamRou-21})
\begin{equs}[eq:trial one body]
\Gammat ^{(1)} = P \Gamma_{\lambda,P} ^{(1)} P + Q \Gamma_0 ^{(1)} Q,
\end{equs}
and
\begin{equs}[eq:trial two body]
\Gammat ^{(2)} = P^{\otimes 2} \Gamma_{\lambda,P} ^{(2)} P^{\otimes 2} + Q^{\otimes 2} \Gamma_0 ^{(2)} Q ^{\otimes 2} + \left( \Gamma_{\lambda,P} ^{(1)} \otimes Q \Gamma_0 ^{(1)} Q + Q \Gamma_0 ^{(1)} Q \otimes \Gamma_{\lambda,P} ^{(1)} \right).
\end{equs}
On the other hand, since the trial state in \eqref{eq:trial state} is factorized (up to the unitary $\cU$) and  the Gaussian state $\Gamma_0$ also satisfies $\Gamma_0=\cU^* \Big(  (\Gamma_{0})_P \otimes (\Gamma_0)_Q \Big) \cU$, we have
\begin{align*}
\cH(\Gammat, \Gamma_0) &=\cH(\Gamma_{\lambda,P} \otimes (\Gamma_0)_Q, (\Gamma_0)_P \otimes (\Gamma_0)_Q)= \cH(\Gamma_{\lambda,P}, (\Gamma_0)_P). 
\end{align*}
Hence, by the variational principle~\eqref{eq:rel-energy}
\begin{align*}
-\log \frac{\cZl}{\cZ_0} \le \cH(\Gammat, \Gamma_0) +  \Tr [\bW \Gammat] = \cH(\Gamma_{\lambda,P}, (\Gamma_0)_P) +  \Tr [\bW \Gammat].
\end{align*}

Note that the localized Gibbs state $\Gamma_{\lambda,P}$ satisfies the same bounds of the full Gibbs state $\Gamma_\lambda$ as in Section \ref{sec:gibbs}. In particular, the analysis remains unchanged when we replace the original one-body Hilbert space $\gH$ by $P\gH$, and also replace the kinetic operator $h$ by $h+\lambda^{1/3}P$ (the latter corresponding to the additional mass term $\lambda^{4/3} \d\Gamma(P)$ in $\bW_P$. {{Therefore, if we decomposed $e_k=e_k^+ +e_k^-$ with $e_k^- = P_N e_k P_N$, then we have \begin{equs}
[eq:variance bound-upper]
\lambda^2  \left\langle \left| \dG (e_k ^+) - \langle \dG (e_k ^+) \rangle_0  \right| ^2 \right\rangle_{\Gammat}
\lesim  \eps^{-2}N^{-1/4} + \eps^{-1}  (\lambda^{3/2}   |k|^2 + \lambda^{3/2}   |k| {N}) + \eps^{-5}\lambda,
\end{equs} 
and
\begin{equs}
[eq:variance bound-2-upper]
\lambda^2  \left\langle \left| \dG (e_k ^-) - \langle \dG (e_k ^-) \rangle_0  \right| ^2 \right\rangle_{\Gammat} 
\lesim \eps^{-4} + \eps^{-5} \lambda^{3/2}   {N^2} + \eps^{-9}\lambda  \lesim \eps^{-5},
\end{equs}
which are similar to \eqref{eq:variance bound} and \eqref{eq:variance bound-2}. From \eqref{eq:variance bound-upper} and \eqref{eq:variance bound-2}, we can argue as in \eqref{eq:fe-lb-2a} to deduce that 	
\begin{equs}[eq:fe-lb-2a-upper]
&\Big| \lambda^2 \left\langle \big|\dG(e_k) - \langle \dG(e_k) \rangle_0\big|^2 \right\rangle_{\Gammat} -\lambda^2 \left\langle \big|\dG(e_k^-) - \langle \dG(e_k^-) \rangle_0\big|^2 \right\rangle_{\Gammat} \Big| \nn \\
&\lesim \eps^{-5}    (   N^{-1/4} + \lambda^{\frac32}   |k|^2 + \lambda^{\frac32}   |k|N  + \lambda) + \eps^{-5}    (   N^{-1/4} + \lambda^{\frac32}   |k|^2 + \lambda^{\frac32}   |k|N  + \lambda)^{\frac12}
\end{equs}
for all $k\in \mathbb{Z}^3$. For the free energy upper bound, we also need the following variant of  \eqref{eq:variance bound-2-upper}: 
\begin{equs}[eq:fe-lb-2b-upper-new]
\lambda^2 \left\langle \big|\dG(e_k) - \langle \dG(e_k) \rangle_0\big|^2 \right\rangle_{\Gammat} \lesim \eps^{-4} + \eps^{-{5}} \lambda^{3/2}   N^2 + {\eps^{-9}}\lambda \lesim \eps^{-5}
\end{equs}
where the right-hand side is uniform in $|k|$. From \eqref{eq:variance bound-upper} and \eqref{eq:variance bound-2-upper}, it suffices to prove \eqref{eq:fe-lb-2b-upper-new} when $k\ne 0$. In this case, $\langle \dG (e_k) \rangle_0=0$ and by \eqref{eq:trial one body} and  \eqref{eq:trial two body} we have 
\begin{align*}
\lambda^2 \left\langle  |\dG (e_k)|^2   \right\rangle_{\Gammat} &= \lambda^2 \Tr ({e_{-k}} \otimes e_k \Gammat^{(2)} ) + \lambda^2  \Tr (|e_k|^2 \Gammat^{(1)}) \\
&= \lambda^2 \Tr (P^{\otimes 2} ({e_{-k}} \otimes e_k) P^{\otimes 2}  \Gamma_{\lambda,P}^{(2)})+ \lambda^2  \Tr (|e_k|^2 \Gammat^{(1)}) \\
&=\lambda^2 \left\langle  |\dG (P e_k P)|^2   \right\rangle_{ \Gamma_{\lambda,P}} +  \lambda^2  \Tr (|e_k|^2 \Gammat^{(1)})  -  \lambda^2  \Tr (|e_k|^2\Gamma_{\lambda,P} ^{(1)}) \lesssim \eps^{-5}.
\end{align*}
Here we used $\tr(e_kQ\Gamma_0^{(1)}Q)=0$ in the second equality and in the last inequality we used 
$$\lambda^2 \left\langle  |\dG (P e_k P)|^2   \right\rangle_{ \Gamma_{\lambda,P}} \lesssim \eps^{-5},$$ which can be proved similarly to  \eqref{eq:variance bound-2-upper}. In summary,   \eqref{eq:fe-lb-2b-upper-new} holds uniformly in $k$. 
Then combining \eqref{eq:fe-lb-2a-upper} and \eqref{eq:fe-lb-2b-upper-new}, we can proceed similarly to \eqref{eq:fe-lb-2} and find that  
\begin{equs}[eq:fe-lb-2-upper]
&	\frac{\lambda^2}{2} \sum_{k\in \bZ^3}  \widehat{ v^\eps} (k) \left\langle \big|\dG(e_k) - \langle \dG(e_k) \rangle_0\big|^2 \right\rangle_{\Gammat} \nn\\
&\le  \frac{\lambda^2}{2} \sum_{k\in \bZ^3}  \widehat{ v^\eps} (k) \left\langle \big|\dG(e_k^-) - \langle \dG(e_k^-) \rangle_0\big|^2 \right\rangle_{\Gamma_{\lambda,P}} + C( \eps^{-8}{N^{-1/8} }+ \lambda^{\delta}).
\end{equs}

Moreover, we have the following analogue of \eqref{eq:HS-input}: 
\begin{equs}[eq:HS-input-upper]
\lambda \Big\|P\sqrt{h}\big(\Gamma_{\lambda,P}^{(1)}-\Gamma^{(1)}_{0,P}\big)\sqrt{h}P\Big\|_{\rm HS} \lesim \eps^{-2} 
\end{equs}
where $\Gamma_{0,P}$ is the non-interacting Gibbs state on $\gF(P\gH)$, with $\Gamma^{(1)}_{0,P}=P \Gamma^{(1)}_{0}P$. Consequently,
\begin{align*}
& \lambda  \Big|  \Tr ( (P-P_N^2) (\Gamma_{\lambda,P}^{(1)}-\Gamma^{(1)}_{0,P}) ) \Big|  =  \lambda   \Big| \Tr ( (P-P_N^2) h^{-1} h^{1/2} (\Gamma_{\lambda,P}^{(1)}-\Gamma^{(1)}_{0,P}) h^{1/2}   ) \Big| \\
&\le \lambda   \Big\| (P-P_N^2) h^{-1} \Big\|_{\rm HS} \Big\|P\sqrt{h}\big(\Gamma_{\lambda,P}^{(1)}-\Gamma^{(1)}_{0,P}\big)\sqrt{h}P\Big\|_{\rm HS}  \lesssim  N^{-1}\eps^{-2}.  
\end{align*}
Combining with  \eqref{eq:trial one body}, we obtain 
\begin{align*}
&\lambda \vartheta^\eps   \left\langle ( \cN - \langle \cN\rangle_0 )  \right\rangle_{\Gammat} = \lambda \vartheta^\eps \Tr(\Gamma_{\Gammat}^{(1)} - \Gamma^{(1)}_{0} ) \\
&=  \lambda \vartheta^\eps \Tr ( \Gamma_{\lambda,P}^{(1)}-\Gamma^{(1)}_{0,P}))  = \lambda \vartheta^\eps \Tr (P_N^2 (\Gamma_{\lambda,P}^{(1)}-\Gamma^{(1)}_{0,P})) + O(N^{-1}\eps^{-3})
\end{align*}
which is an analogue of \eqref{eq:fe-lb-1}. Putting the latter bounds together with \eqref{eq:fe-lb-2-upper} we conclude that 
\begin{equs}[eq:local-W-P-2-PP]
\Tr [\bW \Gammat]  \le  \Tr [\bW_P \Gamma_{\lambda,P}]+ C (\eps^{-8}N^{-1/8} + \lambda^\delta).
\end{equs}
}}
Here the last term on the right-hand side of \eqref{eq:local-W-P-2-PP} comes from the mass term $\lambda^{4/3} \Tr[\d\Gamma(P)\Gamma_{\lambda,P}] \lesssim \lambda^{1/3} \eps^{-2} N^{-1}$ which follows from an argument similar to \eqref{eq:mass-PGlambda}. 	Consequently,
\begin{equs}[eq:rel-fe-up-P]
-\log \frac{\cZl}{\cZ_0} &\le  \cH(\Gamma_{\lambda,P}, (\Gamma_0)_P) + \Tr [\bW_P \Gamma_{\lambda,P}] + C (\eps^{-8}N^{-1/8}+\lambda^\delta)  \nn\\
&=- \log \frac{\Tr e^{- \lambda \dGamma(Ph)- \bW_P} } { \Tr e^{-\lambda \dGamma(Ph)} } + C (\eps^{-8}N^{-1/8}+\lambda^\delta).
\end{equs}
Here in the last equality we used the variational principle due to the choice of $\Gamma_{\lambda,P}$. 

To compute the relevant partition functions on the right hand side of \eqref{eq:rel-fe-up-P}, we use  the coherent-state resolution of the identity~\eqref{eq:resolution_coherent2}:
$$
\1_{\gF(P\gH)}= (\lambda\pi)^{-K} \int_{P\gH} \left|W(u/\sqrt{\lambda})\right\rangle \left\langle W(u/\sqrt{\lambda})\right| \d u,
$$
Here we rescaled \eqref{eq:resolution_coherent2} with $u\mapsto u \lambda^{-1/2}$, and denote $K=\dim P\gH=\Tr P \lesim N^3$. By  the Peierls-Bogoliubov inequality $\la x,e^Ax\ra\geq e^{\la x,Ax\ra}$ we obtain
\begin{equs}[eq:partition_int-PPPP]
\Tr e^{-\lambda \dGamma(Ph)- \bW_P} &= \frac1{(\lambda\pi)^{K}} \int_{P\gH} \Tr \left[ e^{-\lambda\dGamma(Ph) - \bW_P} \left|W\left(u/\sqrt{\lambda}\right)\right\rangle \left\langle W\left(u/\sqrt{\lambda}\right)\right|   \right] \d u \nn\\
&= \frac1{(\lambda\pi)^{K}}  \int_{P\gH}  \left\langle W\left(u/\sqrt{\lambda}\right),  e^{-\lambda\dGamma(Ph) - \bW_P}  W\left(u/\sqrt{\lambda}\right) \right\rangle \d u \nn\\
&\geq \frac1{(\lambda\pi)^{K}} \int_{P\gH}    \exp \left[ - { \left\langle  W\left(u/\sqrt{\lambda}\right), \lambda \left(\dGamma(Ph)+ \bW_P\right) W\left(u/\sqrt{\lambda}\right) \right\rangle}  \right] \d u.
\end{equs}
Then, for $u\in P\gH$, we have the identities
$$
\lambda\left\langle W\left(u/\sqrt{\lambda}\right), \dGamma(Ph) W\left(u/\sqrt{\lambda}\right) \right\rangle = \langle u,h u\rangle,\quad \lambda^{4/3} \left\langle W\left(u/\sqrt{\lambda}\right), \dGamma(P) W\left(u/\sqrt{\lambda}\right) \right\rangle = \lambda^{1/3} \|u\|^2
$$
and
\begin{equs}[eq:part-WP]
&\lambda^2\left\langle W\left(u/{\lambda}^{1/2}\right),  \left| \dG (e_k^-) - \left\langle \dG (e_k^-) \right\rangle_{0} \right| ^2 W\left(u/{\lambda}^{1/2}\right) \right\rangle \nn\\
&= \Big|\langle u, e_k^- u\rangle\Big|^2  -2 \Re\Big[ \lambda \langle u, e_k^- u\rangle {\rm Tr} \left[ \overline{ e_k^{-} } \Gamma_0^{(1)} \right]\Big]  + \lambda^2  \left|{\rm Tr} \left[ e_k^-\Gamma_0^{(1)} \right] \right|^2 + \lambda \| P(e_ku)\|^2  \nn\\
&\le  \Big| \langle u, e_k^- u\rangle - \int\langle u, e_k^- u\rangle\d \mu_0(u) \Big|^2 +C \|u\|^2 \lambda N^3+C\lambda^2N^6.
\end{equs}
The latter bound is an analogue of \eqref{eq:part-local-W-1aaa} and is based on \eqref{eq:ek-Gamma0}.  We  also have
\begin{equs}[eq:part-WP-taueps]
&\lambda \vartheta^\eps \left\langle W\left(u/\lambda^{1/2}\right),  ( \dG (P_N^2) - \left\langle \dG (P^2_N) \right\rangle_{0} )  W\left(u/\lambda^{1/2}\right) \right\rangle \nn\\
&= (a^\eps-6b^\eps-m+1) ( \langle u, P_N^2 u\rangle - \left\langle\langle u, P_N^2 u\rangle\right\rangle_{\mu_0} ) +O (\lambda^{1/2}) \|u\|^2+O(\lambda N^3 \eps^{-1}),
\end{equs}
where the error term comes from the constant $\mathfrak{e}_\lambda=O(\lambda^{1/2})$ in $\vartheta^\eps$.

Summing over $k$ and using $\sum_k |\widehat{ v^\eps}(k)|\lesssim \eps^{-3}$, we find that
\begin{equs}
\left\langle W\left(u/\sqrt{\lambda}\right),  \bW_P W\left(u/\sqrt{\lambda}\right) \right\rangle &\le  \cD[P_N u] -\|u\|^2 (\lambda^{1/3} -  C\lambda N^3\eps^{-3} - C \lambda^{1/2})+C\lambda N^3\eps^{-1} \\&\le  \cD[P_N u]+C\lambda N^{3} \eps^{-1}.
\end{equs}

Inserting the latter bound in~\eqref{eq:partition_int-PPPP} we arrive at
\begin{align} \label{eq:partition_int-PPPPP}
\Tr e^{-\lambda\dGamma(Ph) - \bW_P} \ge (\lambda\pi)^{-K}  \int_{P\gH} \exp \left[- \langle u, hu\rangle -  \cD [P_N u]-C\lambda N^3\eps^{-1} \right]  \d u.
\end{align}
Combining with~ the explicit computation for $\Tr e^{-\lambda\dGamma(Ph)}$ in \cite[(10.12)]{LewNamRou-21}, we find
\begin{equs}[eq:rel-partition-PPPP]
\frac{\Tr e^{-\lambda\dGamma(Ph) - \bW_P} } { \Tr e^{-\lambda\dGamma(Ph)} } &\ge \left[  \prod_{j=1}^K \frac{1}{\lambda\lambda_j} (1-e^{-\lambda\lambda_j}) \right]  \int_{P\gH} \exp \left[- \cD [P_N u] -C\lambda N^3 \eps^{-1} \right]  \d\mu_{0,P}(u),
\end{equs}
where $\d\mu_{0,P}$ is the cylindrical projection of $\d\mu_0$ on $P\gH$ and $0<\lambda_1\le \lambda_2\le...$ are eigenvalues of $h=-\Delta+1$. Thus we deduce from \eqref{eq:rel-fe-up-P} and \eqref{eq:rel-partition-PPPP} that
\begin{equs}[eq:rel-fe-up-P-conclusion]
-\log \frac{\cZl}{\cZ_0} \le - \log \int_{P\gH} e^{- \cD [P_N u]}  \d\mu_{0,P}(u) - \log \left( \prod_{j=1}^K \frac{1}{\lambda\lambda_j} (1-e^{-\lambda\lambda_j}) \right) + C (\eps^{-8}N^{-\frac18}+\lambda^\delta).
\end{equs}

Finally, as already observed in \cite{LewNamRou-21}, we have by Bernoulli's inequality
\begin{equs}[eq:Bernoulli]
\prod_{j=1}^K \left[  \frac{1}{\lambda\lambda_j} (1-e^{-\lambda\lambda_j}) \right] \ge  \prod_{j=1}^K \left[ 1- \frac{\lambda \lambda_j}{2}\right] \ge 1- \sum_{j=1}^K  \frac{\lambda \lambda_j}{2} \ge 1 - C\lambda N^5,
\end{equs}
since $K=\Tr P \le N^3$ and $\lambda_j \simeq j^{2/3}$. Therefore, from \eqref{eq:rel-fe-up-P-conclusion} we conclude the upper bound
\begin{equs}[eq:rel-fe-up-P-upper-final]
-\log \frac{\cZl}{\cZ_0} \le - \log \left(  \int_{\gH} e^{- \cD [P_N u] }   \d\mu_{0}(u) \right)  +  C (\eps^{-8}N^{-1/8}+\lambda^\delta).
\end{equs}
The proof of Theorem \ref{thm:local-W-P} is complete.

\section{Higher-order correlation estimates for quantum Gibbs state} \label{sec:gibbs-higher}

The following bounds will serve as the important input the obtain the convergence of the correlation functions in the next subsection.

\begin{theorem}[Higher-order correlation estimates]\label{thm:higher-moment} Let $\eps \ge \lambda^{\eta}$ for a sufficiently small parameter $\eta>0$. Let $\Gamma_\lambda=\cZ_\lambda^{-1}e^{-\bH_\lambda}$ be defined  in \eqref{eq:Gibbs-state-def-new}. Then for every fixed finite-dimensional projection $\widetilde P$ with eigenfunctions in $\{e_k\}$, we have
	\begin{equs}[eq:moment-m]
		\lambda^k \Tr [\dG(\widetilde P)^k \Gamma_\lambda] \le C_{\widetilde P,k} \eps^{-6}.
	\end{equs}
	for all $k\le 7$. Moreover, we have
	\begin{equs}[eq:moment-m-weak]
		\lambda^k \Tr [\dG(\widetilde P)^k \Gamma_\lambda] \le C_{\widetilde P,k} e^{C\eps^{-2}}
	\end{equs}
	for all $k\ge 1$.
\end{theorem}


We will use the following recent result from \cite[Theorem 3]{DeuNamNap-25}.

\begin{theorem}[Higher-order correlation inequality] \label{thm:correlation-intro-k} Let $A,B$ and $G_t$ as in Theorem \ref{thm:correlation-intro}.
	We assume in addition that $B\ge 0$ and there exists a self-adjoint operator $X\ge 1$ such that $X$ is $A$-relatively bounded and
	\begin{equs}[eq:corr-thm-ass2]
		\text{} [A,X]=[B,X]=0, \quad \pm B\le X,\quad \pm [B,[B,A]] \le b X^{\alpha}
	\end{equs}
	hold with some constants $b > 0$ and $\alpha \in \mathbb{R}$. Then for all even $k\in \mathbb{N}$, we have
	\begin{equs}[eq:corr-thm-conclusion]
		\Tr [B^k e^{-A}] \lesssim_{k}   \sup_{t\in [-1,1]}  \Tr[ ( 1+b^2 X^{2\alpha+k-3}) e^{-A+tB}].
	\end{equs}
\end{theorem}
The bound \eqref{eq:moment-m-weak} is weaker than \eqref{eq:moment-m} for $k\le 7$, but it is helpful for $k>7$ provided that $ \eps \ge |\log \lambda|^{-\eta}$ for a constant  $0<\eta<1/2$. This bound follows from another argument in \cite{LewNamRou-21}, which is based on the Feynman--Kac formula for $\Gamma_\lambda$ and the pointwise bound $\mathbb{W}\ge - C \eps^{-2}$.

\begin{proof}[Proof of Theorem \ref{thm:higher-moment}] We will apply Theorem \ref{thm:correlation-intro-k} to $A=\bH_\lambda$ and $B=\frac{1}{2}\lambda \dG(\widetilde P)$, namely we consider
	$$
	G_t = \frac{1}{Z_t} \exp(-\bH_\lambda+tB), \quad Z_t =\Tr [ \exp(-\bH_\lambda+tB) ] , \quad t\in [-3/2,3/2].
	$$
	
	First, let us verify the condition $\sup_{t\in [-1,1]}|\Tr (B G_t )| \lesssim 1$ from \eqref{eq:CRI-condition-intro}. We will extend the analysis in Section \ref{sec:CV-energy} to include the perturbation $B$ and combine with the uniform bound \eqref{mom:mutildeP-revised} from Theorem \ref{th:partition-revised}. In the following we choose $N$ large enough such that $\tilde P P_N=\tilde P$. 
	
	{{
			Note that for every $t \in [-3/2,3/2]$, the Gibbs state $G_t$ is similar to $\Gamma_\lambda$, except that the kinetic operator $h$ in $\Gamma_\lambda$ is replaced by $h_t = h + \frac{t}{2} \widetilde{P}$ which satisfies $2h\ge h_t \ge h/4$.  The non-interacting Gibbs state $\exp(- \lambda \dG(h_t)/\Tr [ \exp(- \lambda \dG(h_t)]$ can be treated similarly to $\Gamma_0$; in particular all partition functions $\Tr [ \exp(- \lambda \dG(h_t) )$ are comparable for $t\in [-3/2,3/2]$ since 
			\begin{equs}[eq:partition-ul-1-with-tP-non-interacting]
				\left| \partial_t \log \Tr [ \exp(- \lambda \dG(h_t) )  ] \right| &= \frac{\lambda}{2}\frac{ \Tr [ \dG (\widetilde P) \exp(- \lambda \dG(h_t) )]}{\Tr [ \exp(- \lambda \dG(h_t) )  ]}  =\frac{\lambda}{2} \Tr \Big[ \widetilde P \frac{1}{e^{\lambda h_t}-1} \Big] \\
				&\le \Tr[\widetilde P h_t^{-1} ]\lesssim \Tr [\widetilde P h^{-1}] \lesssim C_{\tilde P}. 
			\end{equs} 
			
			Therefore, all of the a priori estimates in Section \ref{sec:gibbs} apply equally well to $G_t$. Moreover, if we denote by $\mu_t$ the Gaussian measure associated with the one-body operator $h_t$, then we have $\dif \mu_t\propto e^{\frac{t}2\|\tilde Pu\|^2}\dif \mu_0$ (we used $|\log\int e^{\frac{t}2\|\tilde Pu\|^2} \dif \mu_0(u) |\lesim C_{\tilde P}$, which can be seen from, e.g., the calculation in \cite[Lemma 5.1]{LewNamRou-21}). Then by a straightforward adaptation of the analysis in Section \ref{sec:CV-energy}, we obtain the following analogue of \eqref{eq:partition-ul-1}:
			\begin{equs}[eq:partition-ul-1-with-tP]
				-\log \frac{Z_t}{ \Tr [ \exp(- \lambda \dG(h_t) ) ] } = -\log \left(  \int  e^{-\cD(P_N u) +  \frac{t}{2} \| \widetilde P  u\|^2} \; \d\mu_{0}(u) \right) + O(\eps^{-8}N^{-1/8}) + O(\lambda^{\delta})
			\end{equs}

			Combing the formula \eqref{eq:partition-ul-1-with-tP} on the partition function $Z_t$ of the perturbed Gibbs state $G_t$ and  \eqref{eq:partition-ul-1} on the partition function $Z_0$ of the non-perturbed Gibbs state $G_0$, we get 
			\begin{align*}
				-\log \frac{Z_t}{Z_0} &=  -\log \left(  \int  e^{-\cD(P_N u) +  \frac{t}{2} \| \widetilde P  u\|^2} \; \d\mu_{0}(u) \right)  + \log \left(  \int  e^{-\cD(P_N u)} \; \d\mu_{0}(u) \right) \\
				&\quad -\log\frac{ \Tr [ \exp(- \lambda \dG(h_t) )  ]}{ \Tr [ \exp(- \lambda \dG(h) )  ]} +  O(\eps^{-8}N^{-1/8}) + O(\lambda^{\delta}). 
			\end{align*}
			The non-interacting partition functions has been estimated via \eqref{eq:partition-ul-1-with-tP-non-interacting}. Moreover, we can choose a smooth cut-off $\1(h\le (N+1)^2/C) \le P_N \le \1(h\le (N+1)^2)$ as in Theorem \ref{th:partition-revised}, which ensures the uniform bound
			\begin{equs}[eq:partition-ul-1-with-tP-aa]
				\left| -\log \left(  \int  e^{-\cD(P_N u) +  \frac{t}{2} \| \widetilde P  u\|^2} \; \d\mu_{0}(u) \right)  + \log \left(  \int  e^{-\cD(P_N u) } \; \d\mu_{0}(u) \right) \right| \le C_{\widetilde P}.
			\end{equs}
			Therefore, we conclude that $|\log (Z_t/Z_0)| \le C_{\widetilde P}$, which is equivalent to 
			\begin{equs}[eq:partition-ul-1-with-tP-bbb]
				1 \lesssim_{\widetilde P} \frac{Z_t}{Z_0}  \lesssim_{\widetilde P} 1, \quad \forall t\in [-3/2,3/2].
			\end{equs}
			
			It is well-known that the mapping $t\mapsto Z_t$ is convex since 
			$$
			\partial^2_t Z_t = \int_0^1 \d s  \Tr [ B \exp(s(-\bH_\lambda+tB)) B \exp((1-s)(-\bH_\lambda+tB)) ] \ge 0.
			$$
			Therefore, $t\mapsto \partial_t Z_t $ is monotone increasing, and \eqref{eq:partition-ul-1-with-tP-bbb} implies that 
			$$
			\Tr [B G_t] =\p_t\log Z_t  \le  \frac2{Z_t}\int_{t}^{t+1/2} \d s (\partial_s Z_s) =  2 \Big(\frac{Z_{t+1/2}}{Z_t} -1\Big) \lesssim_{\widetilde P} 1, \quad \forall t\in [-1,1].
			$$
			By the same argument we also obtain $-\Tr [B G_t]  \lesssim_{\widetilde P} 1$ for all $t\in [-1,1]$. Thus the condition \eqref{eq:CRI-condition-intro} from Theorem \ref{thm:correlation-intro} holds true.
	}}

	Let us apply Theorem \ref{thm:correlation-intro-k} with $X=\lambda \cN+1$. Then clearly $[A,X]=[B,X]=0$. Moreover, by using $[\dG(\widetilde P),\dG(Y)]=\dG([\widetilde P,Y])$, we find that
	$$
	\pm [B,[B, \lambda \dG(h)]] = \pm \frac{\lambda^3}{4}  \dG([ \widetilde P, [\widetilde P,h] ])\le C_{\widetilde P} \lambda^3 \cN,
	$$
	which is similar to \eqref{eq:double-com-kinetic}. Here the condition that $\widetilde P$ has eigenfunctions in $D(h)$ ensures that $\widetilde Ph$ and $h\widetilde P$ are bounded operators. Also, following the analysis in \eqref{eq:double-com-interaction} and using the simple bound $\sum \hat{v}^\eps(k)\lesssim \eps^{-3}$, we obtain
	$$
	\pm [B,[B, \bW]] =  \frac{\lambda^2}{4} [\dG(\widetilde P),[\dG(\widetilde P), \mathbb{W}]]  \le C_{\widetilde P}  \lambda^4  \eps^{-3} \cN^2.
	$$
	Thus we have the last bound in \eqref{eq:corr-thm-ass2}:
	$$
	\pm [B,[B,A]] \le C_{\widetilde P} \lambda^2 \eps^{-3}  (\lambda \cN +1)^2 \le b X^\alpha
	$$
	with $\alpha=2$ and $b= C_{\widetilde P}  \lambda^2 \eps^{-3}$. Thus we can use \eqref{eq:corr-thm-conclusion} and \eqref{eq:partition-ul-1-with-tP-bbb} to deduce that
	\begin{equs}[eq:corr-thm-conclusion-app-Bk]
		\text{} \Tr [B^k  \Gamma_\lambda ] \lesssim_{k}    C_{\widetilde P,k} (1+   \lambda^{4}\eps^{-6} \sup_{t\in [-1,1]}  \Tr[ (\lambda \cN)^{k+1} G_t] ).
	\end{equs}

	For every $t\in [-1,1]$ and $k\ge 1$, we can show that
	\begin{equs}[eq:corr-thm-conclusion-app-Bk-cNk] \Tr[ \cN^{k} G_t] \lesssim_k \lambda^{-3k/2}.
	\end{equs}
	For $k=1$ and $k=2$, the bound \eqref{eq:corr-thm-conclusion-app-Bk-cNk} can be obtained from the variational principle exactly as in \eqref{eq:a-priori-estimate-1}. In this way, we see that the main contribution of   $\Tr[ \cN^{2} G_t]$ comes from the non-interacting value $\Tr[ \cN^{2} e^{  \lambda \dG(h_t)}]/ \Tr [e^{  \lambda \dG(h_t)}] \simeq \lambda^{-3}$). {{For higher $k$, the bound \eqref{eq:corr-thm-conclusion-app-Bk-cNk} can be obtained by adding $\tau \lambda^{3/2}(\cN-N_0)$ to the Hamiltonian $\bH_\lambda + t B$ in $G_t$, which can be treated as a perturbation of $\mW$. This perturbation is very simple since $\cN$ commutes with all relevant operators, and hence the condition \eqref{eq:corr-thm-ass2} hold trivially with $b=0$. We then just need to adjust the mass term and use Theorem \ref{th:partition-revised}. This allows to deduce  \eqref{eq:corr-thm-conclusion-app-Bk-cNk} for all $k\ge 1$ from the result for $k=1$ and a bound similar to \eqref{eq:partition-ul-1-with-tP-bbb}.}}
	
	By inserting \eqref{eq:corr-thm-conclusion-app-Bk-cNk} into \eqref{eq:corr-thm-conclusion-app-Bk}, we conclude that
	$$
	\text{} \Tr [B^k  \Gamma_\lambda ] \lesssim_{k}    C_{\widetilde P,k} (1+   \lambda^{4}\eps^{-6}  \lambda^{-(k+1)/2}) = C_{\widetilde P,k} (1+   \eps^{-6}  \lambda^{(7-k)/2})
	$$
	for all $k\le 7$. In particular,   \eqref{eq:moment-m} holds true.

	The bound \eqref{eq:moment-m-weak} follows from the Hilbert--Schmidt estimate
	$$ \left\|\lambda^{k}\Gamma_\lambda^{(k)}\right\|_{\rm HS}\lesim_k e^{C\eps^{-2}}. $$
	for all $k\ge 1$. The proof of the latter bound is the same with \cite[Lemma 11.3]{LewNamRou-21}. From the heat kernel positivity $e^{-th}(x,y)\ge 0$ and the pointwise estimate $\bW\ge -C\eps^{-2}$, a standard argument using the Trotter product formula
	and the relative bound \eqref{eq:part bound} on partition functions, we obtain the kernel estimate
	$$
	0\le \Gamma_\lambda^{(k)}(X_k;Y_k)\lesim_k e^{C\eps^{-2}} \Gamma_{0}^{(k)}(X_k;Y_k)
	$$
	for all $X_k,Y_k \in (\bT^3)^k$. Thus the desired bound follows from  $\| \lambda^k\Gamma_{0}^{(k)}\|_{\rm HS}\lesim \|h^{-1}\|^k_{\rm HS} \lesim_k 1$.
	The proof of Theorem \ref{thm:higher-moment} is complete.
\end{proof}
\section{Finite-dimensional approximation for classical measures}\label{sec:approximation-classical-measure}

In this section we explain how to relate de Finetti measure of quantum Gibbs state $\Gamma_\lambda$ to the classical Hartree measure $\nu^\eps$. First, in Section \ref{sec:L1-deF-classical} we will discuss the connection between the  de Finetti measure $\mu^\lambda_{P,\lambda}=\mu^\lambda_{P,\Gamma_\lambda}$ and the finite-dimensional approximation $\mu^\eps_N$ of $\nu^\eps$. Then in Section \ref{sec:improved-classical-Hartree} we prove a quantitative convergence from $\mu^\eps_N$ to $\nu^\eps$, which provides additional information over the estimates  in Section \ref{sec:part}.

\subsection{Approximation for de Finetti measure}\label{sec:L1-deF-classical}

We will use the notation in Section \ref{sec:gibbs}. Let  $P=\1(h\le (N+1)^2)$. Let  $\mu^\lambda_{P,\lambda}=\mu^\lambda_{P,\Gamma_\lambda}$ be the de Finetti measure of $\Gamma_\lambda$ as in Theorem \ref{thm:quant deF}, and let $\mu_N^\eps$ be as in Theorem \ref{th:partition-revised}. The two measures can be compared quantitatively as follows.

\begin{theorem}[From de Finetti measure to finite-dimensional Hartree measure]\label{thm:higher-moment-measure} Let $\eps \ge \lambda^{\eta}$ for a sufficiently small parameter $\eta>0$. Let $\Gamma_\lambda=\cZ_\lambda^{-1}e^{-\bH_\lambda}$ be defined  in \eqref{eq:Gibbs-state-def-new}. Then we have
\begin{equs}[eq:Chiribella-app-1]
 \|\mu^\lambda_{P,\lambda}- {\tilde{\mu}^{\eps}_N }\|_{L^1}^2  \le C(\eps^{-8}N^{-1/8} + \lambda^{\delta}).
\end{equs}
Moreover, for every fixed $\varphi\in D(h)$ with $\hat \varphi$ finitely supported, we have the moment bound
\begin{equs}[eq:Chiribella-app-4]
	\int_{P\gH} |\langle \varphi, u\rangle|^{2k}\;  |\d\mu^\lambda_{P,\lambda} - \d{\tilde{\mu}_N^\eps}| (u)  \le C_\varphi (\eps^{-80}N^{-1/8} + {\eps^{-72}}\lambda^{\delta})^{1/14}
\end{equs}
for all $k\le 6$.  If we assume further that $ \eps \ge |\log \lambda|^{-\eta}$ for a constant  $0<\eta<1/2$, then
\begin{equs}[eq:Chiribella-app-4-deF]
	\int_{P\gH} |\langle \varphi, u\rangle|^{2k}\; | \d\mu^\lambda_{P,\lambda} - \d{\tilde{\mu}_N^\eps}| (u) \le C_{\varphi,k} (\eps^{-8} N^{-1/8 } + \lambda^\delta)^{{1/8}}
\end{equs}
for all $k\ge 1$. Here $\tilde{\mu}_N^\eps=\mu_N^\eps\circ P^{-1}$ and  $\delta>0$ is a fixed constant independent of $\lambda,\eps,N$.
\end{theorem}

\begin{proof}

Our starting point is the following identity from \eqref{eq:Chiribella}:
	\begin{equs}[eq:Chiribella-app-0]
		\int_{P\gH}|u^{\otimes k}\ra\la u^{\otimes k}|\;\d\mu^\lambda_{P,\lambda}(u) = k!\lambda^k(\Gamma_\lambda)_P^{(k)} + k! \lambda ^k \sum_{\ell = 0} ^{k-1} {k \choose \ell} (\Gamma_\lambda)_P^{(\ell)}  \otimes_s \one_{\otimes_s ^{k-\ell} P\gH}
	\end{equs}
for all $k\in\N$, where $(\Gamma_\lambda)_P^{(k)} =P^{\otimes k} \Gamma_\lambda^{(k)} P^{\otimes k}$. Thus to the leading order, the correlation functions of $\Gamma_\lambda$, under a suitable projection, is well approximated by the measure $\mu^\lambda_{P,\lambda}$. Our strategy is to relate the measure $\mu^\lambda_{P,\lambda}$ with $\mu_N^\eps$ in $L^1$-norm, and then obtain the desired conclusion by interpolation. 

The smallness of $\mu^\lambda_{P,\lambda}- {\tilde{\mu}^{\eps}_N}$ in $L^1$-norm can be seen from the analysis in Section \ref{sec:CV-energy}. In fact, from the free energy lower bound \eqref{eq:partition-lwb-1} and the corresponding upper bound \eqref{eq:rel-fe-up-P-upper-final}, as well as Pinsker's inequality (see~ e.g. \cite{CarLie-14}), we obtain \eqref{eq:Chiribella-app-1} as
\begin{equs}[eq:Chiribella-app-1-proof]
\frac{1}{2} \|\mu^\lambda_{P,\lambda}- {\tilde{\mu}^{\eps}_N} \|_{L^1}^2 \le \cHcl(\mu^\lambda_{P,\lambda}, {\tilde{\mu}^{\eps}_N})  \le C(\eps^{-8}N^{-1/8} + \lambda^{\delta}).
\end{equs}

In order to put \eqref{eq:Chiribella-app-1} in good use, we will need moment estimates for $\mu^\lambda_{P,\lambda}$ and $\mu_N^\eps$. For every fixed $\varphi \in C^\infty(\mathbb{T}^3)$ with $\hat \varphi$ finitely supported,  using \eqref{mom:muN-eps-revised} from Theorem \ref{th:partition-revised} we have the uniform bound for $N$ large enough
\begin{equs}[eq:Chiribella-app-2]
\text{ }		{\int_{P\gH} |\langle \varphi, u\rangle|^{2k}\;\d\tilde{\mu}_N^\eps =}\int_{P\gH} |\langle \varphi, u\rangle|^{2k}\;\d\mu_N^\eps \le C_{\varphi,k}
\end{equs}
for all $k\ge 1$. Moreover, we deduce from \eqref{eq:Chiribella-app-0} and the moment estimate \eqref{eq:moment-m} that
	\begin{equs}[eq:Chiribella-app-3]
\text{ }		\int_{P\gH} |\langle \varphi, u\rangle|^{2k}\;\d\mu^\lambda_{P,\lambda}(u) \le C_{\varphi} \sum_{\ell=0}^k \lambda^\ell  \langle \varphi^{\otimes \ell}, (\Gamma_\lambda)_P^{(\ell)} \varphi^{\otimes \ell}\rangle \le C_{\varphi} \eps^{-6}
\end{equs}
for all $k\le 7$. Using \eqref{eq:Chiribella-app-1}, \eqref{eq:Chiribella-app-2},  \eqref{eq:Chiribella-app-3} and H\"older's inequality, we find  that for every $k\le 6$,
\begin{equs}[eq:Chiribella-app-4-proof]
 	\int_{P\gH} |\langle \varphi, u\rangle|^{2k}\; | \d\mu^\lambda_{P,\lambda} - \d{\tilde{\mu}_N^\eps}|(u)  \le C_\varphi (\eps^{-6})^{6/7} (\eps^{-8}N^{-1/8} + \lambda^{\delta})^{1/14}.
\end{equs}
This gives the bound in \eqref{eq:Chiribella-app-4}.

If we assume further that $ \eps \ge |\log \lambda|^{-\eta}$ for a constant  $0<\eta<1/2$, then by using  \eqref{eq:moment-m-weak} instead of \eqref{eq:moment-m}, we can replace \eqref{eq:Chiribella-app-3} by
\begin{equs}[eq:Chiribella-app-3b]
\text{ }		\int_{P\gH} |\langle \varphi, u\rangle|^{2k}\;\d\mu^\lambda_{P,\lambda}(u) \le C_{\varphi} \sum_{\ell=0}^k \lambda^\ell  \langle \varphi^{\otimes \ell}, (\Gamma_\lambda)_P^{(\ell)} \varphi^{\otimes \ell}\rangle \le C_{\varphi,k} e^{C \eps^{-2} }
\end{equs}
for all $k\ge 1$. Since ${e^{C\eps^{-2}} \ll \lambda^{-\nu}}$ for every $\nu>0$, we can choose $\Lambda={\lambda^{-\delta/C}}$ for a large but fixed number $C>0$ and  conclude \eqref{eq:Chiribella-app-4-deF} from \eqref{eq:Chiribella-app-1-proof}  and H\"older inequality.
\end{proof}

%

\subsection{Approximation for Hartree measure} \label{sec:improved-classical-Hartree}

In this section, we use analysis of the quantum part to refine the estimates for the classical partition functions given in Section \ref{sec:part}.

\begin{theorem}\label{th:partition-revised-improved}  Let $N \ge \eps^{- 1/ (16\eta)}$ for a small parameter $\eta>0$.  Let $P_N$ be a smooth cut-off as in Theorem \ref{th:partition-revised} and let $\mu^\eps_N$ be defined in \eqref{eq:def-mu-eps-N}.
Then we have
\begin{equs}[sec:9-1]
\text{}	\Big|\log\int e^{-\cD(P_N \Psi)} \d \mu_0(\Psi) - \log \int e^{-\cD(\Psi)} \d \mu_0(\Psi)   \Big| \lesssim \eps^{-8} N^{-1/8 }.
\end{equs}
Moreover, for every fixed $\varphi\in L^2(\bT^3)$ with $\hat \varphi$ is finitely supported, we have
\begin{equs}[sec:9-2]
\text{} \left|  \int |\langle \varphi,u\rangle|^{2k} \d \mu^\eps_N (u) - \int |\langle \varphi,u\rangle|^{2k}   \d \nu^\eps(u)  \right|  \le   C_{\varphi,k} (\eps^{-8} N^{-1/8 })^{{1/2}}.\end{equs}
\end{theorem}

Note that these results can be proved as in \cite[Appendix C]{Bri} using the variational framework in Section \ref{sec:part}, but this approach requires several additional calculations. Here we will derive them as immediate consequences from the quantum results, by using the flexibility of choosing the semiclassical parameter $\lambda$.

\begin{proof} Our starting point is the convergence of the quantum free energy \eqref{eq:partition-ul-1} from Theorem \ref{thm:local-W-P}:
$$
	\left| -\log \frac{\cZ_\lambda}{\cZ_0} + \log \left(  \int  e^{-\cD(P_N u)} \; \d\mu_{0}(u) \right) \right| \le C\eps^{-8}N^{-1/8}
$$
where we have chosen  $\lambda\to 0$ sufficiently fast such that $\lambda^\delta \le \eps^{-8}N^{-1/8}$. Similarly, for any given $M\ge N$, we also obtain the same bound with $P_N$ replaced by $P_M$ with $M \ge N$, provided that that we choose $\lambda^\delta \le \eps^{-8}M^{-1/8}$. Thus by the triangle inequality, we find that for all $M \ge N \to \infty$,
\begin{equs}[bd:par:infinity]
	\Big|\log\int e^{-\cD(P_N \Psi)} \d \mu_0(\Psi) - \log \int e^{-\cD(P_M\Psi)} \d \mu_0(\Psi)   \Big| \lesssim C \eps^{-8} N^{-1/8}.
\end{equs}
Letting $M\to\infty$ we derive \eqref{sec:9-1}.

Concerning \eqref{sec:9-2}, let us take $M\ge N \ge \eps^{- 1/ \eta}$, and choose $\lambda\to 0$ sufficiently fast such that $\lambda^\delta \le \eps^{-8}N^{-1/8}$ and $ \eps \ge |\log \lambda|^{-\eta}$ for a constant  $0<\eta<1/2$. {Using \eqref{eq:Chiribella-app-0} we have
\begin{equs}[eq:Chiribella-app-6]
	\left| 	\int_{P\gH} |\langle \varphi, u\rangle|^{2k} \d\mu^\lambda_{P,\lambda} (u) - k! \lambda^k  \langle \varphi^{\otimes k}, (\Gamma_\lambda)_P^{(k)} \varphi^{\otimes k}\rangle   \right| 
	\le C_{\varphi}\lambda \sum_{\ell=0}^{k-1} \lambda^\ell  \langle \varphi^{\otimes \ell}, (\Gamma_\lambda)_P^{(\ell)} \varphi^{\otimes \ell}\rangle \le C_{\varphi} \lambda e^{\eps^{-2}}.
\end{equs}
Using the fact that $\hat \varphi$ is compact supported, we obtain
\begin{align*}
		\left| 	\int_{P\gH} |\langle \varphi, u\rangle|^{2k}\; \d\mu^\lambda_{P_N,\lambda} (u) - \int_{P\gH} |\langle \varphi, u\rangle|^{2k}\; \d\mu^\lambda_{P_M,\lambda} (u)  \right|\leq C_\varphi \lambda e^{\eps^{-2}}.
\end{align*}
} Then using \eqref{eq:Chiribella-app-4-deF} and the triangle inequality, for every $k\ge 1$ we have
\begin{equs}[eq:Chiribella-app-4-deF-triangle]
	&\bigg| 	\int_{P\gH} |\langle \varphi, u\rangle|^{2k} (
	\dif\mu_M^\eps - \dif \mu_N^\eps)(u) \bigg|
	\\\leq& \bigg|\int_{P\gH}|\la \varphi,u\ra|^{2k}(  \d\mu^\lambda_{P_M,\lambda} - \d\mu_M^\eps)(u)\bigg|+
	\left| 	\int_{P\gH} |\langle \varphi, u\rangle|^{2k}( \d\mu^\lambda_{P_N,\lambda} - \d\mu_N^\eps)(u) \right|
	\\&+\bigg| 	\int_{P\gH} |\langle \varphi, u\rangle|^{2k} \d\mu^\lambda_{P_N,\lambda} (u) - \int_{P\gH} |\langle \varphi, u\rangle|^{2k} \d\mu^\lambda_{P_M,\lambda} (u)  \bigg|
	\\\le& C_{\varphi,k} (\eps^{-8} N^{-1/8 }+\lambda^\delta)^{1/2}.
\end{equs}

Thus $\int_{P\gH} \langle \varphi, u\rangle|^{2k}\; \d\mu_N^\eps(u)$ is a Cauchy sequence. When $M\to \infty$, using the same arguments as in \cite[Lemma 5.3]{LewNamRou-21}, it must converges to $\int  \langle \varphi, u\rangle|^{2k}\; \d\nu^\eps(u)$. From \eqref{eq:Chiribella-app-4-deF-triangle}, by taking $\lambda\to 0$ very fast and $M\to \infty$, we obtain the desired bound in \eqref{sec:9-2}.
 \end{proof}

\section{Convergence of quantum correlation functions} \label{sec:CV-Gamma}

Now we are ready to conclude Theorem \ref{thm:main1}. The convergence of the partition function in \eqref{eq:main-CV-Hartree-1} is a consequence of the estimate \eqref{eq:partition-ul-1} in  Theorem \ref{thm:local-W-P} and the bound \eqref{sec:9-1} in Theorem \ref{th:partition-revised-improved}. Thus it remains to  derive the convergence of the quantum correlation functions. We will approximate the quantum correlation functions using the de Finetti measure and then connect this measure to the classical field theory.

Let us start by establishing a representation for the de Finetti measure which is of independent interest. 

\begin{lemma} \label{lem:deF-measure-f}Let $\Gamma$ be a state on $\gF$ which commutes with the number operator $\cN$. Let $P$ be a finite dimensional projection on $\gH$ and let $\d\mu_{P,\Gamma}^{\lambda}$ be the de Finetti measure at scale $\lambda$ defined in \eqref{eq:Husimi}. For a normalized vector $\varphi \in P\gH$, we denote by $\Gamma_\varphi=\Gamma_{|\varphi\rangle \langle \varphi|}$ the localized state on $\gF (\mathbb{C} \varphi)$. Then for every measurable function $f: [0,\infty)\to [0,\infty)$ we have
	\begin{equs}[eq:f-deF-measure]
	\int_{P\gH} f( |\langle \varphi, u  \rangle|^2) \d\mu_{P,\Gamma}^{\lambda} (u) =\sum_{n=0}^\infty   \int_0^\infty f(\lambda x) e^{-x}\frac{x^n}{n!}\d x \langle \varphi^{\otimes n}, \Gamma_\varphi \varphi^{\otimes n}\rangle.
\end{equs}
\end{lemma}

\begin{proof} By decomposing $u=\langle \varphi, u \rangle \varphi + u_\varphi^\bot$ and using the isometry $\gF(P\gH) \approx \gF(\mathbb{C} \varphi) \otimes\gF( \varphi^\bot_{P\gH})$ (which is similar to \eqref{eq:factorization_Fock_space}), we can interpret the coherent state  $W(u/\sqrt{\lambda}) \in \gF(P\gH)$ as the tensor product of the state $W( \langle \varphi,u\rangle \varphi/\sqrt{\lambda})\in \gF(\mathbb{C} \varphi)$ and the state $ W(u_\varphi^\bot/\sqrt{\lambda}) \in \gF(\varphi^\bot_{P\gH})$. Here $\varphi^\bot_{P\gH}$ is the subspace of $P\gH$ which is orthogonal to $\varphi$. Combining with the resolution of identity \eqref{eq:resolution_coherent2} adapted to $\gF(\varphi^\bot_{P\gH})$, namely
$$
	\1_{\gF(\varphi^\bot_{P\gH})} = \frac{1}{(\pi \lambda)^{\Tr(P)-1} }  \int_{\varphi^\bot_{P\gH}} |W(u_\varphi^\bot/\sqrt{\lambda}) \rangle \langle W(u_\varphi^\bot/\sqrt{\lambda}) | \d u_\varphi^\bot ,
$$
and the definition $\Gamma_\varphi= \Tr_{\gF(\varphi^\bot)} \Gamma_P$, we find that
	\begin{align*}
	\int_{P\gH} f( |\langle \varphi, u  \rangle|^2) \d\mu_{P,\Gamma}^{\lambda} (u) &  = \frac{1}{(\pi \lambda)^{\Tr P}} \int_{P\gH} f( |\langle \varphi, u  \rangle|^2) \langle W(u/\sqrt \lambda), \Gamma_P W(u/\sqrt \lambda) \rangle  \d u \\
	&= \frac{1}{\pi \lambda} \int_{\mathbb{C}\varphi} f( |\langle \varphi, u  \rangle|^2) \langle W( \langle \varphi,u\rangle \varphi/\sqrt{\lambda}),  \Gamma_\varphi W( \langle \varphi,u\rangle \varphi/\sqrt{\lambda}) \rangle  \d ( \langle \varphi, u\rangle \varphi).
\end{align*}
Finally, using \eqref{eq:coherent state} in the form
$$
W( \langle \varphi,u\rangle \varphi/\sqrt{\lambda})= e^{-|\langle \varphi, u\rangle|^2/(2\lambda)}  \bigoplus_{n\ge 0} \frac{ \langle \varphi,u\rangle^n}{\lambda^{n/2}\sqrt{n!}} \varphi^{\otimes n}
$$
and the fact that $\langle \varphi^{\otimes n}, \Gamma_\varphi  \varphi^{\otimes m}\rangle =0$ if $n\ne m$ (since $\Gamma$ commutes with $\cN$), we conclude that
	\begin{align*}
\int_{P\gH} f( |\langle \varphi, u  \rangle|^2) \d\mu_{P,\Gamma}^{\lambda} (u) &=  \frac{1}{\pi \lambda} \int_{\mathbb{C}\varphi} f( |\langle \varphi, u  \rangle|^2) e^{-|\langle \varphi, u\rangle|^2/\lambda}  \sum_{n\ge 0}  \frac{ |\langle \varphi,u\rangle|^{2n}}{\lambda^{n}n!} \langle \varphi^{\otimes n}, \Gamma_\varphi \varphi^{\otimes n} \rangle  \d ( \langle \varphi, u\rangle \varphi) \\
&=  \sum_{n\ge 0} \frac{1}{\pi \lambda}  \int_{\mathbb C} f(|z|^2) e^{-|z|^2/\lambda}   \frac{(|z|^2/\lambda)^n}{n!} \langle \varphi^{\otimes n}, \Gamma_\varphi \varphi^{\otimes n} \rangle \d z \\
& =  \sum_{n\ge 0} \int_0^\infty f(\lambda x) e^{-x}   \frac{x^n}{n!} \d x \langle \varphi^{\otimes n}, \Gamma_\varphi \varphi^{\otimes n} \rangle.
\end{align*}
In the last step we used polar coordinate. The proof of Lemma \ref{lem:deF-measure-f} is complete.
\end{proof}

Now we are ready to conclude the convergence of correlation functions.

\begin{proof}[Proof of \eqref{eq:main-CV-Hartree-2}] Let us show that if $|f(x)-f(y)|\le C |x-y| (1+ x^5+y^5)$, then
	\begin{equs}[eq:f-deF-measure-correlation]
	&\Big|\int f(|\la \varphi, u\ra|^{2})\dif \mu_{P,\lambda}^\lambda-\tr[f(\lambda a^*(\varphi)a(\varphi))\Gamma_{\lambda}]\Big| \le C_\varphi \lambda^{1/2}\eps^{-6}.
	\end{equs}
Under the same notation in Lemma \ref{lem:deF-measure-f}, by using the isometry $\gF \approx \gF(\mathbb{C} \varphi) \otimes\gF( \varphi^\bot_{\cH})$, with $\varphi^\bot_{\cH}$ the subspace of $\cH$ which is orthogonal to $\varphi$, and the operator representation
$$f(\lambda a^*(\varphi)a(\varphi))=\sum_{n\geq0} f(\lambda n)| \varphi^{\otimes n}\ra \la\varphi^{\otimes n}| $$
on $\mathfrak{F}(\mathbb{C}\varphi)$, we can write
$$\tr_\gF[f(\lambda a^*(\varphi)a(\varphi))\Gamma_\lambda] = \tr_{\gF(\mathbb C \varphi)}[f(\lambda a^*(\varphi)a(\varphi))(\Gamma_\lambda)_\varphi]=\sum_{ n\geq 0}f(\lambda n)\la \varphi^{\otimes n},(\Gamma_\lambda)_\varphi\varphi^{\otimes n} \ra.$$
 Combining with Lemma \ref{lem:deF-measure-f}, we obtain
 $$
 \int f(|\la \varphi, u\ra|^{2})\dif \mu_{P,\lambda}^\lambda-\tr[f(\lambda a^*(\varphi)a(\varphi))\Gamma_{\lambda}] = \sum_{n\ge 0} \Big(\frac1{n!}\int_0^\infty e^{-x}x^n f(\lambda x)\dif x-f(\lambda n) \Big) \la \varphi^{\otimes n},(\Gamma_\lambda)_\varphi\varphi^{\otimes n} \ra.
 $$

Moreover, from the condition $|f(x)-f(y)|\le C |x-y| (1+ x^5+y^5)$, by using the triangle inequality and a direct computation via Gamma function, we can estimate for every $n\ge 0$:
	\begin{align*}
		&\Big|\frac1{n!}\int_0^\infty e^{-x}x^n f(\lambda x)\dif x-f(\lambda n)\Big|= \left| \frac1{n!}\int_0^\infty e^{-x}x^n ( f(\lambda x)-f(\lambda n)) \dif x\right|\\
		&\lesssim \frac1{n!}\int_0^\infty e^{-x}x^n  \lambda |x-n| (1+\lambda^5 x^5+\lambda^5 n^5) \dif x \\
		&\lesssim \lambda  \left(  \frac1{n!}\int_0^\infty e^{-x}x^n (x-n)^2 \right)^{1/2}  \left(  \frac1{n!}\int_0^\infty e^{-x}x^n (1+\lambda^{10}x^{10}+\lambda^{10}n^{10})  \right)^{1/2}  \\
		&\lesssim \lambda \sqrt{n+2} \sqrt{1+\lambda^{10}(n+1)^{10}} \lesssim \sqrt{\lambda} (1+ (\lambda n)^{5+1/2}).
	\end{align*}
Summing over $n \ge 0$ we conclude that
\begin{equs}[eq:f-CV-Har-app]
	\Big|\int f(|\la \varphi, u\ra|^{2})\dif \mu_{P,\lambda}^\lambda-\tr[f(\lambda a^*(\varphi)a(\varphi))\Gamma_{\lambda}]\Big|
	&\lesssim \sum_{n\ge 0}   \sqrt{\lambda} (1+ (\lambda n)^{5+1/2})  \la \varphi^{\otimes n},(\Gamma_\lambda)_\varphi\varphi^{\otimes n} \ra \\
	&=  \sqrt{\lambda} \Big(1+\tr[ (\lambda a^*(\varphi)a(\varphi))^{5+1/2}\Gamma_{\lambda}] \Big).
\end{equs}
The bound \eqref{eq:f-deF-measure-correlation} follows from \eqref{eq:f-CV-Har-app} and the a-priori estimate \eqref{eq:moment-m} from Theorem \ref{thm:higher-moment}.

Now let us assume that $|f(x)-f(y)|\le C |x-y| (1+ x^5+y^5)$, which in particular implies that $|f(x)|\le C (1+x^6)$. Then by the moment bound \eqref{eq:Chiribella-app-4} from Theorem \ref{thm:higher-moment-measure} we find that
\begin{equs}[eq:Chiribella-app-4-app]
	\int_{P\gH} f( |\langle \varphi, u\rangle|^{2}) \;  |\d\mu^\lambda_{P,\lambda} - \d\mu_N^\eps| (u)  \le C_\varphi (\eps^{-80}N^{-1/8} + {\eps^{-72}}\lambda^{\delta})^{1/14} \to 0
\end{equs}
with the choice $\eps^{-{640}}\ll N\ll \lambda^{-1/8}$. Combining \eqref{eq:Chiribella-app-4-app}, \eqref{sec:9-2} and \eqref{eq:f-deF-measure-correlation} and using the triangle inequality we obtain
$$
	\left| \tr[f(\lambda a^*(\varphi)a(\varphi))\Gamma_{\lambda}] - \int_{P\gH} f( |\langle \varphi, u\rangle|^{2}) \d\nu^\eps (u) \right|   \le C_\varphi \lambda^{1/2}\eps^{-6} + C_\varphi (\eps^{-80}N^{-1/8} +\eps^{-72} \lambda^{\delta})^{1/14} \to 0.
$$
This is  the desired estimate  \eqref{eq:main-CV-Hartree-2}.
\end{proof}

\begin{proof}[Proof of \eqref{eq:main-CV-Hartree-3}]
From \eqref{eq:main-CV-Hartree-2}, by choosing $f(x)=x^k$ we obtain \eqref{eq:main-CV-Hartree-3} for $k\le 6$. This holds under the assumption that $\lambda^\eta \le \eps$ for a sufficiently small parameter $\eta>0$.

If we assume further that $|\log \lambda|^{-\eta}\le  \eps \to0 $ for a constant  $0<\eta<1/2$, then using \eqref{eq:f-CV-Har-app} with $f(x)=x^k$ and the moment bound  \eqref{eq:moment-m-weak} we have
$$
	\Big|\int f(|\la \varphi, u\ra|^{2})\dif \mu_{P,\lambda}^\lambda-\tr[f(\lambda a^*(\varphi)a(\varphi))\Gamma_{\lambda}]\Big|
	\lesssim_{k,\varphi} \sqrt{\lambda} e^{C\eps^{-2}} \le \lambda^\delta
$$
Combining the latter bound and \eqref{eq:Chiribella-app-4-deF}, \eqref{sec:9-2} with the choice $\eps^{-64}\ll N \ll \lambda^{-1/8}$, we conclude by the triangle inequality that \eqref{eq:main-CV-Hartree-3} holds for for all $k \ge 1$, namely
$$
	\left| \tr[ (\lambda a^*(\varphi)a(\varphi))^k \Gamma_{\lambda}] - \int_{P\gH} |\langle \varphi, u\rangle|^{2k} \d\nu^\eps (u) \right|   \lesssim_{k,\varphi} \lambda^\delta + (\eps^{-8} N^{-1/8 } + \lambda^\delta)^{1/4} \to 0
$$

The proof of Theorem \ref{thm:main1} is complete. 
\end{proof}


\section{Two dimensional case}\label{sec:2D}
In this section we consider the two dimensional case. We have

\begin{theorem}\label{thm:main10-2D} Consider the Gibbs state $\Gamma_\lambda=\cZ_\lambda^{-1}e^{-\bH_\lambda}$ on $\gF (L^2(\bT^2))$, defined similar in \eqref{eq:Gibbs-state-def}  with $v^\eps:\bT^2\to \R$ be defined as in \eqref{eq:v-eps-per-def} with
	$0\le \hat v(k)\lesssim (1+|k|^{2+\delta_0})^{-1}$ for all $k\in \R^2$, $\delta_0>0$ and $\hat v(0)=1$, $|D^m \hat v(k)|\lesssim \frac1{1+|k|^{m}}$ for $m\in \{1,2\}$, and with  the chemical potential
	\begin{equs}[def:gamma-2D]
		\vartheta =  \sum_{k\in\mZ^2}\frac{\lambda}{e^{\lambda(|k|^2+1)}-1} + a^\eps , \quad a^\eps=\int_{\mT^2}v^\eps (y)G(y)\d y.
	\end{equs}
Let $\nu$ be the $\Phi^4_2$ measure associated with $m_0$, defined similarly as in \eqref{e:Phi-measure1-main}. We consider the limit $\lambda,\eps \to 0$ with $ \lambda^{\eta}\le  \eps$ for a sufficiently small constant  $\eta>0$. Then for all $n\ge 1$, we have the convergence of the density matrices
	\begin{equs}[eq:CV-pdm-f-xinf-2D]
 n! \lambda^n \Big\langle \varphi^{\otimes n},\Gamma_\lambda^{(n)}  \varphi^{\otimes n}\Big\rangle  \to  \int  |\langle \varphi, \Phi\rangle|^{2n} \,\d \nu(\Phi),
	\end{equs}
	for all  $\varphi\in L^2(\bT^3)$ with $\hat \varphi$ finitely supported.
\end{theorem}

Theorem \ref{thm:main10-2D} can be interpreted as an extension of the main result in \cite{FKSS23}, which required the stronger condition  $\exp(-|\log \lambda|^{\frac12-c})\le  \eps \to0 $ for $c>0$. This can be proved using our approach to the $\Phi^4_3$ measure, with the analysis simplified in several places. Notably, unlike the $d=3$ case, the assumption $v \geq 0$ is unnecessary here (see the proof of Theorem \ref{thm:main1d=2} below). We now outline the key steps of the proof. Parallel results to Theorem \ref{thm:main10-2D} were also obtained by Jouglas and Rougerie \cite{JouRou-25}, using inputs from \cite{FKSS23} rather than the stochastic PDE method. 

First, we can derive the same results as in Theorem \ref{th:conm} through more elementary arguments. Specifically, the stochastic objects in \eqref{e:wick-tilde} stay in $C_T\bC^{-\kappa}$ for $\kappa>0$ eliminating the need for paracontrolled calculus techniques required in the dynamical $\Phi^4_3$ model.  Furthermore, we can also drop the condition that $v\geq0$, which is necessary for the three dimensional case.  Consider
\begin{equs}[e:Phi_eps-measure1d2]
	\dif\nu^\eps(\Psi)\eqdef  \frac{1}{\sZ_\eps}\exp\bigg(&-\int_{\mathbb T^2} (|\nabla \Psi|^2+m |\Psi|^2) \,\dif x
	-\frac12\int \Wick{|\Psi(x)|^2}v^\eps(x-y)\Wick{|\Psi(y)|^2}\,\dif x\dif y\no
	\\&+{a^\eps\int_{\mathbb T^2} \Wick{|\Psi(x)|^2}\,\dif x}\bigg)\mathcal D \Psi,
\end{equs}
with $a^\eps=\int_{\mT^2}v^\eps (y)G(y)\d y$.
More precisely we have the following results.

\begin{theorem}\label{thm:main1d=2}
	Let $v^\eps$ be as in Theorem \ref{thm:main10-2D}. As $\eps \to 0$, the probability measure $\nu^\eps$ in \eqref{e:Phi_eps-measure1d2} converges weakly to $\nu$ in $\bC^{-\kappa}$, $\kappa>0$, where $\nu$ represents the $\Phi^4_2$ field \begin{equs}[e:Phi-measure1d2]
		\dif\nu(\Phi)\eqdef \frac{1}{\sZ}\exp\bigg(-\int_{\mathbb T^2} \Big(|\nabla \Phi|^2+m |\Phi|^2 \Big)\,\dif x
		+ \frac12\int \Wick{|\Phi(x)|^4}\,\dif x\bigg)\mathcal D \Phi.
	\end{equs} Moreover, any correlation function $\gamma_n^\eps$ associated with $\nu^\eps$ converges to the $n$-point correlation function $\gamma_n$ of the $\Phi^4_2$ field in $\cS'(\mT^{4n})$.
\end{theorem}
\begin{proof} The proof proceeds in precisely the same manner as in the three-dimensional case. Recall that we only use $v\geq0$ in Lemma \ref{cubic}. Here, we present an alternative proof that does not rely on $v\geq0$, achieved by assuming a higher level of regularity specifically $\cZ\in C_T\bC^{-\kappa}$, $\kappa>0$.
	We begin with the following decomposition:
	\begin{align}\label{eq:gcubicd2}
		\langle v^\eps*|f|^2f,\cZ\rangle =\langle v^\eps*|f|^2f,\Delta_{>L}\cZ\rangle +\langle v^\eps*|f|^2f,\Delta_{\leq L}\cZ\rangle .
	\end{align}
	For the second term we  use \eqref{lowerbound} and \eqref{lowerbound1} and \eqref{eq:loc} to have
	\begin{align*}
		|\langle v^\eps*|f|^2f,\Delta_{\leq L}\cZ\rangle|\leq\|(v^\eps *|f|^2)f\|_{L^1}\|\Delta_{\leq L}\cZ\|_{L^{\infty}}\leq  \cV^\eps(f)^{\frac34}2^{2\kappa L}\|\cZ\|_{\bC^{-\kappa}}.
	\end{align*}
	For the first term on the RHS of \eqref{eq:gcubicd2}, we use Besov embedding Lemma \ref{lem:emb},  Lemma \ref{lem:multi} and Lemma \ref{lem:para}, \eqref{eq:loc} and interpolation Lemma \ref{lem:interpolation} to have
	\begin{align*}
		|\langle v^\eps*|f|^2f,\Delta_{> L}\cZ\rangle|
		\lesssim&\,\|v^\eps*|f|^2f\|_{\bB^{2\kappa}_{1,1}}\|\Delta_{> L}\cZ\|_{\bC^{-2\kappa}}
		\lesssim  \|v^\eps*|f|^2\|_{\bB^{4\kappa}_{1+\frac\kappa2}}\|f\|_{H^1}\|\Delta_{> L}\cZ\|_{\bC^{-2\kappa}}
		\\\lesssim&\, \|f\|_{H^{5\kappa}}\|f\|_{L^2}\|f\|_{H^1}2^{-\kappa L}\|\cZ\|_{\bC^{-\kappa}}
		\lesssim \|f\|_{L^2}^{2-5\kappa}\|f\|_{H^1}^{1+5\kappa}2^{-\kappa L}\|\cZ\|_{\bC^{-\kappa}},
	\end{align*}
	where we used
	\begin{align*}
		\|v^\eps*|f|^2f\|_{\bB^{2\kappa}_{1,1}}\lesssim\|v^\eps*|f|^2\|_{\bB^{3\kappa}_{1+\frac\kappa2,1}}\|f\|_{L^{\frac{2+\kappa}\kappa}}+\|v^\eps*|f|^2\|_{L^{\frac1{1-2\kappa}}}\|f\|_{\bB^{3\kappa}_{\frac1{2\kappa},1}}\lesssim\|v^\eps*|f|^2\|_{\bB^{4\kappa}_{1+\frac\kappa2}}\|f\|_{H^1}
	\end{align*}
	in the second step.
	Substituting the above two estimates into \eqref{eq:gcubicd2} and choosing
	$$2^{2\kappa L}\simeq \cV^{\eps}(f)^{\frac18},$$
	we obtain
	\begin{equs}
		|\langle (v^\eps *|f|^2)f, \cZ\rangle |\lesssim \Big[\cV^\eps(f)^{\frac78}+\cV^\eps(f)^{\frac{7}{16}-\frac{5\kappa}4}\|f\|_{H^1}^{1+{5\kappa}}\Big]\|\cZ\|_{\bC^{-\kappa}}.
	\end{equs}
	Applying Young's inequality, we derive
	\begin{equs}
		|\langle (v^\eps *|f|^2)f, \cZ\rangle |\lesssim\delta\cV^\eps(f)+\delta \|f\|_{H^1}^2 +K(\|\cZ\|_{\bC^{-\kappa}}).
	\end{equs}
	The remainder of the proof proceeds with simpler arguments analogous to those used in the three-dimensional case, and the details are therefore omitted.
\end{proof}

Furthermore, for the bounds of partition functions as presented in Theorem \ref{th:partition}, we obtain the same result for $d=2$ without relying on the condition $v\geq0$.

We consider the finite-dimensional approximation of $\nu^\eps$ given by
\begin{align*}
	\d \mu_N^\eps(\Psi)\eqdef \frac1{\sZ_{\eps,N}}\exp\Big(-\cD[P_N\Psi])\Big)\d \mu_0(\Psi),\qquad \sZ_{\eps,N}\eqdef\int \exp\Big(-\cD[P_N\Psi]\Big)\d \mu_0(\Psi),
\end{align*}
with $P_N$ given as in Section \ref{sec:part},
where we set the renormalized potential energy given by
\begin{equs}
	\cD(u)\eqdef \frac12\int (v^\eps*\Wick{|\Psi|^2})\Wick{|\Psi|^2}-(a^{\eps}-m+1)\int \Wick{|u|^2}.
\end{equs}

\bt\label{th:partition-revised2d}
 Suppose the same condition as in Theorem \ref{thm:main1d=2}. For $f:\bC^{-\kappa}\to \mR$ with at most linear growth, it holds that
\begin{equs}[mom:muN-eps-revised2d]
	\sup_{\eps\in(0,1),N\geq \eps^{-1-\kappa}}\int\exp(f(P_N\Psi))\d \mu_N^\eps(\Psi)<\infty.
\end{equs}
Furthermore, for any finite dimensional projection $\widetilde P$ and $c_0, c_1\in\mR$,  it holds that
\begin{equs}[mom:mutildeP-revised2d]
	\Big|\log\int \exp\Big(-\cD(P_N\Psi) +c_0\la \widetilde P\Psi,\Psi\ra + c_1 \int \Wick{|P_N\Psi|^2} \Big)\,\dif \mu_0(\Psi) - \log \sZ_{\eps,N} \Big|\lesssim C_{\widetilde P}+1,
\end{equs}
with $ C_{\widetilde P}$ being a constant depending on $\widetilde P$.  
\et

Proof of Theorem \ref{th:partition-revised2d} is significantly simpler than the $d=3$ case treated in Section \ref{sec:part}. For further details, we refer to \cite[Section 3]{BG18}.

Now for the quantum part, we can proceed similarly to Section \ref{sec:gibbs} and Section  \ref{sec:CV-energy} to obtain the convergence of the quantum free energy. In this way, we also obtain  the bound
\begin{equs}[eq:Chiribella-app-1-proof-2D]
\|\mu^\lambda_{P,\lambda}- \tilde\mu^{\eps}_N \|_{L^1}^2 \le  C(\eps^{-8}N^{-1/8} + \lambda^{\delta}).
\end{equs}
which is similar to \eqref{eq:Chiribella-app-1} (here we used the same notation for the de Finetti measure $\tilde\mu^\lambda_{P,\lambda}$ and the Hartree measure $\tilde\mu^{\eps}_N=\mu^\eps_N\circ P^{-1}$ in finite dimensions). The main difference in 2D case is that we have a much better moment bound
\begin{equs}[eq:moment-m-2D]
\lambda^k \Tr [ \cN^k \Gamma_\lambda] \le C_{k} ( |\log \lambda|)^k
\end{equs}
for all $k \ge 1$. This bound can be obtained by the same way as in \eqref{eq:corr-thm-conclusion-app-Bk-cNk}, and the factor $( |\log \lambda|)^k$ comes from the non-interacting picture
$$
\lambda \Tr [\cN \Gamma_0] = \sum_{p\in \mathbb{Z}^2 }\frac{\lambda}{e^{\lambda(|p|^2+1)}-1} \le C |\log \lambda|.
$$
Consequently, we obtain the following improved version of  \eqref{eq:Chiribella-app-3}:
$$
	\int_{P\gH} |\langle \varphi, u\rangle|^{2k}\;\d\mu^\lambda_{P,\lambda}(u) \le C_{\varphi} \sum_{\ell=0}^k \lambda^\ell  \langle \varphi^{\otimes \ell}, (\Gamma_\lambda)_P^{(\ell)} \varphi^{\otimes \ell}\rangle \le C_{\varphi,k} |\log \lambda|^k
$$
for all $k\ge 1$. Thus by interpolation as in \eqref{eq:Chiribella-app-4-proof}, we deduce that for $\varphi\in C^\infty(\mT^2)$ with $\hat \varphi$ finitely supported
\begin{equs}[eq:Chiribella-app-4-2D]
 	\int_{P\gH} |\langle \varphi, u\rangle|^{2k}\; | \d\mu^\lambda_{P,\lambda} - \d\mu_N^\eps|(u)  \le C_{\varphi,k}  (\eps^{-8}N^{-1/8} + \lambda^{\delta})^{1/4} |\log \lambda|^k \to 0,
\end{equs}
where we choose the cut-off $N$ such that $\eps^{-8}N^{-1/8} \le \lambda^\delta$.

Given \eqref{eq:Chiribella-app-4-2D}, the bound \eqref{eq:CV-pdm-f-xinf-2D} thus follows from the same analysis in Section \ref{sec:CV-Gamma}. This concludes the proof sketch of Theorem \ref{thm:main10-2D}

\appendix
\renewcommand{\appendixname}{Appendix~\Alph{section}}
\renewcommand{\theequation}{A.\arabic{equation}}

\section{Notations and Besov spaces}
\label{sec:pre}

\subsection{Besov spaces}\label{sub:1}
In this section, we recall the definitions and some key properties of Besov spaces and paraproducts. For a more detailed introduction, we refer to \cite{BCD11, GIP15}.
Let $\theta_{-1},\theta\in C_c^\infty(\mR^3)$ be nonnegative radial functions on $\mathbb{R}^d$, such that

i. the support of $\theta_{-1}$ is contained in a ball, and the support of $\theta$ is contained in an annulus;

ii. $\theta_{-1}(z)+\sum_{j\geq0}\theta(2^{-j}z)=1$ for all $z\in \mathbb{R}^d$.

iii. $\textrm{supp}(\theta_{-1})\cap \textrm{supp}(\theta(2^{-j}\cdot))=\emptyset$ for $j\geq1$ and $\textrm{supp}(\theta(2^{-i}\cdot))\cap \textrm{supp}(\theta(2^{-j}\cdot))=\emptyset$ for $|i-j|>1$.

We call the pair $(\theta_{-1},\theta)$ a dyadic partition of unity, and refer to  \cite[Proposition 2.10]{BCD11} for its existence. The Littlewood-Paley blocks are then defined as follows:
$$\Delta_{-1}u=\mathcal{F}^{-1}(\theta_{-1}\mathcal{F}u),\quad \Delta_{j}u=\mathcal{F}^{-1}(\theta(2^{-j}\cdot)\mathcal{F}u), j\geq0.$$

Let $\mathcal{S}'$ be the space of distributions on $\mT^3$. For $\alpha\in \mathbb{R}, p,q\in [1,\infty]$, the H\"{o}lder-Besov space  is defined by
$$\bB^\alpha_{p,q}=\Big\{u\in\mathcal{S}'(\mathbb{T}^d):\|u\|_{\bB^\alpha_{p,q}}=\Big(\sum_{j\geq-1}(2^{j\alpha}\|\Delta_ju\|_{L^p})^q\Big)^{\frac1q}<\infty\Big\},$$
with the usual interpretation as $l^\infty$ norm in case $q=\infty$. For the shift operator $\tau_y$ introduced in \eqref{def:tauy}, it is easy to see that
\begin{equs}[eq:tauy]
	\|\tau_yf\|_{\bB^\alpha_{p,q}}=\|f\|_{\bB^\alpha_{p,q}}.
	\end{equs}

The following embedding results will  be frequently used. 

\bl\label{lem:emb} (i) Let $1\leq p_1\leq p_2\leq\infty$ and $1\leq q_1\leq q_2\leq\infty$, and let $\alpha\in\mathbb{R}$. Then $\bB^\alpha_{p_1,q_1} \subset \bB^{\alpha-d(1/p_1-1/p_2)}_{p_2,q_2}$. (cf. \cite[Lemma~A.2]{GIP15})

(ii) Let $s\in \R$, $1\leq p\leq\infty$, $\delta>0$. Then
$\bB^s_{p,1}\subset \bB^{s}_{p,\infty}\subset \bB^{s-\delta}_{p,1}$.

Here  $\subset$ means  continuous and dense embedding.
\el

We also recall the following interpolation lemma.

\bl\label{lem:interpolation}

Let  $\theta\in[0,1]$ and $\alpha,\alpha_1,\alpha_2\in\mR$ satisfy $ \alpha=\theta \alpha_1+(1-\theta)\alpha_2,$
and $p,q,p_1,q_1,p_2,q_2\in[1,\infty]$ satisfy
$$
\frac{1}{p}=\frac{\theta}{p_1}+\frac{1-\theta}{p_2},\ \ \frac{1}{q}=\frac{\theta}{q_1}+\frac{1-\theta}{q_2}.
$$
It holds that
\begin{align}\label{DQ1}
	\|f\|_{\bB^\alpha_{p,q}}\leq \|f\|_{\bB^{\alpha_1}_{p_1,q_1}}^\theta\|f\|_{\bB^{\alpha_2}_{p_2,q_2}}^{1-\theta}.
\end{align}
(cf. \cite[Lemma 2.7]{ZZZ20})

\el

\subsection{Smoothing effect of heat flow}
\label{sec:heat}

We recall the following  smoothing effect of the heat flow $P_t=e^{t(\Delta-1)}$ (e.g. \cite[Lemma~A.7]{GIP15}, \cite[Proposition~A.13]{MW18}, \cite[Lemma 2.8]{ZZZ20}).
\vskip.10in
\bl\label{lem:heat}  Let $u\in \bB^{\alpha}_{p,q}$ for some $\alpha\in \mathbb{R}, p,q\in [1,\infty]$. Then for every $\delta\geq0$,
$$\|P_tu\|_{\bB^{\alpha+\delta}_{p,q}}\lesssim t^{-\delta/2}\|u\|_{\bB^{\alpha}_{p,q}}.$$
If $0< \beta< 2$, then
$$\|(\mathrm{I}-P_t)u\|_{L^p}\lesssim t^{\frac{\beta}2}\|u\|_{\bB^{\beta}_{p}}.$$
\el

\begin{lemma}\label{lemma:sch} (\cite[Lemma~A.9]{GIP15},\cite[Lemma~2.8, Lemma~2.9]{ZZZ20})
	Let $\alpha\in\R$.
	Then the following bounds hold for $\sI f=\int_0^\cdot P_{\cdot-s}f\,\dif s$
	\begin{equs}
			\|\sI f\|_{L^\infty_T\bC^{2+\alpha}}\lesssim \|f\|_{L_T^{\infty}\bC^{\alpha}}.
		\end{equs}
	If $0\leqslant 2+\alpha < 2$ then
	\begin{equs}
		\|\sI f\|_{C_T^{(2+\alpha)/2}L^\infty}\lesssim \|f\|_{L_T^\infty\bC^{\alpha}}.
		\end{equs}
\end{lemma}

\subsection{Paraproducts and commutators}\label{sec:para}

Now, we recall the following paraproduct introduced by Bony (see \cite{Bon81}). In general, the product $fg$ of two distributions $f\in \bC^\alpha, g\in \bC^\beta$ is well defined if and only if $\alpha+\beta>0$. In terms of Littlewood-Paley blocks, the product $fg$ of two distributions $f$ and $g$ can be formally decomposed as
$$fg=\sum_{j\geq-1}\sum_{i\geq-1}\Delta_if\Delta_jg=f\prec g+f\circ g+f\succ g,$$
with $$f\prec g=g\succ f=\sum_{j\geq-1}\sum_{i<j-1}\Delta_if\Delta_jg, \quad f\circ g=\sum_{|i-j|\leq1}\Delta_if\Delta_jg.$$
We also denote
\begin{align*}
	\succcurlyeq \, \eqdef \, \succ+\circ,
	\qquad \preccurlyeq\,\eqdef\, \prec+\circ.
\end{align*}
For $j\geq0$ we also use the notations
$$S_jf=\sum_{i\leq j-1}\Delta_if,$$
and $\theta_i=\theta(2^{-i}\cdot)$ for $i\geq0$.

It is easy to see that the support of Fourier of $S_jf\Delta_jg$ is contained in an annulus of the form $2^j\sA$. Let $\tilde \theta\in C^\infty_c(\mR^3)$ with support in an annulus such that $\tilde\theta =1$ on $\sA$. Let $K_j=\mathcal{F}^{-1}\theta_j$, $\tilde{K}_j=\mathcal{F}^{-1}\tilde \theta_j$ with
$\tilde\theta_j=\tilde\theta(2^{-j}\cdot)$.

The following results on paraproduct in  Besov space is from \cite{Bon81} (see also \cite[Lemma~2.1]{GIP15},  \cite[Proposition~A.7]{MW18}).

\begin{lemma}\label{lem:para}
	Let  $\beta\in\R$, $p, p_1, p_2, q\in [1,\infty]$ such that $\frac{1}{p}=\frac{1}{p_1}+\frac{1}{p_2}$. Then we have
	\begin{equs}
		\|f\prec g\|_{\bB^\beta_{p,q}}
		&\lesssim\|f\|_{L^{p_1}}\|g\|_{\bB^{\beta}_{p_2,q}},
		\\
			\|f\prec g\|_{\bB^{\alpha+\beta}_{p,q}}
			&\lesssim\|f\|_{\bB^{\alpha}_{p_1,q}}\|g\|_{\bB^{\beta}_{p_2,q}},   \qquad (\mbox{for }\alpha<0)
				\\
				\|f\circ g\|_{\bB^{\alpha+\beta}_{p,q}}
				&\lesssim\|f\|_{\bB^{\alpha}_{p_1,q}}\|g\|_{\bB^{\beta}_{p_2,q}}, \qquad
				(\mbox{for }\alpha+\beta>0).
			\end{equs}
			Furthermore, for $\alpha+\beta>0$
			\begin{align*}
				\|fg\|_{\bB^{\alpha\wedge\beta}_{p,q}}\lesssim\|f\|_{{\bB}^\alpha_{p_1,q}}\|g\|_{{\bB}^\beta_{p_2,q}}.
			\end{align*}
			Moreover, for $\alpha>0$, $p_3, p_4\in [1,\infty]$ satisfy $\frac1p=\frac1{p_3}+\frac1{p_4}$. Then it holds that
			\begin{align*}
				\|fg\|_{\bB^\alpha_{p,q}}\lesssim \|f\|_{\bB^\alpha_{p_1,q}}\|g\|_{L^{p_2}}
				+\|f\|_{L^{p_3}}\|g\|_{\bB^\alpha_{p_4,q}}.
			\end{align*}
		\end{lemma}

		\bl\label{lem:multi}
		(Duality.) Let $\alpha\in (0,1)$, $p,q\in[1,\infty]$, $p'$ and $q'$ be their conjugate exponents, respectively. Then the mapping  $\langle u, v\rangle\mapsto \int \overline u v \,\dif x$  extends to a continuous bilinear form on $\bB^\alpha_{p,q}\times \bB^{-\alpha}_{p',q'}$, and one has $|\langle u,v\rangle|\lesssim \|u\|_{\bB^\alpha_{p,q}}\|v\|_{\bB^{-\alpha}_{p',q'}}$ (cf.  \cite[Proposition~3.23]{MW17}).

		\el

		We also recall the following commutator estimate (\cite[Lemma 2.4]{GIP15}, \cite[Proposition A.9]{MW18}).
		
		\begin{lemma}\label{lem:com1}
			Let $\alpha\in (0,1)$ and $\beta,\gamma\in \R$ such that $\alpha+\beta+\gamma>0$ and $\beta+\gamma<0$, $p,p_1,p_2,p_3\in [1,\infty]$, $\frac{1}{p}=\frac{1}{p_1}+\frac{1}{p_2}+\frac{1}{p_3}$. Then there exist a trilinear bounded operator $C(f,g,h):\bB^\alpha_{p_1}\times \bB_{p_2}^\beta\times \bB^\gamma_{p_3}\to \bB^{\alpha+\beta+\gamma}_{p}$ satisfying
			$$
			\|C(f,g,h)\|_{\bB^{\alpha+\beta+\gamma}_{p}}\lesssim \|f\|_{\bB^\alpha_{p_1}}\|g\|_{\bB^\beta_{p_2}}\|h\|_{\bB^\gamma_{p_3}}
			$$
			and for smooth functions $f,g,h$
			$$
			C(f,g,h)=(f\prec g)\circ h - f(g\circ h).
			$$
		\end{lemma}

	We also recall the following commutators from \cite[Lemma A.1]{CC15}.
		
		\begin{lemma}\label{lem:com2}
			Let $\alpha\in (0,1)$ and $\beta\in \R$, $p,p_1,p_2\in [1,\infty]$, $\frac{1}{p}=\frac{1}{p_1}+\frac{1}{p_2}$. Let $\phi\in\mathcal{S}$, the space of Schwartz functions. Then for every $\eta\leq1$ it holds that
			$$
			\|\varphi(\eps \nabla)(f\prec g)-f\prec \varphi(\eps \nabla)g \|_{\bB^{\alpha+\beta-\eta}_{p}}\lesssim \eps^{\eta}\|f\|_{\bB^\alpha_{p_1}} \|g\|_{\bB^\beta_{p_2}}.
			$$
			Here $\varphi(\nabla)f=\mathcal{F}^{-1}(\varphi\mathcal F f)=(\mathcal{F}^{-1}\varphi)*f$ and the proportional constant is uniform in $\eps$.
Moreover, it holds that for $T>0$
$$\|[\sI,f\prec]g\|_{C_T\bC^{\alpha+\beta+2}}\lesssim \Big(\|f\|_{C_T\bC^{\alpha}}+\|f\|_{C^{\alpha/2}_TL^\infty}\Big)\|g\|_{C_T\bC^{\beta}}.$$
		\end{lemma}
		\begin{proof} The first result follows from \cite[Lemma A.1]{CC15}. The second result follows from the first result and we refer to \cite[Lemma 3.13]{HZZZ24} for a proof.
		\end{proof}
		
			\bl\label{veps} It holds that for $y\in\mathbb{R}^3,\delta\in(0,1), \alpha\in\mathbb{R},p\in[1,\infty]$
		$$\|\tau_y f-f\|_{\bB_p^{\alpha}}\lesssim |y|^\delta\|f\|_{\bB_p^{\alpha+\delta}}.$$
		Moreover,  if $(1+|x|^{\delta})v\in L^1(\mathbb{R}^3)$, then it holds that
		$$\|v^\eps*f-f\|_{\bB_p^{\alpha}}\lesssim \eps^\delta \|f\|_{\bB_p^{\alpha+\delta}}.$$
		\el	
		\begin{proof} The first result follows from \cite[Corollary 2.9]{HZZZ24}. We have
			$$v^\eps*f-f=\int v(y)(f(x-\eps y)-f(x))\,\dif y.$$
			Thus the result follows from the first result.
		\end{proof}
	
We recall the following result from \cite[Proposition 10]{BG18}.

		\begin{lemma}\label{lem:com3}
		Let $\alpha,\beta,\gamma\in \R$ such that $\alpha+\beta+\gamma>0$ and $\beta+\gamma<0$. Then there exist a trilinear bounded operator $D(f,g,h):H^\alpha\times \bC^\beta\times H^\gamma\to \R$ satisfying
		$$
		|D(f,g,h)|\lesssim \|f\|_{H^\alpha}\|g\|_{\bC^\beta}\|h\|_{H^\gamma}
		$$
		and for smooth functions $f,g,h$
		$$
		D(f,g,h)=\int f (g\circ h)  - \int (f\prec g) h .
		$$
	\end{lemma}

For $J_t^N$ introduced in Section \ref{sec:part} we have the following result.

	\begin{lemma}\label{lem:com4} For $t>0$, there exists $\sR_t$ of bounded multilinear forms on $\bC^{-1-\kappa}\times \bC^{-1-\kappa}\times H^{\frac12-\kappa}\times H^{\frac12-\kappa}$, $\kappa>0$, such that, for smooth $ f_1, f_2, g_1,g_2$, it holds that
		\begin{align*}
			\sR_t(f_1,f_2,g_1,g_2)=\int [(J_t^N)^2(g_1\prec f_1)](g_2\prec f_2)-((J_t^N)^2 f_1\circ  f_2) g_1g_2 \d x,
		\end{align*}
		and
		\begin{align*}
			|	\sR_t(f_1,f_2,g_1,g_2)|\lesssim \frac1{( t+1)^{1+\kappa}}\|f_1\|_{\bC^{-1-\kappa}}\|f_2\|_{\bC^{-1-\kappa}}\|g_1\|_{H^{\frac12-\kappa}}\|g_2\|_{H^{\frac12-\kappa}},
		\end{align*}
		with the proportional constant independent of $\eps, N$ and $t$.
		
	\end{lemma}
\begin{proof} We have
\begin{align*}&\Big|\int [(J_t^N)^2(g_1\prec f_1)](g_2\prec f_2)\d x- \int [g_1\prec(J_t^N)^2f_1](g_2\prec f_2)\d x\Big|
\\\lesssim&\,\|(J_t^N)^2(g_1\prec f_1)-g_1\prec(J_t^N)^2f_1\|_{H^{1+2\kappa}}\|g_2\prec f_2\|_{H^{-1-2\kappa}}.
\end{align*}
Now we bound the first term by similar argument as Lemma \ref{lem:com2}: We write $(J_t^N)^2=\mathcal{F}^{-1}\varphi\mathcal{F}$ with $\varphi(k)=\frac{\sigma_t^2(k)\chi_N(k)^2}{2\la k\ra^2}=\frac{-\chi_t(k)\chi_t'(k)\chi_N(k)^2}{\la k\ra (t+1)^2}$. Since the support of Fourier transform of $S_{j-1}f\Delta_jg$ is contained in an annulus of the form $2^j\sA$, we have
			\begin{align*}&(J_t^N)^2S_{j-1}g_1 \Delta_j f_1-S_{j-1}g_1(J_t^N)^2\Delta_jf_1
				\\=&-\int y \mathcal{F}^{-1}(\tilde\theta_j\varphi)(y) \int_0^1\nabla S_{j-1}g_1(x-\tau y)\,\dif \tau \Delta_j f_1(x-y) \,\dif y.
			\end{align*}
	By direct calculation we derive
			\begin{align*}&\|(J_t^N)^2S_{j-1}g_1 \Delta_j f_1-S_{j-1}g_1(J_t^N)^2\Delta_jf_1\|_{L^2}
				\\\lesssim&\int |y|| \mathcal{F}^{-1}(\tilde\theta_j\varphi)(y)|   \,\dif y \|\nabla S_{j-1}g_1\|_{L^{2}}\|\Delta_jf_1\|_{L^{\infty}},
			\end{align*}
		Moreover, we obtain for $\delta\in[0,1]$
\begin{equs}[est:vx] &\big\||x|^\delta\mathcal{F}^{-1}(\tilde\theta(2^{-j}\cdot) \varphi)\big\|_{L^1(\mR^3)}
=2^{-j\delta}\big\||x|^\delta\mathcal{F}^{-1}(\tilde\theta  \varphi(2^j\cdot))\big\|_{L^1(\mR^3)}
\\\lesssim&\,2^{-j\delta}\big\|\mathcal{F}^{-1}[(I-\Delta)^2(\tilde\theta \varphi( 2^j\cdot))]\big\|_{L^2(\mR^3)}=2^{-j\delta}\big\|(I-\Delta)^2(\tilde\theta \varphi( 2^j\cdot))\big\|_{L^2(\mR^3)}
\\\lesssim&\,2^{-j\delta}\Big(1+\sum_{\eta,|\eta|\leq 4}| 2^j|^{|\eta|}\|\partial^\eta \varphi( 2^j\cdot)\|_{L^\infty(supp(\tilde\theta))}\Big)
\lesssim2^{-j(\delta+2-\kappa)}\frac{1}{(t+1)^{1+\kappa}},
\end{equs}
where in the last step we used on the support of  $\tilde \theta \varphi(2^j\cdot)$, $2^j\simeq t+1$. Thus we derive
	\begin{align*}\|(J_t^N)^2(g_1\prec f_1)-g_1\prec(J_t^N)^2f_1\|_{H^{1+2\kappa}}
\lesssim&\,\frac1{( t+1)^{1+\kappa}}\|f_1\|_{\bC^{-1-\kappa}}\|g_1\|_{H^{\frac12-\kappa}},
			\end{align*}
which combined with Lemma \ref{lem:para} implies that
\begin{align*}&\Big|\int [(J_t^N)^2(g_1\prec f_1)](g_2\prec f_2)\d x- \int [g_1\prec(J_t^N)^2f_1](g_2\prec f_2)\d x\Big|
\\\lesssim&\,\frac1{( t+1)^{1+\kappa}}\|f_1\|_{\bC^{-1-\kappa}}\|f_2\|_{\bC^{-1-\kappa}}\|g_1\|_{H^{\frac12-\kappa}}\|g_2\|_{H^{\frac12-\kappa}}.
\end{align*}
Using Lemma \ref{lem:com3} and Lemma \ref{lem:para} we obtain
\begin{align*}&\Big|\int g_1 \Big([(J_t^N)^2f_1]\circ (g_2\prec f_2)\Big)\d x- \int [g_1\prec(J_t^N)^2f_1](g_2\prec f_2)\d x\Big|
\\\lesssim&\,\|(J_t^N)^2f_1\|_{\bC^{1-2\kappa}}\|f_2\|_{\bC^{-1-\kappa}}\|g_1\|_{H^{\frac12-\kappa}}\|g_2\|_{H^{\frac12-\kappa}}
\\\lesssim&\,\frac1{( t+1)^{1+\kappa}}\|f_1\|_{\bC^{-1-\kappa}}\|f_2\|_{\bC^{-1-\kappa}}\|g_1\|_{H^{\frac12-\kappa}}\|g_2\|_{H^{\frac12-\kappa}},
\end{align*}
where in the last step we used \eqref{est:vx} to deduce $\|(J_t^N)^2f_1\|_{\bC^{1-2\kappa}}\lesssim \frac1{( t+1)^{1+\kappa}}\|f_1\|_{\bC^{-1-\kappa}}.$
Similarly using Lemma \ref{lem:com1} we have
\begin{align*}&\Big|\int g_1 \Big([(J_t^N)^2f_1]\circ (g_2\prec f_2)\Big) -((J_t^N)^2 f_1\circ  f_2) g_1g_2\d x\Big|
\\\lesssim&\,\frac1{( t+1)^{1+\kappa}}\|f_1\|_{\bC^{-1-\kappa}}\|f_2\|_{\bC^{-1-\kappa}}\|g_1\|_{H^{\frac12-\kappa}}\|g_2\|_{H^{\frac12-\kappa}}.
\end{align*}
Summarizing the above calculations, the result follows.
\end{proof}

		\renewcommand{\theequation}{B.\arabic{equation}}	
		\section{Proof of Theorem \ref{th:global}, Theorem \ref{th:1} and  Lemma \ref{lem:zz1}} \label{app:pro}
		\begin{proof}[Proof of Theorem \ref{th:global}] For fixed $\eps>0$, due to the smoothing effect of $v^\eps$,  $\cZ_\eps^{\<2c>}\in  C_T\bC^{-\frac12-\kappa}$  $\cZ_\eps^{\<3v>}\in C_T\bC^{-\frac12-\kappa}$  $\bP$-a.s..

			It is standard to derive a local in time solution in $\bC^{-\frac12-\kappa}$ by fixed point arguments in suitable space, following similar arguments as in \cite{DD03, MW17}. We then use invariant measure $\nu^\eps$ to construct global in time solutions. More precisely,
			recall the following potential term from $\nu^\eps$
			$$\cD[\Phi]\eqdef	\frac12\int \Wick{|\Psi(x)|^2}v^\eps(x-y)\Wick{|\Psi(y)|^2}\,\dif x\dif y\no
			-{(a^\eps-6 b^\eps-m+1)\int_{\mathbb T^d} \Wick{|\Psi(x)|^2}\,\dif x},$$
			which can be approximated by the following Galerkin approximation:
			\begin{align*}
				\cD[P_N\Phi]\eqdef&\,	\frac12\int (|\Psi_N(x)|^2-\E |\Psi_N(x)|^2)v^\eps(x-y)(|\Psi_N(y)|^2-\E|\Psi_N(y)|^2)\,\dif x\dif y
				\\&-{(a^\eps- 6b^\eps-m+1)\int (|\Psi_N(x)|^2-\E |\Psi_N(x)|^2)\,\dif x}
				\\=&\,\frac12\int \Big(|\Psi_N(x)|^2-\E |\Psi_N(x)|^2-(a^\eps- 6b^\eps-m+1)\Big)v^\eps(x-y)
				\\&\times\Big(|\Psi_N(y)|^2-\E|\Psi_N(y)|^2-(a^\eps-6 b^\eps-m+1)\Big)\,\dif x\dif y-T^\eps,
			\end{align*}
			with $\Psi_N=P_N\Psi$ for $P_Nf=\cF^{-1}1_{|k|\leq N}\cF f$ and $T^\eps=4\pi^3(a^\eps-6b^\eps-m+1)^2$. We can then apply the same arguments as in \cite[Lemma 5.3]{LewNamRou-21} to have $\cD[P_N\Psi]\to \cD[\Psi]$ in $L^1(\mu_0)$ for Gaussian free field $\mu_0$, as $N\to\infty$.
			We further construct the probability measure
			\begin{align}\label{def:nueps}
				\nu^\eps=\frac1{\sZ_\eps}e^{-\cD}\,\dif \mu_0, \qquad \sZ_\eps=\int e^{-\cD}\,\dif \mu_0,\end{align}
			with $V^\eps\geq -T^\eps$.
			
			Using  Galerkin approximation
			we establish that  the measure $\nu^\eps$ is an invariant measure of the solutions $\Psi^\eps$ to equation \eqref{eq:mainnew}.
			Another way to derive that $\nu^\eps$ is an invariant measure of the solutions to equation \eqref{eq:mainnew} is through the Dirichlet form approach. Since $\nu^\eps$ is absolutely continuous with respect to $\mu_0$, with its density in $L^p(\mu_0)$ for every $p\geq1$, we can readily apply the general Dirichlet form theory from \cite{AR91} to construct a Markov process that leaves $\nu^\eps$ as an invariant measure for equation \eqref{eq:mainnew}. Moreover, by following the same approach as in \cite[Theorem 3.9]{RZZ17}, we conclude that this Markov process coincides with the local solutions to equation \eqref{eq:mainnew}.
			
			Let $\zeta_{\Psi^\eps(0)}$ denote the blow-up time of $\Psi^\eps$ in $\bC^{-\frac12-\kappa}$ staring from $\Psi^\eps(0)$. We then use $\nu^\eps$ to apply Bourgain's argument \cite{Bou} to extend the local-in-time solution to a global one for $\nu^\eps$-a.s. initial data in $\bC^{-\frac12-\kappa}$, meaning that $\bP(\zeta_{\Psi^\eps(0)} = \infty) = 1$ for $\nu^\eps$-a.s. $\Psi^\eps(0) \in \bC^{-\frac12-\kappa}$ (see also \cite{DD03}).
			Moreover, by a general result from \cite{HM18}, we can show that the Markov semigroup formed by the solutions to equation \eqref{eq:mainnew} is strong Feller. This implies that for every $t \geq 0$, the map $\Psi^\eps(0) \to \bP(t < \zeta_{\Psi^\eps(0)})$ is continuous. Consequently, we conclude that $\bP(\zeta_{\Psi^\eps(0)} = \infty) = 1$ for every $\Psi^\eps(0) \in \bC^{-\frac12-\kappa}$. Moreover, by applying \cite[Corollary 3.9]{HM18}, the strong Feller property of the Markov semigroup, and the fact that $\nu^\eps$ is supported on $\bC^{-\frac12-\kappa}$, we conclude that $\nu^\eps$ is the unique invariant measure of the solutions to equation \eqref{eq:mainnew}. The result follows.
		\end{proof}
	
	\begin{proof}[Proof of Theorem \ref{th:1}]	
	During the proof, we use $\|\mZ_\eps(t)\|$, $\|(\mZ_\eps - \mZ)(t)\|$, and $\|\mZ(t)\|$ to denote the quantities $\|\mZ_\eps\|$, $\|\mZ_\eps - \mZ\|$, and $\|\mZ\|$ introduced in Section \ref{sec:dif}, where these norms are now considered on the interval $[0,t]$ for the random fields. We also introduce the following random time: Define for any $L\geq1$
		$$\tau^\varepsilon_L\eqdef\inf\{ t\geq0:\|\psi_\varepsilon(t)\|_{\bC^{-\frac12-\kappa}}\geq L\}\wedge T,\quad\rho_{L}^\varepsilon\eqdef\inf\{ t\geq0:\|\mZ_\eps(t)\|\geq L\}.$$
		We first have the following bound before $\tau_L^\eps\wedge\rho^\eps_{L_1}$, $L,L_1\geq1$:
		Set $$Q^\eps(t)\eqdef\|\psi_\eps(t)\|_{\bC^{-\frac12-\kappa}}+t^{\frac{1+3\kappa}{2}}\|\psi_\eps(t)\|_{\bC^{\frac{1}{2}+2\kappa}}
		+t^{\frac{3+8\kappa}{4}}\|\psi_\eps^{\sharp}(t)\|_{\bC^{1+3\kappa}}+1.$$
		It holds that for $t\leq \tau^\varepsilon_L\wedge\rho_{L_1}^\eps$
		\begin{equs}[boundq]Q^\eps(t)\lesssim C(L,L_1),
		\end{equs}
		with the proportional constant independent of $\eps$.
		In fact, by similar calculations as in \cite[Section 4]{ZZ18}  there exists $q>1$ such that for $0\leq t\leq \rho_{L_1}^\eps\wedge \tau_L^\eps$
		$$Q^\eps(t)^q\leq C(L_1)(\|\psi(0)\|_{\bC^{-\frac12-\kappa}}^q+1)+C(L_1)\int_0^tQ^\eps(r)^{3q}\,\dif r.$$
	 Then Bihari's inequality implies that there exists a short time $t^*=\tilde C(L,L_1)>0$ such that
		$$\sup_{t\in [0,t^*\wedge\tau^\varepsilon_L\wedge\rho_{L_1}^\eps]}Q^\eps(t)\leq C(L_1,L).$$
Consider the solution at time $t^*<t\leq \rho_{L_1}^\eps\wedge \tau_L^\eps$, then it can be viewed as a solution starting from $t-\frac{t^*}2$. A similar argument as above implies for $t^*<t\leq \rho_{L_1}^\eps\wedge \tau_L^\eps$
$$(\frac{t^*}2)^{\frac{1+3\kappa}{2}}\|\psi_\eps(t)\|_{\bC^{\frac{1}{2}+2\kappa}}
		+(\frac{t^*}2)^{\frac{3+8\kappa}{4}}\|\psi_\eps^{\sharp}(t)\|_{\bC^{1+3\kappa}}\lesssim C(L,L_1),$$
which implies \eqref{boundq}.

		For $L\geq0$ define
		$$\tau_L\eqdef\inf\{ t\geq0:\|\phi(t)\|_{\bC^{-\frac12-\kappa}}\geq L\}\wedge T,\qquad
		\bar{\rho}_L\eqdef\inf\{ t\geq0:\|\mZ(t)\|\geq L\},$$
		and
		$$\tilde \rho^\eps\eqdef\inf\{ t\geq0:\|(\mZ-\mZ_\eps)(t)\|\geq \eps^{\frac\kappa4}\}.$$
		Then using Lemma \ref{lem:con1}--Lemma \ref{lem:dif1} and \eqref{dif:esb}, \eqref{bd:cRY}, \eqref{bdd:mZ}  above, \eqref{boundq} and similar argument as \cite[Section 4]{ZZ18} we have for $L, L_i\geq1$ with $i=1,2,3$,
			\begin{equs}[convergence]\sup_{t\in[0,\tau_L\wedge \tau_{L_1}^\varepsilon\wedge \rho_{L_2}^\varepsilon\wedge\bar{\rho}_{L_3}\wedge \tilde \rho^\eps]}\|\psi^{\varepsilon}(t)-\phi(t)\|_{\bC^{-\frac12-\kappa}}\rightarrow^{\bP}0, \quad\varepsilon\rightarrow0.\end{equs}
		Moreover, we have the following estimates: for $\eta>0$
		$$\aligned &\bP(\sup_{t\in[0,T]}\|\psi^\varepsilon-\phi\|_{\bC^{-\frac12-\kappa}}>\eta)\\\leq&\, \bP(\sup_{t\in[0,\tau_L\wedge \tau_{L_1}^\varepsilon\wedge \rho_{L_2}^\varepsilon\wedge\bar{\rho}_{L_3}\wedge \tilde \rho^\eps]}\|\psi^\varepsilon-\phi\|_{\bC^{-\frac12-\kappa}}>\eta)+\bP(\tau_L\wedge\rho_{L_2}^\varepsilon\wedge\bar{\rho}_{L_3}\wedge \tilde \rho^\eps>\tau_{L_1}^\varepsilon)
\\&+\bP(T>\tau_L)+\bP(T>\tilde \rho^\eps)+\bP(T>\rho_{L_2}^\varepsilon)+\bP(T>\bar{\rho}_{L_3}).\endaligned$$
		The first term goes to zero  as $\varepsilon\rightarrow0$ by \eqref{convergence}. Also  for $L_1>L+\eta$
		$$\bP(\tau_L\wedge\rho_{L_2}^\varepsilon\wedge\bar{\rho}_{L_3}\wedge \tilde \rho^\eps>\tau_{L_1}^\varepsilon)\leq \bP(\sup_{t\in[0,\tau_L\wedge \tau_{L_1}^\varepsilon\wedge \rho_{L_2}^\varepsilon\wedge\bar{\rho}_{L_3}\wedge \tilde \rho^\eps]}\|\psi^\varepsilon-\phi\|_{\bC^{-\frac12-\kappa}}>\eta),$$
		which goes to zero by \eqref{convergence}.
		 The third term tends to zero as $ L$ go to $\infty$ by Theorem \ref{th:phi4}. The fourth term goes to zero as $\eps\to0$ by \eqref{bdd:mZ}.
		The last two terms go to zero uniformly over $\varepsilon\in(0,1)$ as $L_2,L_3$ go to $\infty$ by Lemma \ref{lem:bdZ} and \eqref{bd:mZ}.
		Thus the result follows.
		\end{proof}

	\begin{proof}[Proof of Lemma \ref{lem:zz1}] Let $ \tilde\Psi^\eps$ and $\tilde Z$ be solutions to \eqref{eq:mainnew} and \eqref{eq:lin} with general initial conditions, respectively.
		By the general results of \cite{HM18},
		$(\tilde\Psi^\eps, \tilde Z)$ is a Markov process on $(\bC^{-\frac12-\kappa})^{2}$, and we denote by $(P_t^\eps)_{t\geq0}$ the associated Markov semigroup.  To derive the desired structural properties about the limiting measure, we will follow the Krylov-Bogoliubov construction with a specific choice of initial condition that allows to exploit the uniform estimate from Theorem \ref{coming}.
		Namely, we denote by $\tilde\Psi^\eps$ the solution to \eqref{eq:mainnew} starting from
		$Z(0)$ where $Z$ is  the stationary solution to \eqref{eq:lin}, so that the process $\tilde\Psi^\eps-\tilde Z$ starts from the origin. In this case $\tilde Z$ is the same as the stationary solution $Z$.
		By Theorem \ref{coming} for every $T\geq 1$ and $\kappa>0$
		\begin{align*}
			\int_0^T\mathbf{E}\Big(\| (\tilde\Psi^\eps- \tilde Z)(t)\|_{\bC^{-\frac{1+\kappa}2}}^2\Big)\,\dif t & \lesssim T,
		\end{align*}
		where the implicit constant is independent of $T$. In fact, as the uniform bounds are independent of initial data, we can have the same estimates for the solution on $[n,n+1], n\in\mN_0$, i.e. by Proposition \ref{energy}
		\begin{align*}
			&\int_n^{n+1}\mathbf{E}\Big(\| (\tilde\Psi^\eps- \tilde Z)(t)\|_{\bC^{-\frac{1+\kappa}2}}^2\Big)\,\dif t
			\\ \lesssim &\int_n^{n+1}\E\|\psi_l(t)\|_{H^1}^2\,\dif t+\int_n^{n+1}\E\|\psi_h(t)\|_{\bB^{1-2\kappa}_4}^2\,\dif t
			\\&+\int_n^{n+1}\E\|\cZ^{\<3v0m>}(t)\|_{\bC^{-\frac{1+\kappa}2}}^2\,\dif t+\int_n^{n+1}\E\|\sI(\cR(Z))(t)\|_{\bC^{-\frac{1+\kappa}2}}^2\,\dif t
			\\\lesssim&\, \E\|\psi_l(n)\|_{L^2}^2+1\lesssim 1,
		\end{align*}
		where in the last step, we used the fact that the moment bounds for the stochastic terms are uniform with respect to time shifts (see e.g. \cite{MW18}) and Theorem \ref{coming}. In the second step we used  \eqref{boundpsih1} to have
		\begin{align*}\E\int_n^{n+1} \|\psi_h\|^2_{\bB^{1-2\kappa}_4}\,\dif s\lesssim& \,1+\E\Big( K(\|\mZ_\eps\|)\int_n^{n+1}\int_n^s(s-r)^{-1+\frac\kappa2}\|\psi_l(r)\|_{L^4}^2\,\dif r \dif s\Big)
			\\\lesssim &\,1+\E\int_n^{n+1}\|\psi_l\|_{H^1}^2\,\dif s+\E\int_n^{n+1}\cV^\eps(\psi_l)\,\dif s\lesssim1.
		\end{align*}
		Here the proportional constant independent of $n$. Taking sum for $n$ we derive the estimate.
		
		We then apply Krylov-Bogoliubov existence theorem (see \cite[Corollary 3.1.2]{DZ96}) as in the proof of \cite[Lemma 4.2]{SZZ21} to construct an invariant measure $\pi^\eps$ on $(\bC^{-\frac12-\kappa})^2$ for $(P_t^\eps)_{t\geq0}$ and the  desired stationary process $(\Psi^\eps, Z)$, defined to be the unique solution to \eqref{eq:mainnew1} and \eqref{eq:lin} obtained by sampling the initial datum  from $\pi^{\eps}$.
	\end{proof}
	
\renewcommand{\theequation}{C.\arabic{equation}}	
	\section{The ideal Bose gas}\label{sec:ideal-gas}

In the grand canonical ensemble, the non-interacting (ideal) Bose gas in $\bT^3$ at temperature $\lambda^{-1}>0$ and with chemical potential $-\vartheta_0<0$ is described by the Gibbs state on Fock space
\begin{equation*}
    \Gamma_0 = \cZ_0^{-1} e^{-\lambda \dG(-\Delta+\vartheta_0)},\quad \cZ_0 =\Tr e^{-\lambda (-\Delta+\vartheta_0)}.
    \end{equation*}
This model is exactly solvable. In particular, the total number of particles in $\Gamma_0$ is 
\begin{equs}\label{eq:N-kappa0-app}
M = \sum_{k\in \bZ^3} \frac{1}{e^{\lambda (|k|^2+\vartheta_0)}-1},
\end{equs}
which is proportional to $\lambda^{-3/2}$ when $\lambda\to 0$ if $\lambda^{1/2} \lesim \vartheta_0 \lesim \lambda^{-1}$.

The {\em Bose--Einstein condensation} is the phenomenon when the number of particles in the zero-momentum mode $M_0= \Tr [a_0^* a_0 \Gamma_0]=(e^{\lambda \vartheta_0}-1)^{-1}$ is comparable to $M$.
As Bose  \cite{Bose-24} and Einstein \cite{Einstein-24} realized in 1924, in the limit $\lambda\to 0$, i.e. $M\to \infty$, we have
$$
\frac{M_0}{M} \simeq \left[ 1 - \left( \frac{\lambda_{\rm c}}{\lambda} \right)^{3/2} \right]_+,\quad \lambda_{\rm c} =\pi \left( \frac{M}{\zeta(3/2)} \right)^{-2/3},   
$$
which implies a phase transition when $\lambda_c/\lambda$ crosses the critical value $1$. The parameter $\lambda_c^{-1}$ is called the critical temperature (here we are in a fixed volume setting, hence $\lambda_c^{-1}\backsimeq N^{2/3}$. A more detailed analysis using (\ref{eq:N-kappa0-app})
shows that  if we fix $\lambda/\lambda_c= \alpha\in (0,\infty)$, then
\begin{equs}\label{eq:BEC-phase-transition-simple}
\begin{cases}
M_0 \backsimeq M, \quad \vartheta_0 \backsimeq  \lambda^{1/2} &\quad \text{if}\quad  \alpha < 1 \quad \text {(condensed phase)}, \\
M_0 \backsimeq 1, \quad \vartheta_0 \backsimeq   \lambda^{-1} &\quad \text{if}  \quad \alpha > 1 \quad \text {(non-condensed phase)}.
\end{cases}
\end{equs}

To understand further details of the phase transition, we need to zoom in at the critical point with a specific rate of convergence. The choice $\vartheta_0 \backsimeq 1$,
as used in \cite{LewNamRou-21,FKSS22} and in the present paper, is special since it ensures that the contributions from all momentum modes (both zero and nonzero) are comparable, thus naturally leading to the emergence of the $\Phi^4_3$ theory. This choice places us in the non-condensed phase, but just slightly above the critical point. The condensed phase requires $\vartheta_0 \backsimeq \lambda^{1/2}$, which we do not consider here; see \cite{DeuNamNap-25} for recent  results in this case (with  a weaker interaction potential).

Concerning the total number of particles, we have the following expansion in terms of $\lambda$.

\begin{lemma}[Particle number of the ideal Bose gas]\label{lem:density-rho0}  For $\lambda>0$ small we have
\begin{equs}[eq:equivalent-sum-integral-rho0]
\sum_{k\in\mZ^3}\frac{\lambda}{e^{\lambda(|k|^2+1)}-1}=\frac{\pi^{3/2} \zeta\left(\frac{3}2\right)}{\sqrt{\lambda}} - 2\pi^2 +{2\pi^2 } \sum_{\ell\in2\pi\mZ^3\setminus\{0\}}\frac{e^{-|\ell|}}{|\ell|}
 + O\left( \sqrt{\lambda} \right).
\end{equs}
\end{lemma}

This formula is the same as \cite[Eq. (B.5)]{LewNamRou-21}, except for a rescaling by the volume $(2\pi)^3$ and a quantitative error estimate.

\begin{proof}
		Using $\lambda (e^{\lambda (|k|^2+1)}-1)^{-1}= \lambda \sum_{n\geq1} e^{-n\lambda (|k|^2+1)}$ and the Poisson summation formula
		$$\sum_{k\in\mZ^3}\widehat{f}(k)={(2\pi)^{3}}  \sum_{\ell\in2\pi \mZ^3}f(\ell)$$
		for $f(x)=e^{-n\lambda |x|^2}-\1_{[-\frac12,\frac12]^3}*e^{-n\lambda| x|^2}$, we have
		\begin{equs}[eq:Poisson]
			&\lambda\left(\sum_{k\in \mZ^3}\frac{1}{e^{\lambda(|k|^2+1)}-1}-\int_{\R^3}\frac{\d p}{e^{\lambda(|p|^2+1)}-1}\right) \nn\\
			&=\lambda\sum_{n\geq1}e^{-n\lambda}\sum_{k\in\mZ^3}\left(e^{-n\lambda |k|^2}-\int_{(-\frac12,\frac12)^3}e^{-n\lambda |k-p|^2}\,\d p\right)\nn\\
			&={\pi^{3/2}} \lambda \sum_{n\geq1}\frac{e^{-n\lambda}}{(n\lambda)^{\frac{3}2}}\sum_{\ell\in2\pi \mZ^3\setminus\{0\}} e^{-\frac{|\ell|^2}{4n\lambda}}.
		\end{equs}
		For every $\ell\in  2\pi \mZ^3\setminus\{0\}$, we have the convergence of the Riemann sum
		\begin{equs}[eq:Riemann-sum-Yukawa]
			 \sum_{n\geq1}\frac{\lambda}{(4\pi n\lambda)^{\frac{3}2}} e^{-n\lambda - \frac{|\ell|^2}{4n\lambda}}  \to \int_0^\ii \frac{1}{(4\pi t)^{\frac{3}2}} e^{-t -\frac{|\ell|^2}{4t}}\,\d t ={\frac1{4\pi}} \frac{e^{-|\ell|}}{|\ell|},
		\end{equs}
		when $\lambda\to 0$. Here in the last equality we used the Fourier transform of Yukawa potential, see~\cite[Theorem 6.23]{LieLos-01}. In fact, the convergence rate in \eqref{eq:Riemann-sum-Yukawa} can be estimated as
		\begin{align*}
			&\left| \sum_{n\geq1}\frac{\lambda }{(n\lambda)^{\frac{3}2}} e^{-n\lambda - \frac{|\ell|^2}{4n\lambda}} -  \int_0^\ii \frac{1}{t^{\frac{3}2}} e^{-t -\frac{|\ell|^2}{4t}}\,\d t \right| \\
			&= \left| \sum_{n\geq1} \left( \frac{\lambda}{(n\lambda)^{\frac{3}2}} e^{-n\lambda - \frac{|\ell|^2}{4n\lambda}} -  \int_{n-1}^n \frac{\lambda}{(t\lambda)^{\frac{3}2}} e^{-\lambda t -\frac{|\ell|^2}{4\lambda t}}\,\d t \right) \right| \\
			&\le \sum_{n\ge 1} \left| \frac{\lambda}{(n\lambda)^{\frac{3}2}} e^{-n\lambda - \frac{|\ell|^2}{4n\lambda}} -  \int_{n-1}^{n} \frac{\lambda}{(t\lambda)^{\frac{3}2}} e^{-\lambda t -\frac{|\ell|^2}{4\lambda t}}\,\d t \right|  \le  \sum_{n\ge 1} \int_{n-1}^{n} \left|  \partial_t \left(  \frac{\lambda}{(t\lambda)^{\frac{3}2}} e^{-\lambda t -\frac{|\ell|^2}{4\lambda t}} \right) \right| \d t \\
			&\lesim \int_0^\infty \left(\frac{1}{\sqrt \lambda t^{5/2}} + \frac{\sqrt \lambda}{t^{3/2}} + \frac{|\ell|^2}{4\lambda^{3/2} t^{7/2}} \right)  e^{-\lambda t -\frac{|\ell|^2}{4\lambda t}} \d t \\
			&\lesim \int_0^\infty  \left(\frac{1}{\sqrt \lambda t^{5/2}} \left( \frac{\lambda t}{|\ell|^2}\right)^{3}  + \frac{\sqrt \lambda}{t^{3/2}} \left( \frac{\lambda t}{|\ell|^2}\right)^{2}  + \frac{|\ell|^2}{\lambda^{3/2} t^{7/2}} \left( \frac{\lambda t}{|\ell|^2}\right)^{4}  \right) e^{-\lambda t} \d t   \\
			&\lesim  \frac{\lambda^{5/2}}{|\ell|^4} \int_0^\infty t^{1/2} e^{-\lambda t} \d t \lesim \frac{\lambda}{|\ell|^4}.
		\end{align*}
		Here we used $e^{-|\ell|^2/(4t \lambda)} \lesssim_s (t\lambda/|\ell|^2)^s$ for all $s\in \{2,3,4\}$. Since $|\ell|^{-4}$ is summable in $2\pi \mathbb{Z}^3\backslash\{0\}$, we deduce from \eqref{eq:Poisson} and \eqref{eq:Riemann-sum-Yukawa}
		\begin{equs}[eq:equivalent-sum-integral-rho0]
			\sum_{k\in\mZ^3}\frac{\lambda}{e^{\lambda(|k|^2+1)}-1}= \int_{\R^d}\frac{\lambda \d k}{e^{\lambda (|k|^2+1)}-1}+ \sum_{\ell \in 2\pi \mathbb{Z}^3\backslash\{0\}} {2\pi^2} \frac{e^{-|\ell|}}{|\ell|} + O\left( \sqrt{\lambda} \right).
		\end{equs}
		The conclusion follows from  from the well-known formula
		\begin{equs}[eq:expansion_integral]
			\int_{\R^3}\frac{\lambda^{3/2} \d k}{e^{\lambda (|k|^2+1)}-1} =  \int_{\R^3}\frac{\d k}{e^{|k|^2+\lambda}-1}
			=  \pi^{3/2}\zeta\left(\frac{3}2\right) -
			2 \pi^2 \sqrt{\lambda} +O(\lambda)_{\lambda \to0^+}.
		\end{equs}
	\end{proof}


\begin{thebibliography}{BDLSVD19}
	\bibitem[AM01]{ArnMoo-01}
	{ P.~Arnold and G.~Moore}, {{BEC} transition temperature of a dilute
		homogeneous imperfect {B}ose gas}, {\em Phyical Review Letters}, 87,
	p.~120401, 2001.
	
		\bibitem[AK20]{AK17} S. Albeverio and S. Kusuoka. The invariant measure and the flow associated to the $\phi^4_3$-quantum field model. {\em Ann. Sc. Norm. Super. Pisa Cl. Sci.} (5) 20  no. 4, 1359--1427, 2020.

		\bibitem[AR91]{AR91} S. Albeverio and M. R\"ockner. Stochastic differential equations in infinite dimensions: solutions via Dirichlet forms.
	{\em Probab. Theory Related Fields}, 89(3),  347--386, 1991.
	
	\bibitem[BCCH21]{BCCH21}Y. Bruned, A. Chandra, I. Chevyrev, and M. Hairer. Renormalising
	SPDEs in regularity structures, {\em Journal of the European Mathematical Society}, Vol:23,  869-947, 2021.

	\bibitem[BCD11]{BCD11} H. Bahouri, J.-Y. Chemin, R. Danchin,  Fourier analysis and nonlinear
	partial differential equations, vol. 343 of Grundlehren der Mathematischen
	Wissenschaften [Fundamental Principles of Mathematical Sciences]. Springer, Heidelberg,
	2011.
	
	\bibitem[Bos24]{Bose-24} {S. N. Bose},  Plancks Gesetz und Lichtquantenhypothese, {\em Zeitschrift für Physik} 26, pp. 178--181, 1924.


	\bibitem[BHZ19]{BHZ19} Y. Bruned, M. Hairer, and L. Zambotti. Algebraic renormalisation of
	regularity structures. {\em Invent. Math.} 215, no. 3,  1039--1156, 2019.

	\bibitem[Bon81]{Bon81} J.-M. Bony,  Calcul symbolique et propagation des singularit\'{e}s pour les \'{e}quations
	aux d\'{e}riv\'{e}es partielles non lin\'{e}aires. {\em Ann. Sci. \'{E}cole Norm. Sup.} (4) 14, no. 2,
	209--246, 1981.
	
	\bibitem[Bou94]{Bou}J. Bourgain, Periodic nonlinear Schrödinger equation and invariant measures, {\em Comm. Math.
	Phys.} 166, no. 1, 1--26, 1994.
	
	
	\bibitem[Bou96]{Bourgain-96}
	\leavevmode\vrule height 2pt depth -1.6pt width 23pt, { Invariant measures
		for the 2d-defocusing nonlinear {S}chr{\"o}dinger equation}, {\em Comm. Math.
	Phys.}, 176, pp.~421--445, 1996.
	
	\bibitem[Bou97]{Bourgain-97}
	\leavevmode\vrule height 2pt depth -1.6pt width 23pt, { Invariant measures
		for the {G}ross-{P}itaevskii equation}, {\em J. Math. Pures Appl.}, 76,
	pp.~649--02, 1997.
	
	\bibitem[BB14a]{BouBul-14a}
	{ J.~Bourgain and A.~Bulut}, { {Almost sure global well posedness for the
			radial nonlinear Schr\"odinger equation on the unit ball I: the 2D case}},
{\em	Annales I. H. Poincare (C)}, 31, pp.~1267--1288, 2014.
	
	\bibitem[BB14b]{BouBul-14b}
	\leavevmode\vrule height 2pt depth -1.6pt width 23pt, { {Almost sure global
			well posedness for the radial nonlinear Schr\"odinger equation on the unit
			ball II: the 3D case}}, {\em Journal of the European Mathematical Society}, 16, pp.~1289--1325, 2014.
		
		\bibitem[BBS19]{BauBrySla-19}
		{ R.~Bauerschmidt, D.~Brydges, and G.~Slade}, { Introduction to a
			Renormalisation Group Method}, vol.~2242 of Lecture Notes in Mathematics,
		Springer Singapore, 2019.
		
		\bibitem[BBHLV99]{BayBlaiHolLalVau-99}
		{ G.~Baym, J.-P. Blaizot, M.~Holzmann, F.~Lalo{\"e}, and D.~Vautherin}, {
			The transition temperature of the dilute interacting {B}ose gas}, {\em Physical
		Review Letters}, 83, pp.~1703--1706, 1999.
		
		\bibitem[BBHLV01]{BayBlaiHolLalVau-01}
		\leavevmode\vrule height 2pt depth -1.6pt width 23pt, { {Bose-Einstein
				transition in a dilute interacting gas}}, {\em European Physical Journal B}, 24, p.~107, 2001.
	
	\bibitem[Bri22]{Bri}B. Bringmann. Invariant Gibbs measures for the three-dimensional wave equation with a Hartree nonlinearity I:
	Measures. {\em Stoch. PDE: Anal Comp}, 10, 1-89, 2022.

\bibitem[BC23]{BC23} B. Bringmann and S. Cao. A para-controlled approach to the stochastic Yang-Mills equation in two dimensions.
arXiv:2305.07197, May 2023.

\bibitem[BC24]{BC24} B. Bringmann and S. Cao. Global well-posedness of the stochastic Abelian-Higgs equations in two dimensions.
arXiv:2403.16878,  2024.

\bibitem[BC24a]{BC24a} B. Bringmann and S. Cao. Global well-posedness of the dynamical sine-Gordan model up to $6\pi$.
arXiv:2410.15493,  2024.

\bibitem[BG20]{BG18} N. Barashkov, M. Gubinelli,  A variational method for $\Phi^4_3$, {\em Duke Math. J.} 169  no. 17, 3339--3415, 2020.

\bibitem[BTT13]{BurThoTzv-13}
{ N.~Burq, L.~Thomann, and N.~Tzvetkov}, {\em Long time dynamics for the one
	dimensional non linear {S}chr\"{o}dinger equation}, Ann. Inst. Fourier
(Grenoble), 63, pp.~2137--2198, 2013.

\bibitem[C96]{Cardy-96}
{ J.~Cardy}, { Scaling and renormalization in statistical physics},
vol.~5 of Cambridge Lecture Notes in Physics, Cambridge University Press,
Cambridge, 1996.

\bibitem[CC18]{CC15} R. Catellier and K. Chouk, Paracontrolled distributions and the 3-dimensional stochastic
quantization equation. {\em Ann. Probab.}, 46(5):2621--2679, 2018.

\bibitem[CDS15]{CacSuz-14}
{ F.~Cacciafesta and A.-S. {de Suzzoni}}, { Invariant measure for the
	{S}chr\"odinger equation on the real line}, {\em J. Func Anal.}, 269,
pp.~271--324, 2015.



\bibitem[CCHS22]{CCHS20} A. Chandra, I. Chevyrev, M. Hairer, and H. Shen. Langevin dynamic for the 2D Yang-
	Mills measure. {\em Publ. Math.
Inst. Hautes \'Etudes Sci.}, 136:1--147, 2022.


\bibitem[CCHS24]{CCHS22} A. Chandra, I. Chevyrev, M. Hairer, and H. Shen. Stochastic quantisation of Yang-Mills-Higgs in 3D.
 {\em Invent. math.}, 237, 541--696, 2024.

\bibitem[CFW24]{CFW24}A. Chandra, G. de Lima Feltes, H. Weber, A priori bounds for 2-d generalised Parabolic Anderson Model, arXiv:2402.05544, 2024.


	
	\bibitem[CGW22]{CGW2020} A. Chandra, T. S. Gunaratnam, and H. Weber. Phase transitions for $\phi^4_3$, {\em Comm. Math. Phys.}, 392,  691-–782, 2022.

\bibitem[CH16]{CH16} A. Chandra and M. Hairer.
\newblock An analytic BPHZ theorem for regularity structures.
\newblock {\em arxiv:1612.08138}.

\bibitem[CHS18]{CHS18} A. Chandra, M. Hairer, and H. Shen. The dynamical sine-Gordon model in the full subcritical regime, August 2018.

\bibitem[CL14]{CarLie-14}
{ E.~A. Carlen and E.~H. Lieb}, { Remainder terms for some quantum
	entropy inequalities}, {\em J. Math. Phys.}, 55,  p.~042201, 2014.

	 \bibitem[DD03]{DD03}G. Da Prato, A. Debussche,  Strong solutions to the stochastic quantization equations. {\em Ann.
Probab.}, 31(4):1900–1916, 2003.

\bibitem[DNN25]{DeuNamNap-25} A. Deuchert, P. T. Nam, and M. Napiórkowski. The Gibbs State of the Mean-Field Bose Gas. Preprint 2025.

\bibitem[DNY19]{DNY19} Y. Deng, A. R. Nahmod, and H. Yue. Invariant Gibbs measures and global strong solutions for nonlinear Schr\"odinger
equations in dimension two. arXiv:1910.08492, October 2019.

\bibitem[DNY22]{DNY22} Y. Deng, A. R. Nahmod, and H. Yue. Random tensors, propagation of randomness, and nonlinear dispersive
equations. {\em Invent. Math.}, 228(2):539--686, 2022

\bibitem[Duc21]{Duc21} P. Duch. Flow equation approach to singular stochastic PDEs. arXiv:2109.11380, September 2021.

	\bibitem[DZ96]{DZ96} G. Da Prato and J. Zabczyk. Ergodicity for infinite-dimensional systems, volume 229 of London Mathematical Society Lecture Note Series. Cambridge University Press, Cambridge, 1996.
	
\bibitem[DGR24]{DGR24} P.Duch, M. Gubinelli, P. Rinaldi, Parabolic stochastic quantisation of the fractional $\Phi^4_3$ model
in the full subcritical regime. 2024.

\bibitem[DNN25]{DeuNamNap-25} A. Deuchert, P. T. Nam, and M. Napiórkowski. The Gibbs State of the Mean-Field Bose Gas. Preprint 2025.

\bibitem[Ein24]{Einstein-24} A. Einstein.  Quantentheorie des einatomigen idealen Gases. {\em K\"onigliche Preu$\beta$ische Akademie der Wissenschaften, Sitzungsberichte} (1924), 261--267.


 \bibitem[FKSS17]{FroKnoSchSoh-17}
{J.~Fr{\"o}hlich, A.~Knowles, B.~Schlein, and V.~Sohinger}, {Gibbs
	measures of nonlinear {S}chr{\"o}dinger equations as limits of many-body
	quantum states in dimensions {$d \leqslant 3$}}, {\em Commun. Math. Phys.}, 356, pp.~883--980, 2017.

\bibitem[FKSS19]{FroKnoSchSoh-19}
{J.~Fr\"{o}hlich, A.~Knowles, B.~Schlein, and V.~Sohinger}, {A
	microscopic derivation of time-dependent correlation functions of the 1 {$D$}
	cubic nonlinear {S}chr\"{o}dinger equation}, {\em Adv. Math.}, 353,
pp.~67--115, 2019.

	\bibitem[FKSS22]{FKSS22} J. Fr\"ohlich, A. Knowles, B. Schlein, and V. Sohinger, The mean-field limit of quantum Bose
	gases at positive temperature, {\em J. Amer. Math. Soc.}
	 35, no. 4, 955--1030, 2022.
	
	\bibitem[FKSS23]{FKSS23} J. Fr\"ohlich, A. Knowles, B. Schlein, and V. Sohinger, The Euclidean $\Phi^4_2$
	theory as a limit of an interacting Bose gas, To appear in {\em J. Eur. Math. Soc.}, arXiv:2201.07632., 2023.
	
	\bibitem[FSS76]{FSS76}J. Fr\"ohlich, B. Simon, and T. Spencer, Infrared bounds, phase transitions and continuous
	symmetry breaking, {\em Comm. Math. Phys.} 50,  no. 1, 79--95, 1976.
	
	\bibitem[GIP15]{GIP15} M. Gubinelli, P. Imkeller, N. Perkowski, Paracontrolled distributions and singular PDEs, {\em Forum Math. Pi} 3 no. 6, 2015.
	\bibitem[GH19]{GH18} M. Gubinelli, M. Hofmanov\'{a}, Global solutions to elliptic and parabolic $\psi^4$ models in Euclidean space. {\em Comm. Math. Phys.}, 368(3):1201-1266, 2019.
	
	
	
	\bibitem[GH21]{GH18a} M. Gubinelli, M. Hofmanov\'{a}, A PDE construction of the Euclidean $\Phi^4$ quantum field theory. {\em Comm. Math. Phys.}, 384(1):1--75, 2021.
	
	\bibitem[GJ87]{GliJaf-87}
	{ J.~Glimm and A.~Jaffe}, { Quantum Physics: A Functional Integral Point
		of View}, Springer-Verlag, 1987.
	
	\bibitem[GJS74]{GliJafSpe-74}
	{ J.~Glimm, A.~Jaffe, and T.~Spencer}, { The {W}ightman axioms and
		particle structure in the {$\mathscr{P}(\phi)_{2}$} quantum field model},
	{\em Ann. of Math.} (2), 100 (1974), pp.~585--632.

\bibitem[GM24]{GM24} M. Gubinelli and S.-J. Meyer. The FBSDE approach to sine-Gordon up to $6\pi$. arXiv:2401.13648, January 2024.


\bibitem[GP17]{GP17}M. Gubinelli, N. Perkowski, KPZ reloaded,  {\em Comm. Math. Phys.}, 349(1):165--269, 2017.

\bibitem[GRS75]{GueRosSim-75}
{ F.~Guerra, L.~Rosen, and B.~Simon}, { The {${\bf P}(\phi)_{2}$}
	{E}uclidean quantum field theory as classical statistical mechanics. {I},
	{II}}, {\em Ann. of Math.} (2), 101, pp.~111--189, 1975; ibid. (2) 101\,
191--259, 1975.

		\bibitem[Hai14]{Hai14} M. Hairer, A theory of regularity structures. {\em Invent. Math.} 198(2), 269--504, 2014.
		
		\bibitem[HB03]{HolBay-03}
		{ M.~Holzmann and G.~Baym}, { Condensate density and superfluid mass
			density of a dilute {Bose-Einstein} condensate near the condensation
			transition}, {\em Physical Review Letters}, 90, p.~040402, 2003.
		
		\bibitem[HIN17]{HIN17} M. Hoshino, Y. Inahama, N. Naganuma, 	Stochastic complex Ginzburg-Landau equation with
		space--time white noise, {\em Electron. J. Probab.} 22  no. 104, 1--68, 2017.
	
	
	
	\bibitem[Hos18]{Hos} M. Hoshino, Global well-posedness of complex Ginzburg--Landau equation with a space-time white noise, {\em Ann. Inst. Henri Poincar\'e Probab. Stat.} 54, no. 4, 1969--2001, 2018.
	
	\bibitem[HLS09]{HaiLewSol_2-09}
	C.~Hainzl, M.~Lewin, and J.~P. Solovej, { The thermodynamic
		limit of quantum {C}oulomb systems. {P}art {II}. {A}pplications}, {\em Advances in
	Math.}, 221, pp.~488--546, 2009.
	
	\bibitem[HM18a]{HM18a}	M. Hairer and K. Matetski. Discretisations of rough stochastic PDEs. {\em Ann. Probab.}, 46(3):1651--1709, 2018.
	
	\bibitem[HM18]{HM18} M. Hairer and J. Mattingly. The strong Feller property for singular stochastic PDEs. {\em Ann. Inst. Henri Poincar\'e Probab. Stat.}, 54(3):1314--1340, 2018.
	\bibitem[HS16]{HS16}  M. Hairer and H. Shen. The dynamical sine-Gordon model. {\em Comm. Math. Phys.}, 341(3):933–989, 2016.
	\bibitem[HS22]{HS22} M. Hairer and P. Sch\"onbauer. The support of singular stochastic PDEs. {\em Forum Math. Pi}, 10:No. e1, 127,
	2022.
	
		\bibitem[HZZZ24]{HZZZ24} Z. Hao, X. Zhang, R. Zhu, X. Zhu. Singular kinetic equations and applications. {\em Ann. Probab.}, Vol. 52, No. 2, 576--657, 2024.

		   \bibitem[JLM85]{JLM85} G. Jona-Lasinio and P. K. Mitter. On the stochastic quantization of field theory. {\em Comm.
Math. Phys.}, 101(3):409-436, 1985.

	   \bibitem[JP23]{JP21} A. Jagannath and N. Perkowski. A simple construction of the dynamical $\Phi^4_3$ model. {\em Trans. Amer. Math. Soc.} 376, no. 3, 1507–1522, 2023.

\bibitem[JR25]{JouRou-25} {L. Jouglas and N. Rougerie}, {\em in preparation}. 	   
	   
\bibitem[Kup16]{Kup16}  A. Kupiainen. Renormalization group and stochastic PDEs. {\em Ann. Henri Poincar\'e,} 17(3):497–535, 2016.

\bibitem[KPS01]{KasProSvi-01}
{ V.~A. Kashurnikov, N.~V. Prokof'ev, and B.~V. Svistunov}, { Critical
	temperature shift in weakly interacting {B}ose gas}, {\em Physical Review Letters},
87, p.~120402, 2001.

\bibitem[L11]{Lewin-11}
{ M.~Lewin}, { Geometric methods for nonlinear many-body quantum
	systems}, {\em J. Funct. Anal.}, 260, pp.~3535--3595, 2011.



\bibitem[LL01]{LieLos-01}
{E.~H. Lieb and M.~Loss}, { Analysis}, vol.~14 of Graduate Studies in
Mathematics, American Mathematical Society, Providence, RI, 2nd~ed., 2001.



\bibitem[LNR15]{LewNamRou-15}
M. Lewin, P. T. Nam, and N. Rougerie. {Derivation of
	nonlinear {G}ibbs measures from many-body quantum mechanics}, {\em J. {\'E}c.
polytech. Math.}, 2, pp.~65--115, 2015.

\bibitem[LNR18]{LewNamRou-18} M. Lewin, P. T. Nam, and N. Rougerie. {Gibbs measures based
	on {1D} (an)harmonic oscillators as mean-field limits}, {\em J. Math. Phys.}, 59, p.~041901, 2018.

	\bibitem[LNR21]{LewNamRou-21} M. Lewin, P. T. Nam, and N. Rougerie. Classical field theory limit of many-body quantum Gibbs states in 2D and 3D. {\em Invent. Math.} 224, 315--444, 2021
		
	\bibitem[LOT23]{LOT21} P. Linares, F. Otto, and M. Tempelmayr. The structure group for quasi-linear equations via universal enveloping
algebras. {\em Commun. Am. Math. Soc.} 3, 1–64, 2023.

\bibitem[LOTT24]{LOTT21} P. Linares, F. Otto, M. Tempelmayr, and P. Tsatsoulis. A diagram-free approach to the stochastic estimates in
regularity structures. {\em Invent. Math.} 237, no. 3, 1469--1565, 2024.
	
	
	\bibitem[LRS88]{LebRosSpe-88}
	{ J.~L. Lebowitz, H.~A. Rose, and E.~R. Speer}, { Statistical mechanics
		of the nonlinear {S}chr{\"o}dinger equation}, {\em J. Statist. Phys.}, 50,
	pp.~657--687, 1988.
	
	\bibitem[LSSY05]{LieSeiSolYng-05}
	{E.~H. Lieb, R.~Seiringer, J.~P. Solovej, and J.~Yngvason}, { The
		mathematics of the {B}ose gas and its condensation}, Oberwolfach {S}eminars,
	Birkh{\"a}user, 2005.
	
	\bibitem[MW17]{MW18} J.-C. Mourrat and H.Weber. The dynamic $\Phi^4_3$
	model comes down from infinity. {\em Comm.
		Math. Phys.}, 356(3):673--753, 2017.

	\bibitem[MW17a]{MW17} J.-C. Mourrat and H. Weber. Global well-posedness of the dynamic $\Phi^4$ model in the
	plane. {\em Ann. Probab.}, 45(4):2398--2476, 2017.


\bibitem[MW20]{MW20}A. Moinat and H. Weber. Space-time localisation for the dynamic $\Phi^4_3$
model. {\em Communications on Pure and Applied
Mathematics}, 73(12):2519--2555, 2020.

\bibitem[N73]{Nelson-73}
{ E.~Nelson}, { Construction of quantum fields from {M}arkoff fields}, {\em J.
Funct. Anal.}, 12, pp.~97--112, 1973.

\bibitem[OSSW18]{OSSW18} F. Otto, J. Sauer, S. Smith, and H. Weber. Parabolic equations with rough coefficients and singular forcing, March
2018.

\bibitem[OOT24]{OOT24}T. Oh, M. Okamoto, L. Tolomeo,  Focusing $\Phi^4_3$-model with a Hartree-type nonlinearity,  {\em Memoirs of the American Mathematical Society}
 Volume 304, Number 1529, 2024.

 \bibitem[OT18]{OhTho-18}
 { T.~Oh and L.~Thomann}, { A pedestrian approach to the invariant gibbs
 	measures for the 2-d defocusing nonlinear schr{\"o}dinger equations}, {\em Stoch
 PDE: Anal Comp.}, 6, pp.~397--445, 2018.

\bibitem[OW19]{OW19} F. Otto and H. Weber. Quasilinear SPDEs via rough paths. {\em Arch. Ration. Mech. Anal.}, 232(2):873–950, 2019.

		 \bibitem[PW81]{PW81} G. Parisi,  Y. S. Wu. Perturbation theory without gauge fixing. Sci. Sinica 24,
no. 4,  483–496, 1981.

\bibitem[RS72]{ReeSim1}
{ M.~Reed and B.~Simon}, { Methods of {M}odern {M}athematical {P}hysics.
	{I}. Functional analysis}, Academic Press, 1972.

\bibitem[RS16]{RouSer-16}
{ N.~{Rougerie} and S.~{Serfaty}}, { Higher dimensional coulomb gases and
	renormalized energy functionals}, {\em  Comm. Pure Appl. Math.}, 69 (2016),
pp.~519--605, 2016.


		\bibitem[RZZ17]{RZZ17} M. R\"ockner, R. Zhu, and X. Zhu. Restricted Markov uniqueness for the stochastic quantization of $P(\Phi)_2$ and its
		applications. {\em J. Funct. Anal.}, 272(10),  4263--4303, 2017.
		
		\bibitem[RZZ17a]{RocZhuZhu-16}
		{ M.~R\"ockner, R.~Zhu, and X.~Zhu}, { Ergodicity for the stochastic
			quantization problems on the {2D}-torus}, {\em Comm. Math. Phys.}, 352,
		pp.~1061--1090, 2017.
		
		\bibitem[R15]{Rougerie-LMU}
		{ N.~{Rougerie}}, { De {F}inetti theorems, mean-field limits and
			{B}ose-{E}instein condensation}, ArXiv e-prints,  (2015).

\bibitem[She21]{She21} H. Shen. Stochastic quantization of an Abelian gauge theory. {\em Comm. Math. Phys.}, 384(3):1445–1512, 2021.

\bibitem[S74]{Simon-74}
{ B.~Simon}, { The {$P(\phi )_{2}$} {E}uclidean (quantum) field theory},
Princeton University Press, Princeton, N.J., 1974.
\newblock Princeton Series in Physics.

\bibitem[S05]{Simon-05}
\leavevmode\vrule height 2pt depth -1.6pt width 23pt, { Functional
	integration and quantum physics}, AMS Chelsea Publishing, Providence, RI,
second~ed., 2005.

\bibitem[S66]{Symanzik-66}
{ K.~Symanzik}, { Euclidean quantum field theory. {I}. {E}quations for a
	scalar model}, {\em J. Mathematical Phys.}, 7,  (1966), pp.~510--525, 1966.
	
	\bibitem[SSZZ22]{SSZZ20} H. Shen, S. Smith, R. Zhu, and X. Zhu. Large $N$ limit of the $O(N)$ linear sigma model via stochastic
	quantization. {\em Ann. Probab.} 50(1), 131--202, 2022.
	
	\bibitem[SZZ22]{SZZ21} H. Shen, R. Zhu, and X. Zhu. Large $N$ limit of the $O(N)$ linear sigma model in 3D. {\em Comm.
		Math. Phys.}, 394 no.3, 953--1009. 2022.
		
			\bibitem[SZZ25]{SZZ23} H. Shen, R. Zhu, and X. Zhu. Large $N$ limit and $1/N$ expansion of invariant observables in $O(N)$ linear
		$\sigma$-model via SPDE. {\em Probability Theory and Related Fields}, 191:853–932, 2025.
		
	\bibitem[Sta13]{Sta13} H. R. Stahl. Proof of the BMV conjecture. {\em Acta Math.} 211 (2), 255--290, 2013.
		

	\bibitem[Tri78]{Tri78} H. Triebel. Interpolation theory, function spaces, differential operators, volume 18 of North-Holland Mathematical
	Library. North-Holland Publishing Co., Amsterdam-New York, 1978.
	
	\bibitem[TT10]{ThoTzv-10}
	{\sc L.~Thomann and N.~Tzvetkov}, {Gibbs measure for the periodic
		derivative nonlinear schr{\"o}dinger equation}, {\em Nonlinearity}, 23,
	p.~2771, 2010.
	
	\bibitem[TW18]{TsaWeb-18}
	{ P.~Tsatsoulis and H.~Weber}, { Spectral gap for the stochastic
		quantization equation on the 2-dimensional torus}, {\em Ann. Inst. Henri
	Poincar\'{e} Probab. Stat.}, 54, pp.~1204--1249, 2018.
	
	\bibitem[T08]{Tzvetkov-08}
	{ N.~Tzvetkov}, { Invariant measures for the defocusing nonlinear
		{S}chr{\"o}dinger equation}, {\em Ann. Inst. Fourier (Grenoble)}, 58,
	pp.~2543--2604, 2008.
	
		\bibitem[ZJ89]{ZinnJustin-89}
	{ J.~Zinn-Justin}, { Quantum Field Theory and Critical Phenomena}, Oxford
	University Press, 1989.
	
	\bibitem[ZJ13]{ZinnJustin-13}
	\leavevmode\vrule height 2pt depth -1.6pt width 23pt, { Phase transitions
		and renormalization group}, Oxford Graduate Texts, Oxford University Press,
	Oxford, 2013.
	\newblock Paperback edition of the 2007 original [MR2345069].

	\bibitem[ZZ15]{ZZ15} Rongchan Zhu, Xiangchan Zhu, Three-dimensional Navier-Stokes equations driven by space-time white noise, {\em Journal of Differential Equations}
	, 259,  9, 5,  2015,  4443-4508
	
	\bibitem[ZZ18]{ZZ18} R. Zhu and X. Zhu. Lattice approximation to the dynamical $\Phi^4_3$
	model. {\em Ann. Probab.}, 46(1):397--455, 2018.
	
	\bibitem[ZZZ22]{ZZZ20}X. Zhang, R. Zhu, X. Zhu, Singular HJB equations with applications to KPZ on the real line, {\em Probab. Theory Relat. Fields}  183 no. 3-4, 789--869, 2022
	

\end{thebibliography}
\end{document}